\pdfoutput=1
\documentclass[11pt,notitlepage]{article}
\usepackage[utf8]{inputenc} 
\usepackage{geometry}
\usepackage{verbatim}
\usepackage{float}
\usepackage{amsmath}
\usepackage{amsfonts} 
\usepackage{setspace}
\usepackage[authoryear]{natbib}
\usepackage{hyperref}
\usepackage{etoolbox}
\makeatletter
\patchcmd{\NAT@citex}
  {\@citea\NAT@hyper@{\NAT@nmfmt{\NAT@nm}\hyper@natlinkbreak{\NAT@aysep\NAT@spacechar}{\@citeb\@extra@b@citeb}\NAT@date}}
  {\@citea\NAT@nmfmt{\NAT@nm}\NAT@aysep\NAT@spacechar\NAT@hyper@{\NAT@date}}
  {}{\PackageWarning{natbib}{year-only citation link patch 1 failed}}
\patchcmd{\NAT@citex}
  {\@citea\NAT@hyper@{\NAT@nmfmt{\NAT@nm}\hyper@natlinkbreak{\NAT@spacechar\NAT@@open\if*#1*\else#1\NAT@spacechar\fi}{\@citeb\@extra@b@citeb}\NAT@date}}
  {\@citea\NAT@nmfmt{\NAT@nm}\NAT@spacechar\NAT@@open\if*#1*\else#1\NAT@spacechar\fi\NAT@hyper@{\NAT@date}}
  {}{\PackageWarning{natbib}{year-only citation link patch 2 failed}}
\makeatother
\usepackage{xurl}
\usepackage{graphicx}
\usepackage{subcaption}
\usepackage{booktabs}
\usepackage{makecell}
\usepackage{changepage}
\usepackage{tikz}
\usetikzlibrary{positioning, calc, arrows.meta}
\usepackage{threeparttable}
\usepackage{adjustbox}
\usepackage{rotating}
\usepackage{pifont}
\newcommand{\cmark}{\ding{51}}
\usepackage{enumitem}
\usepackage{amsthm}

\theoremstyle{definition}
\newtheorem{prediction}{Prediction}

\usepackage{caption}
\usepackage{amssymb}

\usepackage{multicol}
\usepackage{pdflscape}
\usepackage{fancyhdr}
\usepackage{amsthm}

\newtheorem{lemma}{Lemma}
\newcommand\averagereactions {70.92}
\newcommand\feOneBlueCommentsMeanPct {0.8}
\newcommand\feOneBlueNegativeMeanPct {56}
\newcommand\feOneBlueOffensiveMeanPct {35}
\newcommand\feOneBlueReactionsMeanPct {10}
\newcommand\feOneBlueSupportiveReactionsMeanPct {9.5}
\newcommand\feOneCommentsApproxN {1,500}

\newcommand\feOneCommentsBlueVsRedCROnePText {<0.01}

\newcommand\feOneConservativeBlueVsRedCROnePText {<0.01}

\newcommand\feOneInformativeMaxMeanPct {11}
\newcommand\feOneInformativeMinPairwiseCROnePText {>0.10}

\newcommand\feOneNegativeBlueVsRedCROnePText {<0.01}
\newcommand\feOneOffensiveBlueContrastsCROnePText {<0.01}

\newcommand\feOnePostsN {150}
\newcommand\feOneReachApproxN {135,000}
\newcommand\feOneReactionsApproxN {12,000}

\newcommand\feOneReactionsBlueVsRedCROnePText {<0.01}

\newcommand\feOneRedCommentsMeanPct {1.3}
\newcommand\feOneRedSupportiveReactionsMeanPct {7}
\newcommand\feOneSharesApproxN {650}
\newcommand\feOneSupportiveCommentsMeanPct {0.3}
\newcommand\feOneSupportiveCommentsMinPairwiseCROnePText {>0.20}

\newcommand\feOneSupportiveReactionsBlueVsRedCROnePText {<0.01}
\newcommand\feOneUniqueLinkClicksApproxN {1,750}
\newcommand\feTwoAdCPMMeanUSD {9.6}
\newcommand\feTwoAdFrequencyMean {1.12}
\newcommand\feTwoAdReachApproxMeanN {14,600}
\newcommand\feTwoAdSpendMeanUSD {150}
\newcommand\feTwoAdsN {72}
\newcommand\feTwoAdsPerArmN {18}
\newcommand\feTwoAdsPerPostN {6}

\newcommand\feTwoAllAnyCROnePText {<0.01}

\newcommand\feTwoAllAnyEffectPP {0.065}
\newcommand\feTwoAllAnyRelativePct {13.5}
\newcommand\feTwoAllAnyRelativeWholePct {13}

\newcommand\feTwoAllControlMeanPct {0.485}
\newcommand\feTwoAllMeanPct {0.53}

\newcommand\feTwoAllMixedCROnePText {<0.01}

\newcommand\feTwoAllMixedEffectPP {0.062}
\newcommand\feTwoAllMixedRelativePct {12.7}

\newcommand\feTwoAllOpposeCROnePText {<0.01}

\newcommand\feTwoAllOpposeEffectPP {0.087}
\newcommand\feTwoAllOpposeRelativePct {18.0}

\newcommand\feTwoAllOpposeSupportEffectRatio {2}

\newcommand\feTwoAllOpposeVsSupportCROneP {0.059}

\newcommand\feTwoAllSupportCROnePText {<0.01}

\newcommand\feTwoAllSupportEffectPP {0.047}
\newcommand\feTwoAllSupportRelativePct {9.7}

\newcommand\feTwoArmMeanSharePct {25}

\newcommand\feTwoBlueAllAnyEffectPP {0.031}

\newcommand\feTwoBlueAllAnyRelativeWholePct {6}

\newcommand\feTwoBlueAllControlCoarseMeanPct {0.53}

\newcommand\feTwoBlueClustersN {24}

\newcommand\feTwoBlueInteractionsOpposeVsSupportCROnePText {<0.05}

\newcommand\feTwoClickCostOpposeDeclinePct {15}

\newcommand\feTwoClicksAnyCROnePText {<0.05}

\newcommand\feTwoClicksAnyEffectPP {0.016}

\newcommand\feTwoClicksAnyRelativeWholePct {7}

\newcommand\feTwoClicksControlMeanPct {0.228}
\newcommand\feTwoClicksMeanPct {0.24}

\newcommand\feTwoClicksMixedEffectPP {0.012}

\newcommand\feTwoClicksOpposeCROnePText {<0.01}

\newcommand\feTwoClicksOpposeEffectPP {0.034}
\newcommand\feTwoClicksOpposeRelativePct {14.8}
\newcommand\feTwoClicksOpposeRelativeWholePct {15}

\newcommand\feTwoClicksOpposeVsSupportCROnePText {<0.01}

\newcommand\feTwoClicksSupportEffectPP {0.003}

\newcommand\feTwoDisplayedToxicityMean {0.37}

\newcommand\feTwoExpansionsAnyCROnePText {<0.01}
\newcommand\feTwoExpansionsAnyEffectCoarsePP {0.05}
\newcommand\feTwoExpansionsAnyEffectPP {0.049}
\newcommand\feTwoExpansionsAnyRelativePct {19.2}

\newcommand\feTwoExpansionsControlCoarseMeanPct {0.25}
\newcommand\feTwoExpansionsControlMeanPct {0.254}
\newcommand\feTwoExpansionsMeanPct {0.29}
\newcommand\feTwoExpansionsMinPairwiseCROnePText {>0.90}

\newcommand\feTwoExpansionsMixedCROnePText {<0.01}

\newcommand\feTwoExpansionsMixedEffectPP {0.049}

\newcommand\feTwoExpansionsOpposeCROnePText {<0.01}

\newcommand\feTwoExpansionsOpposeEffectPP {0.048}

\newcommand\feTwoExpansionsSupportCROnePText {<0.01}

\newcommand\feTwoExpansionsSupportEffectPP {0.049}

\newcommand\feTwoFemaleSharePct {47.7}

\newcommand\feTwoInteractionsAnyEffectPP {0.003}

\newcommand\feTwoInteractionsClicksOpposeVsSupportCROnePText {<0.01}

\newcommand\feTwoInteractionsControlMeanPct {0.020}
\newcommand\feTwoInteractionsMeanPct {0.02}

\newcommand\feTwoInteractionsMixedContrastsCROnePText {<0.10}

\newcommand\feTwoInteractionsMixedEffectPP {0.002}

\newcommand\feTwoInteractionsOpposeCROnePText {<0.05}

\newcommand\feTwoInteractionsOpposeEffectPP {0.009}

\newcommand\feTwoInteractionsOpposeRelativeWholePct {43}

\newcommand\feTwoInteractionsOpposeVsSupportCROnePText {<0.01}

\newcommand\feTwoInteractionsSupportEffectPP {-0.003}

\newcommand\feTwoMaleSharePct {52.3}

\newcommand\feTwoPageViewsAnyCROnePText {<0.10}

\newcommand\feTwoPageViewsAnyEffectPP {0.016}
\newcommand\feTwoPageViewsAnyRelativePct {9.3}

\newcommand\feTwoPageViewsControlMeanPct {0.171}
\newcommand\feTwoPageViewsMeanPct {0.18}

\newcommand\feTwoPageViewsMixedEffectPP {0.017}

\newcommand\feTwoPageViewsOpposeCROnePText {<0.10}

\newcommand\feTwoPageViewsOpposeEffectPP {0.019}
\newcommand\feTwoPageViewsOpposeRelativePct {11.2}

\newcommand\feTwoPageViewsSupportEffectPP {0.011}

\newcommand\feTwoPostReachLowerBoundN {80,000}

\newcommand\feTwoReachApproxMillions {1}
\newcommand\feTwoReachMillions {1.05}
\newcommand\feTwoReachN {1,054,015}

\newcommand\feTwoRedAllAnyCROnePText {<0.01}

\newcommand\feTwoRedAllAnyEffectPP {0.096}

\newcommand\feTwoRedAllAnyRelativeWholePct {21}

\newcommand\feTwoRedAllControlCoarseMeanPct {0.45}

\newcommand\feTwoRedAllOpposeCROnePText {<0.01}
\newcommand\feTwoRedAllOpposeEffectCoarsePP {0.13}

\newcommand\feTwoRedClicksOpposeCROnePText {<0.05}
\newcommand\feTwoRedClicksOpposeContrastsCROnePText {<0.05}

\newcommand\feTwoRedClicksOpposeEffectPP {0.045}

\newcommand\feTwoRedClustersN {24}

\newcommand\feTwoRedInteractionsOpposeCROnePText {<0.01}
\newcommand\feTwoRedInteractionsOpposeContrastsCROnePText {<0.01}

\newcommand\feTwoRedInteractionsOpposeEffectPP {0.015}
\newcommand\feTwoRedInteractionsOpposeLevelRatio {2}

\newcommand\feTwoSelectedPostsN {12}
\newcommand\feTwoSeniorSharePct {7.5}
\newcommand\feTwoStrataN {18}

\newcommand\feTwoSwingAllAnyCROnePText {<0.01}

\newcommand\feTwoSwingAllAnyEffectPP {0.064}

\newcommand\feTwoSwingAllAnyRelativeWholePct {13}

\newcommand\feTwoSwingAllControlCoarseMeanPct {0.48}

\newcommand\feTwoSwingAllMaxEffectPP {0.07}
\newcommand\feTwoSwingAllMinEffectPP {0.05}

\newcommand\feTwoSwingClicksOpposeCROnePText {<0.05}

\newcommand\feTwoSwingClicksOpposeEffectPP {0.037}

\newcommand\feTwoSwingClustersN {24}
\newcommand\feTwoToxicityRestrictedAdsN {60}
\newcommand\feTwoYoungAdultSharePct {57.3}
\newcommand\feTwoYoungestSharePct {8.7}

\newcommand\attritionDropAttentionFail {3,896}
\newcommand\attritionDropInvalidConditionsPREATTRITION {5,077}

\newcommand\attritionFinalN {3,868}

\newcommand\balanceNControl {1,189}
\newcommand\balanceNOppose {1,392}
\newcommand\balanceNSupport {1,287}

\newcommand\responseRateControl {0.997}
\newcommand\responseRateFtestP {0.134}
\newcommand\responseRateOppose {0.990}
\newcommand\responseRateSupport {0.992}

\newcommand\tabArmCEightaOnePostExposureAttitudesPairwiseP {0.014}

\newcommand\surveyAttitudeOpposeSD {0.12}
\newcommand\surveyCognitiveDemocratSD {0.67}
\newcommand\surveyCognitiveIndependentSD {0.38}

\newcommand\surveyDonationAmountPct {10}
\newcommand\surveyDonationAmountUSD {2.6}
\newcommand\surveyDonationControlMean {0.73}
\newcommand\surveyDonationDemocratPctDecline {6.4}
\newcommand\surveyDonationOpposePP {5.3}
\newcommand\surveyDonationOpposePctDecline {7.2}
\newcommand\surveyDonationRepublicanPctDecline {12.6}
\newcommand\surveyEmotionalIndependentSD {0.4}
\newcommand\surveyEmotionalRepublicanSD {0.35}
\newcommand\surveyExpectCommentersConservative {17.3}

\newcommand\surveyExpectCommentersProgressive {47.6}
\newcommand\surveyNewsletterControlMean {0.19}

\newcommand\surveyReadCommentsSometimesPct {86}

\newcommand\surveyTimeSpentMaxPct {43}
\newcommand\surveyTimeSpentMaxSec {27}
\newcommand\surveyTimeSpentMinPct {39}
\newcommand\surveyTimeSpentMinSec {25}

\begin{document}

\title{\Huge \vspace{-1cm} The Influence of the Vocal Few:\\Evidence from Social Media Comments}

\author{\Large Dante Donati and Lena Song\thanks{Donati: Columbia Business School. dd3137@gsb.columbia.edu. Song: University
of Illinois Urbana-Champaign. lenasong@illinois.edu. We thank Charles Amuzie from Color of Change for support. We thank George Beknazar-Yuzbashev, Luca Braghieri, Yiting Deng, Sarah Eichmeyer, Ruben Enikolopov, Rafael Jim\'enez-Dur\'an, Gita Johar, Maria Petrova, Andrea Prat, Carlo Schwarz, Andrey Simonov, and Yanwen Wang, and  seminar and conference participants at the Virtual Quantitative Marketing Seminar, University of Chicago, University of Alberta, University of California Berkeley, Bass FORMS Conference, MIT Conference on Digital Experimentation, Berlin School of Economics, Carnegie Mellon University, University of Illinois Urbana-Champaign, Marketing Science, Columbia University, New York University, Workshop on Platform Analytics, AOM, Early Career Behavioral Economics Conference, and CESifo Venice Summer Institute: Workshop on Digital Platforms for helpful feedback. We are grateful to Anna Bezhanishvili, Seungwoo Kim, Jungyun Kim, Thomas Lilly, Daniel Merlau, and Navtej Singh for excellent research assistance. The research was approved by the Institutional Review Boards at Columbia University (AAAU8166). This study was registered in the American Economic Association Registry for randomized controlled trials under trial numbers AEARCTR-0013812 and AEARCTR-0017850 \citep{donati2024rct13812,donati2026rct17850}. We acknowledge financial support from the Russell Sage Foundation (Grant G-2309-44994), the Digital Future Initiative and the Bernstein Center at Columbia Business School, and the Provost Office at Columbia University. The authors declare no conflict of interest. All errors are the authors' own.}} 

\date{September 2026} 
\maketitle

\vspace{-0.5cm}\begin{abstract}

\singlespacing
\vspace{-0.4cm}
Online comment sections let a small number of vocal individuals reach far beyond their own networks. We conduct a large-scale field experiment on Facebook that randomizes the presence and stance of comments beneath posts for a racial justice organization, reaching around one million U.S. users. Opposing comments increase reactions, comments, and link clicks by \feTwoClicksOpposeRelativeWholePct{}--\feTwoInteractionsOpposeRelativeWholePct{} percent relative to no comments, whereas supportive comments have little effect. A complementary survey experiment shows that similar opposing comments make attitudes less progressive and reduce donations to the organization. Through a common feature of online platforms, the vocal few can exert outsized influence.

\end{abstract}\par
\hspace{0.35cm}\makebox[\textwidth][l]{\small \textbf{Keywords}: 
Social Media, Comments, User-Generated Content, Social Influence, Field Experiments
}\par
\hspace{0.35cm}\makebox[\textwidth][l]{\small  \textbf{JEL codes}: C93, D72, D83, D91, J15, L82, L86}

\medskip{}

\newpage{}

\section{Introduction}

Digital technologies have changed the structure and scale of social influence. Online platforms, such as social media, have extended social influence well beyond offline networks and expanded people's exposure to others' opinions.  Unlike traditional media with primarily one-way broadcasting,  a defining feature of social media is horizontal communication that allows users to share information and opinions among peers  \citep{zhuravskaya2020political}.

Comment sections on social media provide a natural setting to study horizontal communication. Comments are publicly visible and widely read: in our survey of U.S. adults, \surveyReadCommentsSometimesPct{} percent of respondents report sometimes, often, or very often reading comments on social media, and more than 50 percent report doing so often or very often. At the same time, comments are often generated by a small and potentially unrepresentative subset of users \citep{kim2025disproportionate}, mirroring the concentration of content production in the hands of a small minority of accounts \citep{grinberg2019fake}. As comments compete for attention with the post, they can shape how others receive it, even when those who comment are not representative. This effect may be further amplified by engagement-maximizing algorithms that make posts with comments more visible.
Understanding the causal effects of comments, and of the views they express, is therefore central to evaluating how social media platforms shape discourse and how they should be governed.

This paper presents a conceptual framework motivated by salience theory (e.g., \citealt{bordalo2022salience}) and studies how comments shape other users' subsequent on-platform engagement and off-platform attitudes and behavior in two complementary pre-registered experiments. The framework predicts that the presence of comments raises attention to the post regardless of stance, but that only opposing comments, being the more salient, raise subsequent engagement.
In collaboration with Color of Change, the largest online racial justice organization in the United States, we ran a large-scale field experiment in which approximately \feTwoReachApproxMillions{} million Facebook users were randomly assigned to see posts about racial justice that had identical content but differed in the presence and stance of the comment section. To examine downstream attitudinal and behavioral effects, we complement the on-platform experiment with a survey experiment on Prolific.

We show that comments causally affect user engagement, attitudes, and behavior, with effects varying by stance. The presence of comments increases attention to posts, but opposing comments further increase on-platform engagement while shifting attitudes in a less progressive direction and reducing donations to a progressive cause. These results suggest that the comment section is an integral component of platform design that can shape outcomes well beyond the platform itself, creating a tension between online engagement and broader offline social objectives, a trade-off that engagement-based ranking algorithms may further amplify \citep{acemoglu2024model}.

To establish these findings, we begin by describing how users engage with racial justice content and analyzing the discussions that emerge in the comment sections. Racial justice provides a useful setting because it is a domain of intense public debate and highly polarized views \citep{pew_racial_attitudes_2024}, allowing us to study cross-cutting exposure in a context where the stakes are socially important. To collect comments and reactions---Facebook's one-click responses to a post, such as \emph{like}, \emph{love}, and \emph{angry}---we advertised posts to approximately \feOneReachApproxN{} Facebook users covering several dimensions of racial justice, including voting rights, environmental justice, and criminal justice. We show that engagement patterns differ sharply across areas with different ideological compositions. Users in conservative areas commented at higher rates  than those in progressive areas, yet they reacted less frequently. Moreover, comments were overwhelmingly negative and more likely to be offensive in conservative areas.
These findings suggest that engagement does not always imply endorsement, especially in ideologically opposed communities. This supports concerns that the voices visible on social media are more extreme than the audience they reach \citep{bail2021breaking}.

We then leverage these comments from real user discussions to estimate their causal impact on subsequent engagement.
Isolating the effects of comments from those of the original posts is empirically challenging: posts that receive many early comments may have inherent characteristics that make them more engaging ex ante, and platform algorithms tend to recommend posts with higher early activity, further increasing their visibility and subsequent engagement.
To address this challenge, we build on existing platform features to design a pipeline that manipulates comment visibility and stance as a novel research instrument. Using this pipeline, we provide causal evidence that the presence and stance of pre-existing comments  influence the behavior of other users.

We implement this pipeline in a large-scale field experiment on Facebook. To isolate the effect of comment stance from that of the post itself, we \emph{re‑marketed} a subset of the posts from the previous phase of the study---that is, we ran a second paid campaign delivering the same posts, each now pre‑populated with two comments, to an audience that had not seen them before. This second campaign reached roughly \feTwoReachApproxMillions{} million new users combined into \feTwoStrataN{} clusters of ZIP codes, across areas with different ideological compositions. Using Meta’s A/B testing infrastructure, in each cluster, we randomized participants into four conditions: (1) no comments (control), (2) supportive comments only, (3) opposing comments only, and (4) a mixed condition with one supportive and one opposing comment. We measured actions to expand the post and comment section,  user interactions with the post (reactions, comments, and shares), and direct traffic to the organization's website.

We made several design choices to address potential concerns. To minimize violations of the Stable Unit Treatment Value Assumption (SUTVA), a real‑time filtering pipeline hid all new comments, minimizing the influence that users exposed to the same post could have on one another. To limit differences in other visible interactions, we equalized the number of shares and balanced the number of reactions across conditions.  To address the risk of divergent delivery—the tendency of ad algorithms to learn and serve treatment arms to different user types based on early engagement \citep{braun2025b, eckles2018field}—we followed and augmented best practices from recent work \citep{burtch2025characterizing}. Specifically, we split budgets evenly across arms, launched all ads simultaneously, imposed a one-impression cap per user, and optimized for reach rather than engagement. In addition, we ran the campaign over several weeks to achieve audience saturation for each ad, ensuring that nearly all users within a defined area were reached, to further limit algorithmic learning and divergence. We verified balance across gender, age, and delivery metrics such as impressions and costs.

We show that comment sections significantly influence subsequent user engagement: the  opinions of a small, vocal minority can shape the behavior of a larger audience. Comparing the treatment conditions to the control condition in which posts were displayed without any comments, we find that displaying any pre-populated comments increases all subsequent engagement with the post by \feTwoAllAnyEffectPP{} percentage points, a \feTwoAllAnyRelativeWholePct{} percent rise relative to the baseline ($p\feTwoAllAnyCROnePText{}$). Further dissecting by types of engagement, we show that the presence of comments increases the likelihood of users clicking to expand the post by about \feTwoExpansionsAnyEffectCoarsePP{} percentage points on a \feTwoExpansionsControlCoarseMeanPct{} percent baseline. Consistent with the fact that comment stance is only revealed after a user clicks to expand, there are no significant differences between supportive, opposing, or mixed conditions. However, comment stance has pronounced effects on downstream on-platform engagement: ads with opposing comments generate significantly more interactions and link clicks than those with supportive comments ($p\feTwoInteractionsClicksOpposeVsSupportCROnePText{}$). Relative to the control, opposing stances increase comments, reactions, and shares by roughly \feTwoInteractionsOpposeRelativeWholePct{} percent ($p\feTwoInteractionsOpposeCROnePText{}$) and raise click-through rates by \feTwoClicksOpposeEffectPP{} percentage points (a \feTwoClicksOpposeRelativeWholePct{} percent increase, $p\feTwoClicksOpposeCROnePText{}$; advertising costs per click and per interaction fall proportionately).

These effects vary with the political leaning of the area. The effect of the comment section is largest in conservative areas, where the baseline engagement with the posts is lowest. Stance also matters in these areas: opposing comments substantially increase interactions and click-through rates, whereas the corresponding estimates in progressive areas are small and imprecise. This suggests that identity congruence may affect the impact of comments, as an opposing comment contrasts with the post for every viewer but aligns with a conservative viewer's own stance, although with ad-level data we cannot separate this channel from differences in attention. The influence of comments, therefore, depends on both their stance and the composition of the audience.

We make several design choices in our field experiment to preserve ecological validity. First, we examine the effects of comment sections on a large, diverse population of social media users across a wide range of ZIP codes. Second, users were unaware they were part of an experiment and interacted with the posts as they naturally would. Third, we used comments from real users, ensuring that the content shown to new users reflected genuine Facebook user responses rather than researcher‑generated or AI‑generated text. Together, these choices allow us to directly study real-world online discourse on a socially important and divisive issue.

At the same time, on-platform engagement metrics capture only part of the picture. To assess how comment sections shape beliefs, attitudes, and costly off-platform behavior, we conducted a complementary survey experiment on Prolific, mirroring the design, creatives, and organic comments used in the Facebook experiment. We find that both supportive and opposing comments increase the time spent viewing the posts by up to \surveyTimeSpentMaxPct{} percent relative to control. However, opposing comments shift perceptions of the organization and related issues in a less progressive direction, and reduce incentivized donations by \surveyDonationOpposePctDecline{} percent, with no significant differences between Republicans and Democrats. In contrast, opposing comments make Democrats angrier, more annoyed, and less curious and reflective about the post and its comments than they do Republicans. Thus, although opposing comments amplify engagement on the platform and traffic to the website, they weaken support for the organization and its message off the platform. Taken together, the evidence highlights a tension between short-run engagement gains and downstream attitudinal and behavioral consequences. A back-of-the-envelope fundraising analysis further suggests that the net benefit of tolerating opposing comments depends on the quality of the additional traffic they generate.

The results have several implications for content producers, platforms, and policymakers. For advertisers and campaign managers, opposing comments raise engagement, but stricter moderation may better protect the organization's reputation and downstream support. Importantly, the objectives of platforms that maximize engagement may not align with those of firms, nonprofits, or political actors that care about their legitimacy or about policy support. Regulations such as the European Union's Digital Services Act and the United Kingdom's Online Safety Act already place growing responsibility on platforms for the content they host, including comments as well as posts.\footnote{Digital Services Act (2022, EU): \url{https://eur-lex.europa.eu/eli/reg/2022/2065/oj/eng}; Online Safety Act (2023, UK): \url{https://www.legislation.gov.uk/ukpga/2023/50/contents}.} Our results support including comment sections in these regulations. Without moderation, comment sections from the vocal few can disproportionately shape the attitudes and behavior of the broader audience.

Our paper makes several contributions. First, we show that the opinions of a small and unrepresentative set of individuals shape the on-platform behavior and the off-platform attitudes and behavior of a much larger audience. Our paper builds upon and expands the literature on the drivers and consequences of user-generated content (UGC), which has largely studied product reviews.\footnote{%
See \citet{luca2015user} for a review.} Social media comments represent a distinct and understudied form of UGC. They can express opinions on important issues, with societal rather than commercial consequences.
A related literature shows that visible social signals change how people respond to content. Cues showing which friends liked an ad raise engagement \citep{bakshy2012social, tucker2014social}.
Our paper shows that even when comments come from selected and unknown users rather than friends or influencers, they can influence the many people who read the post, including those who never comment.

Second, our findings speak to an emerging literature on the divergence between engagement and other outcomes on digital platforms. \cite{beknazar2024model} show theoretically that platforms may have incentives to promote harmful yet engaging content, and \cite{beknazar2025toxic} provide supporting experimental evidence. Relatedly, \cite{germano2026ranking} show that algorithms placing greater weight on social interaction signals can increase engagement while also increasing misinformation and polarization.
We document a trade-off between on-platform engagement and off-platform outcomes that operates through the stance of the comment sections. 

Third, we contribute to the growing literature on the economics of social media \citep{zhuravskaya2020political, aridor2024economics}. A recent wave of field experiments has studied how platform design choices shape user behavior and attitudes, focusing on algorithmic curation of the news feed \citep{levy2021social, nyhan2023like, guess2023social, guess2023reshares, gauthier2026political}. These studies have advanced our understanding of the role of algorithms, but less is known about how the social context surrounding a post shapes user responses. The comment section is also a different lever from the feed: feed ranking is set by the platform, while the composition of a comment section can be managed by whoever posts, which makes it available to any firm, nonprofit, or campaign. 

Finally, we contribute a methodology for experimentally varying the social context around a post in the field. As manipulating comments is challenging without platform cooperation or the installation of software such as browser extensions, past work has largely used them as an outcome rather than a treatment (e.g., \citealt{moehring2024personalized}). We address this challenge by developing an experimental pipeline that leverages platform features to randomize comment visibility and stance. The infrastructure outlined in the paper can be adapted to study the impact of comments in other settings.

The paper is organized as follows: Section 2 introduces a conceptual framework of the influence of the comment section, Section 3 presents the study setting, Section 4 describes the generation and analysis of engagement from real users, Section 5 presents causal evidence on the impact of the comment section on online engagement, Section 6 shows its effects on offline attitudes and behavior, and Section 7 concludes.

\section{Conceptual Framework\label{sec:conceptual_framework}}
Comment sections are ubiquitous online. On Facebook and YouTube, the two most widely used platforms in the United States, they appear directly beneath the post. On Reddit, they are the primary mode of discussion. Beyond social media, online news outlets such as The New York Times, The Wall Street Journal, and CNN maintain dedicated comment sections at the end of articles. This proximity means that comments are often read alongside the post or article itself. Such discussions can expose readers to a range of viewpoints and, in principle, support productive conversation on divisive issues. They can also devolve into dismissive or even hostile exchanges. In either case, the exchange is not confined to its participants. Comments are delivered to the same audience as the post, including many who read them without responding.

Those who comment are not drawn at random from the audience. To conceptualize the influence of a comment section dominated by the vocal few, we use a stylized framework motivated by salience theory \citep{bordalo2012salience,bordalo2013salience,bordalo2018diagnostic,bordalo2022salience}. The salience of a comment section is captured by three features: its prominence before any of its content is visible, how its content contrasts with the post, and how surprising that content is given what the viewer anticipated. Matching our empirical setting in the field experiment, we model a viewer who first decides whether to open the comment section and, if she does, observes the comments alongside the post. She then decides whether to engage, and her attitudes update in response to what she has seen. In line with the experimental design, the framework takes the composition of the comment section as given and does not model who chooses to comment. Appendix~\ref{app:framework} provides details of the setup and derivations. There are two main predictions:

\begin{prediction}[Presence and stance act on different margins of engagement]
\label{pred:margins}
The comment counter makes the post more prominent, and because stance is not revealed until the post is expanded, the presence of comments raises post expansions regardless of stance. Once the comment section is opened, the opposing section is more salient, so it increases interactions and link clicks relative both to a post without comments and to one with a supportive section.
\end{prediction}

\begin{prediction}[Opposing comments move attitudes away from the post]
\label{pred:attitudes}
Because comments shift attention away from the post and toward the comment section, an opposing section can pull attitudes toward its own position and against that of the post and, if willingness to donate is increasing in support for the organization, lower donations. A section that repeats the post's position has little effect, since averaging the post with a section that restates it leaves the position unchanged. 
\end{prediction}

The two predictions highlight a trade-off between engagement and downstream objectives. Greater salience of the comment section raises the total attention people pay to the post and its comments, which increases the probability of engagement. At the same time, salience shifts the allocation of that attention from the post toward the comments, which increases the weight placed on the comments relative to the post when attitudes are formed. When the comments oppose the post, the first effect produces more interactions and clicks, while the second pulls attitudes away from the post. Supportive comments, which are less salient and restate the post, raise post expansions like any comment section but have little effect on interactions, clicks, or attitudes.

Salience responds to how much the comments surprise the viewer and contrast with the post, not to how representative they are. Content produced by a vocal few can therefore influence a much larger audience even when it is not representative. We focus on salience as the channel through which comments affect engagement, and in Section~\ref{sec:survey} we discuss other channels that could affect offline outcomes, such as learning about prevailing norms.

\section{Background and Setting}

\subsection{Comment Sections and Racial Justice}

We study comment sections in the context of racial justice, one of the most divisive issues in the United States. In 2020, the George Floyd protests sparked a racial reckoning across the country. On social media, discussions about race and racial justice, such as those with the hashtag \#blacklivesmatter, surged in 2020 \citep{anderson2020blmsurge}. Yet racial attitudes remain divided in the U.S. According to a Pew Research Center survey in 2024, 80 percent of Democrats or Democratic-leaning independents say that White people benefit from advantages in society that Black people do not have, while only 22 percent of Republicans or Republican-leaning independents expressed the same view \citep{pew_racial_attitudes_2024}.

Our study provides evidence on how these divisions manifest in online discourse and how they influence subsequent users' on-platform engagement and off-platform attitudes and behavior. We partner with Color of Change, the largest online racial justice organization in the United States. As a nonprofit, Color of Change uses digital platforms to mobilize supporters to hold institutions accountable. We designed social media posts in line with Color of Change's brand guidelines to ensure the ecological validity of our study.

\subsection{Social Media Advertising as a Research Tool}
To reach a large audience, we deliver the posts as sponsored content on Facebook using Meta Ads Manager under \textit{The Public Square}, a research account created for the project. Our design leverages social media advertising for several reasons. First, understanding the role of comment sections on advertisements is inherently important, as social media ads represent a significant share of the trillion-dollar advertising industry. In 2025, global social media ad spend exceeded \$275 billion, including about \$100 billion in the U.S.\footnote{DataReportal, \textit{Digital 2026: Global Overview Report}, \url{https://datareportal.com/reports/digital-2026-global-overview-report}.} While existing literature has studied social media ad effectiveness (e.g., \citealt{gordon2019comparison}), the role of user comments in enhancing or diminishing ad impact remains underexplored. Comments on ads allow users to share information and opinions, potentially influencing future viewers, akin to other forms of UGC. Despite significant interest from firms in leveraging UGC and social influence, there is little empirical evidence on the role of social media comments.

Second, using Facebook advertising allows us to access a large and diverse sample while minimizing experimenter demand effects. Recent work (e.g., \citealt{donati2025adaptive, donati2024can}) 
has explored social media ads as a research tool, as they enable the delivery of content to a broad audience and allow researchers to observe user behavior in the natural context in which the content is usually encountered. 

Third, we develop a novel pipeline for comment section manipulation using existing features in the Meta Ads Manager. This enables us to manage comment sections across a large number of posts and systematically collect relevant platform data. Because our approach relies solely on tools already available on the platform, it yields practical insights for social media managers seeking to moderate and manage comment sections, as well as a methodological contribution for researchers seeking to vary social context in an experimental setting at scale. 

\subsection{Post Design and Pre-testing}
To generate the comment sections used in our analysis, we produced ad creatives and copy on five issues identified with our partner organization: voter suppression, environmental justice, criminal justice and police reform, education reform, and technology fairness.\footnote{See Appendix \ref{app:background} for a detailed description of these issues.} The organization reviewed and approved all materials, so the creatives are comparable to content the organization would itself run rather than researcher-produced stimuli. We pre-tested the designs and advertised the best-performing ones to elicit the organic comments that form the basis of our main experiment.

Appendix Figure \ref{fig:banners} displays the creatives and their headlines exactly as they would appear to Facebook users. Appendix \ref{app:background} reports the pre-tests used to select among them and to assess their performance under alternative delivery configurations.

\subsection{Facebook Audience Selection}

We reach Facebook users via sponsored content. A key advantage of social media advertising is that it enables both broad distribution and precise control over audience targeting \citep{aridor2024Experiments}.

To examine how responses vary by audience characteristics, we use ZIP codes as a targeting criterion. We group ZIP codes that are similar in political ideology and racial composition into strata, each pooled into an audience group to target. Within each ideology--race category, we construct several independent strata. Homogeneous strata ensure that arms are compared on similar users rather than on systematically different areas, since Meta's A/B tool randomizes within a single audience. It also lets us compare outcomes across areas with different ideological leaning. 
The ZIP code characteristics were collected from several sources:

\begin{itemize}
    \item \textbf{Meta Audience Estimates}: We use information on audience size provided by Meta, which reports the estimated number of users advertisers could potentially reach over a given period.\footnote{Meta Business Help Center, ``About Estimated Audience Size,'' \url{https://www.facebook.com/business/help/1665333080167380?id=176276233019487}.} This data was collected through the Marketing API for each ZIP code on October 15, 2024.\vspace{-0.2cm}

    \item \textbf{Voting Behavior}: To proxy political preferences and ideology at the ZIP code level, we rely on the 2020 voting results. These come at the precinct level and were obtained from The Upshot.\footnote{The Upshot, \textit{Presidential Precinct Map 2020}, \url{https://github.com/TheUpshot/presidential-precinct-map-2020}.} We assigned each precinct to its nearest ZIP code according to Euclidean distance of the centroids using GIS software, and then aggregated voting information across all precincts matched with the same ZIP code.\vspace{-0.2cm}

    \item \textbf{Population and Racial Composition}: We use 2020 Census data to obtain information on the total and Black populations in each ZIP code and compute the share of Black residents.\footnote{U.S.\ Census Bureau, 2020 ZIP Code Tabulation Area shapefiles, \url{https://www.census.gov/cgi-bin/geo/shapefiles/index.php?year=2020&layergroup=ZIP+Code+Tabulation+Areas}.}
\end{itemize}

We categorize ZIP codes into three ideology groups based on the Republican vote share: \emph{Blue} (Republican vote below 30 percent), \emph{Swing} (between 45 and 55 percent), and \emph{Red} (above 70 percent). Each ideology group was further divided into two subgroups based on racial composition (low and high share of Black population relative to the group median). Hence, a total of six ideology-race categories were created. Each ZIP code is used only once during our experiments. Because Meta's A/B tool randomizes users within a single audience, this ensures that each user is exposed to the content in a single, well-defined condition.

\section{Generation and Analysis of Engagement \label{sec:generation}}

To analyze the impact of comment sections, we begin with a large-scale social media campaign designed to elicit engagement with posts about racial justice. This initial phase generates user comments and reactions in a naturalistic setting, and provides direct evidence on how individuals engage with divisive content across audience types and topics. This engagement forms the basis for subsequent experimental manipulation.\footnote{The generation and analysis of engagement, as well as the main field and survey experiments, were pre-registered before being fielded. The analysis follows the pre-analysis plans, and Appendix~\ref{app:deviations}
documents the few deviations.}

Generating engagement is an important feature of our design for several reasons. First, comments reflect genuine, realistic user behavior, which is crucial for ensuring external validity. Unlike researcher- or AI-generated content, organic comments capture the authentic language, tone, and perspectives that users produce and encounter on social media. Second, these comments are themselves substantively important to study. By targeting specific audiences through Facebook's ad infrastructure, we can link engagement patterns to detailed audience characteristics---such as demographics and ZIP code–level ideology---providing richer insights than are typically available. Third, using organic comments enhances the ethical integrity of the experiment, as users interact with content created by other real users rather than being unknowingly exposed to artificially constructed narratives.

\subsection{Methodology}

Between January 13 and February 10, 2025, we ran a Facebook ad campaign featuring five banners—one for each issue—that had achieved higher click-through rates (CTR) in pre-tests. 
Each banner promoted content related to racial justice, covering topics in education, environmental justice, criminal justice and police reform, technology fairness, and voting rights. 
The campaign was optimized to maximize engagement with the posts—showing ads to users most likely to react, comment, or share—in order to collect authentic interactions that reflected spontaneous responses to important yet divisive content.

To examine variation in engagement across audiences and ad characteristics, we created 30 distinct audience strata. 
Each stratum consisted of sets of ZIP codes randomly sampled and grouped by ideological similarity and racial composition, with an estimated Facebook audience size of about 800{,}000 users on average (we excluded ZIP codes used in pre-tests). 
The 30 strata corresponded to six combinations of political ideology (conservative, moderate, and liberal) and racial composition (above or below the median share of Black residents within each ideological category). For each of the six combinations (e.g., conservative areas with a below-median Black population share), we constructed five audience groups, yielding a total of 30 audience strata. Within each stratum, we used Facebook's native A/B testing tool to randomly allocate users to be potentially exposed to one of the five issue banners. 
In total, the campaign included \feOnePostsN{} posts (30 strata × 5 topics), reaching approximately \feOneReachApproxN{} individuals and generating \feOneReactionsApproxN{} reactions, \feOneUniqueLinkClicksApproxN{} unique link clicks, \feOneCommentsApproxN{} direct comments,\footnote{Direct comments are those directed at the Facebook page/post itself, as opposed to replies to other users’ comments.} and \feOneSharesApproxN{} shares.

The randomization from A/B testing enables comparisons of engagement patterns across topics while holding the potential audience composition constant. 
The resulting data allow us to characterize the intensity and nature of engagement and to identify patterns of interaction by political ideology and demographic composition. 
However, we do not interpret potential differences as causal effects of content. 
Although audiences are randomly assigned to potential exposure, actual exposure is determined by Facebook's ad-delivery algorithm, which endogenously allocates impressions based on predicted engagement probabilities \citep{braun2025b}. 
As a result, the observed engagement patterns reflect the joint influence of both content characteristics and algorithmic delivery.

We focus on engagement outcomes visible to other users: comments and reactions. A \emph{reaction} is a one-click response from a fixed menu of emoji (\emph{like}, \emph{love}, \emph{care}, \emph{laugh}, \emph{wow}, \emph{sad}, \emph{angry}), shown only as an aggregate count. A \emph{comment} is written text posted under the user's name and profile picture. Commenting is thus both more effortful and more public than reacting, which is why the two can move in opposite directions across audiences and why the comment section need not represent the audience that sees the post.

First, we compare overall engagement levels and rates (expressed as a percentage of total reach) across ideological groups, computing confidence intervals using standard errors clustered at the advertisement level (\feOnePostsN{} posts).
Second, we assess the \emph{valence} of these interactions. 
For reactions, we classify \emph{likes}, \emph{loves}, and \emph{cares} as supportive of the posts.\footnote{Other reaction types, such as \emph{laugh}, \emph{wow}, \emph{sad}, and \emph{angry}, are context-dependent and therefore harder to interpret; we focus on those that most reliably convey positive engagement.}
For comments, we use GPT\textendash4 to characterize their political stance, sentiment, informativeness, and offensiveness, and the Google Perspective API to measure toxicity. Because a reply is often hard to interpret by itself, we include a description of the post in the prompt, so each comment is coded in the context of the content it responds to. Appendix \ref{app:Codebook} documents  NLP-based measures used in the paper, including the response scales and the API settings.

\subsection{Engagement Across Locations}

We first document patterns of engagement---both comments and reactions---across different areas. As shown in Figure~\ref{fig:comments_reactions_levels}, there are pronounced differences in how users engage with racial justice content across areas with different ideological compositions. When considering all interactions irrespective of their valence, the number of comments increases substantially from liberal to conservative areas. We find a similar pattern for the rate of commenting in Appendix Figure~\ref{fig:comments_reactions_rates}. Comment rate rises from about \feOneBlueCommentsMeanPct{} percent of reach in Blue ZIP codes to \feOneRedCommentsMeanPct{} percent in Red ones (\(p\feOneCommentsBlueVsRedCROnePText{}\)).
Reactions, by contrast, follow the opposite pattern, declining from roughly \feOneBlueReactionsMeanPct{} percent of reach in Blue areas to 7.5 percent in Red areas (\(p\feOneReactionsBlueVsRedCROnePText{}\)).
These differences indicate that users in more conservative areas are less likely to engage through quick, low-effort reactions but more likely to participate vocally by commenting on posts. 
In more progressive areas, engagement occurs primarily through reactions, suggesting a more silent mode of interaction. Because reactions are far more common than comments, summing the two measures across the subfigures indicates that users in progressive areas are, as expected, more likely to engage with racial justice posts overall.\footnote{As these are differences across locations, we interpret these results as descriptive of who is exposed to and reacts to racial justice content when it is targeted at an area, not as a causal effect of ideology. Conservative and progressive areas differ on many other characteristics that plausibly affect commenting, such as age composition.} 

\begin{figure}[h]
\caption{Comment and Reaction Counts by Location}
\label{fig:comments_reactions_levels}
\centering
\includegraphics[width=0.6\textwidth]{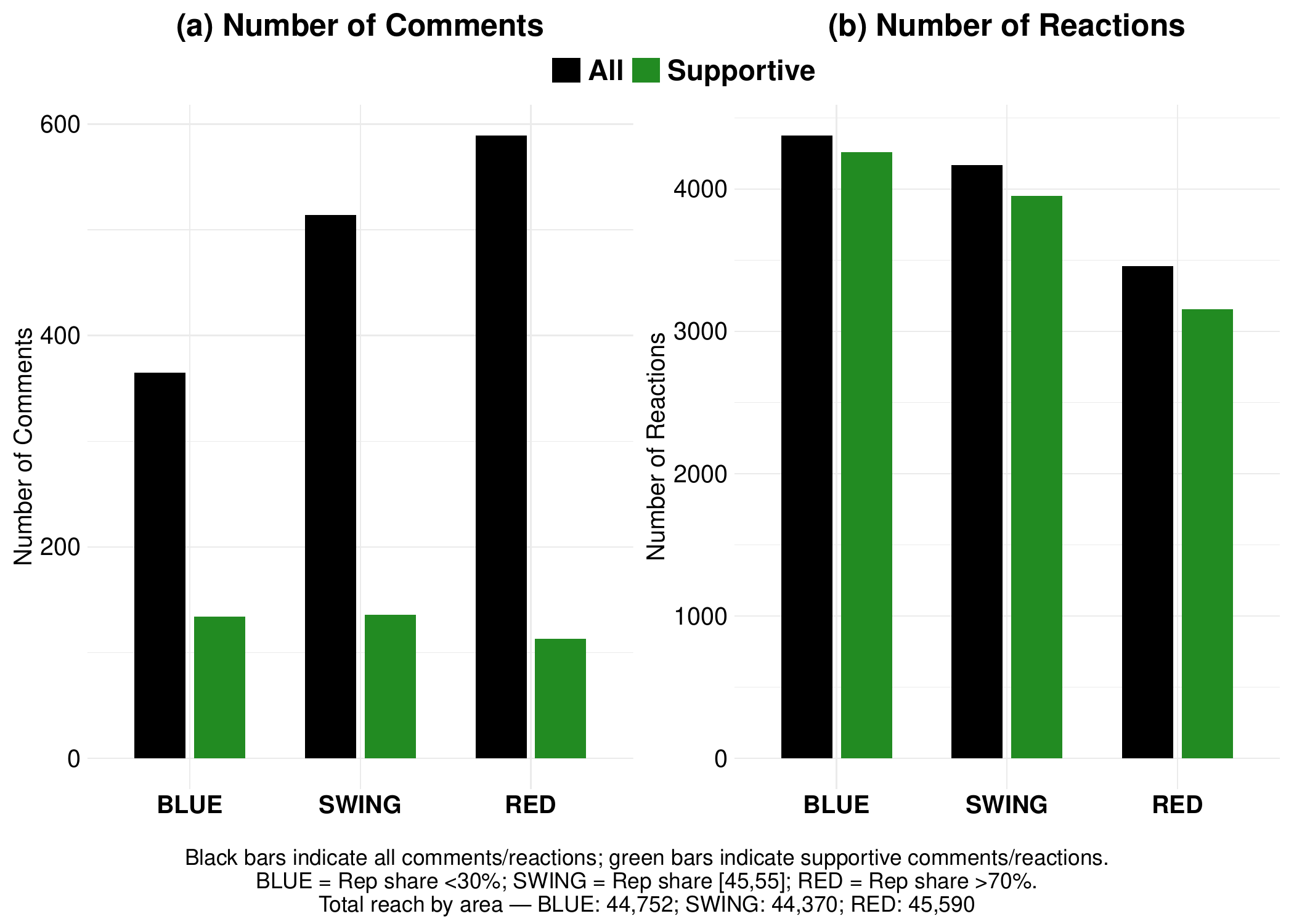} 
\begin{minipage}{\textwidth}
\scriptsize
\noindent \textit{Notes:} This figure reports the total number of comments and reactions generated during the initial Facebook campaign, separately by location. Black bars denote all comments/reactions, while green bars denote supportive comments/reactions. Areas are grouped as BLUE (Republican vote share below 30\%), SWING (Republican vote share between 45\% and 55\%), and RED (Republican vote share above 70\%). Supportive reactions include likes, loves, and cares.
\end{minipage}

\end{figure}

When focusing on \emph{supportive engagement}—defined as reactions or comments expressing agreement or approval—the ideological gradient is notably flatter for comments. 
Supportive comments remain consistently low across areas, ranging around \feOneSupportiveCommentsMeanPct{} percent of reach, with no statistically significant differences between Blue, Swing, and Red ZIP codes (\(p\feOneSupportiveCommentsMinPairwiseCROnePText{}\)).
In contrast, supportive reactions decline significantly from roughly \feOneBlueSupportiveReactionsMeanPct{} percent in Blue areas to about \feOneRedSupportiveReactionsMeanPct{} percent in Red areas (\(p\feOneSupportiveReactionsBlueVsRedCROnePText{}\)), mirroring the overall reaction pattern.
This suggests that while the overall volume of vocal participation (comments) rises in conservative areas, supportive responses remain relatively small and stable. 
Taken together, these results imply that ideological context shapes not only the intensity but also the \emph{type} of engagement: audiences in liberal areas interact more through silent reactions, whereas those in conservative areas engage more vocally, using comments more frequently to express or debate opposing views.

We present the results separately for each issue in Appendix Figure~\ref{fig:comments_reactions_levels_issue} (Appendix Figure~\ref{fig:comments_reactions_rates_issue} reports the rates). Issues such as voter suppression, as well as criminal justice and police reform, generate substantially more vocal engagement in Red ZIP codes, producing large differences relative to Blue ZIP codes. In contrast, topics such as education reform and technology fairness elicit lower levels of vocal engagement overall and exhibit little variation across areas. This pattern suggests that users in more conservative areas are not uniformly more vocal; rather, specific issues within the broader domain of racial justice appear to trigger heightened vocal engagement and often dissent. Unlike comments, reactions follow the general pattern in which users in more progressive areas are more likely to react, and this pattern holds consistently across all issues.

While the different issues allow us to show that engagement patterns differ even within racial justice, by design all of the posts are progressive. Therefore, we interpret the results as descriptive of how people respond to progressive content, rather than content in general. Nevertheless, the results highlight that the comment section under progressive content is disproportionately from conservative areas, so the vocal few are not representative of the underlying audience.

\subsection{Content and Tone of the Comment Section}

Appendix Figure~\ref{fig:comments_characteristics} shows how the content and tone of comments vary across areas with different ideological compositions, with all rates expressed as a share of total comments.%
\footnote{Appendix Tables \ref{tab:phase1length} and \ref{tab:phase1quality} report additional pre-specified analyses on the characteristics of the comment sections, including length, depth, and other measures of conversation quality and opinions' diversity.} 
Comments originating from conservative areas are substantially more likely to express a conservative stance and a negative sentiment. 
The share of conservative-leaning comments rises from about 55 percent in Blue ZIP codes to roughly 75 percent in Red areas (Panel~(a), \(p\feOneConservativeBlueVsRedCROnePText{}\)), while the share of comments with negative sentiment increases from around \feOneBlueNegativeMeanPct{} percent to nearly 80 percent (Panel~(b), \(p\feOneNegativeBlueVsRedCROnePText{}\)).
Panel~(c) shows that the prevalence of offensive language increases from roughly \feOneBlueOffensiveMeanPct{} percent in Blue areas to almost 50 percent in Swing and Red areas, with a statistically significant difference relative to Blue areas (\(p\feOneOffensiveBlueContrastsCROnePText{}\)).
Finally, Panel~(d) indicates that the share of informative comments—those providing factual content or elaboration—remains relatively low, between 7.5 and \feOneInformativeMaxMeanPct{} percent on average, with no statistically significant differences across locations (\(p\feOneInformativeMinPairwiseCROnePText{}\)). We find similar patterns with our pre-specified measures of conversation quality (Appendix Table~\ref{tab:phase1quality}). In the average post, about a third of direct comments engage the existing conversation rather than shifting topic or attacking, fewer than one in five give a reason for the position they take, and fewer than half are civil.

Taken together, these results suggest that while the tone and stance of comments vary, comment sections on divisive issues are generally not very informative across geography.\footnote{This could be because commenters simply prefer to express a stance, or because expressive content draws higher engagement than informative content \citep{lee2018advertising}.}  
Rather than stating a conservative position with relevant arguments, an opposing comment typically voices disagreement without justification.

\section{The Impact of Comments on On-platform Engagement \label{sec:comment_experiment}}
In this section, we provide causal evidence on the effect of the comment section. Specifically, we study how the presence and stance of a comment section influence subsequent users' on-platform engagement with the content.

\subsection{Design \label{sec:experiment_design}}
The experiment ran from March 26 to April 13, 2025, and reached about \feTwoReachApproxMillions{} million Facebook users. Figure~\ref{fig:design} provides an overview of the experimental design. We manipulate the comments that participants see below our ads in a new ad campaign, using Meta's A/B testing tool and an automated pipeline that hides new comments from users.

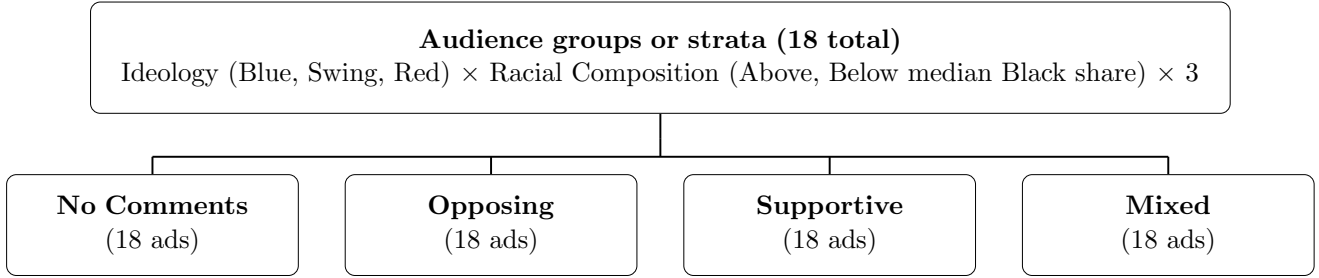
\begin{figure}[h]
\centering
\caption{Design Overview}
\label{fig:design}
\hspace*{-0.5cm}\scalebox{0.92}{
\begin{tikzpicture}[
  strata/.style={draw, rounded corners, align=center, inner sep=10pt,
                 text width=0.95\linewidth, minimum height=1.6cm},
  box/.style={draw, rounded corners, align=center, inner sep=8pt,
              text width=0.22\linewidth, minimum height=1.45cm},
  line/.style={thick}
]
\node[box] (c0) {\textbf{No Comments} \\ (\feTwoAdsPerArmN{} ads)};
\node[box, right=0.04\linewidth of c0] (c1) {\textbf{Opposing} \\ (\feTwoAdsPerArmN{} ads)};
\node[box, right=0.04\linewidth of c1] (c2) {\textbf{Supportive} \\ (\feTwoAdsPerArmN{} ads)};
\node[box, right=0.04\linewidth of c2] (c3) {\textbf{Mixed} \\ (\feTwoAdsPerArmN{} ads)};
\node[coordinate] (mid) at ($(c0.west)!0.5!(c3.east)$) {};
\node[strata, above=1.6cm of mid] (strata)
{\textbf{Audience groups or strata (\feTwoStrataN{} total)}\\
Ideology (Blue, Swing, Red) $\times$ 
Racial Composition (Above, Below median Black share)   $\times$  3 };
\coordinate (split) at ($(strata.south) + (0,-0.35)$);
\draw[line] (strata.south) -- (split);
\draw[line] ($(c0.north)+(0,0.25)$) -- ($(c3.north)+(0,0.25)$);
\draw[line] (split) -- ($(split |- c0.north)+(0,0.25)$);
\foreach \c in {c0,c1,c2,c3} {
  \draw[line] ($(\c.north)+(0,0.25)$) -- (\c.north);
}
\end{tikzpicture}
}
\begin{minipage}{\textwidth}
\scriptsize
\noindent \textit{Notes:} This figure summarizes the experimental design. Users are grouped into \feTwoStrataN{} audience groups or strata, defined by the interaction of ZIP code ideology (Blue, Swing, Red) and racial composition (above or below the median Black population share). Within each ZIP code group, users are randomly assigned via Meta's A/B testing infrastructure to one of four comment conditions: (i) no comments (control), (ii) opposing comments only, (iii) supportive comments only, or (iv) mixed comments.
\end{minipage}
\end{figure}

We investigate how exposure to the comment section and the different narratives expressed in the comment section of a post (collected in the Facebook campaign described above) affect individuals' subsequent engagement with that post (views and comments), as well as their intentions (clicks).

\subsubsection{Intervention Design} 

To select and stratify audience types, we follow a similar approach as the one described in Section \ref{sec:generation}. We exclude ZIP codes used in the comment generation campaign, and create \feTwoStrataN{} audience groups, with three groups for each combination of political ideology (Blue, Swing, Red) and racial composition (above or below the median share of Black residents within each ideological category).

Within each audience group, we use A/B testing to randomly split its population into four conditions: no visible comments (No Comments), visible comments that include both opposing and supportive comments (Mixed), opposing comments only (Opposing), and supportive comments only (Supportive). Figure \ref{fig:S2_arms} provides an example of what users could see under each condition. 

Throughout the paper, Supportive and Opposing describe a comment's stance toward the post beneath which it appears, not its position on an absolute political scale. A supportive comment endorses the claim the post advances; an opposing comment rejects it. Because every post in our experiment advocates a progressive position on racial justice, opposing comments in our setting are also comments a reader would code as conservative. We use the relational terms deliberately, since the object we manipulate is the presence of visible agreement or disagreement with the content.

On Facebook, comments on ads are not visible by default. Users must actively click the comment counter or comment button, or expand the ad container, to view the comment section. Importantly, Facebook algorithmically ranks comments based on engagement and other signals.\footnote{Meta Newsroom, ``Making Public Comments More Meaningful,'' June 2019, \url{https://about.fb.com/news/2019/06/making-public-comments-more-meaningful/}.}
 To mitigate potential ranking effects, we display exactly two comments in each comment condition, thereby minimizing within-condition variation in comment visibility. In addition, at the start of the campaign, the share counter is equalized across conditions (12 shares).\footnote{We equalized the number of shares across posts by sharing them ourselves before the experiment.} The comment counter is absent in the control condition and equalized across the comment arms that are part of the same A/B test, displaying 20, 30, or 40 comments. The reaction counter is similarly balanced across posts within the same A/B test, set at approximately 60, 70, or 80. Figure \ref{fig:S2_arms} shows an example of an A/B test in which the comment counter is set to 30 in the treatment arms and the reaction counter is approximately 80.

\begin{figure}[t]
\begin{adjustwidth}{-0.4cm}{-0.6cm}
\caption{Example Experimental Conditions}
\label{fig:S2_arms}
   \begin{subfigure}[b]{0.262\textwidth}
        \centering
    \includegraphics[width=1\textwidth]{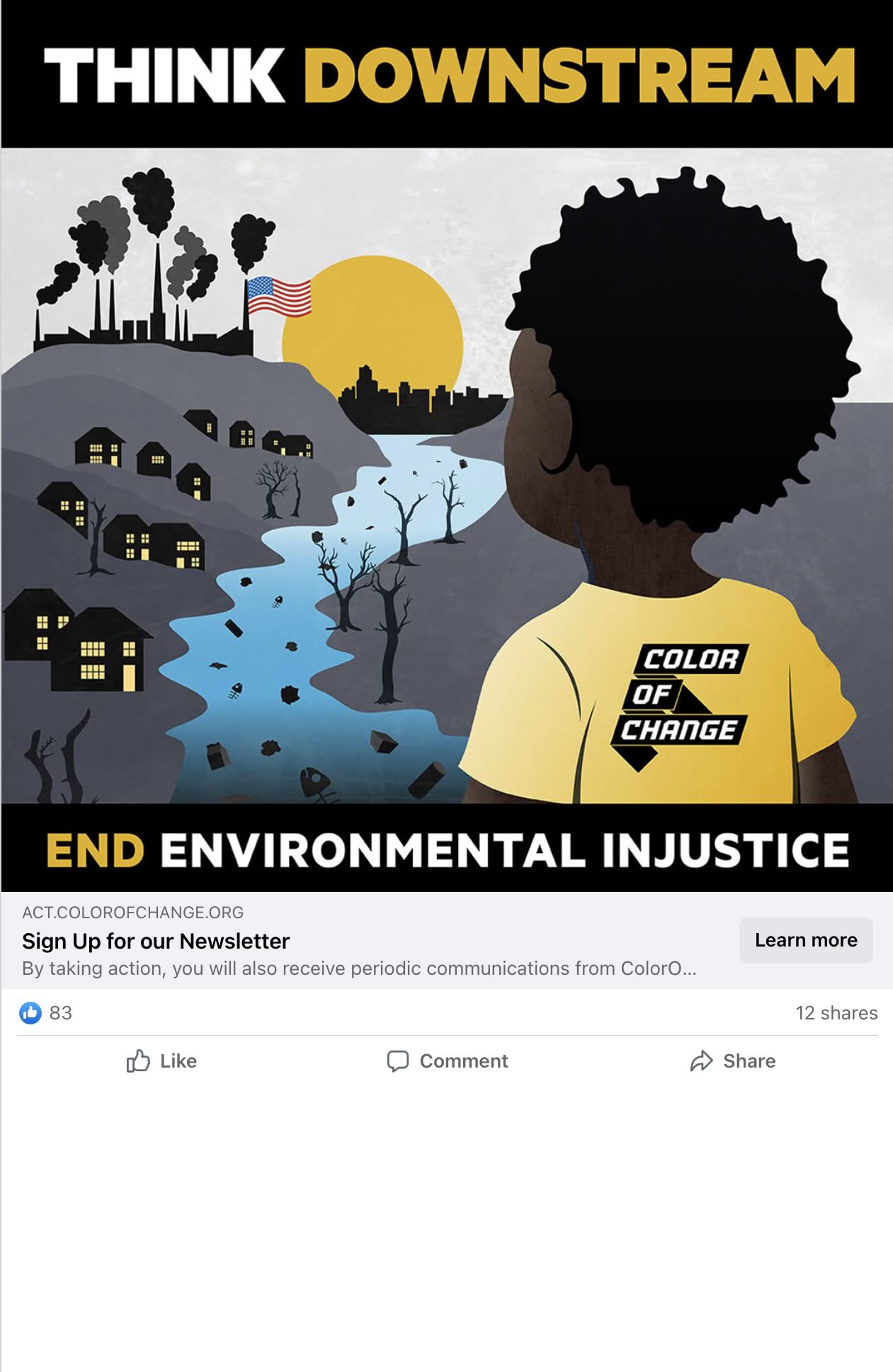} 
        \caption{No Comments}
    \end{subfigure}
    \begin{subfigure}[b]{0.255\textwidth}
        \centering
    \includegraphics[width=1\textwidth]{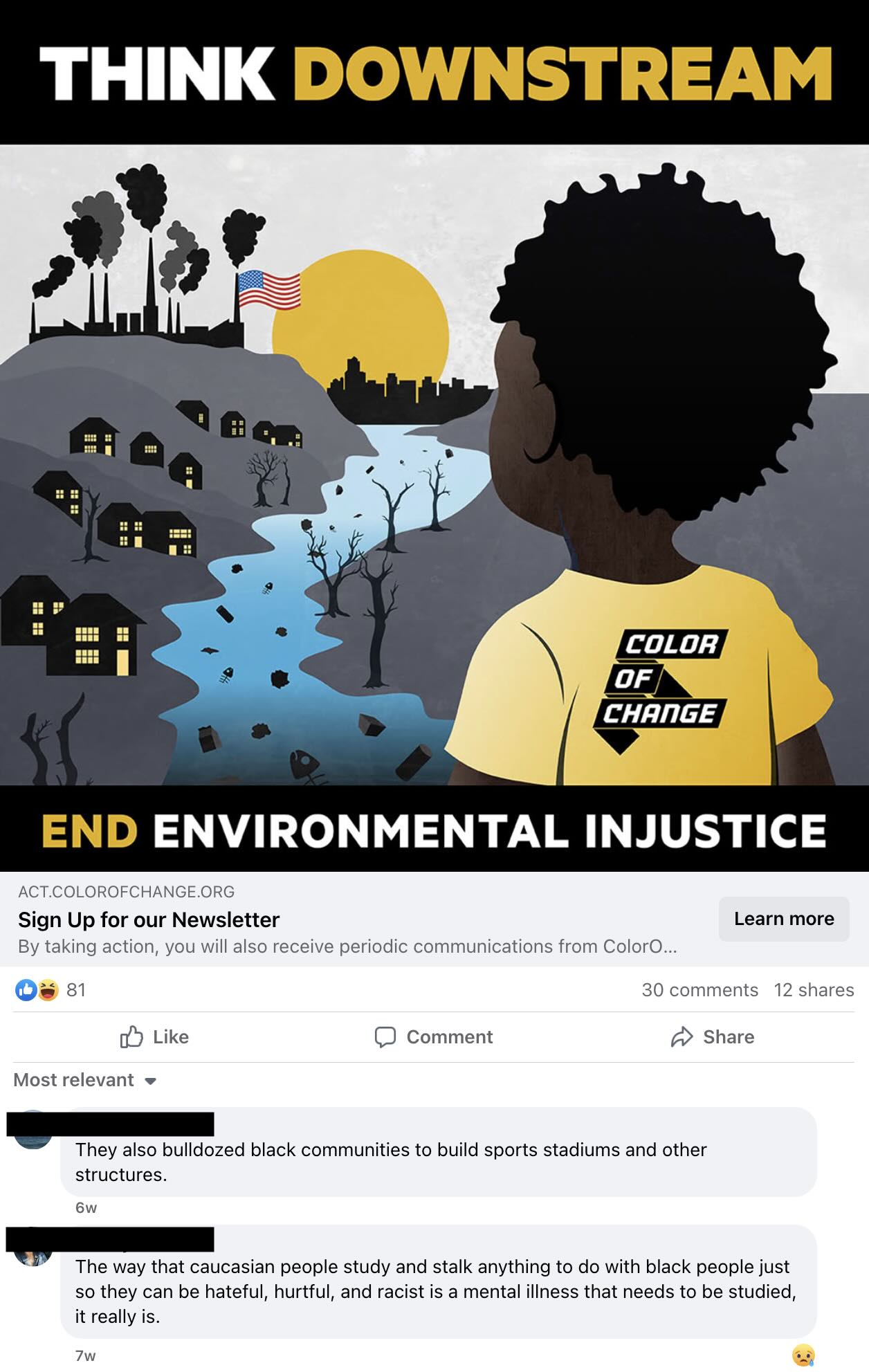}
        \caption{Supportive}
     \end{subfigure} \vspace*{0.4cm}
    \begin{subfigure}[b]{0.254\textwidth}
        \centering
    \includegraphics[width=1\textwidth]{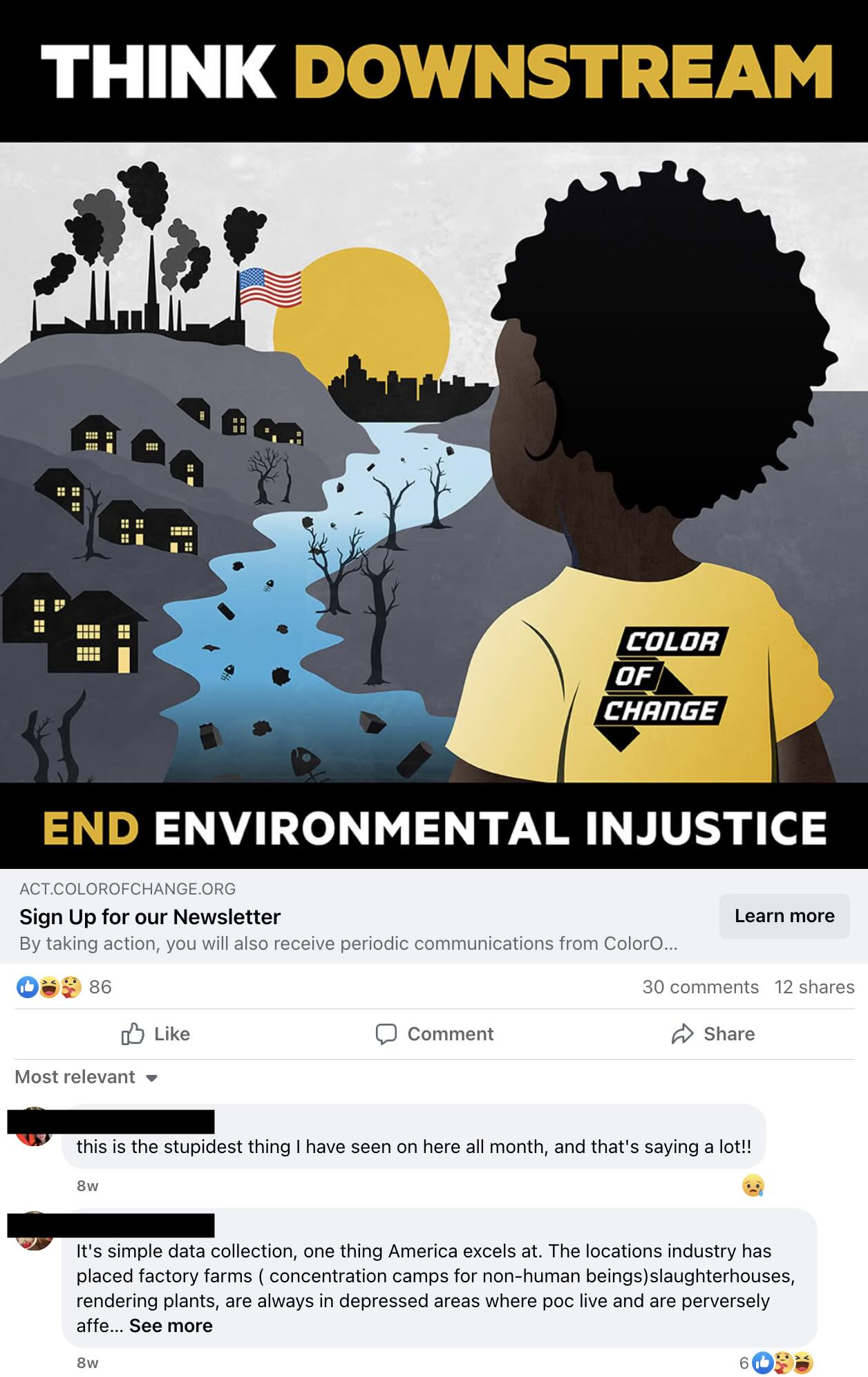} 
        \caption{Mixed}
    \end{subfigure}
    \begin{subfigure}[b]{0.27\textwidth}
        \centering
    \includegraphics[width=1\textwidth]{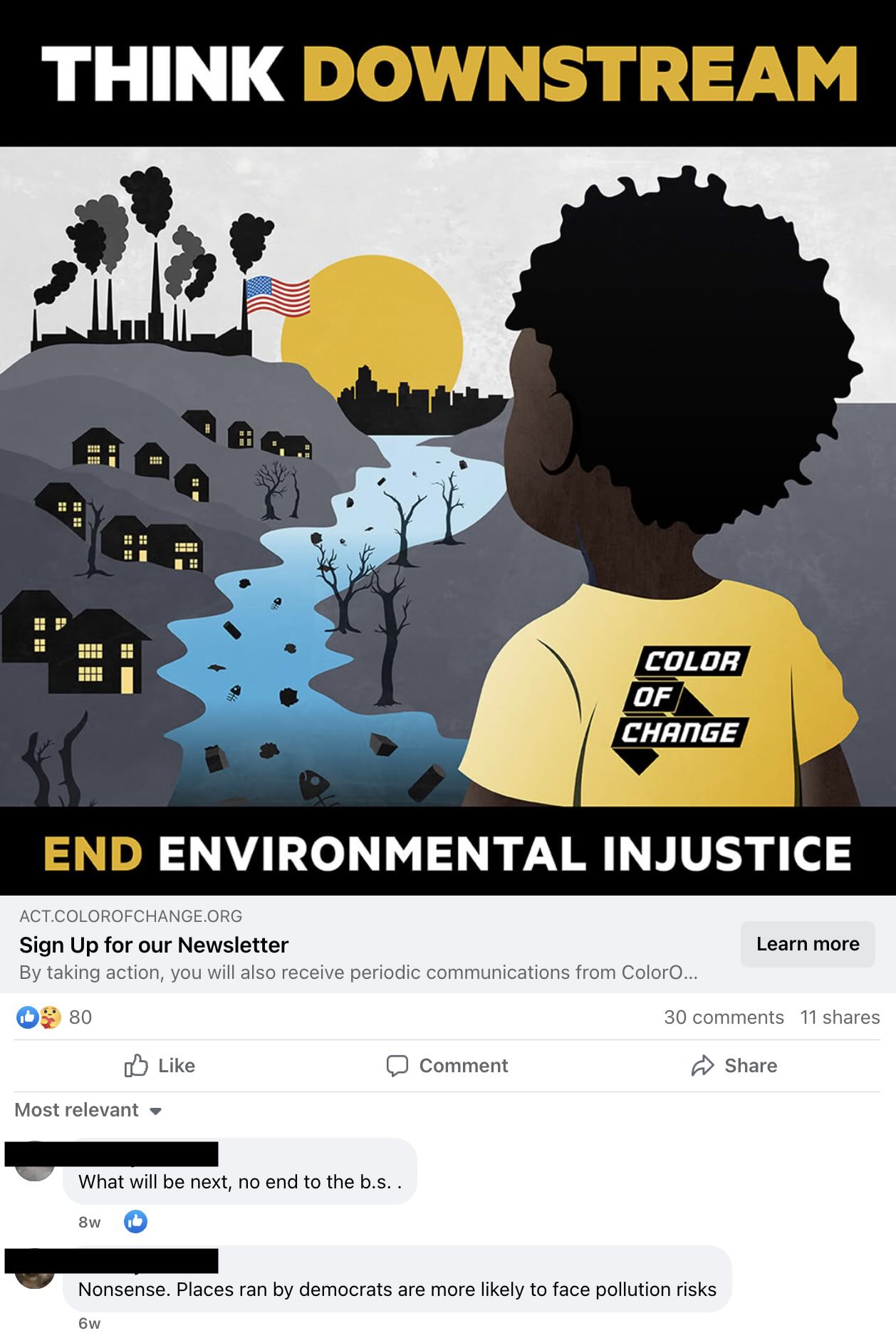}
        \caption{Opposing}
    \end{subfigure}
    \end{adjustwidth}

\begin{minipage}{\textwidth}
\scriptsize
\noindent \textit{Notes:} This figure displays  screenshots of the four experimental conditions as they appeared to users in their Facebook feed. Panel (a) shows the No Comments condition (control), in which the post is displayed without a visible comment section. Panel (b) shows the supportive condition, in which two supportive comments are pre-populated below the post. Panel (c) shows the Mixed condition, in which one supportive and one opposing comment are displayed. Panel (d) shows the opposing condition, in which two opposing comments are pre-populated below the post. All other features of the post (e.g., the image, caption, call-to-action link, number of shares) are held constant across conditions. Aggregate like and reaction counts vary naturally but are balanced across conditions. 
\end{minipage}

\end{figure}

\subsubsection{Issue Selection} From the five issues used in the comment generation phase, we selected environmental justice to focus on in order to maximize statistical power. This topic was chosen because it ranks near the middle in terms of overall engagement in the earlier comment-generation campaign, ensuring sufficient variation in user responses without being dominated by extreme levels of attention or controversy. In addition, the topic remained timely and relevant at the time of the experiment.

\subsubsection{Post Selection} We showed a subset of posts from the comment generation campaign with existing interactions to new audience groups. To select the posts for the treatment conditions, we identified triplets of posts using the following procedure. As part of the analysis in Section \ref{sec:generation}, comments are classified using a 5-point scale for political ideology: Strongly progressive or left-leaning, Slightly or moderately progressive, Centrist/unclear or no explicit stance, Slightly or moderately conservative, or Strongly conservative or right-leaning. In general, comments that support the original post or express pro-racial justice views are classified as progressive, while those that oppose the post or its message are classified as more conservative. To select posts for the Mixed, Opposing, and Supportive conditions, we focused on posts that have at least two supportive and two opposing direct comments (i.e., comments that are directly responding to the post, rather than comments that reply to another comment). For each triplet of posts with a similar number of reactions, we assigned one post each to the Mixed, Opposing, and Supportive conditions. 

For the No Comments condition, we selected posts that are not used in any other conditions and have a number of reactions similar to the average number of reactions in each triplet group. 

We selected \feTwoSelectedPostsN{} posts, organized into three groups of four. Posts are grouped by reaction count, so posts within a group have similar numbers of reactions. Within each group, we have one post in each of Mixed, Opposing, and Supportive, and No Comments (control). All \feTwoSelectedPostsN{} posts use the same ad banner and headline, with identical image, copy, and call-to-action link. We advertise each post to six distinct audience groups, one in each of the six ideology-race categories, creating \feTwoAdsN{} advertisements shown to around one million individuals.

\subsubsection{Comment Selection}
To create the Opposing and Supportive conditions, we selected two comments that matched the relevant stance. For the Mixed condition, we selected one opposing comment and one supportive comment. These comments remained visible to users in the main experiment, while all other direct comments and replies were hidden.
We did so using Facebook's existing moderation tool---the \textit{Hide} comment option available to page owners (see Figure \ref{fig:hide2}). This tool makes a comment invisible to other users. 

Comments were selected by rule and manual inspection. From each eligible post we kept only direct comments, discarding replies to other users. Among these we gave priority to comments at the two extremes of the five-point stance scale. All displayed comments are coded 1 or 5, except one opposing comment coded 4. Posts within a group were assigned to the Mixed, Opposing, and Supportive conditions starting from a random draw, and we finalized the assignment and the displayed comments by manual inspection, checking that each post had suitable comments for its condition. No displayed comment was edited, shortened, or rewritten, and none was written by the research team. Appendix Table~\ref{tab:displayed_comments} reports the full text of every displayed comment with its GPT-4 scores. 

Most of these comments are expressive and often convey a position without justification. Our comment selection procedure allows us to identify the effect of the stance of a comment section under a progressive post, holding the post content and the audience fixed and other visible interactions balanced.\footnote{Because every post takes a progressive position, comments that oppose (support) the post are also comments a reader could call conservative (progressive). We therefore cannot separate the effect of disagreement (agreement) with the post from the effect of simply being exposed to conservative (progressive) comments.}

Finally, to better isolate the effect of the number of comments versus the narrative of the comments on outcomes, we manipulated the number of comments in the posts displayed to audience groups. In the No Comments condition, we deleted as many comments as possible so that the comment counter is zero or close to zero when it is first delivered to the audience in the experiment.\footnote{Some comments, such as those violating Facebook policy, are not visible to the research team and therefore cannot be deleted. As a result, for some comment sections, it may not be possible to reduce the comment counter to zero.} In the triplet of posts used for the treatment conditions, we equalized the number of initial comments across conditions by deleting comments or by posting comments ourselves and hiding them, while allowing the total number of comments to vary across triplets (ranging from 20 to 40). This procedure ensured that the comment counter in the No Comments condition was substantially lower than the counter displayed in the treatment conditions.

\subsubsection{SUTVA Violations}
A potential threat to the internal validity of social media experiments is the violation of the Stable Unit Treatment Value Assumption, or SUTVA \citep{aridor2024Experiments}. 
This concern is particularly relevant to our research question, given the inherently social nature of comment sections. Because outcomes for one user may depend on the treatment or behavior of others exposed to the same post, interference is a first-order threat, and reducing it through design rather than only through analysis is preferable where feasible. New comments posted by users could influence the perceptions or behaviors of other users, thereby confounding the treatment effects. For example, a new comment supportive of racial justice posted in the opposing condition would alter the intended composition of the comment section. 

To address this concern, we implemented a real-time, automated comment-hiding pipeline that immediately hides any new user comments from other users once they are posted.\footnote{The commenter will not know that their comment has been hidden; the comment remains visible to the commenter as well as to the page owner.} This design ensured that the No Comments condition displayed no comments and that users in the treatment conditions saw only the comments corresponding to the assigned stance. It also helped minimize interactions among users exposed to the same post, thereby mitigating possible SUTVA violations resulting from social interactions within the comment section.\footnote{Other forms of visible engagement on the post, such as the number of likes and shares, remained stable over the treatment period. At the start of the experiment, the average post had \averagereactions{} reactions and 12 shares; the additional engagement generated over the course of the experiment is small relative to these baseline levels.}

Hiding new comments comes at a cost. In the field, an early comment draws replies and the composition of a section shifts over the life of a post. Existing work shows that what is already visible changes what later users write: exposure to higher quality comments leads strangers to write higher quality comments \citep{berry2017discussion}. Our pipeline shuts down this channel for internal validity: otherwise users exposed to the same post would determine each other's treatment, and the stance we assigned would not be the stance later users saw. We therefore estimate the effect of a fixed comment section on the users who read it, not the joint effect of that section and the discussion it would generate.

\subsubsection{Divergent delivery} 
Field experimentation in online display advertising presents several challenges to causal inference \citep{johnson2023inferno}. In particular, even within A/B tests, advertising platforms' algorithms may optimize campaign delivery over time for predicted user–ad relevance. As a result, different users can be targeted across experimental conditions based on engagement early in the campaign, generating algorithmic selection bias over time \citep{eckles2018field, braun2025b}. This threat to internal validity---the divergent delivery bias---poses a key concern when the goal is to identify the causal effect of specific post features, rather than the joint effect of algorithmic delivery and post features. 

To minimize the risk of divergent delivery and ensure that our estimates reflect the causal effect of the comment section on behavior, rather than the effect of the comment section and platform delivery, we followed and augmented best practices from recent work \citep{burtch2025characterizing}. First, we split budgets evenly across arms, launched all ads simultaneously, and sought to cap exposure at one impression per user. Second, we optimized the campaign for reach rather than engagement. Third, we set a budget large enough to saturate the predefined audience and ran the campaign over multiple weeks to ensure that nearly all users within a defined area were reached. These design features minimize differences in ad delivery across conditions \citep{braun2025b}. Consequently, any observed differences across conditions can be more confidently interpreted as the causal effect of the comment section on user behavior, rather than as artifacts of algorithmic delivery dynamics.

Since we equalized the number of shares and average reactions across conditions and held the post content constant, the only element that varies across conditions is the comment section. Together with the precautionary measures described above, this design allows us to isolate the effect of the comment section.

\subsection{Data}

Our data come from two sources. First, we obtain engagement metrics from the Meta Ads Manager dashboard,\footnote{Meta Ads Manager, \url{https://www.facebook.com/business/tools/ads-manager}.} which reports daily outcomes for each advertisement disaggregated by gender and age group. Second, we collect the full text and timestamps of all user-generated comments directly from the corresponding Facebook posts. 

\subsubsection{Outcomes}

We analyze a set of engagement outcomes that capture on-platform activity. Specifically, we focus on four metrics: 1) all engagement, defined as any unique user activity related to the ad, including clicks on the ad container, link clicks, profile clicks, reactions, comments, shares, saves, and other observable actions; 2) post expansions, our proxy for attention to the post and its comment section, defined as clicks to expand the ad panel, including clicks to view the comments (this outcome is not directly reported in the Meta Ads Manager interface and is constructed by differencing out other available metrics);\footnote{\textit{Post Expansions = Clicks All - Page Engagement}. \textit{Clicks All} captures all user clicks on an ad, including clicks on the ad container, link clicks, profile clicks, reactions, comments, shares, saves, and other interactions. \textit{Page Engagement} includes all identifiable interactions with the ad or the advertiser's page (link clicks, profile clicks, reactions, comments, shares, saves, and other interactions), except for ad container clicks. The residual therefore isolates clicks on the ad container that expand the post without generating any other observable user activity.} 3) interactions, defined as the sum of unique reactions (likes and other emoji responses), comments, and shares (reposting); and 4) unique link clicks, capturing whether a user clicked on the external link, which we interpret as intent to learn more about the campaign.\footnote{The last three outcomes need not add up exactly to all engagement. For example, a user who both comments and clicks on the link is counted twice when the outcomes are summed but only once in all engagement.}\textsuperscript{,}\footnote{As a robustness check, Section \ref{sec:robust} reports results using landing-page views, which capture downstream off-platform activity. This measure has two limitations. First, landing-page views are inferred via the Facebook Pixel rather than directly recorded by Facebook, which may introduce measurement error. Second, unlike unique link clicks, landing-page views are not unique at the user level.}

All of these actions can be taken with or without exposure to the comments themselves. Across many platforms, including Facebook, the comment section is not displayed by default. In the feed, users see the post and its aggregate counters, but comment text appears only if they expand the container. Unless otherwise specified, we report the engagement rates computed as the ratio of each outcome to total reach, which is defined as the number of distinct users who were shown the ad at least once.

In addition to these engagement outcomes, we collect the full text and timestamps of all user-generated comments. This allows us to examine not only the volume but also the \textit{stance} of user discourse. We use a large language model (GPT-4) to classify each comment as \textit{supportive} or \textit{non-supportive} of the organization's message, based on its semantic content and sentiment. Analogously, we categorize user reactions according to their valence: ``likes,'' ``loves,'' and ``cares'' are coded as supportive, while other reaction types (such as ``angry,'' ``sad,'' or ``wow'') are coded as neutral or not supportive. These additional measures allow us to quantify how pre-existing comment narratives shape the tone and direction of subsequent engagement, thereby linking the stance of visible comments to the ideological composition and sentiment of later user responses. 

\subsubsection{Sample Characteristics}

Appendix Table \ref{tab:summary_stats} reports descriptive statistics for the sample used in the main experiment. The final dataset comprises \feTwoReachN{} observations at the user level, the sum of the unique users reached by each of the \feTwoAdsN{} ads. Treatment assignment is well balanced across arms, with roughly \feTwoArmMeanSharePct{} percent of individuals allocated to each of the four conditions (Control, Supportive, Mixed, and opposing comments).
Engagement outcomes exhibit substantial variation, reflecting the skewed and infrequent nature of user activity on social media platforms. On average, \feTwoAllMeanPct{} percent of reached users engaged with the post in any form, \feTwoExpansionsMeanPct{} percent expanded the ad panel to view the comment section, and \feTwoClicksMeanPct{} percent clicked on the external link. Interaction rates—comprising reactions, comments, and shares—averaged \feTwoInteractionsMeanPct{} percent of reach, while \feTwoPageViewsMeanPct{} percent of users visited the organization’s landing page.

The reached audience broadly tracks the demographic composition of Meta's U.S. adult audience, although some differences remain. Men account for \feTwoMaleSharePct{} percent of reached users and women for \feTwoFemaleSharePct{} percent, compared with 46.8 percent and 53.3 percent, respectively, in Meta's combined U.S. adult (18+) ad audience across Facebook, Instagram, and Messenger. By age, the reached audience is concentrated among users aged 25–44, who represent \feTwoYoungAdultSharePct{} percent of the sample compared with 43.6 percent of Meta's U.S. adult audience, whereas users aged 18–24 and 65+ are relatively underrepresented (\feTwoYoungestSharePct{} vs. 18.0 percent and \feTwoSeniorSharePct{} vs. 12.7 percent, respectively). Taken together, these patterns indicate that the campaign reached a demographically broad segment of U.S. Meta users, but with somewhat greater representation of men and middle-aged users than in the overall adult audience.

\subsubsection{Balance Checks}
Table \ref{tab:balance} reports covariate balance across the four experimental conditions. The sample comprises approximately \feTwoReachMillions{} million individuals, evenly distributed across treatment arms, with about 263{,}000 users per group. We examine three observable individual-level characteristics: gender, middle-aged status (ages 35–64), and senior status (ages 65 and above). Mean values and standard deviations are shown by group, and pairwise differences with the control arm are tested using two-sample \textit{t}-tests with standard errors clustered at the advertisement level (see Section~\ref{sec:Strategy} for details).

Across all covariates, differences between treatment and control groups are small in magnitude and statistically insignificant. The \textit{p}-values from the corresponding tests uniformly exceed conventional significance thresholds, indicating that random assignment produced well-balanced groups across key demographic dimensions. This balance supports the internal validity of the experimental design and suggests that any subsequent differences in engagement outcomes can be  attributed to the randomized variation in comment visibility and stance, rather than to pre-existing differences in audience composition.

\begin{table}[t]
\centering
\caption{Balance Checks: Individual-level Covariates}
\label{tab:balance}
\small
\resizebox{\textwidth}{!}{\begin{tabular}{lccccccc}
\toprule\toprule
& \multicolumn{4}{c}{\textit{Group Mean / (SD)}} & \multicolumn{3}{c}{\textit{t}--test \textit{p}--value} \\
\cmidrule(lr){2-5}\cmidrule(lr){6-8}
Variable &
(1) Control & (2) Supportive & (3) Mixed & (4) Opposing &
(1)–(2) & (1)–(3) & (1)–(4) \\[0.3em]
\midrule
Male
 & \makecell{0.522 \\ (0.500)}
 & \makecell{0.522 \\ (0.500)}
 & \makecell{0.524 \\ (0.499)}
 & \makecell{0.525 \\ (0.499)}
 & 0.982 & 0.898 & 0.877 \\
\midrule
Middle-Aged (35-64)
 & \makecell{0.538 \\ (0.499)}
 & \makecell{0.536 \\ (0.499)}
 & \makecell{0.537 \\ (0.499)}
 & \makecell{0.535 \\ (0.499)}
 & 0.867 & 0.944 & 0.812 \\
\midrule
Senior (65+)
 & \makecell{0.074 \\ (0.262)}
 & \makecell{0.076 \\ (0.266)}
 & \makecell{0.075 \\ (0.263)}
 & \makecell{0.074 \\ (0.262)}
 & 0.736 & 0.937 & 0.973 \\
\midrule
Observations
 & 263{,}706 & 262{,}246 & 263{,}197 & 264{,}866 &  &  &  \\[0.3em]
\bottomrule
\end{tabular}
}
\vspace{0.3em}

\begin{minipage}{\textwidth}
\scriptsize
\noindent \textit{Notes:} Each observation is a user. Standard errors in the t-tests are clustered at the ad level (\feTwoAdsN{} ads).
\end{minipage}
\end{table}

Appendix Table \ref{tab:balance_adlevel} reports balance checks for ad-level cost and performance metrics across the four experimental conditions. Each observation corresponds to one advertisement, for a total of \feTwoAdsN{} ads evenly distributed across treatment arms. We compare total spend, cost per mille (CPM), frequency, reach, and spend per user to verify that Meta's delivery algorithm exposed ads in each treatment arm to comparable audience sizes and costs.

Mean values are virtually identical across groups, and none of the pairwise differences relative to the control group are statistically significant. Total spend per ad averages approximately \$\feTwoAdSpendMeanUSD{}, with CPMs around \$\feTwoAdCPMMeanUSD{} and mean reach near \feTwoAdReachApproxMeanN{} users. The estimated \textit{p}-values for all tests are well above conventional significance thresholds, confirming that the experimental conditions were implemented under comparable delivery and budget parameters. These results indicate that Meta's optimization algorithm did not differentially allocate impressions or spending across treatment arms, reinforcing the internal validity of our causal design.

The frequency and reach metrics further support the correct implementation of the experimental design. Our objective was to saturate audiences such that each individual would be reached at most once. The estimated average frequency of approximately \feTwoAdFrequencyMean{} shows that repeat exposure was limited though not eliminated. Combined with an upper-bound estimated audience size of about 14{,}518 users per ad on average, it is consistent with complete audience saturation and balanced delivery across conditions. Appendix Table \ref{tab:audience_saturation} reports saturation numbers under different scenarios. Under the most conservative estimate, we reached over 95 percent of the potential audience in the areas we targeted.

\subsection{Empirical Strategy} \label{sec:Strategy}
We observe outcomes for \feTwoAdsN{} ads, one for each combination of audience stratum \( z \in \{1, \dots, \feTwoStrataN{}\} \) and condition \( k \in \{\text{No Comments}, \text{Opposing}, \text{Supportive}, \text{Mixed}\} \). Within each stratum, individuals are randomly assigned to one of the four ads. For each ad, we record the number of unique individuals reached and the number who took a given action (e.g., clicking the link or reacting to the post), disaggregated by day, gender, and age group. In the main analysis, we aggregate across days, so each observation covers the full campaign.

To estimate effects at the individual level, we construct a synthetic dataset in which each observation is one individual reached by an ad. The construction is deterministic: each cell defined by ad, gender, and age group is replicated once per individual reached, and \( Y = 1 \) is assigned to exactly the number of individuals in the cell who took the action. Let \( p_{zg}^k \) denote the observed share of individuals in stratum \( z \), condition \( k \), and gender--age group \( g \) who took action \( Y \). We model the individual outcome as Bernoulli, \( Y_{izg}^k \sim \text{Ber}(p_{zg}^k) \), independent and identically distributed within each cell. Random assignment within strata justifies this assumption. The synthetic data let us include individual-level covariates and stratum fixed effects, and report coefficients as changes in the probability that a user takes the action.\footnote{The point estimates are numerically identical to weighted least squares on the cell-level proportions, weighted by cell reach, a standard property of regression on grouped data.}

We then estimate the following linear probability model by ordinary least squares:
\begin{equation}
Y_{iz} = \alpha + \sum_{k} \beta_k T_i^k + X_{iz}' \gamma + \delta_z + \varepsilon_{iz},
\label{eq:main}
\end{equation}
where \( Y_{iz} \) is the outcome of individual \( i \) in stratum \( z \) in the synthetic dataset; \( T_i^k \) indicates assignment to condition \( k \in \{\text{Opposing}, \text{Supportive}, \text{Mixed}\} \), with \textit{No Comments} as the omitted category; \( X_{iz} \) is a vector of gender and age-group indicators, as well as their interactions with the stratum fixed effects; \( \delta_z \) denotes stratum fixed effects, which the tables label \textit{ZIP Code Set FEs} because each stratum is a set of ZIP codes; and \( \varepsilon_{iz} \) is an error term. The coefficients \( \beta_k \) capture the causal effect of assignment to each comment condition relative to the control.

We estimate heterogeneous effects by re-estimating equation~\eqref{eq:main} separately by area ideology. Ideology is measured at the ZIP code group level, not for individuals. A Red area is one where the Republican vote share is high, not one in which every user is conservative, and vote shares correlate with other characteristics such as population density. We therefore interpret these results as heterogeneity across locations.

We cluster standard errors at the stratum--treatment level, corresponding to the \feTwoAdsN{} ads. This is the level at which the comment section is delivered. Individuals within a cell see the same post and, in the treatment arms, the same comments. Clustering allows arbitrary dependence within cells, so inference does not rely on the i.i.d.\ assumption behind the synthetic data analysis. It requires that cells be independent of one another. The design supports this: strata are mutually exclusive sets of ZIP codes, ZIP codes used in the comment generation campaign are excluded, and the comment-hiding pipeline prevents users who see the same post from influencing each other. Section~\ref{sec:robust} reports robustness to alternative specifications and inference methods, including estimates computed directly on the \feTwoAdsN{} ad-level rates instead of relying on the synthetic data, and wild cluster bootstrap $p$-values.

\subsection{Results}
\subsubsection{Main Results \label{sec:main_results}}

We present estimates of the causal impact of the comment section on on-platform user engagement, based on the linear probability model described above, with standard errors clustered at the advertisement level. Figure~\ref{fig:results_others}(a) reports the effects on overall engagement, while Figures~\ref{fig:results_others}(b–d) present the effects on post expansions, interactions, and link clicks. Effects are reported in percentage points (pp). Full regression results are provided in Appendix Table~\ref{main_combined}.

\begin{figure}[H]
\centering

\caption{The Impact of the Comment Section on On-platform User Engagement\\
\textit{Outcomes are expressed in \% of total reach}}
\label{fig:results_others}

\vspace{0.5em}

\begin{subfigure}[b]{0.48\textwidth}
    \centering
    \caption{All Engagement}
    \includegraphics[width=\textwidth]{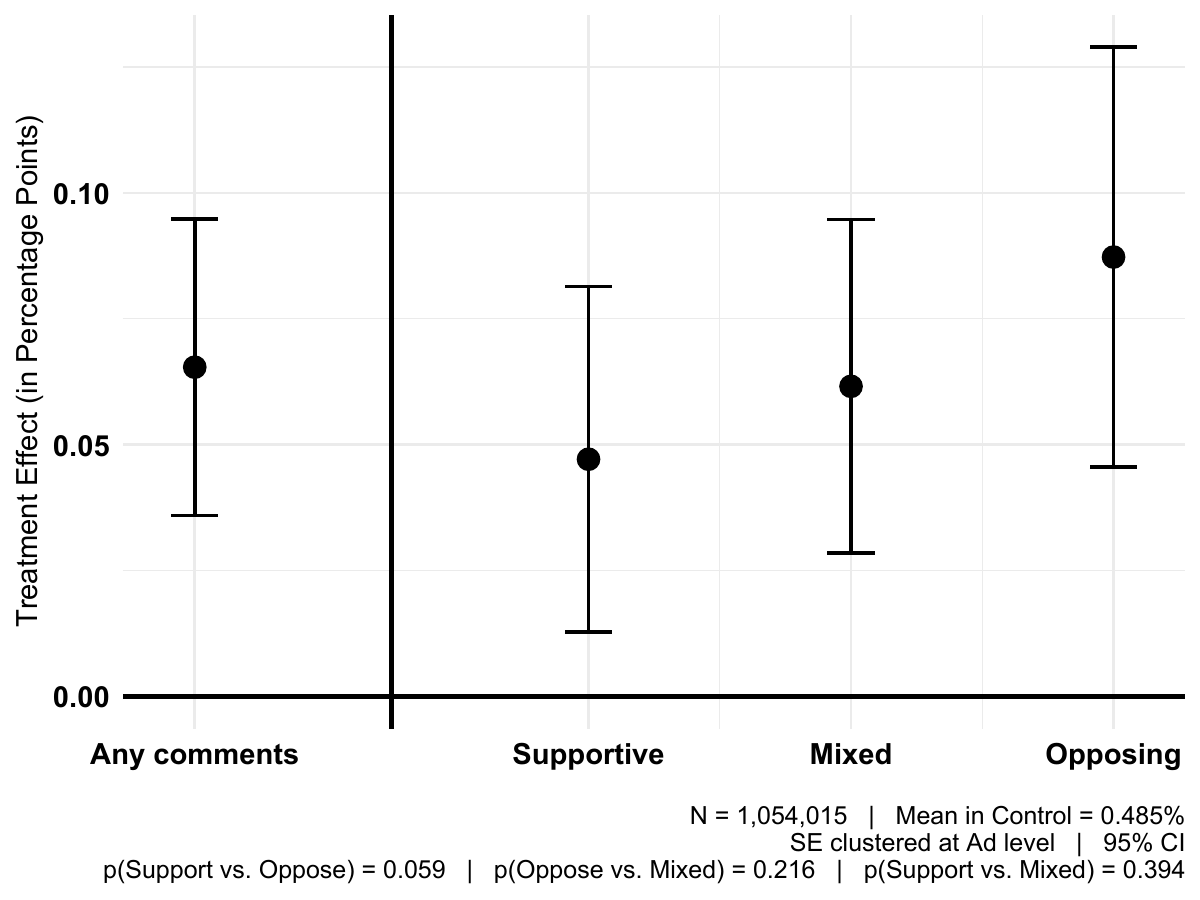}
\end{subfigure}
\hfill
\begin{subfigure}[b]{0.48\textwidth}
    \centering
    \caption{Post Expansions}
    \includegraphics[width=\textwidth]{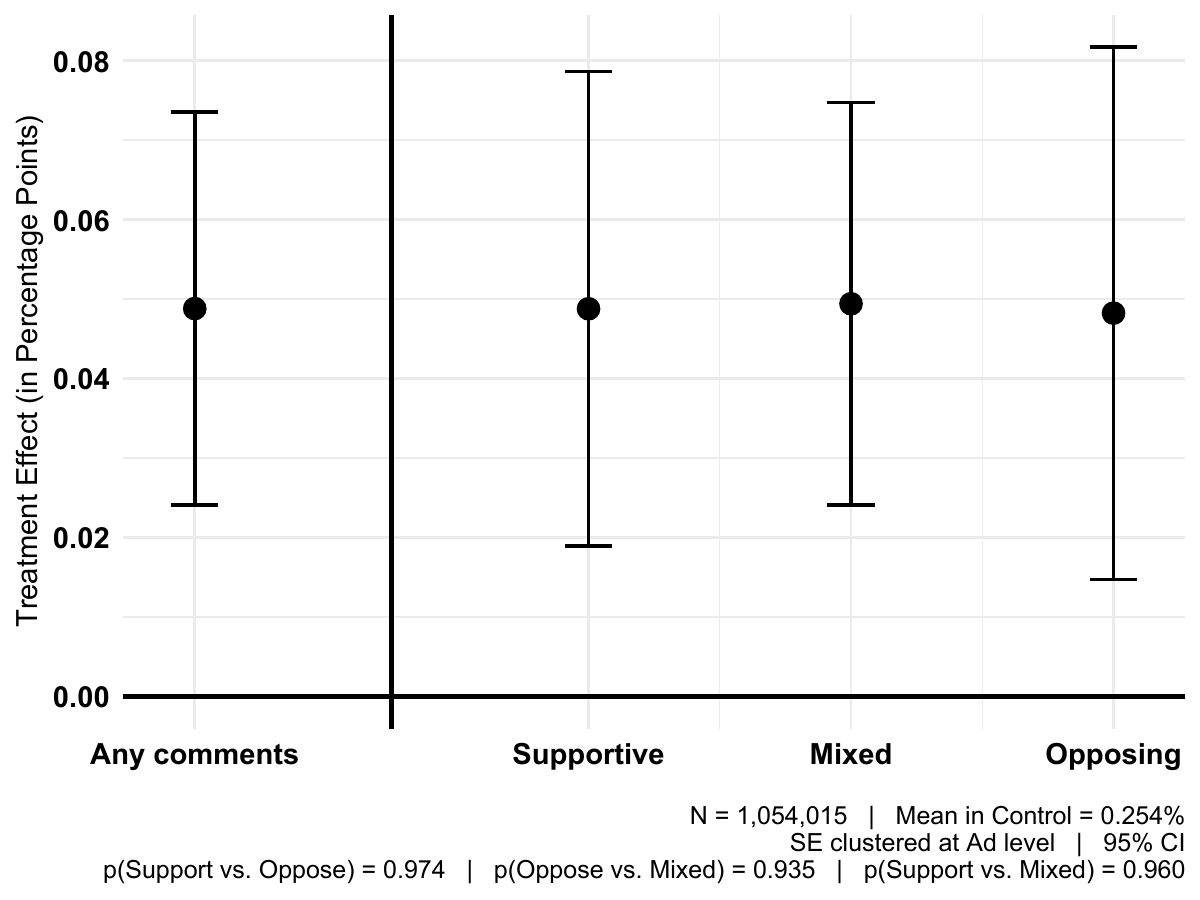}
\end{subfigure}

\vspace{0.4cm}

\begin{subfigure}[b]{0.48\textwidth}
    \centering
    \caption{Interactions (reactions, comments, shares)}
    \includegraphics[width=\textwidth]{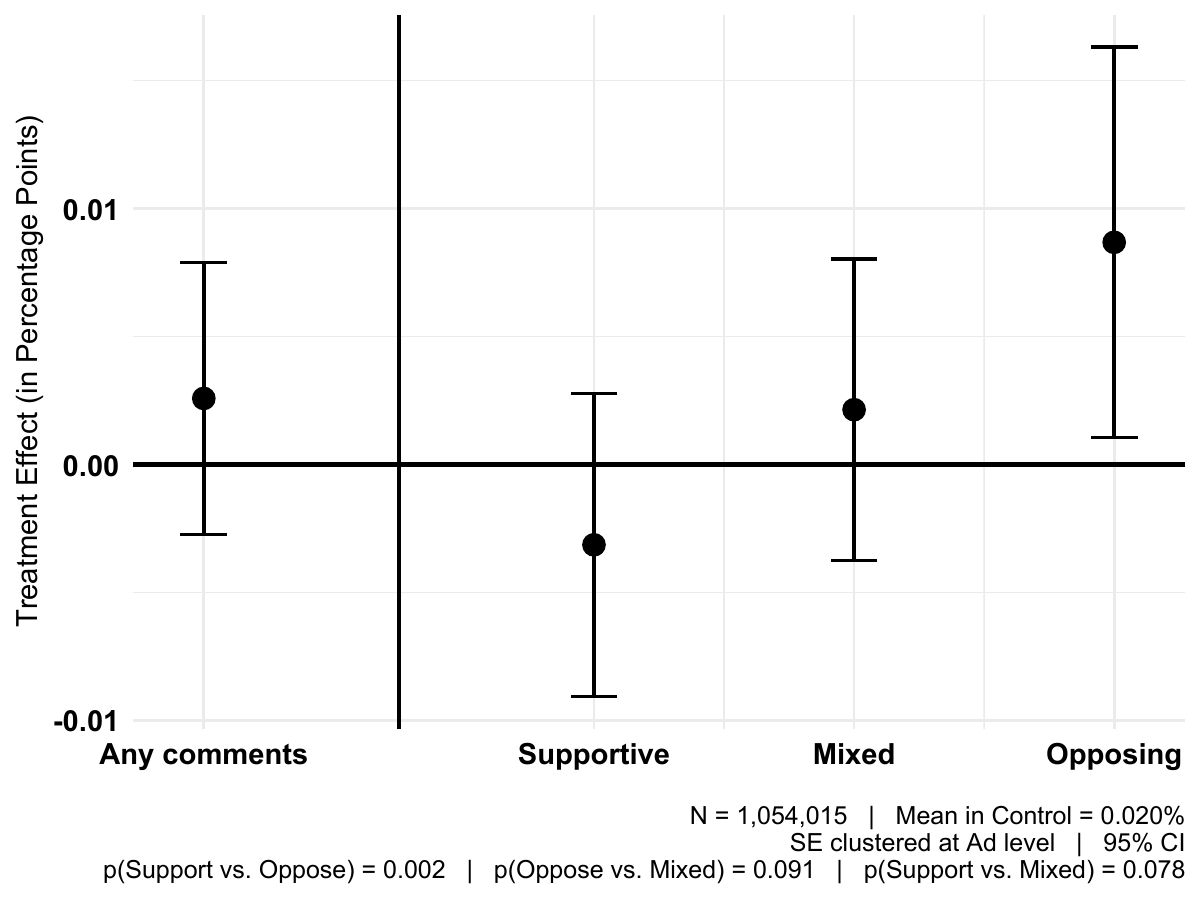}
\end{subfigure}
\hfill
\begin{subfigure}[b]{0.48\textwidth}
    \centering
    \caption{Link Clicks}
    \includegraphics[width=\textwidth]{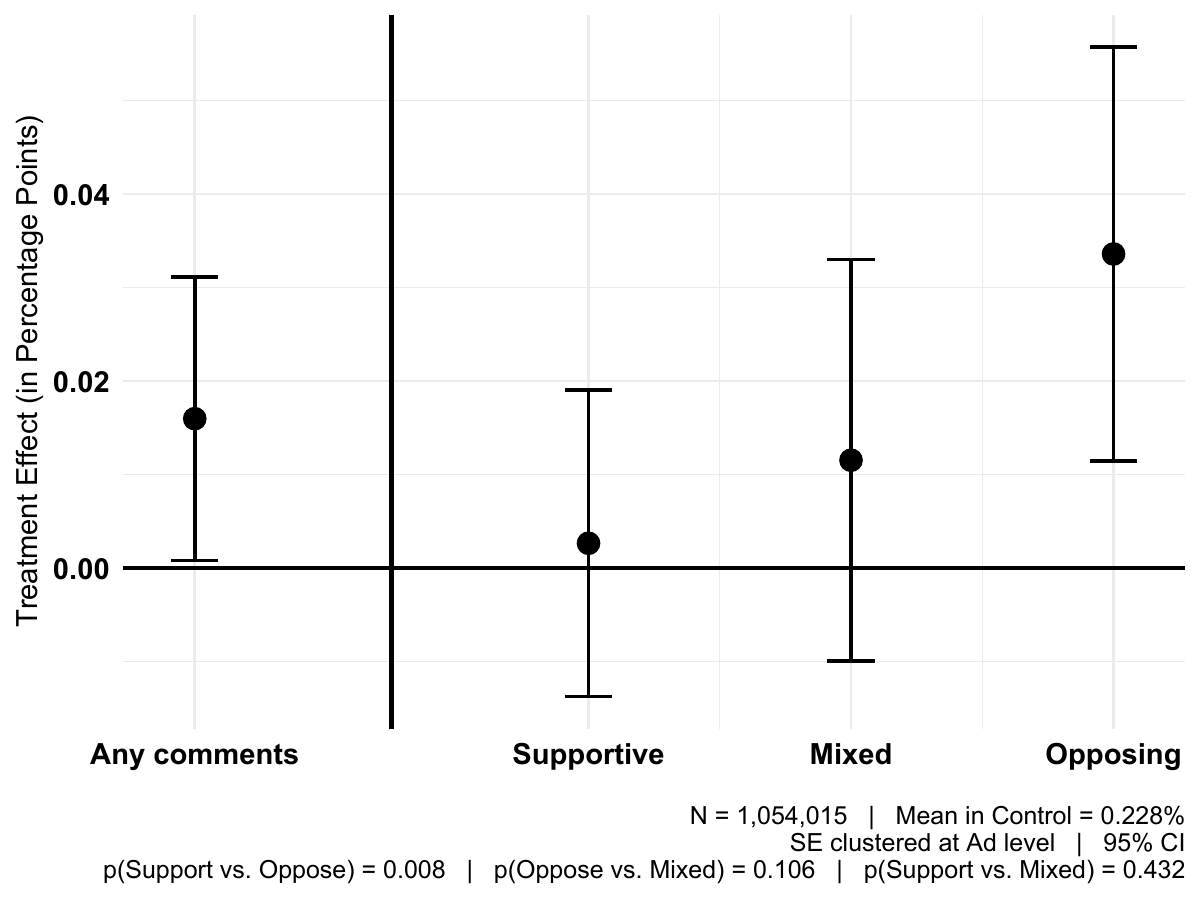}
\end{subfigure}

\vspace{0.3em}

\begin{minipage}{\textwidth}
\scriptsize
\noindent \textit{Notes:} This figure reports treatment effects of comment visibility and 
stance on on-platform engagement outcomes, expressed in percentage points of total reach. Panel (a) reports all engagement, panel (b) post expansions, panel (c) interactions (comments, reactions, and shares), and panel (d) unique link clicks. ``Any comments'' pools the three comment-treatment arms. Estimates are from regression models with the No Comments arm as the omitted category, including gender and age-group controls, ZIP code set fixed effects, and all pairwise interactions among gender, age group, and ZIP code set. Standard errors are clustered at the advertisement level (\feTwoAdsN{} ads). Vertical lines represent 95\% confidence intervals. The sample includes \feTwoReachN{} reached users.
\end{minipage}

\end{figure}

\paragraph{All Engagement}
Figure~\ref{fig:results_others}(a) shows that displaying any comments increases total engagement by \feTwoAllAnyEffectPP{} pp (\(p\feTwoAllAnyCROnePText{}\)), corresponding to a \feTwoAllAnyRelativePct{} percent increase relative to the control mean of \feTwoAllControlMeanPct{} percent.
By stance, opposing comments generate the largest increase (\feTwoAllOpposeEffectPP{} pp, \(p\feTwoAllOpposeCROnePText{}\); \feTwoAllOpposeRelativePct{} percent), followed by Mixed (\feTwoAllMixedEffectPP{} pp, \(p\feTwoAllMixedCROnePText{}\); \feTwoAllMixedRelativePct{} percent) and supportive (\feTwoAllSupportEffectPP{} pp, \(p\feTwoAllSupportCROnePText{}\); \feTwoAllSupportRelativePct{} percent).
The difference between opposing and supportive comments is statistically significant at the 10 percent level (\(p=\feTwoAllOpposeVsSupportCROneP{}\)) and sizable, as the effect of opposing comments is nearly \feTwoAllOpposeSupportEffectRatio{} times as large.
The results indicate that not only does the presence of comments matter for subsequent engagement, but also their stance. In particular, critical remarks draw substantially more overall engagement than supportive ones, suggesting that negative commentary tends to amplify overall user activity around the content. We next examine which specific actions drive this pattern.

\paragraph{Post Expansions}
Figure~\ref{fig:results_others}(b) shows that any comments increase the probability that users expand the ad panel to view the comment section by \feTwoExpansionsAnyEffectPP{} pp (\(p\feTwoExpansionsAnyCROnePText{}\)), corresponding to a \feTwoExpansionsAnyRelativePct{} percent increase over the baseline mean of \feTwoExpansionsControlMeanPct{} percent.  Interpreting post expansions as a proxy for attention, these findings indicate that the presence of comments per se increases attention, independent of their ideological orientation.
All three stance conditions produce virtually identical effects—Opposing (\feTwoExpansionsOpposeEffectPP{} pp, \(p\feTwoExpansionsOpposeCROnePText{}\)), Mixed (\feTwoExpansionsMixedEffectPP{} pp, \(p\feTwoExpansionsMixedCROnePText{}\)), and supportive (\feTwoExpansionsSupportEffectPP{} pp, \(p\feTwoExpansionsSupportCROnePText{}\))—and none of the pairwise differences are significant (\(p\feTwoExpansionsMinPairwiseCROnePText{}\)).
This pattern matches Prediction~\ref{pred:margins}: stance is visible only after the post is expanded, so the decision to open the comment section should not differ across stances. As such, this evidence provides an additional sanity check on the successful randomization for our study, suggesting that differential ad delivery across experimental conditions is unlikely to explain our results.

\paragraph{Interactions}
Figure~\ref{fig:results_others}(c) indicates that the pooled ``any comments'' effect on reactions, comments, and shares is small and statistically insignificant (\feTwoInteractionsAnyEffectPP{} pp).
However, stance-specific estimates reveal that opposing comments substantially increase interactions by \feTwoInteractionsOpposeEffectPP{} pp (\(p\feTwoInteractionsOpposeCROnePText{}\)), a \feTwoInteractionsOpposeRelativeWholePct{} percent rise relative to the control mean of \feTwoInteractionsControlMeanPct{} percent.
For mixed comments, the coefficient is positive but statistically insignificant (\feTwoInteractionsMixedEffectPP{} pp), whereas for supportive comments it is negative and insignificant (\feTwoInteractionsSupportEffectPP{} pp). Differences between opposing and supportive conditions are highly significant (\(p\feTwoInteractionsOpposeVsSupportCROnePText{}\)), whereas differences involving the Mixed condition are smaller and only marginally significant (\(p\feTwoInteractionsMixedContrastsCROnePText{}\)).
This pattern suggests that antagonistic comments elicit greater visible participation from subsequent users.

\paragraph{Link Clicks}
Figure~\ref{fig:results_others}(d) shows that the presence of any comments raises the probability of clicking on the external link by \feTwoClicksAnyEffectPP{} pp (\(p\feTwoClicksAnyCROnePText{}\)), representing a \feTwoClicksAnyRelativeWholePct{} percent increase over the baseline mean of \feTwoClicksControlMeanPct{} percent.
Opposing comments again produce the largest effect (\feTwoClicksOpposeEffectPP{} pp, \(p\feTwoClicksOpposeCROnePText{}\); \feTwoClicksOpposeRelativePct{} percent), while the coefficients for Mixed (\feTwoClicksMixedEffectPP{} pp) and supportive (\feTwoClicksSupportEffectPP{} pp) comments are smaller and not significant.
The difference between opposing and supportive arms is statistically significant (\(p\feTwoClicksOpposeVsSupportCROnePText{}\)), whereas comparisons involving the Mixed condition are not.
These results indicate that opposing comments are particularly effective in increasing click-through activity, a key metric for organizations seeking to drive website traffic.

\paragraph{Valence of Interactions}
The additional engagement generated by opposing comments appears to be predominantly supportive in tone. Comparing effects on all interactions with effects on \emph{supportive interactions}, defined as positive reactions (\emph{likes}, \emph{loves}, \emph{cares}), comments expressing agreement with the organization's message, and shares, opposing comments raise all interactions by 0.0082 pp and supportive interactions by 0.0075 pp (both \(p<0.05\)), as shown in Appendix Figure~\ref{fig:results_interactions_valence} and Appendix Table~\ref{tab:interactions_valence}.\footnote{These estimates are from a specification using data retrieved directly from the posts, without user-level covariates, and with heteroskedasticity-robust standard errors. Results for all interactions are therefore very similar but not identical to those reported above. Minor discrepancies may arise if users deleted their comments or removed reactions after the initial data extraction.}
Mixed and supportive comment sections have no statistically significant effects, and effects on non-supportive or ambiguous interactions are small and imprecisely estimated. One possibility is that visible disagreement increases the likelihood that users who support the original message choose to express their agreement. At the same time, the initial presence of negative comments may reduce the propensity of other users holding similar views to post additional negative responses. Either way, these patterns suggest that opposing comments alter the composition of subsequent engagement.

\paragraph{Discussion}
Across all outcomes, the presence of a comment section modestly increases on-platform user engagement. However, much of this overall effect reflects greater attention to the post itself, as captured by post expansions, and does not vary by the stance of existing comments. In contrast, comment stance shapes the composition of subsequent engagement. These patterns are consistent with Prediction~\ref{pred:margins}: the presence of comments increases post expansion regardless of stance, and stance drives subsequent engagement.  Opposing comments consistently generate higher rates of interactions and link clicks relative to the No Comments condition, whereas supportive comments do not significantly outperform the control. Differences between the opposing and supportive conditions are substantial, while differences across the remaining conditions are generally small. These engagement gains also lower unit advertising costs: with respect to the control, spend per interaction falls by 50 percent under opposing comments and spend per link click by \feTwoClickCostOpposeDeclinePct{} percent (Appendix Table~\ref{tab:costs}). Taken together, these findings suggest that opposing discourse amplifies participation and interest more than supportive commentary, pointing to a potential trade-off between engagement amplification and polarization in online comment sections. This trade-off echoes findings on other platform features: \citet{germano2026ranking}, for instance, document a similar tension in the context of engagement-maximizing ranking algorithms.  In Section~\ref{sec:survey}, we further examine how these engagement patterns affect attitudes and off-platform behavior in a survey experiment.

\subsubsection{Heterogeneous Effects by Area Ideology}

We next examine whether these effects vary with the prevailing political ideology of the area, distinguishing between \emph{Blue} (mostly liberal), \emph{Swing} (mixed), and \emph{Red} (mostly conservative) ZIP codes. Figure~\ref{fig:hte_political} and Appendix Tables~\ref{tab:blue_zip}--\ref{tab:red_zip} report the estimates.

\begin{figure}[h]
\centering

\caption{Heterogeneous Treatment Effects by Area Ideology\\
\textit{Outcomes are expressed in \% of total reach}}
\label{fig:hte_political}

\vspace{0.5em}

\begin{adjustwidth}{-.1cm}{-.1cm}

\begin{subfigure}[b]{0.48\textwidth}
    \centering
    \caption{All Engagement}
    \includegraphics[width=\textwidth]{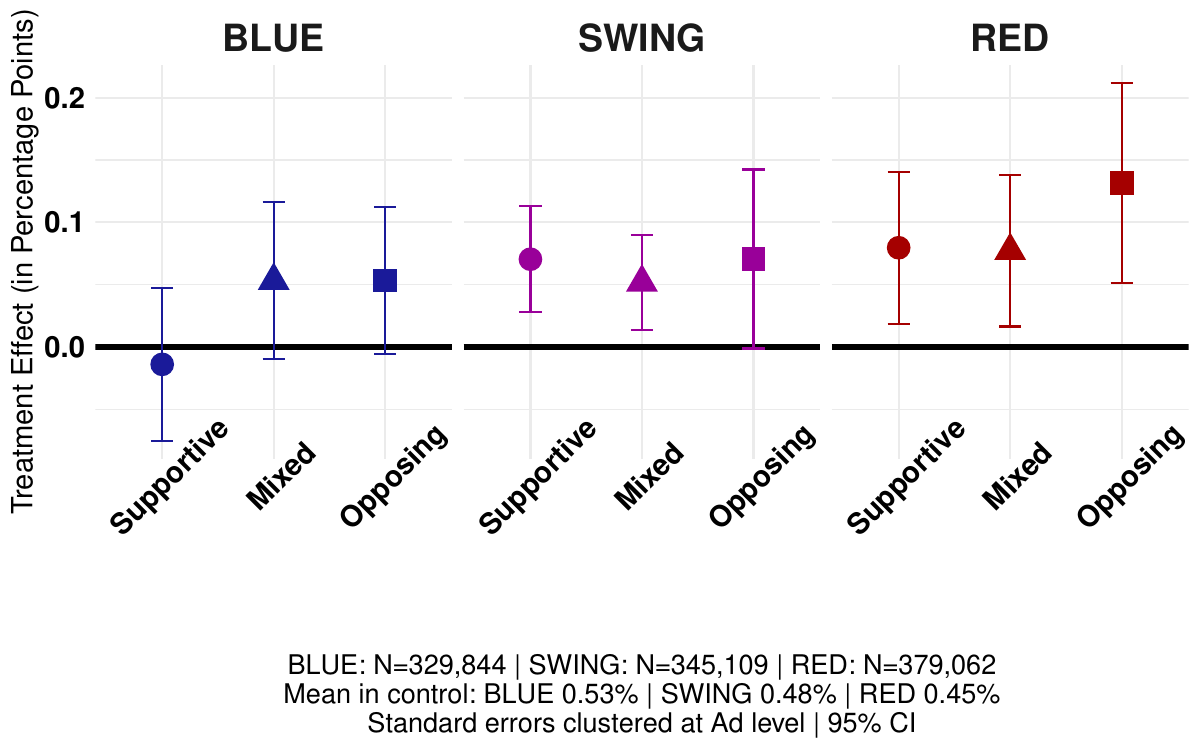}
\end{subfigure}
\hfill
\begin{subfigure}[b]{0.48\textwidth}
    \centering
    \caption{Post Expansions}
    \includegraphics[width=\textwidth]{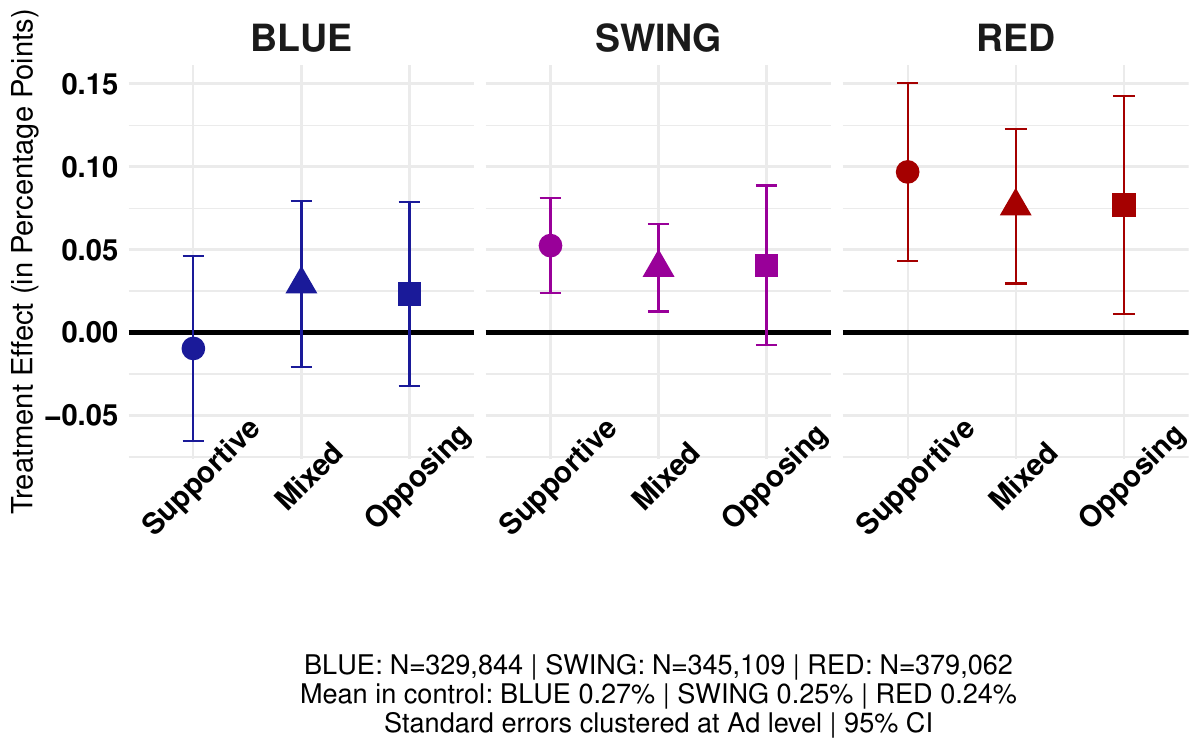}
\end{subfigure}

\vspace{0.4cm}

\begin{subfigure}[b]{0.48\textwidth}
    \centering
    \caption{Interactions (Reactions, Comments, Shares)}
    \includegraphics[width=\textwidth]{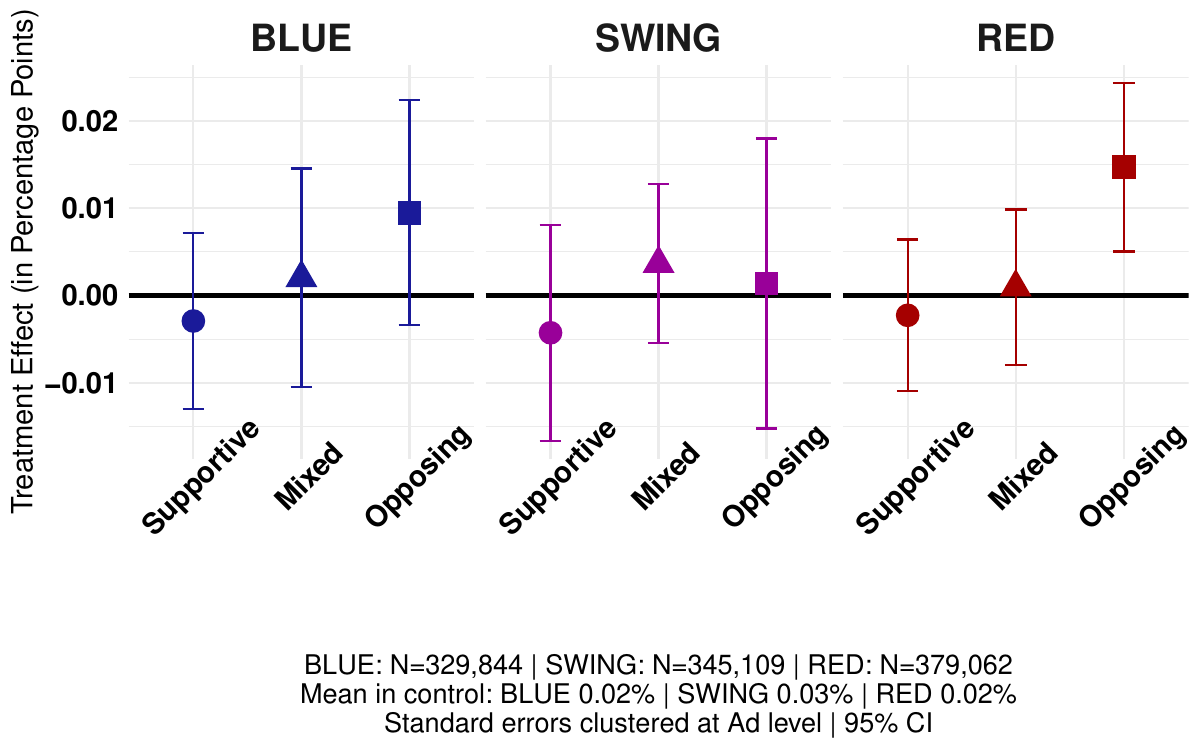}
\end{subfigure}
\hfill
\begin{subfigure}[b]{0.48\textwidth}
    \centering
    \caption{Link Clicks}
    \includegraphics[width=\textwidth]{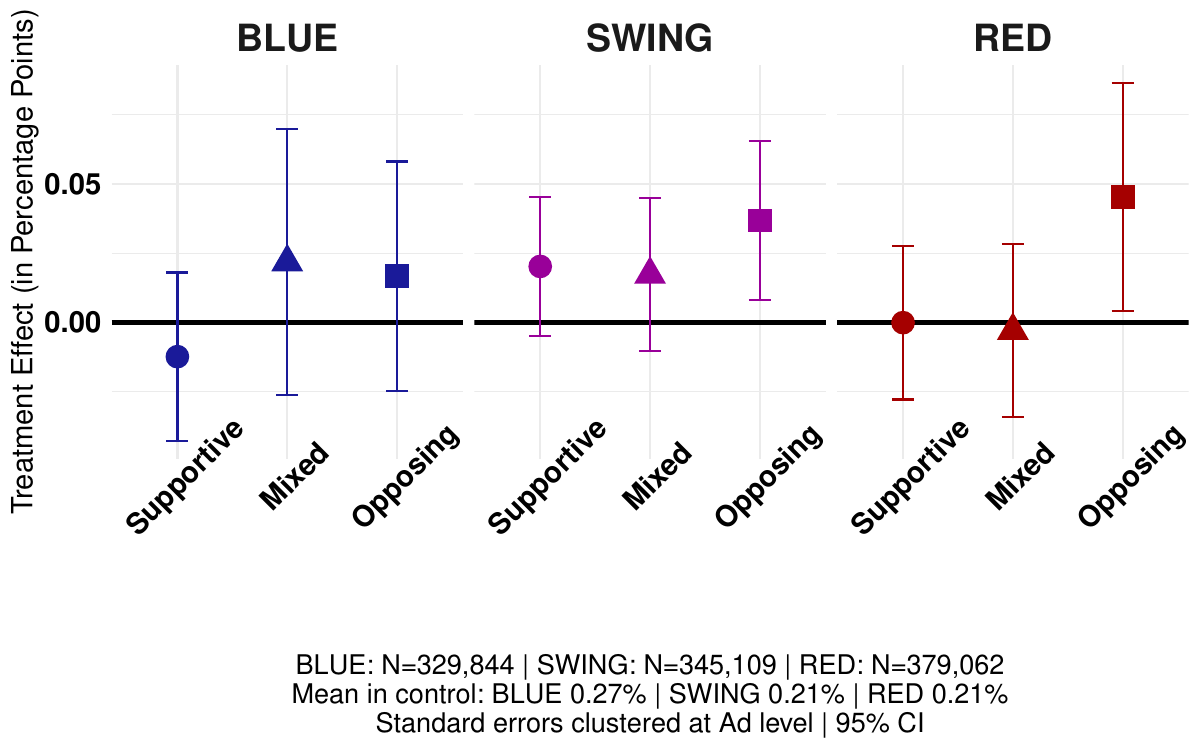}
\end{subfigure}

\end{adjustwidth}

\vspace{0.3em}

\begin{minipage}{\textwidth}
\scriptsize
\noindent \textit{Notes:} This figure reports heterogeneous treatment effects across areas with different ideological compositions. Panels show treatment effects on all engagement, post expansions, interactions, and unique link clicks, each expressed in percentage points of total reach. BLUE, SWING, and RED areas are defined using Republican vote share cutoffs. Estimates are from regression models estimated separately by area type, with the No Comments arm as the omitted category. Standard errors are clustered at the advertisement level, and vertical lines represent 95\% confidence intervals.
\end{minipage}

\end{figure}

The effect of the presence of a comment section rises with the conservativeness of the area. Any comments raise total engagement by \feTwoBlueAllAnyEffectPP{} pp in Blue areas, \feTwoSwingAllAnyEffectPP{} pp in Swing areas (\(p\feTwoSwingAllAnyCROnePText{}\)), and \feTwoRedAllAnyEffectPP{} pp in Red areas (\(p\feTwoRedAllAnyCROnePText{}\)), or \feTwoBlueAllAnyRelativeWholePct{}, \feTwoSwingAllAnyRelativeWholePct{}, and \feTwoRedAllAnyRelativeWholePct{} percent of the respective control means. The estimate for Blue areas is statistically insignificant. In Swing and Red areas, most of this increase consists of post expansions, our proxy for attention, which respond similarly to all three stances.  Baseline engagement runs in the opposite direction: in the control condition, overall engagement averages \feTwoBlueAllControlCoarseMeanPct{} percent in Blue areas, \feTwoSwingAllControlCoarseMeanPct{} percent in Swing areas, and \feTwoRedAllControlCoarseMeanPct{} percent in Red areas, consistent with greater alignment between the campaign's message and prevailing ideology in more progressive areas.\footnote{This pattern is consistent with prior literature (e.g., \citealt{song2024} in the racial justice context), which finds that individuals are more likely to engage with social media content that aligns with their pre-existing attitudes.} The comment section thus draws attention in areas where the post itself draws less attention. By making visible that the post has generated discussion, it may increase the salience of content that users in conservative areas would otherwise overlook.

Differences across comment stances are most pronounced in Red areas. There, opposing comments raise total engagement by \feTwoRedAllOpposeEffectCoarsePP{} pp (\(p\feTwoRedAllOpposeCROnePText{}\)), interactions by \feTwoRedInteractionsOpposeEffectPP{} pp (\(p\feTwoRedInteractionsOpposeCROnePText{}\)), about \feTwoRedInteractionsOpposeLevelRatio{} times the control rate, and link clicks by \feTwoRedClicksOpposeEffectPP{} pp (\(p\feTwoRedClicksOpposeCROnePText{}\)), whereas mixed and supportive comments have no detectable effect on interactions or link clicks despite raising post expansions as much as opposing comments do. The differences between the Opposing condition and the other conditions are statistically significant for interactions (\(p\feTwoRedInteractionsOpposeContrastsCROnePText{}\)) and link clicks (\(p\feTwoRedClicksOpposeContrastsCROnePText{}\)). In Swing areas, all three stances raise total engagement by \feTwoSwingAllMinEffectPP{}--\feTwoSwingAllMaxEffectPP{} pp and opposing comments raise link clicks by \feTwoSwingClicksOpposeEffectPP{} pp (\(p\feTwoSwingClicksOpposeCROnePText{}\)), but the stances do not differ significantly from one another. In Blue areas, the estimates are small and imprecise, and only the difference in interactions between opposing and supportive comments is statistically significant (\(p\feTwoBlueInteractionsOpposeVsSupportCROnePText{}\)).

These patterns suggest that the comment counter raises attention, most where baseline attention to the post is lowest, and once the section is opened, the opposing section is the most salient. Identity congruence may reinforce the latter effect in conservative areas, where an opposing comment both contrasts with the post and aligns with the viewer's own stance, as discussed in Appendix~\ref{app:discussion}. Because we observe ad-level aggregates, however, we cannot link post expansions to subsequent actions at the user level, and so cannot tell whether the larger stance effects in Red areas arise because more users open the comment section there or because those who open it respond differently to what they see.

\subsubsection{{Robustness and Placebo Tests}\label{sec:robust}}

We conduct several robustness checks to assess the sensitivity of our results. First, we re-estimate all main models using alternative combinations of control variables and fixed effects (Appendix Tables \ref{tab:robustness_anycomment_controls} and \ref{tab:robustness_stance_controls}), and examine robustness to alternative estimators by employing a logit specification (Appendix Table \ref{tab:logit}). Second, we compute engagement rates directly at the ad level, without the synthetic data construction, and report the corresponding estimates together with heteroskedasticity-robust t-tests of equality across experimental conditions (Appendix Table \ref{tab:outcomes_model_free}). The estimates remain robust across all these exercises.

We also check robustness to the inference method, reporting wild cluster bootstrap $p$-values for all estimates  \citep{cameron2008bootstrap}. Under the wild cluster bootstrap, the effect of any comments on overall engagement and the differences between opposing and supportive comments in interactions and link clicks remain significant at the 5 percent level. Other estimates become less precise, particularly within areas, where only \feTwoBlueClustersN{} clusters are available (Appendix Tables~\ref{main_combined} and \ref{tab:blue_zip}--\ref{tab:red_zip}). Clustering at the ZIP-code-set level leaves our main conclusions unchanged (Appendix Table~\ref{tab:main_zip_cluster}).

We conduct three additional exercises. We first perform a ZIP code exclusion sensitivity analysis, dropping one of the \feTwoStrataN{} ZIP-code sets at a time and re-estimating the main specifications; the resulting coefficients remain stable, indicating that the findings are not driven by particular geographic areas (Appendix Figures \ref{all_Engagement_one_zipcode_set_out}-\ref{link_clicks_one_zipcode_set_out}). We next repeat this exercise at the level of the stimulus, dropping one stimulus (i.e., \feTwoAdsPerPostN{} ads and more than \feTwoPostReachLowerBoundN{} people) at a time, so that each estimate excludes one realization of the comment section. The coefficients are again stable, indicating that the results are not driven by a particular stimulus or comment section (Appendix Figures \ref{all_Engagement_one_post_out}--\ref{link_clicks_one_post_out}).  Moreover, we implement randomization-inference (permutation) tests based on repeated placebo reassignments of treatment status \citep{young2019channeling}. For all statistically significant effects in the main analysis, the observed estimates lie consistently in the tails of the corresponding placebo distributions, with one-sided permutation p-values below 0.02, suggesting that these effects are unlikely to be generated by chance under the original randomization scheme (Appendix Figures \ref{all_engagement_permutation}--\ref{link_clicks_permutation}).

Furthermore, we test the impact of comments on landing-page views, an alternative measure of downstream behavior beyond link clicks. Appendix Table~\ref{tab:pageviews} suggests that the presence of any comments increases off-platform engagement by \feTwoPageViewsAnyEffectPP{} percentage points (\(p\feTwoPageViewsAnyCROnePText{}\)), corresponding to a \feTwoPageViewsAnyRelativePct{} percent increase relative to the control mean of \feTwoPageViewsControlMeanPct{} percent.
Opposing comments yield the largest increase (\feTwoPageViewsOpposeEffectPP{} percentage points, \(p\feTwoPageViewsOpposeCROnePText{}\); \feTwoPageViewsOpposeRelativePct{} percent), followed by mixed (\feTwoPageViewsMixedEffectPP{} percentage points) and supportive comments (\feTwoPageViewsSupportEffectPP{} percentage points), both of which are not statistically significant. However, we interpret these findings with caution since landing-page views are inferred via pixel-based tracking and are prone to measurement error.

One potential concern is that the higher engagement driven by opposing comments reflects not their stance per se, but other textual features correlated with stance. In particular, toxicity has been shown to drive up engagement \citep{beknazar2025toxic}. To address this, we classify comments shown in each experimental condition using the Google Perspective API and find that average toxicity levels across the three treatment conditions are similar at around \feTwoDisplayedToxicityMean{}, so the comments in Opposing are no more toxic than those in other conditions. As a further robustness check, we exclude posts containing comments that exceed the Perspective API's recommended toxicity threshold of 0.7.\footnote{Perspective API developer documentation, ``About the API: Score,'' \url{https://developers.perspectiveapi.com/s/about-the-api-score}.} Appendix Table \ref{tab:toxicity}
shows that our results remain robust to this restriction, confirming that the engagement effect of opposing comments is not merely an artifact of toxic content. More broadly, toxicity and stance are conceptually distinct: toxicity is a fixed characteristic of the content itself, whereas the effect of stance is relative to the post. Moreover, identity congruence between the comment and the viewer could play a role. The heterogeneity we document is therefore unlikely to be driven by toxicity or other textual features whose effects do not vary systematically with location ideology.

\section{\makebox[\textwidth][l]{The Impact of Comments on Attitudes and Off-Platform Outcomes}}
\label{sec:survey}

The field experiment provides causal evidence that comment sections influence engagement on the platform. However, it does not allow us to elicit users' beliefs or attitudes. Nor can we directly test the psychological mechanisms, such as emotional reactions driven by ideological congruence between users and content.

To address these limitations, we conduct a complementary survey experiment. This design preserves key features of the organic online environment, including real social media posts, authentic comments collected from Facebook users, and behavioral incentives, while allowing us to measure beliefs, attitudes, and incentivized behaviors at the individual level. By embedding experimentally manipulated comment sections inside a controlled survey setting, we isolate how the presence and stance of comments affect attention, perceptions, attitudes, and incentivized donations.

Comments can affect attitudes and behavior through several channels. Our framework in Section~\ref{sec:conceptual_framework} predicts that, because comments shift attention away from the post, opposing comments move the audience's attitudes toward the comments and away from the post, while supportive comments, which restate the post, have little effect (Prediction~\ref{pred:attitudes}). The framework is reduced form, and attitudes and behavior could respond through several other channels. Comments could persuade, if users update on the arguments or information they carry, or backfire, if opposing views trigger reactance \citep{bail2018exposure}. They could shift perceived norms, if users read the comment section as evidence of what others believe \citep{bursztyn2020misperceived}, or signal the post's credibility \citep{muchnik2013social}. They may also have little effect, since political and social attitudes are difficult to change \citep{haaland2023beliefs}. Whether comments affect attitudes and behavior in practice is therefore an empirical question.

\subsection{Design}

\subsubsection{Recruitment and Sample}

We recruited approximately 5,000 participants from Prolific. Eligibility criteria required participants to be U.S. residents aged 18–64, to have completed at least 20 previous Prolific submissions, and to maintain an approval rate of at least 95 percent. Pilot participants were excluded. To enable the study of heterogeneous effects, recruitment was stratified to achieve approximate balance across self-identified political ideology (conservative, moderate, liberal) and gender.\footnote{Participants were informed that the survey would take approximately 10 minutes and that they could earn a bonus of up to \$1 based on performance in incentivized belief elicitation, in addition to being entered into a \$100 lottery.} From an initial sample of \attritionDropInvalidConditionsPREATTRITION{} eligible respondents, we excluded those who failed the attention check, leaving \attritionDropAttentionFail{} participants. Our final sample consists of the \attritionFinalN{} respondents who completed the survey. 

Appendix Table~\ref{tab:survey_sample_vs_us} compares the composition of the analysis sample with the U.S. adult population. Our sample is gender-balanced but more educated, and includes more Republicans and Democrats and fewer independents than the U.S. average. As shown in Appendix Table~\ref{tab:survey_response_rates}, almost everyone who was randomized completed the survey (\responseRateControl{} in the control arm, \responseRateSupport{} under supportive comments, and \responseRateOppose{} under opposing comments), and completion rates are balanced across experimental arms ($p=\responseRateFtestP{}$).\footnote{Including respondents who failed the attention check does not qualitatively change the results.} The arms differ somewhat in size (\balanceNControl{}, \balanceNSupport{}, and \balanceNOppose{} respondents in the control, supportive, and opposing arms). These differences are present at assignment, and baseline characteristics are balanced across arms in the analysis sample as shown in Appendix Table~\ref{tab:survey_balance}.

\subsubsection{Survey Structure}
The survey proceeded in four stages. Participants first completed baseline measures capturing political ideology, racial attitudes, beliefs about others' views, social media usage, and demographics. Next, the visibility and stance of the comment section were randomized at the participant level. Participants were then shown three social media posts promoting racial justice, mirroring those used in our Facebook campaigns, with the comment section displayed according to their assigned condition. Finally, participants completed post-exposure measures capturing attention, beliefs, attitudes, behavior, emotional responses, and open-ended questions. To minimize experimenter demand effects and preserve ecological validity, participants were instructed to react to each post as if they encountered it naturally on their Facebook feed.

\subsubsection{Treatment Conditions}

Participants were randomly assigned with equal probability to one of three conditions. In the \textit{No Comments} control condition, posts were displayed without any visible comment section. In the \textit{Supportive Comments} condition, each post displayed two supportive comments. In the \textit{Opposing Comments} condition, each post displayed two opposing comments.\footnote{As pre-specified in the Pre-Analysis Plan, we omitted the \textit{Mixed} condition from the survey experiment to increase statistical power.}

The posts themselves were identical across conditions. The comments were real comments collected from Facebook users during the comment generation campaign described in Section~\ref{sec:generation}. To protect user privacy, names and profile images were replaced with neutral placeholders. In the comment conditions, we displayed two comments of the same stance, one using a male name and one using a female name, in order to avoid confounding stance with perceived commenter gender. Within a treatment condition, the stance of comments was consistent across all posts shown to a participant. Thus, the only dimension varying across individuals was the presence and stance of the comment section. Appendix Figure \ref{fig:stimuli_survey} provides an example. Each participant viewed three posts in random order, one each on environmental justice, criminal justice and police reform, and education, all under the same condition.

\subsection{Outcomes}

We measure attention, beliefs, attitudes, behavior, and emotional responses.\footnote{For survey instruments, see Appendix \ref{app:SurveyInstrument}.} When multiple measures capture a common construct, we aggregate them into standardized indices using inverse-covariance weighting following \citet{anderson2008multiple}. All component variables are re-oriented so that higher values reflect greater alignment with the organization or post. Results for pre-specified secondary outcomes are reported in Appendix \ref{app:SecOutcomes} for completeness.

\paragraph{Time Spent}

We measure the total time each participant spent viewing the three posts,  using both time (in seconds) and $\log(1 + \text{time})$ to reduce the influence of outliers.

\paragraph{Beliefs about Others and Commenters}

We measure beliefs about others using an incentivized question in which participants estimated the share of participants in the study who agreed with a conservative statement on racial inequality.  Responses were incentivized using a quadratic scoring rule, with potential earnings of up to \$1 based on accuracy relative to the share observed in the study. This measure captures whether comments shift perceptions of others' views, which we interpret as a proxy for perceived social norms. 

\paragraph{Attitudes toward Racial Issues and the Organization}
We construct a Post-Exposure Attitudes Index aggregating measures of attitudes toward the organization, perceived importance of the racial justice issues covered in the posts, willingness to discuss political issues separately with progressives and with conservatives, and attitudes toward Black Americans. All components are re-oriented so that higher values reflect greater alignment with the organization and are combined using inverse-covariance weighting.

\paragraph{Behavioral Outcomes}

We measure two costly behavioral outcomes. Participants were given the opportunity to subscribe to the organization's mailing list by voluntarily providing their email address. Separately, they were enrolled in a lottery and asked whether they would donate any winnings to Color of Change, and if so, how much. If selected as winners, the chosen amount was deducted from their prize and transferred directly to the organization. We analyze both the extensive margin, defined as an indicator for making a positive donation, and the intensive margin, defined as the committed donation amount.

\paragraph{Curiosity, Reflection, Anger, and Annoyance}

As a measure of curiosity and reflection, we construct an index combining participants' reported interest in seeing additional comments and their assessments of how thought-provoking they found the post and the comments. Because perception of comments applies only to participants in the treatment conditions, we construct this index for these two groups and compare them. For these participants, we also construct an index combining Likert-scale measures of anger and annoyance triggered by the comments.

\subsection{Empirical Strategy}

Let $Y_i$ denote an outcome for participant $i$. We estimate:

\begin{equation}
Y_i = \alpha + \beta_1 \mathit{Supportive}_i + \beta_2 \mathit{Opposing}_i + \tau' X_i + \varepsilon_i,
\label{eq:main_survey}
\end{equation}
where $Y_i$ is the outcome for respondent $i$, $\alpha$ is a constant, and $\mathit{Supportive}_i$ and $\mathit{Opposing}_i$ are indicators for assignment to the respective comment conditions, and the omitted category is No Comments. The vector $X_i$ contains pre-specified covariates: baseline racial attitudes, political ideology, beliefs about others' views, and demographics. Standard errors are robust to heteroskedasticity. 

The coefficients $\beta_1$ and $\beta_2$ identify the causal effect of exposure to supportive and opposing comments, respectively, relative to no comments. The contrast $\beta_2 - \beta_1$ captures the differential effect of opposing versus supportive comments. 

\subsubsection{Descriptive Statistics and Balance Checks}

Appendix Figure \ref{fig:social_read_comments} (Panel a) reports the distribution of responses to the question of how often the respondent reads or checks comments on social media. Approximately \surveyReadCommentsSometimesPct{} percent of survey participants report that they sometimes, often, or very often read comments on social media, and more than 50 percent report doing so often or very often. These patterns are even more pronounced  among respondents who report higher social media use (Panel b). This evidence confirms that comment sections are a salient feature of online content consumption and motivates our focus on their effects.

Appendix Table~\ref{tab:survey_balance} reports baseline characteristics by experimental condition and balance tests. Across all observable covariates, the treatment groups are well-balanced relative to the control group. Differences in means are small in magnitude, and the vast majority of pairwise t-tests fail to reject equality at conventional significance levels, with only a few imbalances in ethnicity, race, and beliefs about others' views. Importantly, joint orthogonality F-tests fail to reject equality across all pairwise comparisons, confirming that treatment assignment is orthogonal to observed baseline characteristics. To account for these marginal imbalances and improve precision, as specified in the pre-analysis plan, we include in the main specifications baseline controls for age, gender, education, race, ethnicity, party affiliation, political views, racial attitudes, and beliefs about others’ views. The results are qualitatively unchanged across alternative control specifications.

\subsection{Results}\label{sec:SurveyResults}

\paragraph{Time Spent}
As shown in Appendix Figure~\ref{fig:survey_attention} and Appendix Table~\ref{tab:survey_arm_engagement}, across the three posts shown to  respondents, both supportive and opposing comments significantly increase viewing time by approximately \surveyTimeSpentMinSec{}--\surveyTimeSpentMaxSec{} seconds, a \surveyTimeSpentMinPct{}--\surveyTimeSpentMaxPct{} percent increase relative to the No Comments condition. 
These results indicate that people read comments, and the presence of a comment section makes users spend more time on the post regardless of stance. The results are similar when using log viewing time (Appendix Figure \ref{fig:survey_attention_log}). 

\paragraph{Attitudes}
Figure~\ref{fig:survey_attitudes} shows that opposing comments shift attitudes in a less progressive direction relative to the control condition by approximately \surveyAttitudeOpposeSD{} standard deviations, in line with Prediction~\ref{pred:attitudes}. The effects are concentrated in measures of NGO favorability and the perceived importance of education equity and environmental justice in the context of racial issues, although the coefficient  for the latter is only marginally significant. Opposing comments also significantly reduce the willingness to discuss political issues with conservatives.\footnote{The overall effect remains negative and statistically significant when we exclude the cross-partisan interaction measure from the attitude index. This pre-specified measure can be interpreted as capturing a reduction in affective polarization, though it maps less directly onto progressiveness of attitudes than the other components of the index.} In contrast, supportive comments have comparatively small and statistically insignificant effects, and the two arms differ significantly ($p=\tabArmCEightaOnePostExposureAttitudesPairwiseP{}$). Thus, exposure to opposing comments shifts attitudes away from the organization’s position, with potential implications for support of both the organization and the cause it represents. 

\paragraph{Off-platform Behavior}
As shown in Figure~\ref{fig:survey_behavior}, opposing comments significantly reduce the likelihood of making an incentivized donation relative to the control condition by \surveyDonationOpposePP{} percentage points, corresponding to a \surveyDonationOpposePctDecline{} percent decline. They also reduce the amount donated by \$\surveyDonationAmountUSD{}, or about \surveyDonationAmountPct{} percent relative to the control group (Appendix Figure~\ref{fig:donations}).
Supportive comments have no statistically significant effect. 
Neither treatment meaningfully affects newsletter sign-up. 
Thus, while opposing comments increase attention, they reduce financial support to the organization. Appendix Table~\ref{tab:survey_arm_main} reports the corresponding regression estimates for the attitude index and the two behavioral outcomes.

\begin{figure}[t]
    \centering
        \caption{Post-Exposure Racial Attitudes}
    \label{fig:survey_attitudes}
    \includegraphics[width=0.8\textwidth]{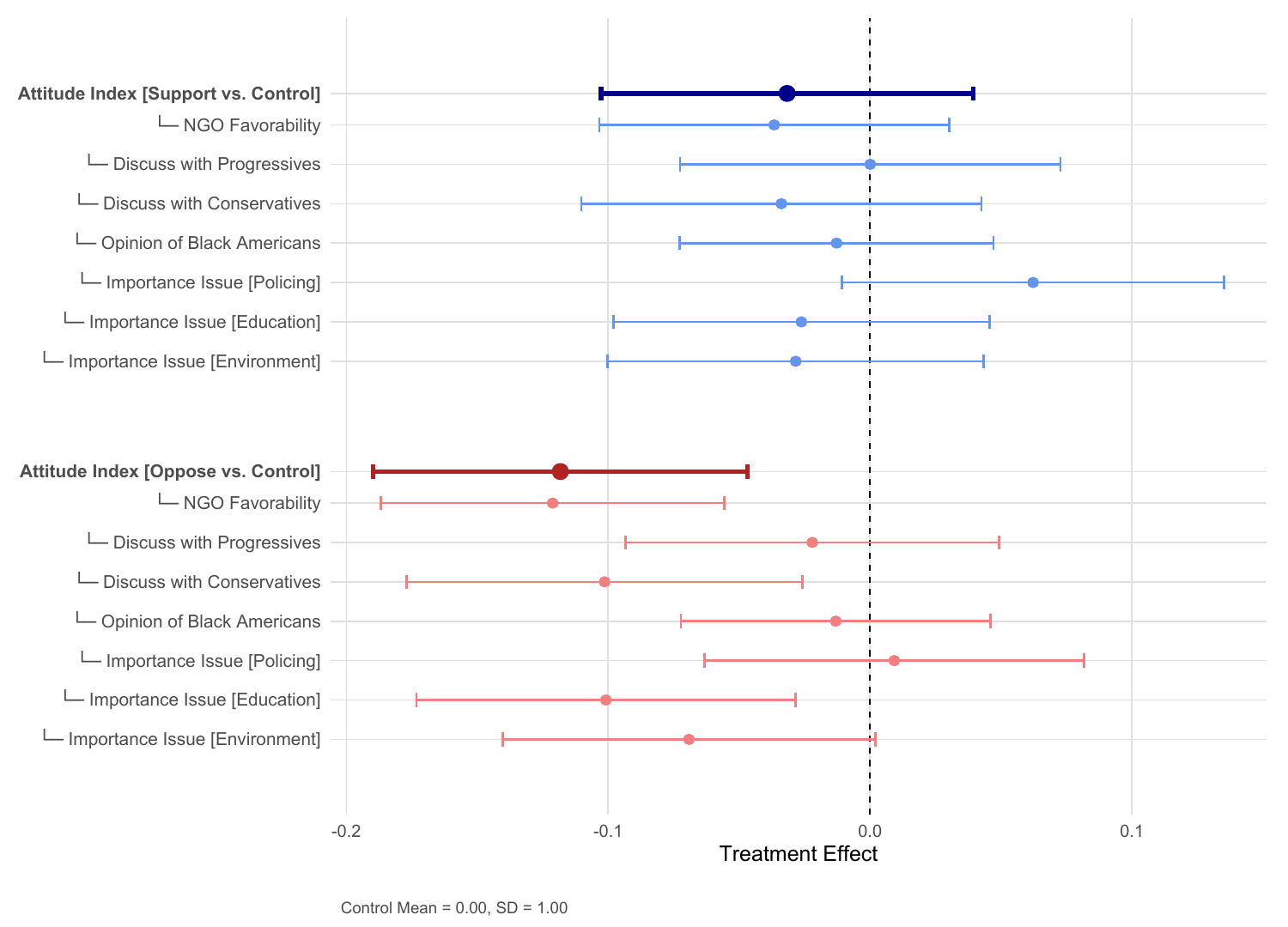}
\begin{minipage}{\textwidth}
\scriptsize
\noindent \textit{Notes:} This figure reports treatment effects of supportive and opposing comments, relative to the control, on post-exposure racial attitudes in the survey experiment. The main outcome is an attitude index, with higher values indicating greater alignment with the organization’s position. The figure also reports treatment effects on the index components, including NGO favorability, willingness to discuss with progressives and conservatives, opinions of Black or African Americans, and the perceived importance of policing, education, and environmental issues. Outcomes are standardized to mean zero and standard deviation one in the control group. Horizontal lines represent 95\% confidence intervals.
\end{minipage}

\end{figure}

\begin{figure}[t]
    \centering
    \caption{Donations and Newsletter Sign-Up (Yes/No)}
    \label{fig:survey_behavior}
    \includegraphics[width=0.8\textwidth]{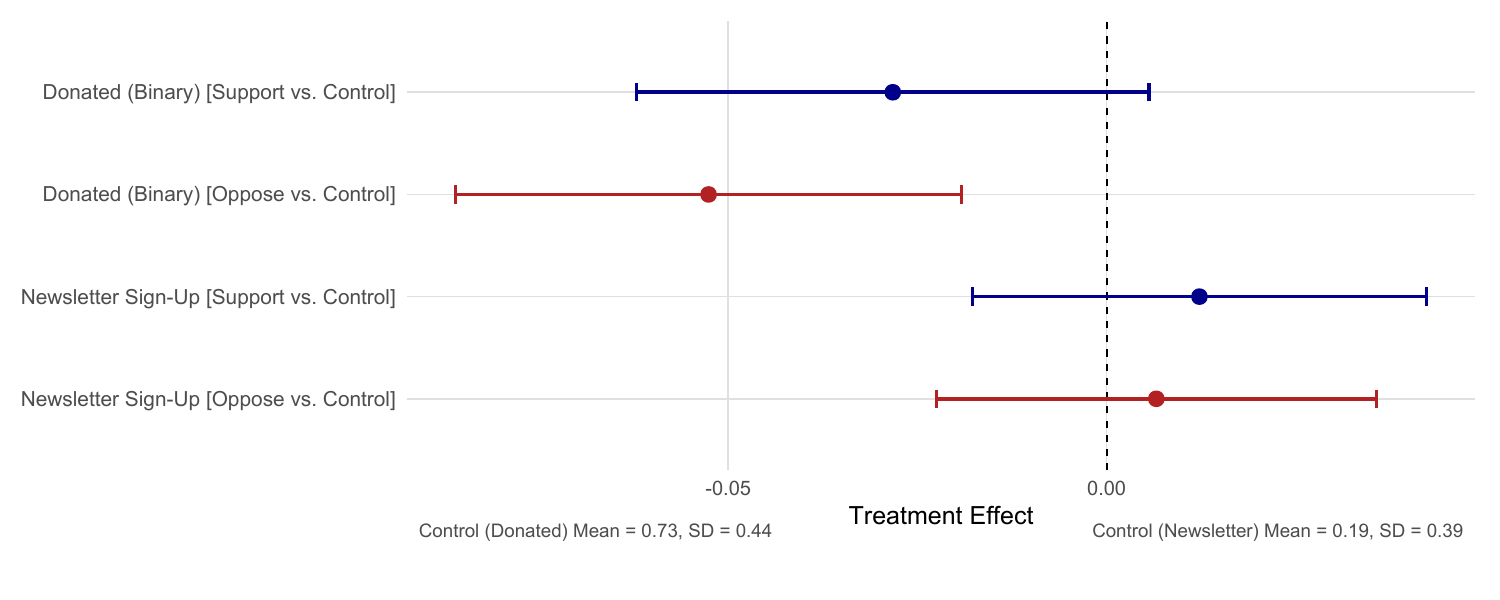}
\begin{minipage}{\textwidth}
\scriptsize
\noindent \textit{Notes:} This figure reports treatment effects of supportive and opposing comments, relative to the control, on off-platform behavioral outcomes in the survey experiment. Outcomes include a binary indicator for newsletter sign-up and a binary indicator for making a positive donation in the incentivized donation task. The control-group mean is \surveyDonationControlMean{} for donation and \surveyNewsletterControlMean{} for newsletter sign-up. Horizontal lines represent 95\% confidence intervals.
\end{minipage}

\end{figure}

\paragraph{Heterogeneity}
We further examine heterogeneity in behavioral and attitudinal outcomes by pre-specified moderators including baseline racial attitudes, party affiliation, gender, and age (Appendix Figures \ref{fig:hte_time_spent}-\ref{fig:hte_newsletter}). Overall, we find limited evidence of heterogeneous effects across these dimensions. Estimated treatment effects are broadly similar across groups. In particular, opposing comments reduce willingness to donate for both Democrats and Republicans by \surveyDonationDemocratPctDecline{} percent and \surveyDonationRepublicanPctDecline{} percent, respectively, relative to the donation rate in the control group. 

In contrast, there is substantial heterogeneity in emotional responses. Appendix Figure~\ref{fig:survey_emotions} and Appendix Table~\ref{tab:survey_arm_engagement} show that opposing comments generate higher levels of anger and annoyance, particularly when visible comments conflict with respondents' party affiliation. Relative to supportive comments, opposing comments increase reported anger and annoyance by more than one standard deviation among Democrats, have a smaller positive effect among Independents (about \surveyEmotionalIndependentSD{} standard deviations), and reduce negative emotional responses among Republicans by roughly \surveyEmotionalRepublicanSD{} standard deviations.  A closely related pattern appears for curiosity and reflection (Appendix Figure~\ref{fig:cognitive_responses_primary_overall_by_party_affiliation}): relative to supportive comments, opposing comments reduce an index combining curiosity about additional comments and how thought-provoking respondents find the post and its comments by roughly \surveyCognitiveDemocratSD{} standard deviations among Democrats and \surveyCognitiveIndependentSD{} standard deviations among Independents, with no detectable effect among Republicans. These findings indicate that comment stance shapes not only engagement, but also emotional and cognitive reactions, and that identity congruence plays an important role in determining how users respond to visible disagreement.

To provide suggestive evidence on mechanisms behind the attitude and donation effects, we examine beliefs about others by party affiliation. Appendix Figure \ref{fig:survey_social_norms} reports responses to an incentivized belief-elicitation question using a quadratic scoring rule, asking participants what share of survey respondents they believe agreed or strongly agreed with the statement ``Black people could be just as well off as white people if they would only try harder''. The outcome is reverse-coded, so that higher values indicate more progressive perceived norms. Among Democrats, opposing comments reduce the perceived prevalence of progressive views, consistent with social norms as a potential channel. At the same time, the shift is smaller and insignificant among Republicans and absent among Independents, suggesting that norm updating alone cannot fully account for the effects on attitudes and donations. Other channels, such as persuasion or learning about the organization's legitimacy, may also play a role.

\paragraph{Discussion}

The survey evidence reveals a clear tension between engagement and downstream outcomes. Comment sections—especially those featuring opposing remarks—increase engagement but can shift attitudes and reduce financial support, highlighting a trade-off between amplifying engagement and preserving support for the organization and the underlying message. 

The experiment on Facebook shows that comment sections on progressive posts are disproportionately populated by conservative voices, and both their presence and stance shape subsequent on-platform behavior. Opposing comments increase clicks and interactions in the field setting, particularly in conservative areas. However, the survey evidence shows that exposure to opposing comments shifts attitudes in a less progressive direction and reduces incentivized donations.

To quantify this trade-off, we conduct a back-of-the-envelope cost-benefit analysis for fundraising campaigns in Appendix~\ref{sec:cost-benefit}. Under a benchmark calibration that combines the observed increase in click-through rates with the estimated decline in donation propensity, abstracting from changes in delivery efficiency or audience composition, opposing comments can still raise expected donations. However, this result relies on the assumption that the additional users induced to click are no less likely to convert than baseline users. Our evidence suggests that this assumption may be too strong, since opposing comments attract relatively more traffic from less progressive areas, making a deterioration in traffic quality plausible. 
While this exercise should be interpreted with caution given its simplifying assumptions, it highlights that the net value of tolerating opposing comments depends not only on the increase in traffic they generate, but also on the quality of that traffic.

\section{Conclusion}

This paper provides causal evidence that comment sections shape both on-platform engagement and downstream attitudes and behavior. In a large-scale Facebook field experiment conducted in collaboration with a leading racial justice organization, we show that the presence of a comment section on progressive posts increases engagement, and that comment stance affects how users respond to content. Opposing comments, in particular, amplify interactions (e.g., comments and reactions) and link clicks relative to supportive comments. Because comment sections are often populated by the vocal few, these findings imply that the visible opinions of a few users can meaningfully shape the on-platform behavior of many others.

Our complementary survey experiment reveals that higher visible engagement does not necessarily imply greater support. While opposing comments increase attention, they also shift attitudes to be less progressive and reduce incentivized donations to the organization. Exposure to counter-attitudinal comments generates anger and annoyance among identity-incongruent users, underscoring that visible disagreement affects not only engagement behavior but also emotional responses and evaluations of the organization. Together, the evidence points to a tension between short-run on-platform engagement gains and potential off-platform longer-run costs.

These results speak to broader concerns about online discourse. Comment sections can amplify divisive  narratives in ways that may not reflect the broader audience, potentially complicating the ideal of social media as a deliberative public sphere. At the same time, our findings do not imply that removing comment sections is a straightforward solution: comments increase attention and facilitate participation, and in some contexts may stimulate supportive expression.
The challenge lies in how visible discourse is structured and moderated.

For platforms and content producers, moderation involves clear trade-offs. Tolerating contentious or opposing comments can boost engagement and reduce advertising costs, but may undermine brand safety and shift attitudes away from the content producer's objectives. Importantly, platform incentives to maximize engagement may not align with the incentives of firms, nonprofits, or political actors seeking to preserve brand integrity or policy support. Understanding these misalignments is central to current debates over platform governance and content moderation.

A limitation of our study is that our evidence comes from a single organization and a single issue, racial justice. The comment sections in our experiments involve little deliberation and informational content, and those who comment are a small and unrepresentative subset of the audience. This mirrors the concentration of post production on platforms \citep{grinberg2019fake} and is common for news and political issues \citep{kim2025disproportionate}. We show that under these conditions the vocal few can exert outsized influence, even without any prior connection to the audience they reach. In settings where comments are more deliberative or informative, other mechanisms may play a role and the effects may differ.

Our study points to two promising directions for future research. Methodologically, we introduce a scalable experimental pipeline that manipulates comment visibility and stance using platform-native tools while preserving ecological validity through organic user-generated content. One direction is to extend this framework to other domains, such as commercial products or public health campaigns, to better understand how visible online discourse shapes behavior and to inform the design of comment environments that balance engagement with broader social and organizational goals. A second direction is to move beyond one-time exposure and study how repeated exposure, user interactions, and algorithmic ranking jointly shape engagement and attitudes over time.

\singlespacing
\bibliographystyle{aea}
\bibliography{bibliography}

\clearpage

\appendix
\setcounter{page}{1}
\onehalfspacing
\begin{center}
\textbf{\LARGE Online Appendix%
}
\par\end{center}

\bigskip{}

\begin{center}
\LARGE The Influence of the Vocal Few:\\Evidence from Social Media Comments
\par\end{center}

\begin{center}
\Large \emph{Dante Donati and Lena Song}
\par\end{center}

\newpage{}

\clearpage{}

\pagestyle{fancy}
\lhead{Online Appendix} \rhead{The Influence of the Vocal Few; Donati and Song}

\section{Framework Appendix\label{app:framework}}
\setcounter{figure}{0}
\renewcommand{\thefigure}{A\arabic{figure}}
\renewcommand{\theHfigure}{A\arabic{figure}}

This appendix states the setup and results of the framework described in Section~\ref{sec:conceptual_framework}.

\subsection{Setup}
A post is shown either without comments, $c=\emptyset$, or with a section that is supportive or opposing, $c\in\{S,O\}$. Define $x(c)$ as the attitude the section expresses toward the organization and its cause, measured relative to the post's own position, with
\[
x(\emptyset)=x(S)=0,
\qquad
x(O)=-1.
\]
A supportive section repeats the post's position, while an opposing section expresses a different view and therefore provides the greatest contrast.

As in our empirical setting and on many other platforms, we assume that a comment counter is visible next to a post before the comment section is opened, making the post and the presence of comments more prominent. Denote this prominence by $\phi(c)$, measured relative to a post with no comments and common to $S$ and $O$, with $\phi(\emptyset)=0$ and $\phi(S)=\phi(O)=\phi>0$. 

When comments are present, the viewer expects the section to be opposing with probability $\widehat q\in(0,1)$, which depends on her belief about who comments: a viewer who thinks progressives do most of the commenting below a progressive post holds low $\widehat q$.

Applying the three channels through which a stimulus captures attention \citep{bordalo2022salience}, the salience of section $c$ is
\begin{equation}
s(c)=\underbrace{\phi(c)}_{\text{prominence}}
+\underbrace{|x(c)|}_{\text{contrast}}
+\underbrace{\gamma\big(-\log p_c\big)}_{\text{surprise}},
\qquad \gamma\geq0,
\label{eq:salience}
\end{equation}
where $p_c$ is the probability the viewer, having seen the counter but not yet the comments, assigns to the section being of stance $c$. Surprise is strictly decreasing in $p_c$: the less a section was expected, the more salient it is. A counter showing no comments reveals $c=\emptyset$, so $p_\emptyset=1$ and $s(\emptyset)=0$. When the counter shows comments, $p_S=1-\widehat q$ and $p_O=\widehat q$, so
\begin{equation}
\begin{aligned}
s_S&=\phi+0+\gamma\big[-\log(1-\widehat q)\big],\\[4pt]
s_O&=\phi+1+\gamma\big[-\log\widehat q\big],
\end{aligned}
\qquad\text{so}\qquad
s_O-s_S=1+\gamma\log\frac{1-\widehat q}{\widehat q}.
\label{eq:levels}
\end{equation}
An opposing section always has higher contrast, and it is also more surprising if the viewer expected a supportive discussion, $\widehat q<1/2$.

Let $a(s)$ denote the attention the post and its comment section jointly receive, and $w(s)\in[0,1]$ the share of that attention devoted to the comment section, both strictly increasing in $s$.

\paragraph{Engagement.} Because stance is not visible until the post is expanded, the viewer decides in two stages. Let $R$ indicate that she expands the post, which reveals the comment section if there is one, and let $A(c)$ indicate that she takes some further action, such as reacting or clicking a link, when the post is shown with section $c$. Let $G_R$ and $G_A$ be strictly increasing response functions that map the level of attention into the probability of each action. This captures the assumption that greater attention makes a response more likely by increasing the opportunity to process and act on the stimulus. Then

\begin{equation}
\Pr(R=1\mid c)=G_R\big(a(\phi(c))\big),
\qquad
\Pr\big(A(c)=1\mid R\big)=
\begin{cases}
G_A\big(a(\phi(c))\big) & \text{if } R=0,\\[2pt]
G_A\big(a(s(c))\big) & \text{if } R=1.
\end{cases}
\label{eq:engagement}
\end{equation}
Before the viewer expands the post, a post with a positive comment counter is more salient than a post without comments through prominence, which raises the probability of expanding the post, $G_R(a(\phi))>G_R(a(0))$, equally for $S$ and $O$. A viewer who does not open the section can still act in response to the post, in which case the level of attention is $a(\phi(c))$. For a viewer who opens the section, the salience of its content also enters, so attention is $a(s(c))$ and the probability of subsequent engagement is $G_A(a(s(c)))$.

\paragraph{Attitudes.} 

Let $\mu(c)$ denote the change in the viewer's attitude toward the organization and its cause relative to seeing the post alone. The comment section moves it by shifting weight from the post toward the section's own position, in proportion to the attention each receives, so that salience enters as a decision weight distorted toward what draws the eye \citep{bordalo2012salience,bordalo2013salience}: 
\begin{equation}
\mu(c)=\big(1-w(s(c))\big)\cdot 0+w(s(c))\,x(c)=w\big(s(c)\big)\,x(c).
\label{eq:attitude}
\end{equation}
Equation~\eqref{eq:attitude} assumes that comments move attitudes only through how far the position they express departs from the post, so supportive comments, which restate the post, leave attitudes unchanged. If supportive comments instead reinforced the post or shifted perceived social norms, they could also move attitudes, as discussed in Section~\ref{sec:SurveyResults}.

\subsection{The salience trade-off\label{app:proof}}

\begin{lemma}[The salience trade-off]
\label{prop:tradeoff}
Suppose that the viewer expects a supportive discussion, $\widehat q<1/2$, so that $s_O>s_S$. Relative to a supportive section, an opposing section
\begin{enumerate}
\item[(i)] leaves the probability of opening unchanged but, conditional on opening, strictly raises the probability of subsequent engagement, $\Pr\big(A(O)=1\mid R=1\big)=G_A(a(s_O))>G_A(a(s_S))=\Pr\big(A(S)=1\mid R=1\big)$, and
\item[(ii)] strictly lowers the attitude, $\mu(O)=-w(s_O)<0=\mu(S)$, which holds for any $\widehat q$ since $x(S)=0$.
\end{enumerate}
\end{lemma}

\begin{proof}
Write $g(\widehat q)$ for $s_O-s_S$ in equation \eqref{eq:levels}. If $\gamma=0$ then $g\equiv1>0$. If $\gamma>0$ and $\widehat q<1/2$, then $(1-\widehat q)/\widehat q>1$, so the logarithm is positive and $g(\widehat q)>1>0$. In either case $s_O>s_S$.

For (i), the opening probability is $G_R(a(\phi))$ under both stances because the comment content is not yet visible. If $R=0$, the probability of a subsequent action is $G_A(a(\phi))$ under both stances. If $R=1$, strict monotonicity of $a$ and $G_A$, together with $s_O>s_S$, imply $G_A(a(s_O))>G_A(a(s_S))$. Thus, opposing comments raise the unconditional probability of subsequent engagement whenever the section is opened with positive probability. The difference is $G_R(a(\phi))$$\big[G_A(a(s_O))-G_A(a(s_S))\big]>0$.

For (ii), equation \eqref{eq:attitude} gives $\mu(c)=w(s(c))\,x(c)$. Since $x(S)=0$, $\mu(S)=0$ for any $w$. Since $x(O)=-1$, $\mu(O)=-w(s_O)<0$.
\end{proof}

By \eqref{eq:engagement}, opening depends on prominence alone, so expansions rise with the presence of comments but do not differ across stances. Once the comment section is open, greater salience raises the level of attention and hence the probability of a subsequent action. By \eqref{eq:attitude}, the same rise in salience shifts the allocation of attention toward a section carrying $x(O)=-1$, pulling the evaluation away from the post, while reweighting a supportive section with $x(S)=0$ leaves the evaluation unchanged. The stance that generates the most engagement is therefore the one that moves attitudes against the post.

The condition $\widehat q<1/2$ is sufficient but not necessary. If $\gamma=0$, then $s_O>s_S$ for every $\widehat q$. If $\gamma>0$, then by equation \eqref{eq:levels}, $s_O>s_S$ holds if and only if
\[
\widehat q<\overline q(\gamma)\equiv\frac{1}{1+e^{-1/\gamma}}>\tfrac12 ,
\]
so the ordering reverses only once surprise runs against opposing content strongly enough to offset the contrast effect.

Appendix Figure~\ref{fig:expected_commenter_ideology} reports a related belief measured in the survey experiment. Among respondents in the control group, \surveyExpectCommentersProgressive{}\% expect progressives to be the more likely commenters and only \surveyExpectCommentersConservative{}\% expect conservatives. Most respondents therefore do not expect conservatives to be more likely to comment, even though users in conservative areas comment at higher rates as shown in Section~\ref{sec:generation}, consistent with them not anticipating an opposing comment section.

\begin{figure}[H]
    \centering
    \caption{Expected Ideology of Commenters in the Control Condition}
    \label{fig:expected_commenter_ideology}
    \includegraphics[width=0.8\textwidth]{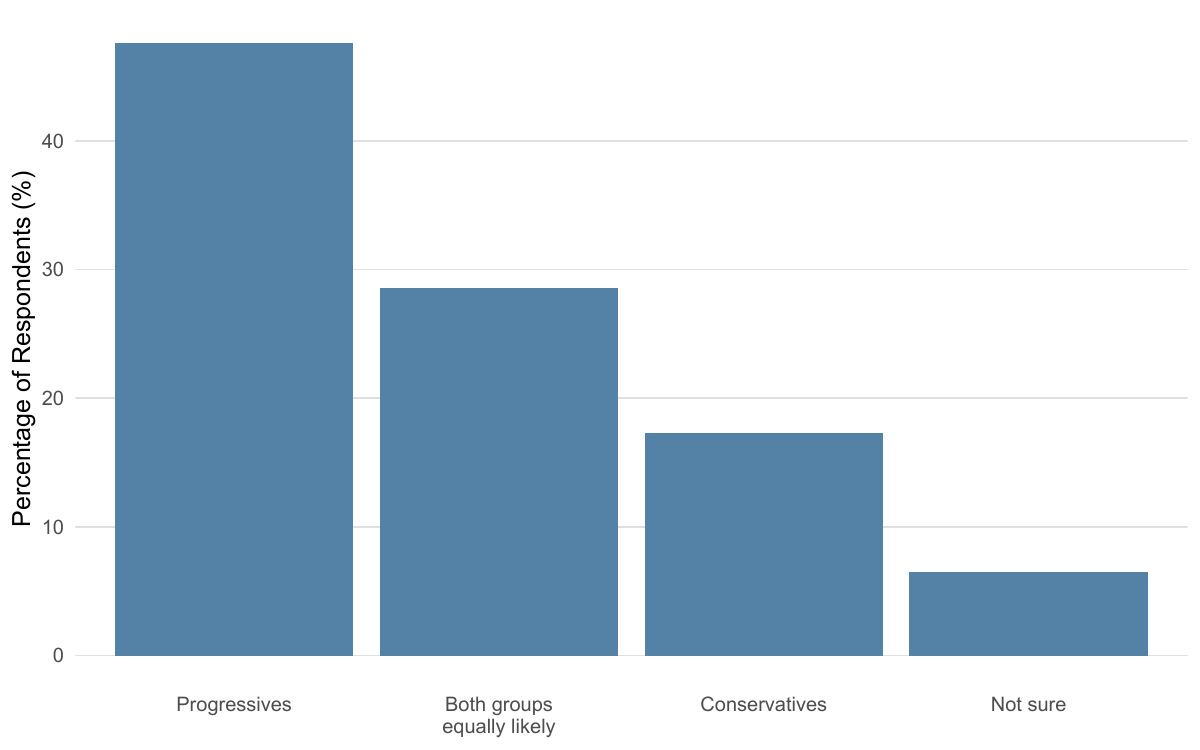}
\begin{minipage}{\textwidth}
\scriptsize
\noindent \textit{Notes:} This figure reports beliefs about who comments among the \balanceNControl{} survey respondents assigned to the control condition, who saw the post without a comment section. Respondents were asked which group they believe is more likely to comment on the post. 
\end{minipage}
\end{figure}

\subsection{Discussion\label{app:discussion}}
Equation \eqref{eq:salience} should be interpreted as a reduced-form application of salience. In \citet{bordalo2012salience} and \citet{bordalo2013salience}, salience is defined relative to payoffs or attributes in a choice set. Here the reference is the post, which is what surrounds the comments on the screen, and the surprise term $\gamma(-\log p_c)$ is an addition motivated by the broader account in \citet{bordalo2022salience} of attention being drawn to what is unexpected.

The framework abstracts from viewer ideology. Salience is defined relative to the post, whereas identity congruence is defined relative to the viewer. A comment opposing a progressive post contrasts with the post for every viewer, but aligns with a conservative viewer's own position and conflicts with a progressive viewer's. The same comment can therefore invite approval from one viewer and anger from another. An extension of the model is to let the payoff of acting depend on identity congruence. This may explain why differences across comment stances are most pronounced in conservative areas: for a conservative viewer, salience and identity congruence both push toward engaging with an opposing section, whereas for a progressive viewer they may move in opposite directions.

\section{Pre-registration and Pre-analysis Plan}\label{app:deviations}

Both Facebook and Prolific experiments were pre-registered in the American Economic Association Registry
for randomized controlled trials (under AEARCTR-0013812 and
AEARCTR-0017850), before the
corresponding experiment was fielded. We follow the two plans in  the experimental designs, audience construction and stratification,
randomization, treatment conditions, post and comment selection, estimating
equations, covariate sets, and index construction. The field experiment plan also
pre-specified heterogeneity analyses by political ideology, racial composition, age, and gender. We report
heterogeneity by the political ideology of the area; the analyses by racial composition, age, and gender
are available upon request. 

Two outcome-level deviations apply to the
field experiment. First, we add post expansions, which was not pre-specified, as
our proxy for seeing the comment section.  Second, we do not report treatment effects
on newsletter sign-ups recorded by the organization's website as the campaign
generated too few sign-ups to support inference on this outcome, so we instead
report unique link clicks. We also compare conditions using the linear
probability model in equation~\eqref{eq:main}, together with model-free
comparisons of reach-weighted means and heteroskedasticity-robust $t$-tests
(Appendix Table~\ref{tab:outcomes_model_free}), rather than the $\chi^{2}$ tests of
proportions named in the plan. Finally, the realized sample is also smaller than anticipated: the plan targeted 2,500,000 users
reached, whereas the campaign reached \feTwoReachN{} across \feTwoStrataN{} ZIP code sets.
Reach per ZIP code set was in line with the plan, so the shortfall reflects the number of
sets fielded rather than under-delivery within them. 

For the survey experiment, the realized analysis sample of
\attritionFinalN{} respondents is smaller than the 4,800 anticipated in the plan,
primarily because fewer conservatives were available on Prolific than the ideology quotas required. Descriptive analyses specified in the plan, including those based on the open-response questions, are available upon request.

\section{Background Appendix\label{app:background}}

\setcounter{figure}{0}
\renewcommand{\thefigure}{C\arabic{figure}}
\renewcommand{\theHfigure}{C\arabic{figure}}

\setcounter{table}{0}
\renewcommand{\thetable}{C\arabic{table}}
\renewcommand{\theHtable}{C\arabic{table}}

\subsection{List of Issues and Description}
Here are the racial justice issues and their descriptions presented to content designers to guide the creation of posts for each issue: 

\begin{itemize}
\item Voter Suppression:
Black communities face deliberate barriers like restricted polling access, strict ID laws, and voter roll purges, aimed at limiting their voting power. Misinformation campaigns also target Black voters to reduce turnout, undermining fair representation. Breaking down these barriers is crucial to ensure Black voices are heard in democratic processes.

 \item Environmental Justice:
Black communities often live near pollution sources like factories and highways, leading to higher rates of health issues such as asthma. These neighborhoods are frequently overlooked in clean-up efforts and lack green spaces. Environmental justice aims to provide Black communities with clean air, safe water, and healthy environments.

 \item Criminal Justice and Police Reform:
Black communities experience disproportionate police violence, profiling, and harsher sentencing. This systemic bias erodes trust in law enforcement and perpetuates disadvantages. Police reform is essential for fair treatment, accountability, and ensuring Black communities feel protected, not targeted, by the justice system.

 \item Education Reform: 
Black students often attend underfunded schools with fewer resources, larger classes, and limited access to advanced courses. These disparities create achievement gaps and limit future opportunities. Education reform seeks equitable funding and support to provide Black students with the quality education they deserve.

 \item Technology Fairness:
Black communities face systemic biases in technology, from algorithmic discrimination in hiring and lending to facial recognition tools that disproportionately misidentify Black individuals. These inequities perpetuate existing racial disparities and limit opportunities. Ensuring technology fairness involves designing inclusive systems, addressing bias in algorithms, and creating tools that serve all communities equitably.

\end{itemize}

\begin{figure}[H]
\caption{Ad Banners and Headlines}
        \label{fig:banners}
        \centering
    \begin{subfigure}[b]{0.48\textwidth}
        \centering
    \includegraphics[width=0.44\textwidth]{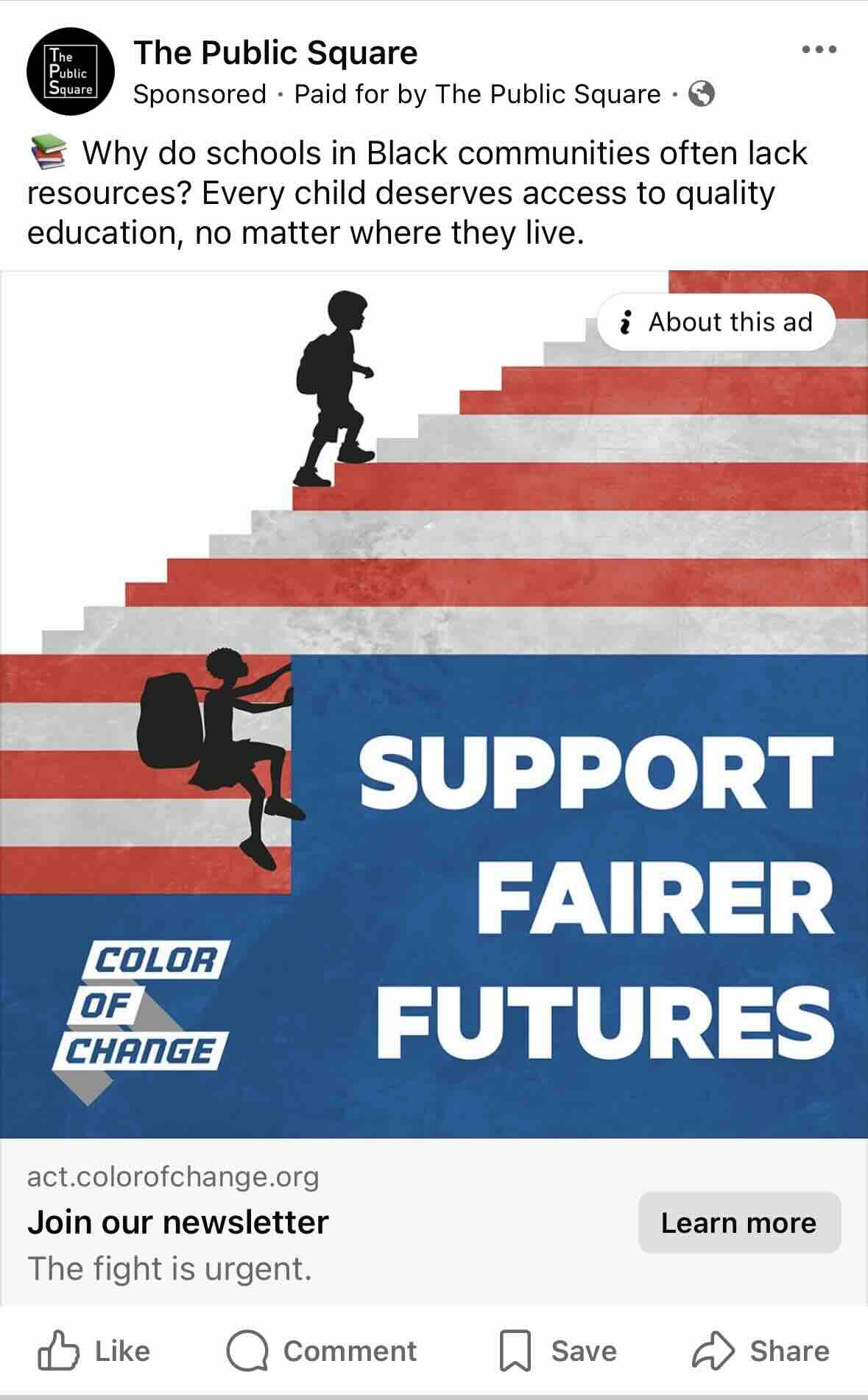} \,
    \includegraphics[width=0.43\textwidth]{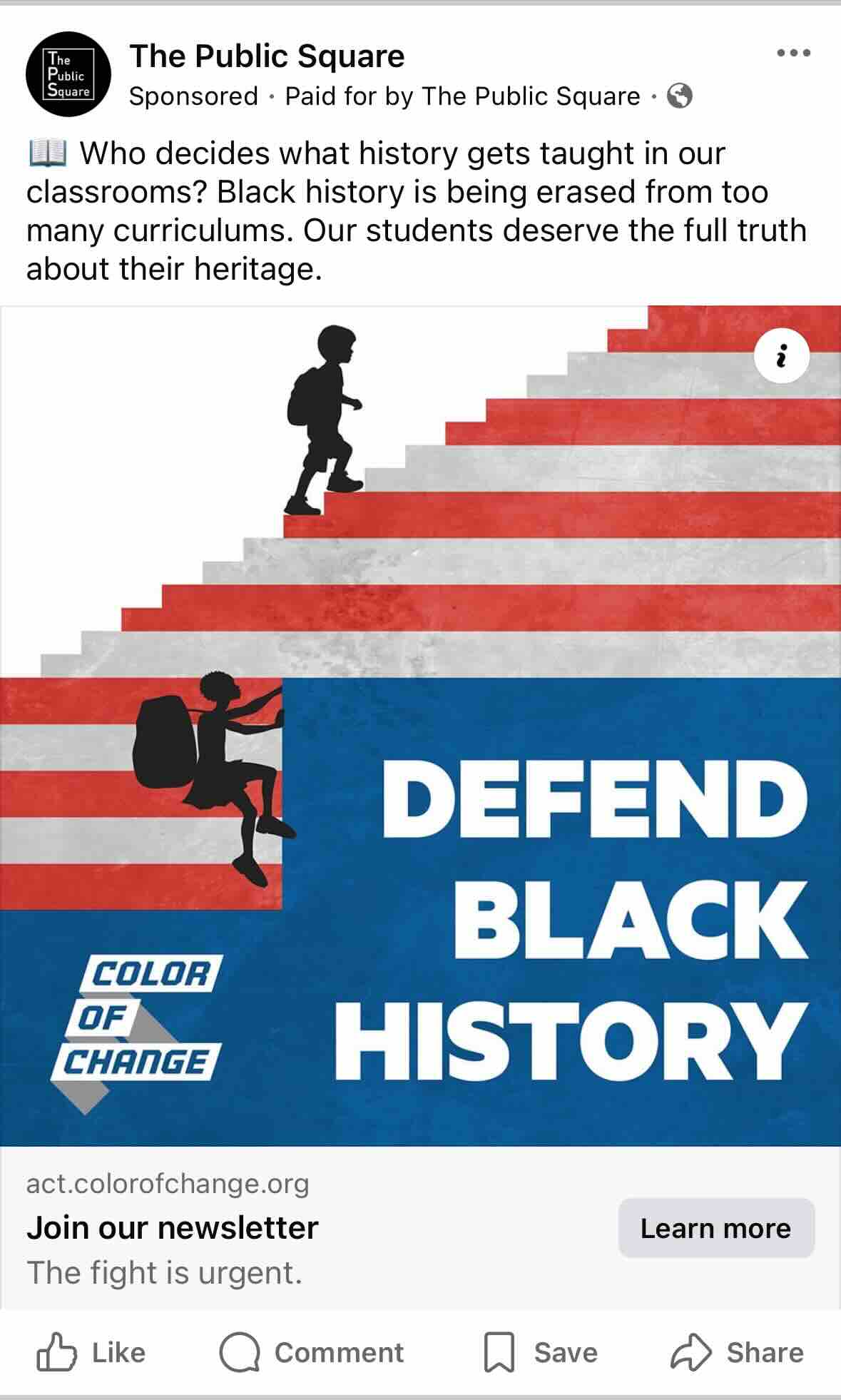}
        \caption{Education}
    \end{subfigure}
    \hfill
    \begin{subfigure}[b]{0.48\textwidth}
        \centering
    \includegraphics[width=0.44\textwidth]{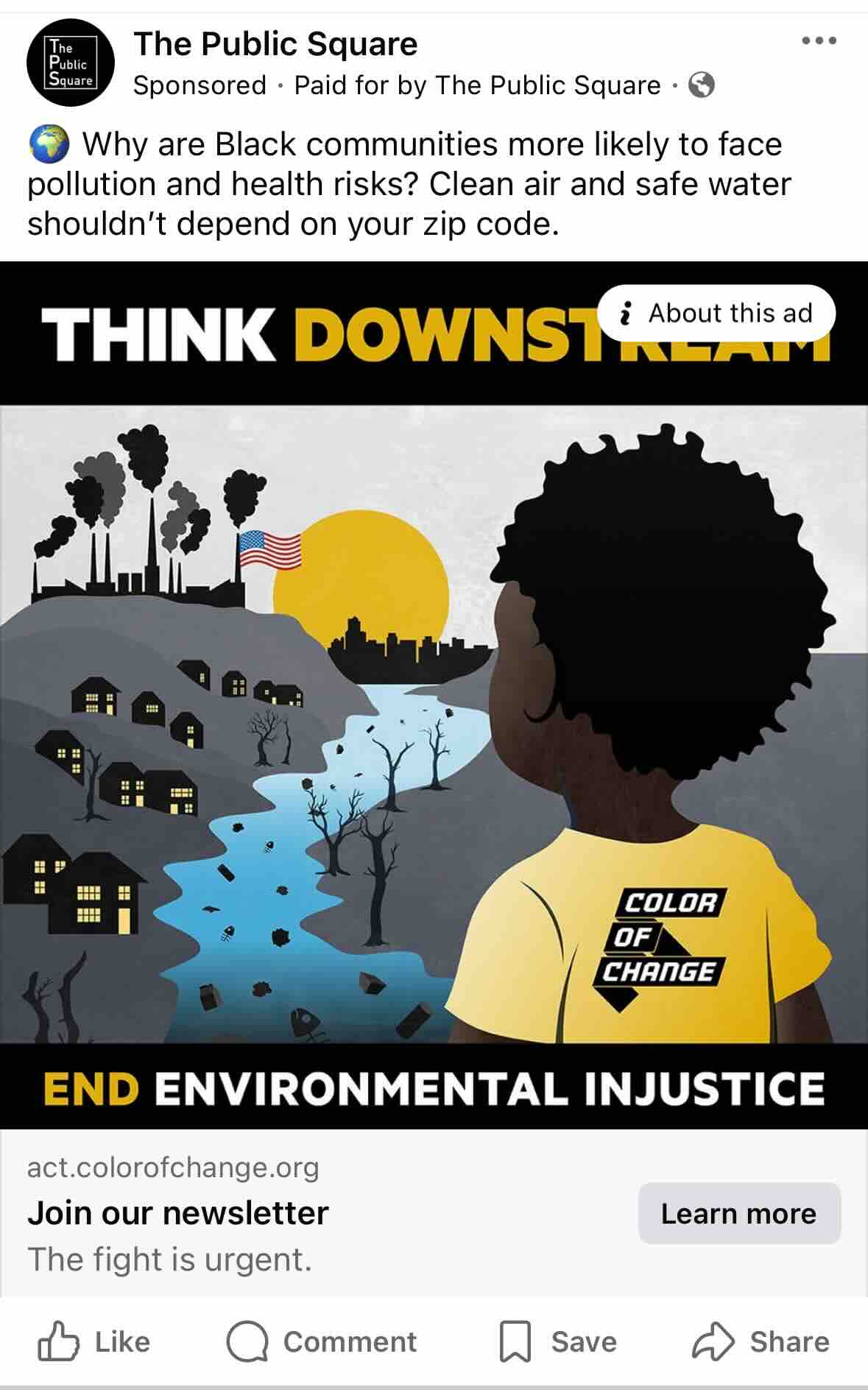} \,
    \includegraphics[width=0.44\textwidth]{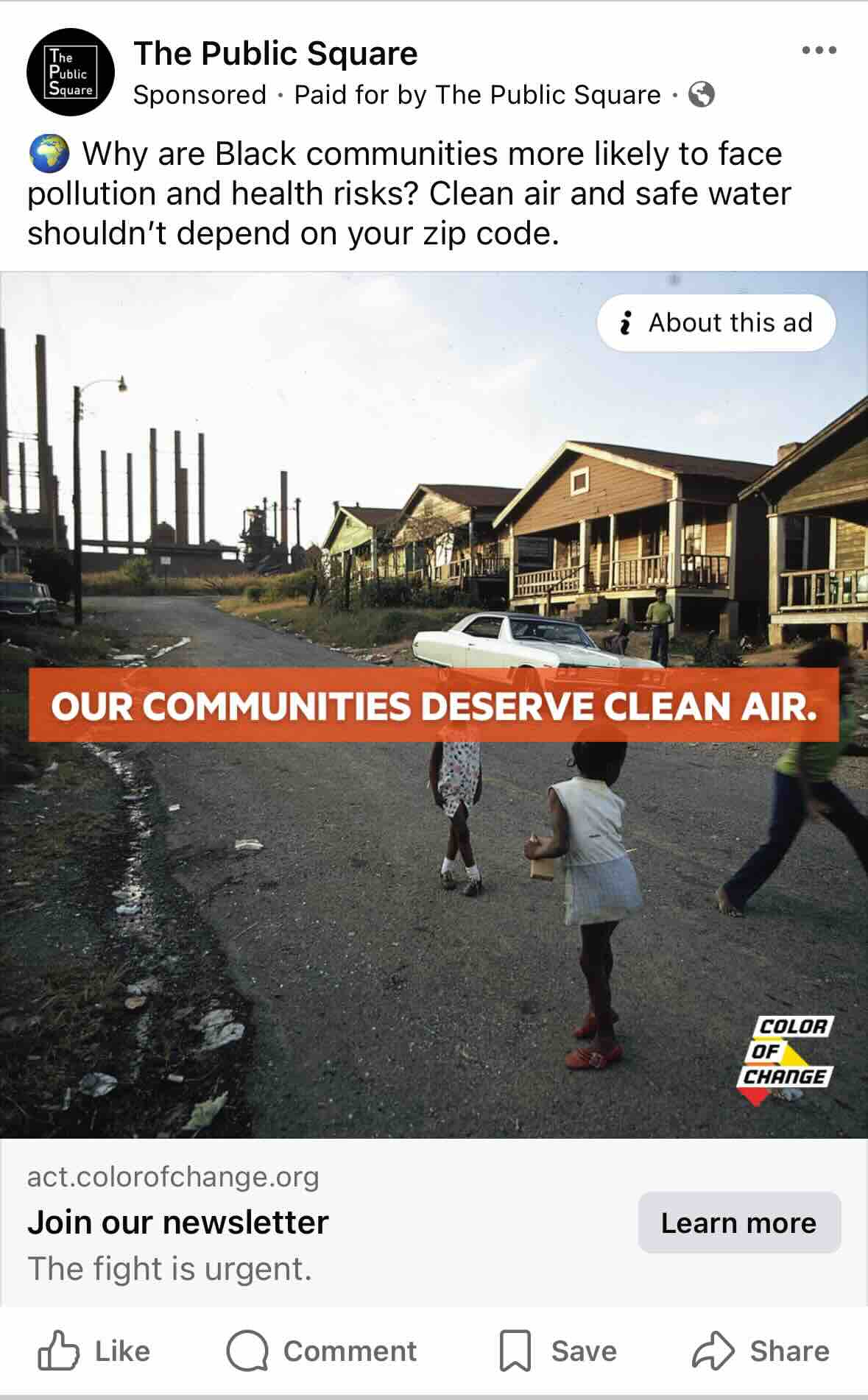}
        \caption{Environment}
     \end{subfigure} \vspace*{0.4cm}
    
    \begin{subfigure}[b]{0.48\textwidth}
        \centering
    \includegraphics[width=0.44\textwidth]{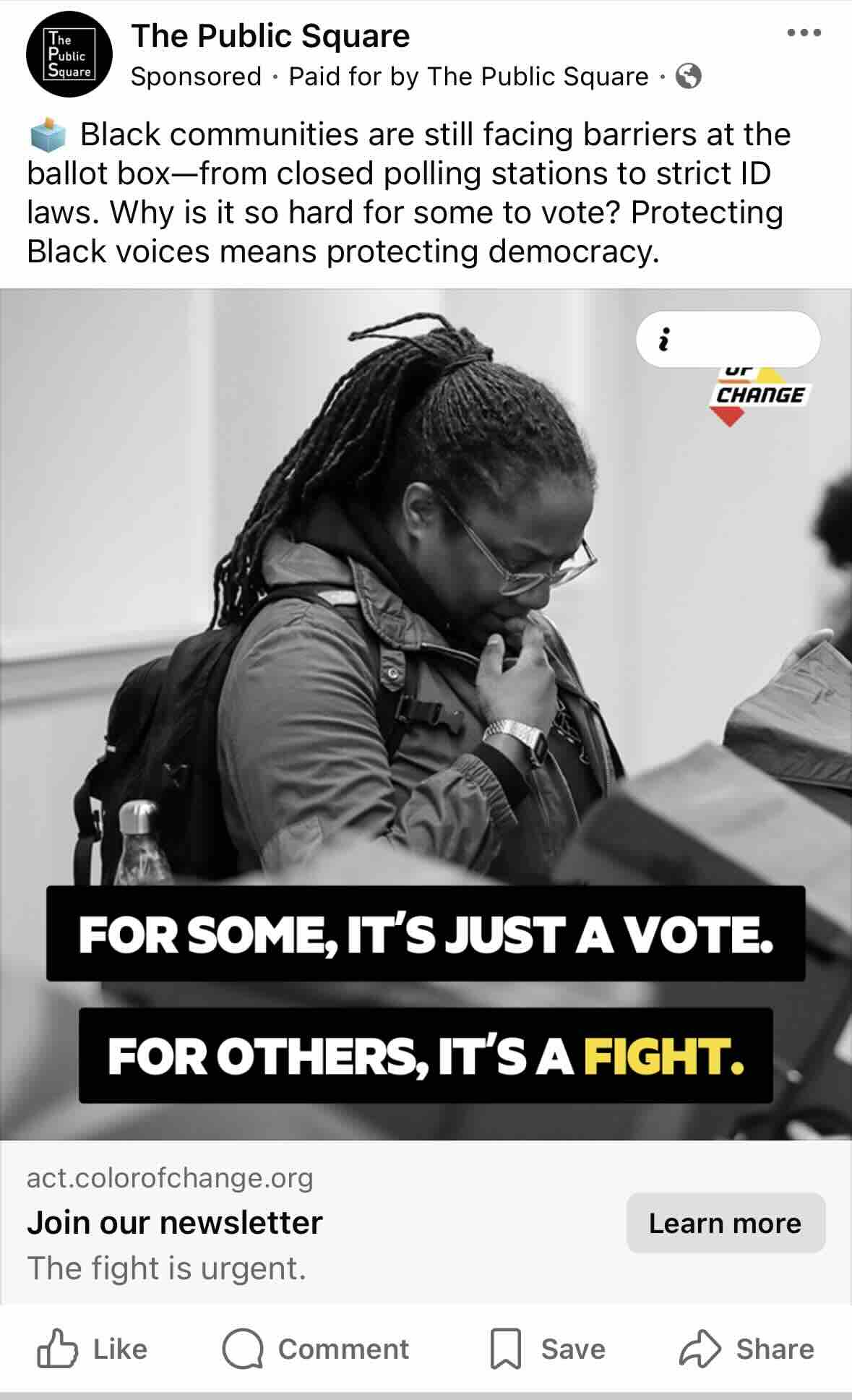} \, 
    \includegraphics[width=0.44\textwidth]{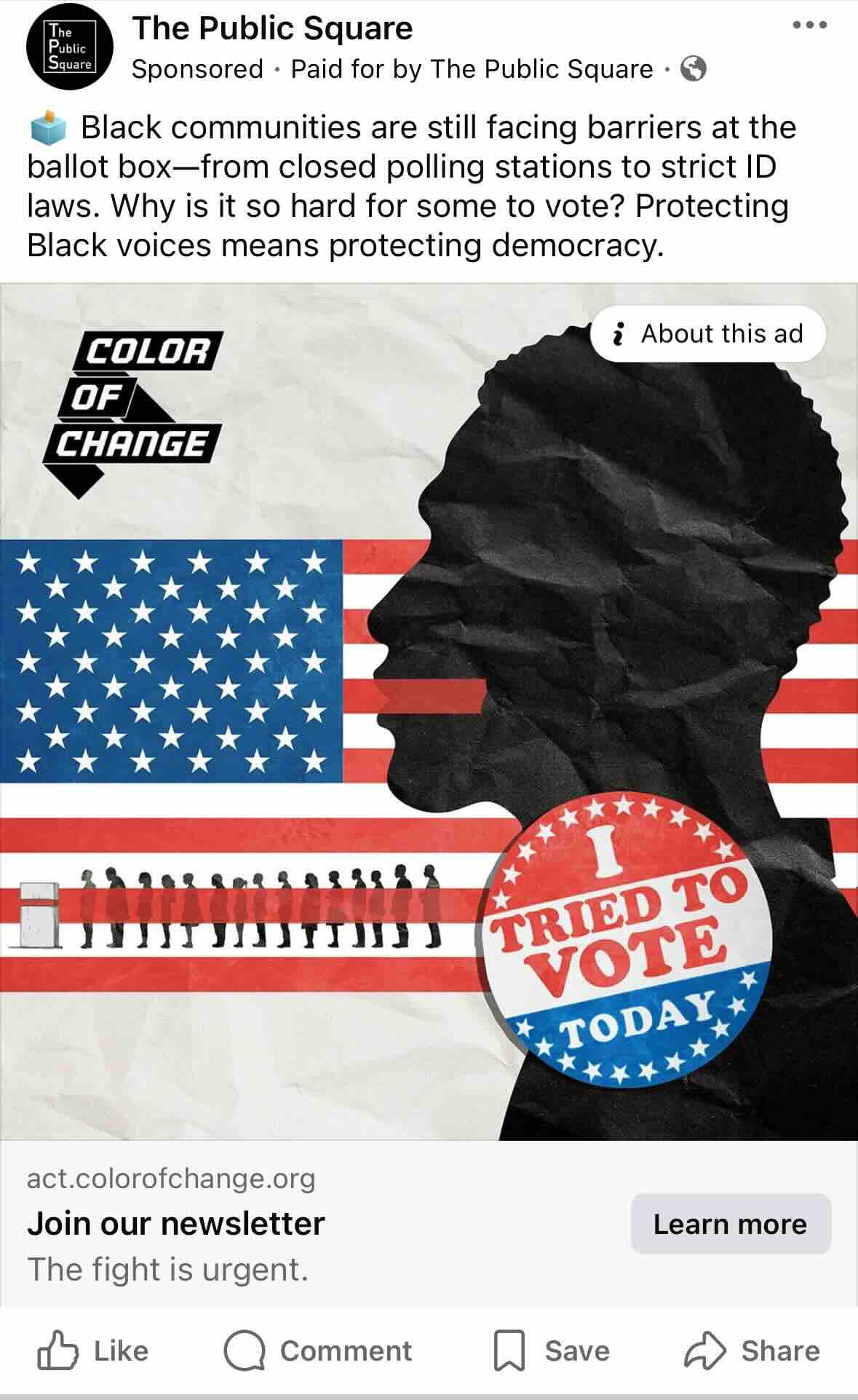}
        \caption{Voting}
    \end{subfigure}
     \hfill  
    \begin{subfigure}[b]{0.48\textwidth}
        \centering
    \includegraphics[width=0.44\textwidth]{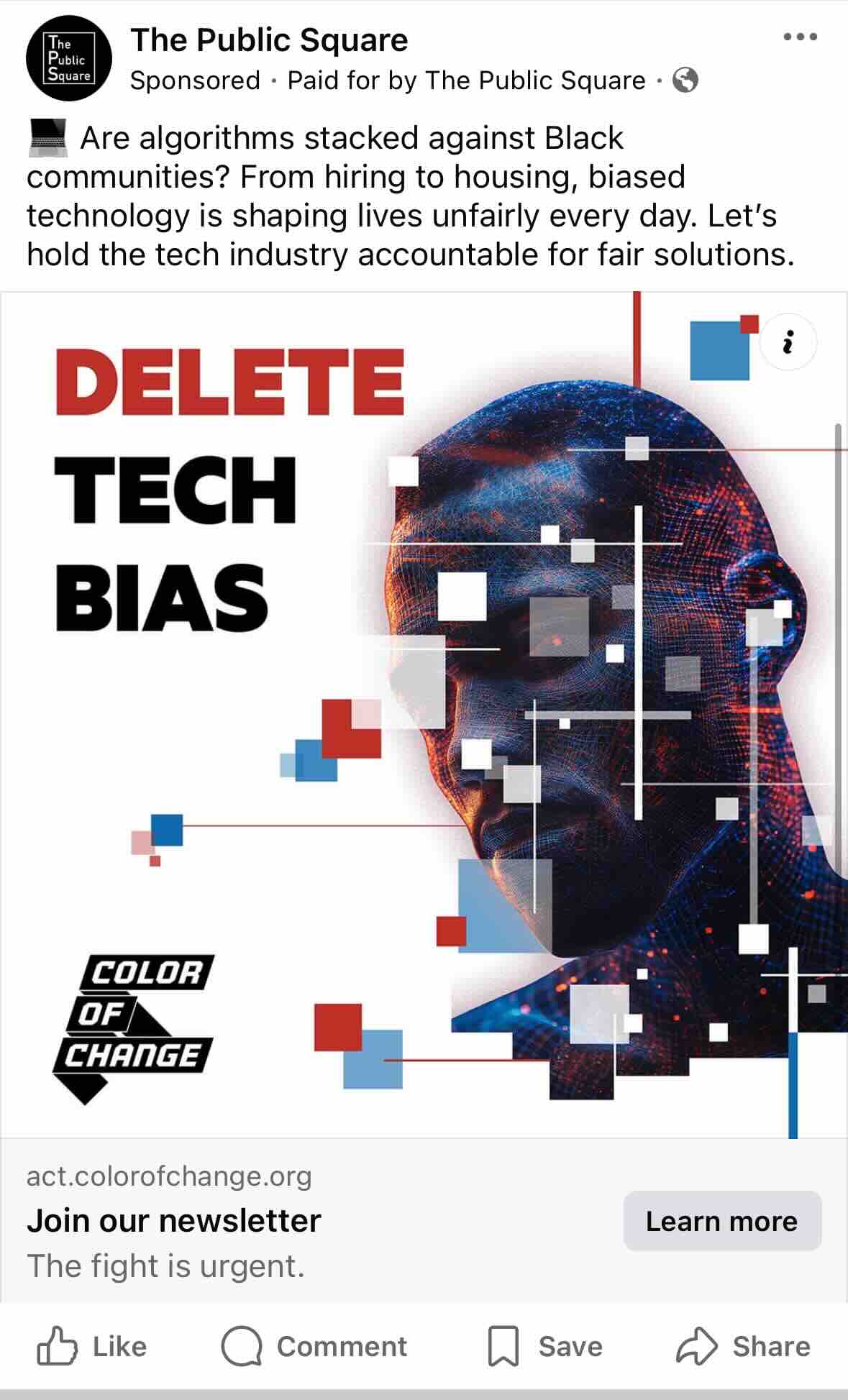} \, 
    \includegraphics[width=0.44\textwidth]{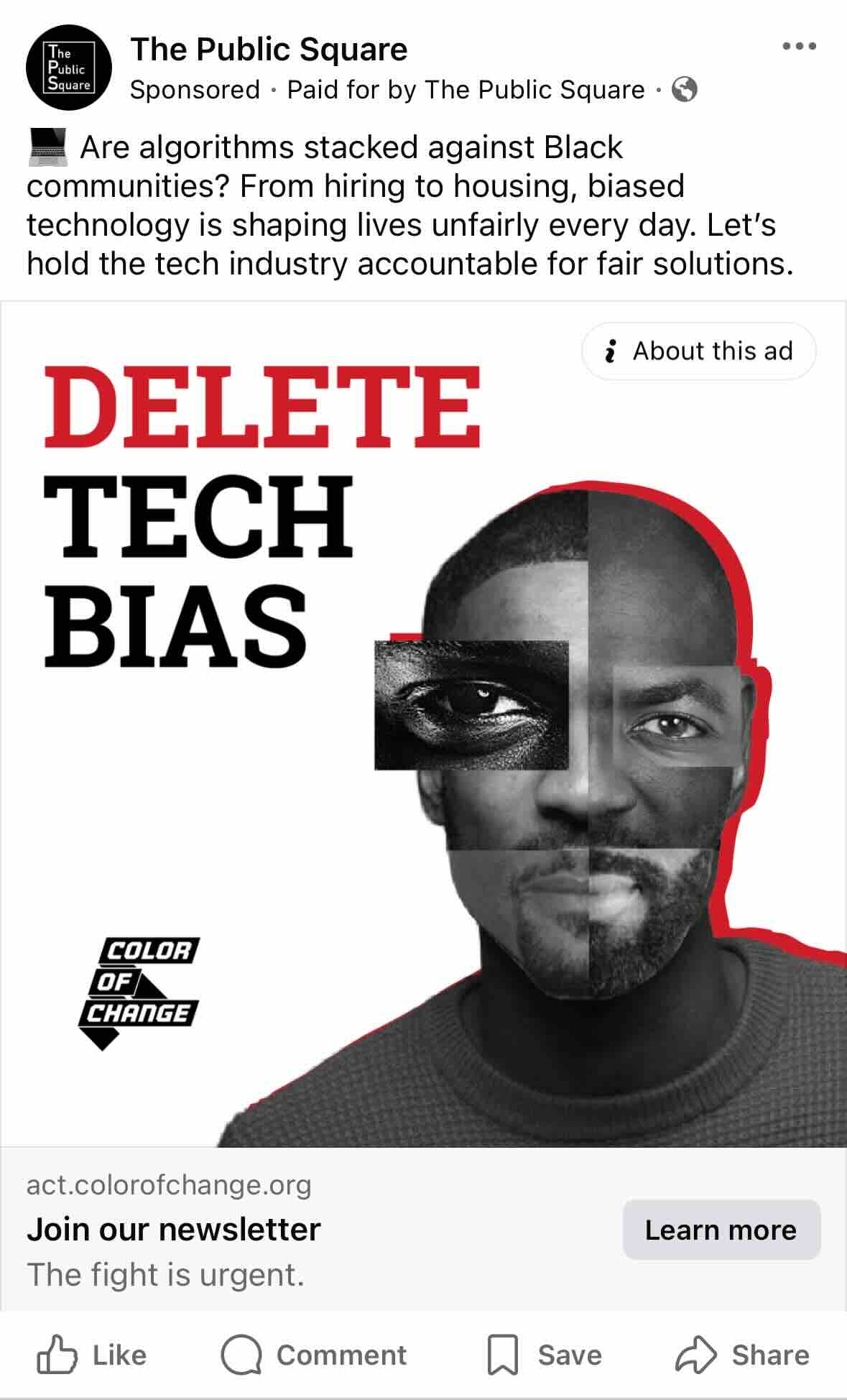}
        \caption{Technology}
    \end{subfigure} \vspace*{0.4cm}
    
    \begin{subfigure}[b]{0.48\textwidth}
        \centering
    \includegraphics[width=0.44\textwidth]{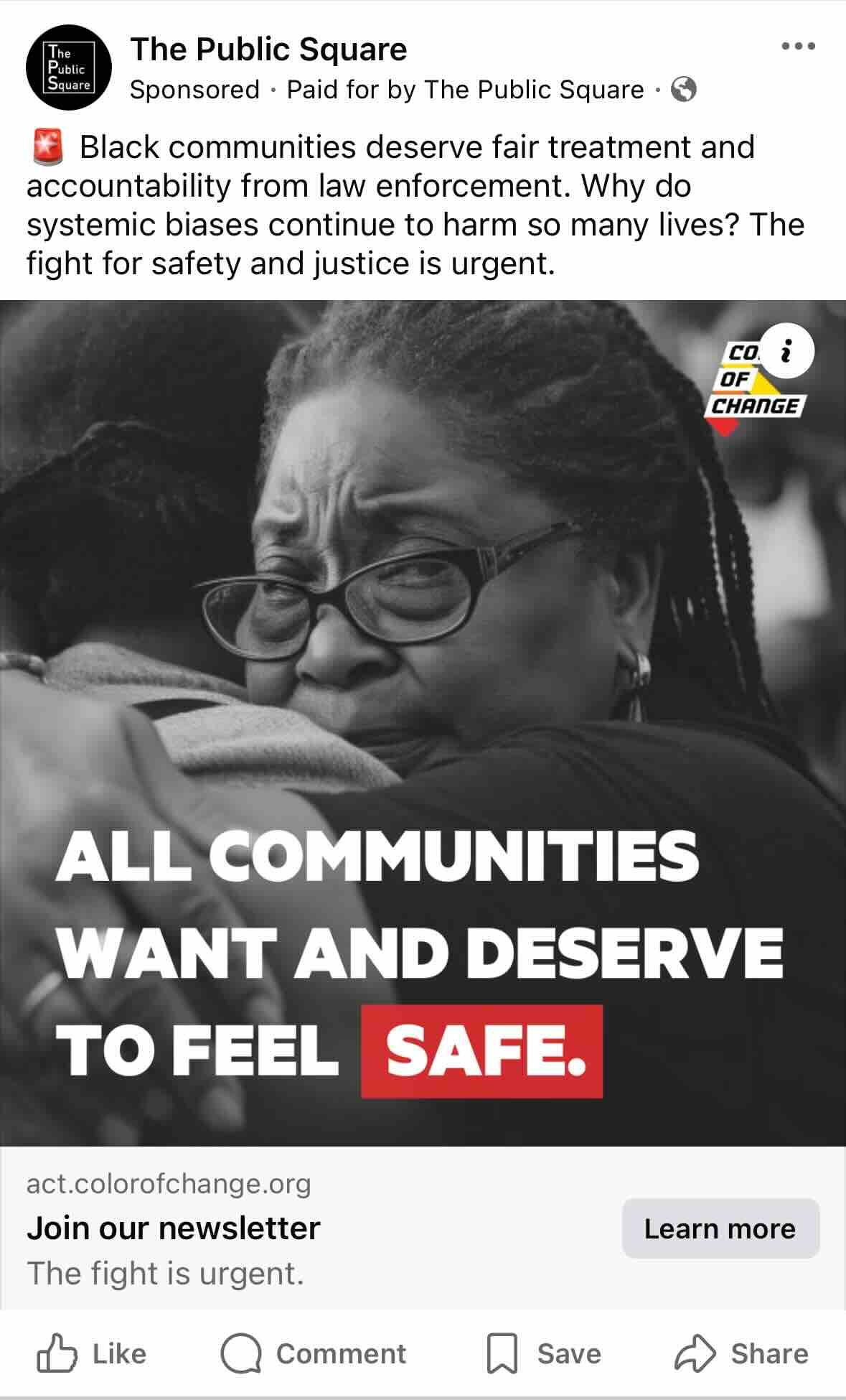} \, 
    \includegraphics[width=0.44\textwidth]{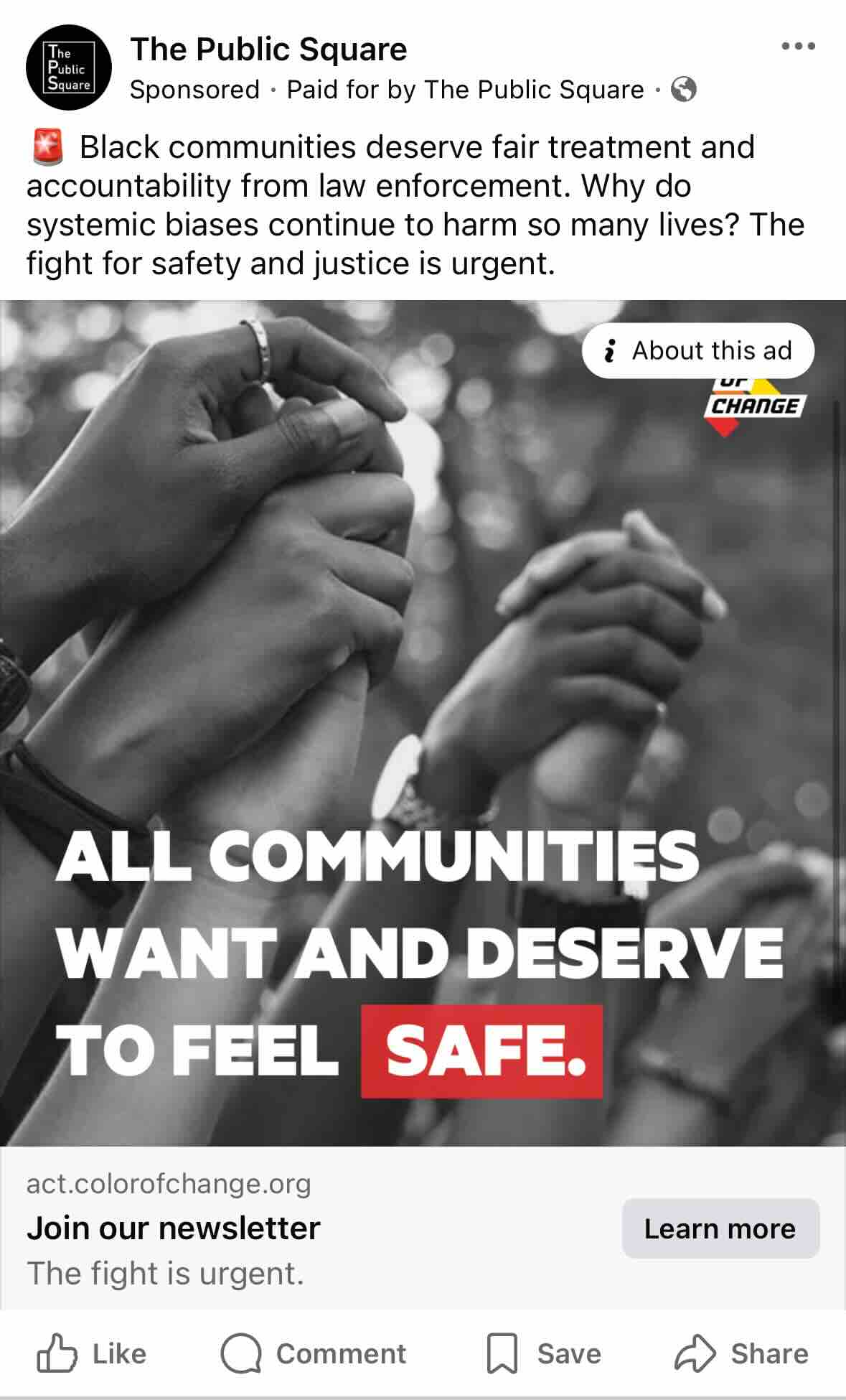}
        \caption{Police}
    \end{subfigure}
\begin{minipage}{\textwidth}
\scriptsize
\noindent \textit{Notes:} This figure displays the ad banners and headlines used in the pre-tests.
\end{minipage}

\end{figure}

\subsection{Pre-tests}

We conducted several pre-tests to systematically select the posts used in the study and to refine the campaign parameters. In Pre-test A, we used Facebook's A/B testing tool across all 10 banners to identify, within each issue, which posts were most likely to generate a high number of clicks.\footnote{In Pre-test A, we specified the audience (ZIP codes with a high share of progressive populations), budget (\$50 per banner), duration (1 week), and optimization goal (reach), without imposing a frequency cap.} 
Table \ref{tab:pretestA} summarizes the click-through rates (CTRs)---the ratio of link clicks over reach---for this test. These vary between 0.12\% and 0.25\%. For each issue, the banners with higher CTR are displayed on the right in Appendix Figure \ref{fig:banners}.

We conducted two additional tests, Pre-tests B and C, where we capped the frequency at one impression per person. Test B was conducted with a large potential audience (approximately 200,000 users per banner), while Test C targeted a smaller potential audience (approximately 6,000 users per banner). These adjustments were made to simulate a campaign that closely resembles the one planned for our main experiment. 

Table \ref{tab:pretestBC} presents the aggregate results for Pre-tests B and C.\footnote{The banners on education were excluded from these tests due to their low performance in Pre-test A.} Notably, the CTRs for Pre-tests B and C are significantly lower than those reported for Pre-test A. This discrepancy arises because Pre-test A did not impose a frequency cap, allowing users to see each banner an average of two times and thereby increasing the likelihood of clicks. By contrast,  Pre-tests B and C adopt a configuration similar to that used in the main experiment, in which we saturate an audience group by imposing a frequency cap of one impression per user. While this approach may result in lower CTRs, it is essential to mitigate potential divergent delivery bias in Facebook A/B tests caused by the ad platform's algorithm \citep{burtch2025characterizing}, as further discussed in Section \ref{sec:experiment_design}.

\begin{table}[H] \centering
  \caption{Results from Pre-test A}
  \label{tab:pretestA}
\begin{tabular}{llrrr}
\toprule
      Issue & Ad Name &  Link Clicks &  Reach &  CTR (\%) \\
\midrule
 technology &  pixels &           15 &   6115 &    0.245 \\
 technology &     man &           13 &   6683 &    0.195 \\
     voting &    lady &           14 &   5777 &    0.242 \\
     voting &    flag &           11 &   5792 &    0.190 \\
     police &    lady &           13 &   6146 &    0.212 \\
     police &   hands &           13 &   6315 &    0.206 \\
environment &     kid &           13 &   6235 &    0.209 \\
environment &  street &           11 &   6290 &    0.175 \\
  education &  future &            9 &   6216 &    0.145 \\
  education & history &            6 &   4884 &    0.123 \\
\bottomrule
\end{tabular}

\vspace{0.3em}

\begin{minipage}{\textwidth}
\scriptsize
\noindent \textit{Notes:} This table reports link clicks, reach, and click-through rates (CTR) for each ad creative tested in Pre-test A, used to select the best-performing banner and headline for each issue prior to the main experiment.
\end{minipage}
\end{table}

\begin{table}[H] \centering
  \caption{Results from Pre-tests B and C}
  \label{tab:pretestBC}
\begin{tabular}{llrrr}
\toprule
       Issue & Ad Name &  Link Clicks &  Reach &  CTR (\%) \\
\midrule
      voting &    lady &           32 &  26794 &    0.119 \\
 environment &     kid &           33 &  28351 &    0.116 \\
      police &   hands &           32 &  27771 &    0.115 \\
  technology &  pixels &           29 &  28208 &    0.103 \\
\bottomrule
\end{tabular}

\vspace{0.3em}

\begin{minipage}{\textwidth}
\scriptsize
\noindent \textit{Notes:} This table reports link clicks, reach, and click-through rates (CTR) for each ad creative tested in Pre-tests B and C. Link clicks and reach are pooled across the two tests.
\end{minipage}
\end{table}

\section{Descriptive Evidence on Engagement: Additional Results}

\setcounter{figure}{0}
\renewcommand{\thefigure}{D\arabic{figure}}
\renewcommand{\theHfigure}{D\arabic{figure}}

\setcounter{table}{0}
\renewcommand{\thetable}{D\arabic{table}}
\renewcommand{\theHtable}{D\arabic{table}}

\begin{figure}[H]
\caption{Comment and Reaction Rates by Valence and Location}
\label{fig:comments_reactions_rates}
\centering
\includegraphics[width=0.65\textwidth]{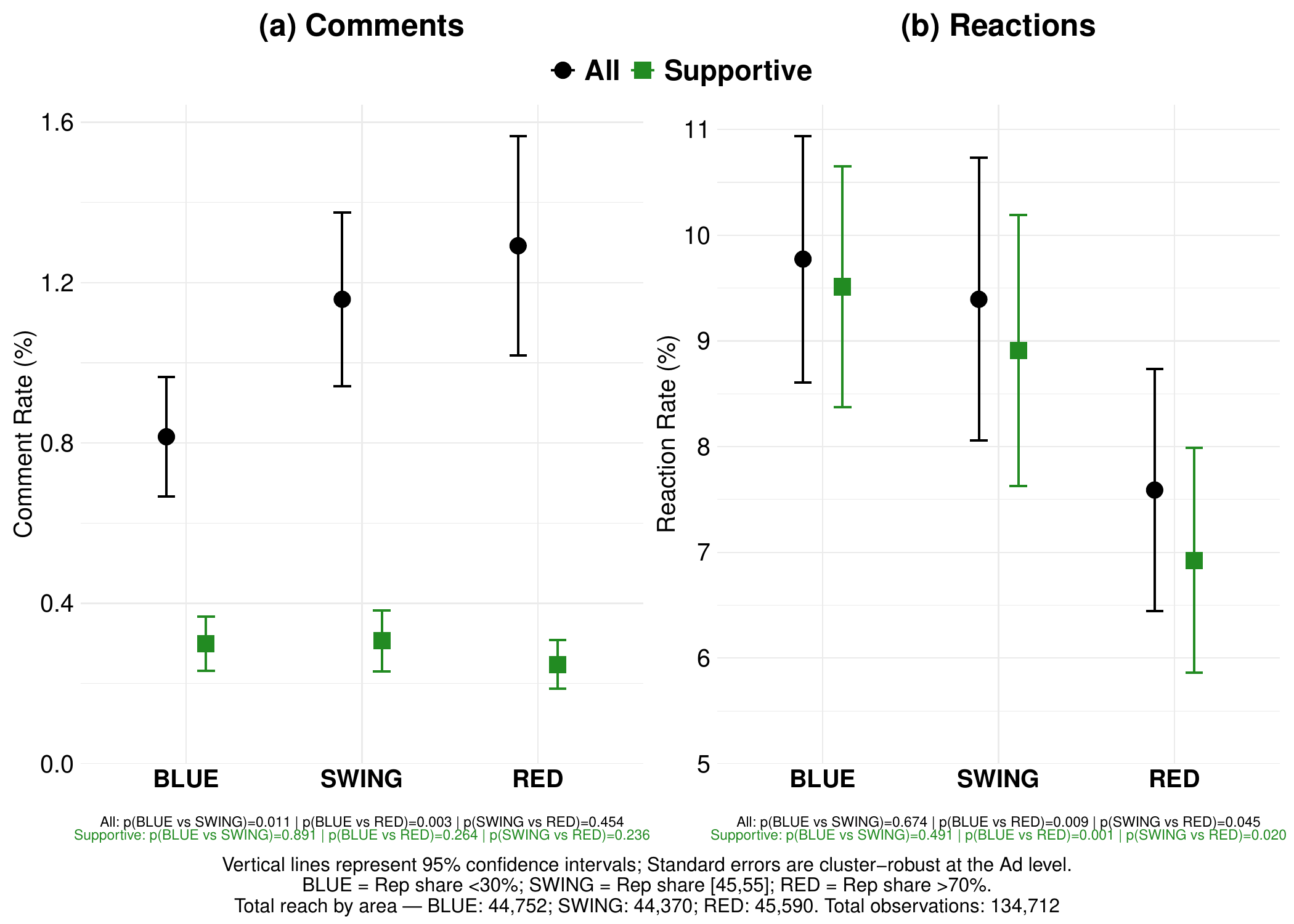}
\begin{minipage}{\textwidth}
\scriptsize
\noindent \textit{Notes:} This figure reports comment and reaction rates, expressed as percentages of total reach, by location. Black markers denote all comments/reactions and green markers denote supportive comments/reactions. Supportive reactions include likes, loves, and cares. Vertical lines represent 95\% confidence intervals, and standard errors are cluster-robust at the advertisement level. Total reach is 44,752 in BLUE areas, 44,370 in SWING areas, and 45,590 in RED areas.
\end{minipage}
\end{figure}

\begin{figure}[H]
\caption{Comment and Reaction Counts by Issue and Location}
\label{fig:comments_reactions_levels_issue}
\centering
\includegraphics[width=0.8\textwidth]{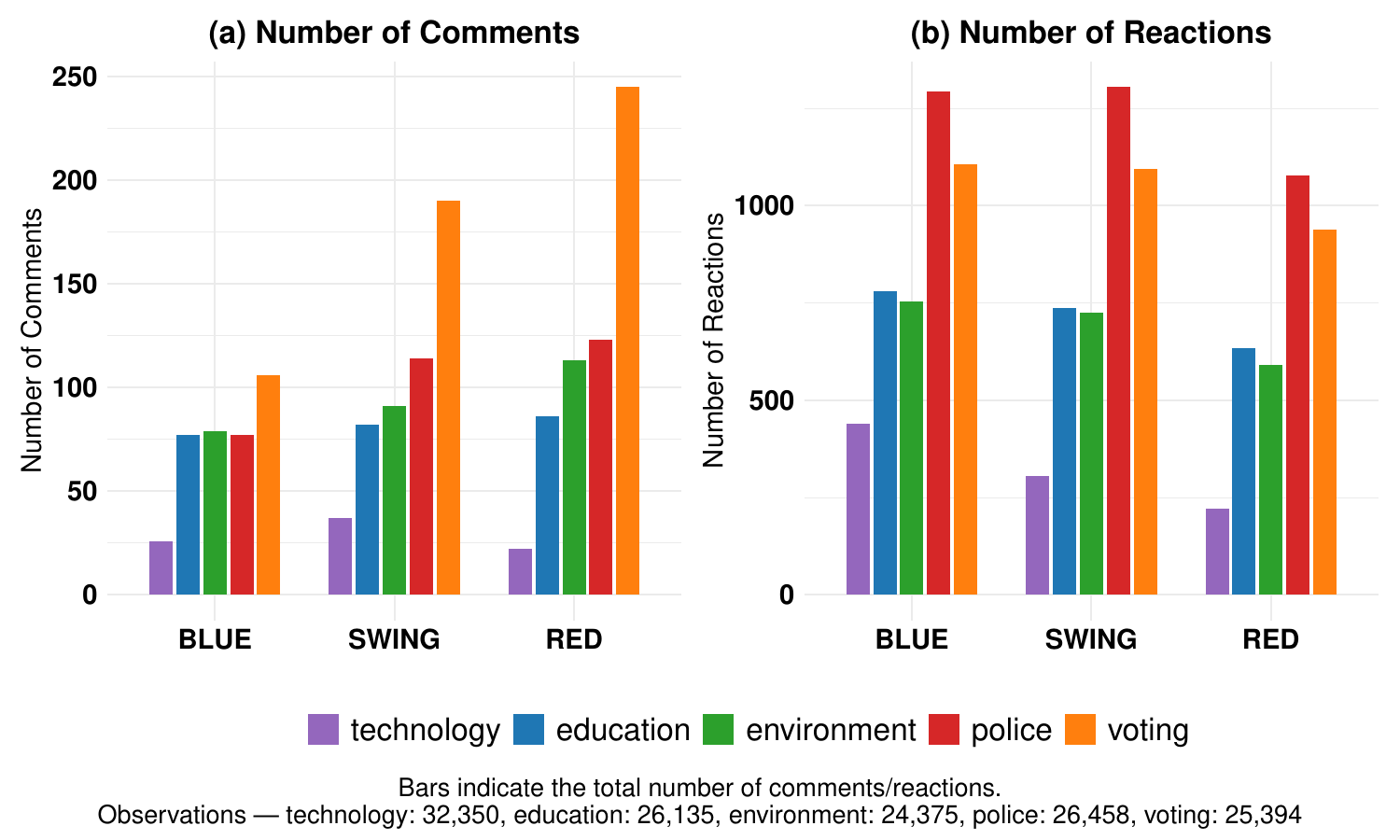} 
\begin{minipage}{\textwidth}
\scriptsize
\noindent \textit{Notes:} This figure reports the total number of comments and reactions generated during the initial Facebook campaign, by issue and location. Colors denote issue areas: technology fairness, education reform, environmental justice, criminal justice and police reform, and voting rights. Areas are grouped as BLUE, SWING, and RED according to Republican vote share.
\end{minipage}

\end{figure}

\begin{figure}[H]
\caption{Comment and Reaction Rates by Location and Issue}
\label{fig:comments_reactions_rates_issue}
\centering
\includegraphics[width=0.85\textwidth]{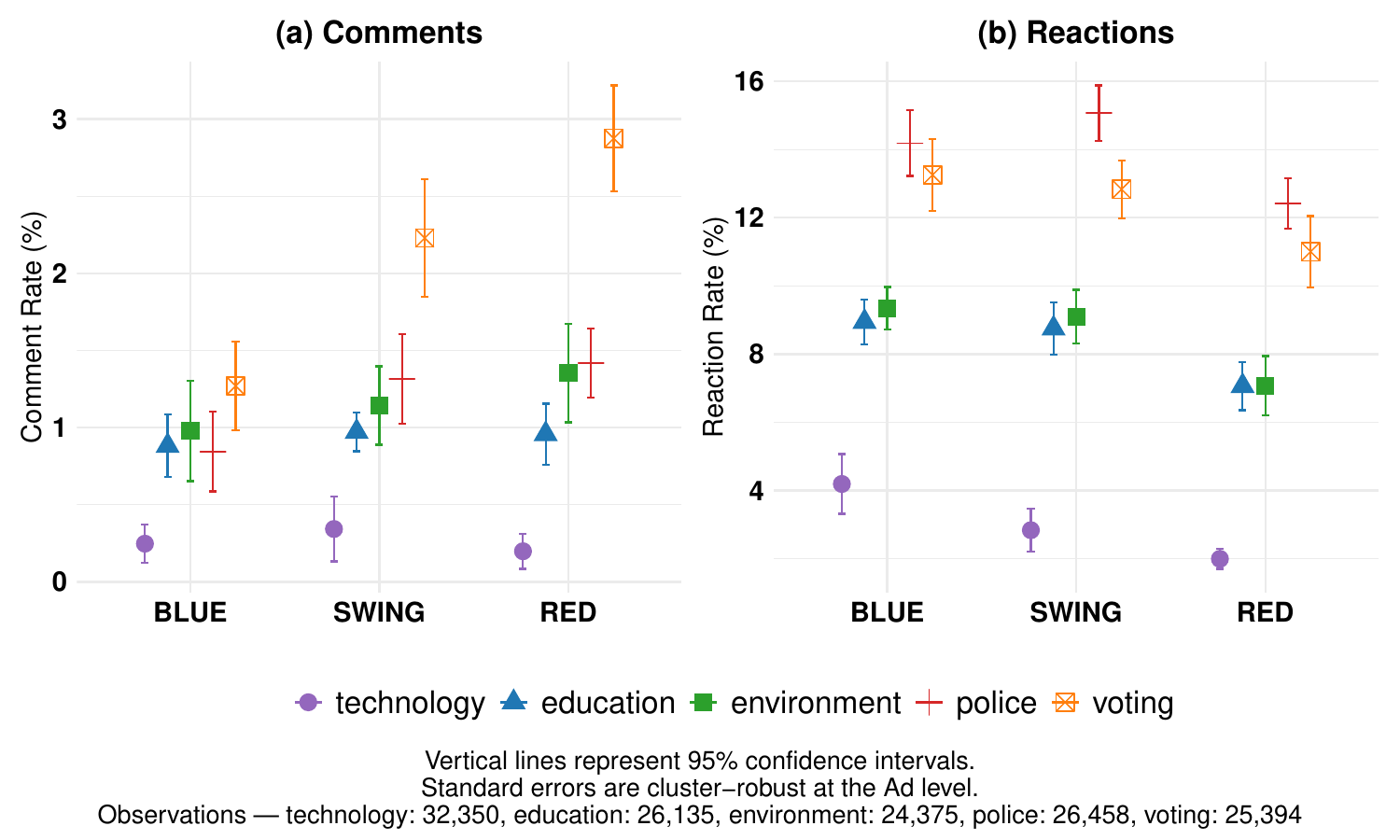} 
\begin{minipage}{\textwidth}
\scriptsize
\noindent \textit{Notes:} This figure reports comment and reaction rates, expressed as percentages of total reach, by issue and location. Colors denote issue areas: technology fairness, education reform, environmental justice, criminal justice and police reform, and voting rights. Vertical lines represent 95\% confidence intervals, and standard errors are cluster-robust at the advertisement level. Observations are 32,350 for technology, 26,135 for education, 24,375 for environment, 26,458 for police, and 25,394 for voting.
\end{minipage}
\end{figure}

\begin{figure}[H]
\caption{Comment Characteristics by Location}
\label{fig:comments_characteristics}
\centering
\includegraphics[width=0.8\textwidth]{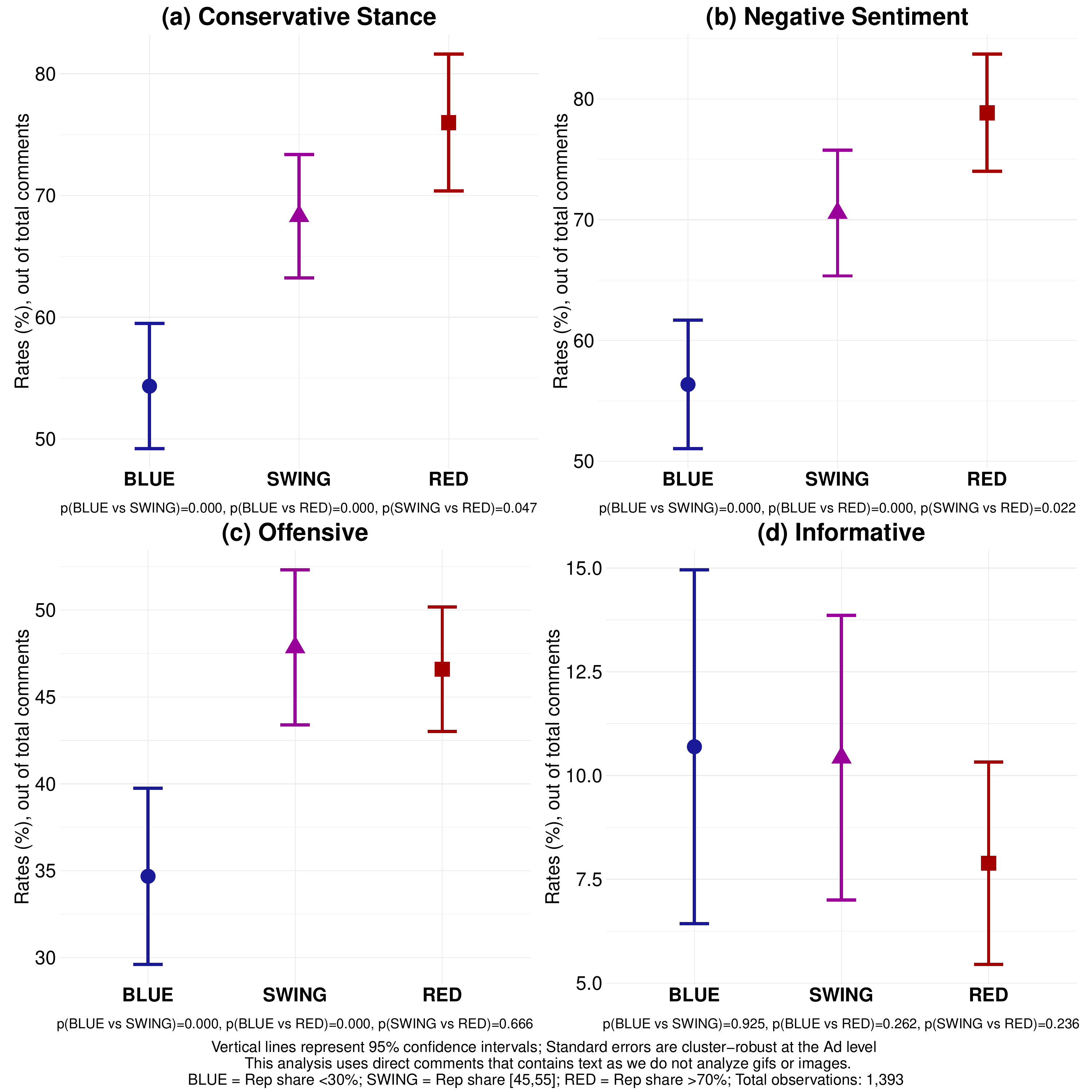} 
\begin{minipage}{\textwidth}
\scriptsize
\noindent \textit{Notes:} This figure reports characteristics of direct text comments by location, expressed as percentages of total comments. Panel (a) shows the share of comments classified as conservative in stance, panel (b) the share with negative sentiment, panel (c) the share classified as offensive, and panel (d) the share classified as informative. The analysis is restricted to direct comments containing text and excludes GIFs or images. Vertical lines represent 95\% confidence intervals, and standard errors are cluster-robust at the advertisement level. Total observations are 1,393 comments.
\end{minipage}
\end{figure}

\begin{table}[H] \centering 
  \caption{Length and Depth of the Comment Section} 
  \label{tab:phase1length} 
\scalebox{0.7}{\begin{tabular}{@{\extracolsep{5pt}}lccccc} 
\\[-1.8ex]\hline 
\hline \\[-1.8ex] 
 & \multicolumn{5}{c}{\textit{Dependent variable:}} \\ 
\cline{2-6} 
 & Comment Size & Words per Comment & Sentences per Comment & Sentence Length & Comment Depth \\ 
\\[-1.8ex] & (1) & (2) & (3) & (4) & (5)\\ 
\hline \\[-1.8ex] 
 Red & $-$9.350 & $-$1.046 & 0.128 & $-$0.912$^{*}$ & 0.101 \\ 
  & (8.177) & (1.772) & (0.131) & (0.536) & (0.232) \\ 
  Swing & 10.798 & 2.780 & 0.271$^{*}$ & $-$0.597 & 0.060 \\ 
  & (11.076) & (2.360) & (0.150) & (0.554) & (0.169) \\ 
  Constant & 115.705$^{***}$ & 24.619$^{***}$ & 2.244$^{***}$ & 10.765$^{***}$ & 0.786$^{***}$ \\ 
  & (7.132) & (1.526) & (0.109) & (0.463) & (0.119) \\ 
 \hline \\[-1.8ex] 
Observations & 2,593 & 2,593 & 2,593 & 2,593 & 1,393 \\ 
R$^{2}$ & 0.002 & 0.002 & 0.001 & 0.001 & 0.0001 \\ 
\hline 
\hline \\[-1.8ex] 

\end{tabular}
}
\vspace{0.3em}

\begin{minipage}{\textwidth}
\scriptsize
\noindent \textit{Notes:} Heteroskedasticity-robust standard errors (HC1) are reported. The unit of observation is the comment. Comment Size is the number of characters excluding spaces. Words per Comment is the number of words. Sentences per Comment is the number of sentences. Sentence Length is token count divided by sentence count. Comment Depth is the number of direct replies to a parent comment, and uses parent comments only. Omitted baseline category is BLUE. $^{*}p<0.10$; $^{**}p<0.05$; $^{***}p<0.01$.
\end{minipage}
\end{table}

\begin{table}[H] \centering 
  \caption{Quality and Diversity of the Comment Section} 
  \label{tab:phase1quality} 
\scalebox{0.8}{

\begin{tabular}{@{\extracolsep{5pt}}lccccc} 
\\[-1.8ex]\hline 
\hline \\[-1.8ex] 
 & \multicolumn{5}{c}{\textit{Dependent variable:}} \\ 
\cline{2-6} 
 & Reciprocity & Justification & Respect & Simpson's Index & Variance in Political Stance \\ 
\\[-1.8ex] & (1) & (2) & (3) & (4) & (5)\\ 
\hline \\[-1.8ex] 
 Red & $-$0.0003 & 0.035 & $-$0.112$^{**}$ & $-$0.049 & $-$0.513$^{***}$ \\ 
  & (0.049) & (0.047) & (0.049) & (0.038) & (0.196) \\ 
  Swing & $-$0.057 & 0.001 & $-$0.144$^{***}$ & $-$0.030 & $-$0.260 \\ 
  & (0.046) & (0.039) & (0.048) & (0.038) & (0.194) \\ 
  Constant & 0.331$^{***}$ & 0.182$^{***}$ & 0.528$^{***}$ & 0.567$^{***}$ & 1.903$^{***}$ \\ 
  & (0.035) & (0.031) & (0.036) & (0.028) & (0.153) \\ 
 \hline \\[-1.8ex] 
Observations & 145 & 145 & 145 & 134 & 134 \\ 
R$^{2}$ & 0.014 & 0.006 & 0.067 & 0.013 & 0.054 \\ 
\hline 
\hline \\[-1.8ex] 

\end{tabular}
}
\vspace{0.3em}

\begin{minipage}{\textwidth}
\scriptsize
\noindent \textit{Notes:} Heteroskedasticity-robust standard errors (HC1) are reported. The unit of observation is the post. Reciprocity, Justification, and Respect are binary indicators coded for each direct comment.
Each measure is averaged across the direct comments of a post with equal weight. Reciprocity captures whether participants stay on topic and engage with prior comments rather than shifting to unrelated themes, making personal attacks, or using rhetorical questions that hinder argumentation. Justification captures whether participants who express or defend a viewpoint provide reasons or arguments for their position. Respect captures whether participants interact civilly, without threats, insults, vulgar language, identity attacks, humiliating language, or silencing expressions. For the Simpson’s Index and Variance in Political Stance, posts with only one usable comment are excluded. Simpson's Index is computed as \(1 - \sum_{k=1}^{K} p_k^2\), where \(p_k\) denotes the share of comments in political-stance category \(k\) within a post. Political Stance is coded on a five-point scale: (1) strongly progressive or left-leaning, (2) slightly or moderately progressive, (3) centrist, unclear, or no explicit stance, (4) slightly or moderately conservative, and (5) strongly conservative or right-leaning. Omitted baseline category is BLUE. $^{*}p<0.10$; $^{**}p<0.05$; $^{***}p<0.01$.
\end{minipage}
\end{table}

\section{Field Experiment: Additional Results and Robustness }

\setcounter{figure}{0}
\renewcommand{\thefigure}{E\arabic{figure}}
\renewcommand{\theHfigure}{E\arabic{figure}}

\setcounter{table}{0}
\renewcommand{\thetable}{E\arabic{table}}
\renewcommand{\theHtable}{E\arabic{table}}

\begin{table}[htbp]
\centering
\caption{Comments Displayed in the Field Experiment}
\label{tab:displayed_comments}
\footnotesize
\renewcommand{\arraystretch}{1.15}
\begin{tabular}{@{}c p{10.5cm} cccc@{}}
\toprule
Group & Comment (verbatim) & St & Se & In & Of \\
\midrule
\addlinespace
\multicolumn{6}{@{}l}{\textit{\textbf{Supportive comments}}} \\
\addlinespace[2pt]
1 & They also bulldozed black communities to build sports stadiums and other structures. & 1 & 3 & 2 & 1 \\
1 & The way that caucasian people study and stalk anything to do with black people just so they can be hateful, hurtful, and racist is a mental illness that needs to be studied, it really is. & 1 & 2 & 1 & 2 \\
2 & Yep put polluting plants upstream so the water gets polluted, then allow poor people to live nearby for work & 1 & 4 & 2 & 1 \\
2 & Environmental racism & 1 & 4 & 1 & 1 \\
3 & Because unfortunately still in 2025 we are racially bias. And anyone who says that ain't true lives in a world all their own. & 1 & 4 & 1 & 1 \\
3 & The asshats laughing have a huge lack of morals. & 1 & 3 & 1 & 2 \\
\addlinespace
\multicolumn{6}{@{}l}{\textit{\textbf{Mixed comments}}} \\
\addlinespace[2pt]
1 & It's simple data collection, one thing America excels at. The locations industry has placed factory farms ( concentration camps for non-human beings)slaughterhouses, rendering plants, are always in depressed areas where poc live and are perversely affected. You'd never see a pig farm or slaughterhouse placed in suburban or high income areas Oh no! Can't remind the mindless what their violent diets require. They might regain their long lost conscience and feel disgusted at how highly intelligent, sentient animals are violated and tortured. They'd gave to smell death. & 1 & 4 & 2 & 2 \\
1 & this is the stupidest thing I have seen on here all month, and that's saying a lot!! & 5 & 1 & 1 & 2 \\
2 & No green new scam wasting millions let people choose their car & 5 & 1 & 1 & 1 \\
2 & Racism, (im)pure and simple. & 1 & 5 & 1 & 1 \\
3 & I THINK MOST WHITE FOLKS DON'T CARE WHAT HAPPENS TO THE BLACK PEOPLE AS THEY ALWAYS HAVE. & 1 & 2 & 1 & 2 \\
3 & This is the stupidest thing I read today. Black communities do not have more pollution or any different water than anyone else! & 5 & 1 & 1 & 1 \\
\addlinespace
\multicolumn{6}{@{}l}{\textit{\textbf{Opposing comments}}} \\
\addlinespace[2pt]
1 & Nonsense. Places ran by democrats are more likely to face pollution risks & 5 & 1 & 1 & 1 \\
1 & What will be next, no end to the b.s. . & 5 & 1 & 1 & 1 \\
2 & Sooooooooooooooooooooooooo, WOKE \& more DEI, B/S!!!! Thank God we're moving back to "Common Sense"! & 5 & 1 & 1 & 2 \\
2 & Typical Progressive racism claptrap. & 5 & 1 & 1 & 2 \\
3 & it depends on the politicians you elect in those cities. Has nothing to do with the color of your skin. & 4 & 2 & 2 & 1 \\
3 & Enough with the race-baiting bullshit & 5 & 1 & 1 & 2 \\
\bottomrule
\end{tabular}

\begin{minipage}{\textwidth}
\vspace{4pt}\scriptsize
\noindent \textit{Notes:} This table reports the exact text of all comments displayed in the field experiment (Section~\ref{sec:comment_experiment}), with their GPT-4 scores. Commenter names and profile images are withheld. Each treatment post displayed two comments: Supportive posts two progressive comments, Opposing posts two conservative comments, and Mixed one of each; the three groups (Grp) correspond to the post triplets matched on reactions. Columns: St = Political Stance (1 = strongly progressive $\ldots$ 5 = strongly conservative); Se = Sentiment toward the ad (1 = highly negative $\ldots$ 5 = very positive); In = Informativeness (1 = not $\ldots$ 3 = highly informative); Of = Offensiveness (1 = not $\ldots$ 3 = highly offensive). Emoji have been removed for typesetting.
\end{minipage}
\end{table}

\begin{figure}[H]
\caption{Example of Opposing Comment and Hide Option}
\label{fig:hide2}
\centering
\includegraphics[width=0.7\textwidth]{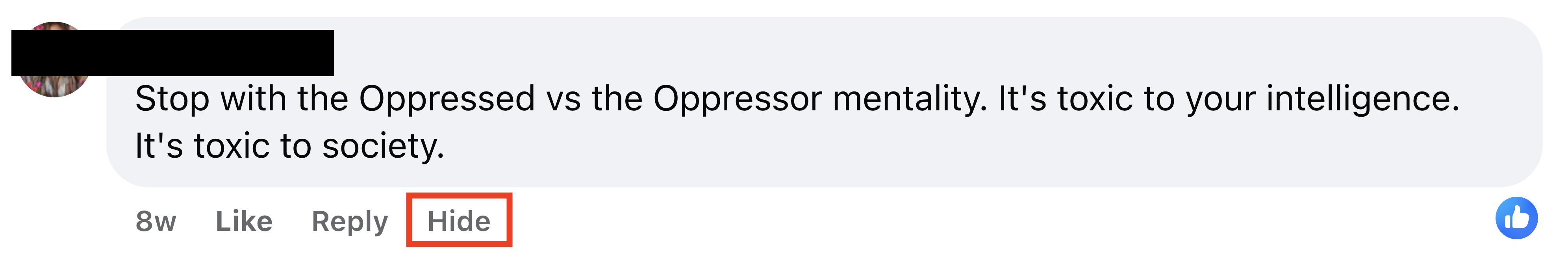} 
\begin{minipage}{\textwidth}
\scriptsize
\noindent \textit{Notes:} This figure presents an example comment and the hide option on Facebook. 
\end{minipage}
\end{figure}

\begin{table}[H]
\centering
\caption{Summary Statistics}
\label{tab:summary_stats}
\small
\begin{tabular}[t]{lccc}
\toprule\toprule
Variable name & Mean (\%) & St. Dev. & US Meta users 18+ (\%)\\
\midrule
\addlinespace[0.3em]
\multicolumn{4}{l}{\textbf{Treatment assignment (obs.)}}\\
\hspace{1em}Arm: Control (no comments) & 25.019 & 43.312 & -- \\
\hspace{1em}Arm: Supportive        & 24.881 & 43.232 & -- \\
\hspace{1em}Arm: Mixed             & 24.971 & 43.284 & -- \\
\hspace{1em}Arm: Opposing          & 25.129 & 43.376 & -- \\
\addlinespace[0.3em]\midrule
\multicolumn{4}{l}{\textbf{Main Outcomes}}\\
\hspace{1em}All engagement         & 0.535  & 7.292 & -- \\
\hspace{1em}Post expansions        & 0.290  & 5.380 & -- \\
\hspace{1em}Interactions           & 0.022  & 1.487 & -- \\
\hspace{1em}Link clicks            & 0.240  & 4.889 & -- \\
\hspace{1em}Page views             & 0.183  & 4.276 & -- \\
\addlinespace[0.3em]\midrule
\multicolumn{4}{l}{\textbf{Demographics}}\\
\hspace{1em}Gender: females        & 47.672 & 49.946 & 53.3 \\
\hspace{1em}Gender: males          & 52.328 & 49.946 & 46.8 \\
\addlinespace[0.3em]
\hspace{1em}Age: 18--24            & 8.692  & 28.172 & 18.0 \\
\hspace{1em}Age: 25--34            & 30.165 & 45.897 & 24.4 \\
\hspace{1em}Age: 35--44            & 27.158 & 44.478 & 19.2 \\
\hspace{1em}Age: 45--54            & 16.262 & 36.902 & 14.0 \\
\hspace{1em}Age: 55--64            & 10.241 & 30.319 & 11.8 \\
\hspace{1em}Age: 65+               & 7.482  & 26.309 & 12.7 \\
\addlinespace[0.3em]\midrule
\multicolumn{1}{l}{\textbf{Observations:}} & \multicolumn{2}{c}{\textbf{1,054,015}} & \\
\bottomrule
\end{tabular}

\vspace{0.3em}

\begin{minipage}{\textwidth}
\scriptsize
\noindent \centering \textit{Notes:} Values are expressed as percentages of reach. US Meta user data are from DataReportal \citep{kemp2025us}.
\end{minipage}

\end{table}

\begin{table}[H]
\centering
\caption{Balance Checks: Ad-level Delivery Metrics}
\label{tab:balance_adlevel}
\small
\resizebox{\textwidth}{!}{\begin{tabular}{lccccccc}
\toprule\toprule
& \multicolumn{4}{c}{\textit{Group Mean / (SD)}} & \multicolumn{3}{c}{\textit{t}--test \textit{p}--value} \\
\cmidrule(lr){2-5}\cmidrule(lr){6-8}
Variable &
(1) Control & (2) Supportive & (3) Mixed & (4) Opposing &
(1)–(2) & (1)–(3) & (1)–(4) \\[0.3em]
\midrule
Total Spend
 & \makecell{150.763 \\ (1.028)}
 & \makecell{150.276 \\ (1.202)}
 & \makecell{150.395 \\ (0.722)}
 & \makecell{150.542 \\ (1.078)}
 & 0.196 & 0.218 & 0.531 \\
\midrule
CPM
 & \makecell{9.594 \\ (2.130)}
 & \makecell{9.639 \\ (2.075)}
 & \makecell{9.626 \\ (2.065)}
 & \makecell{9.548 \\ (2.046)}
 & 0.949 & 0.964 & 0.947 \\
\midrule
Frequency
 & \makecell{1.123 \\ (0.026)}
 & \makecell{1.118 \\ (0.020)}
 & \makecell{1.116 \\ (0.023)}
 & \makecell{1.118 \\ (0.023)}
 & 0.502 & 0.356 & 0.520 \\
\midrule
Reach
 & \makecell{14650.333 \\ (3234.248)}
 & \makecell{14569.222 \\ (3150.366)}
 & \makecell{14622.056 \\ (3136.580)}
 & \makecell{14714.778 \\ (3103.458)}
 & 0.939 & 0.979 & 0.952 \\
\midrule
Spend per User
 & \makecell{0.011 \\ (0.003)}
 & \makecell{0.011 \\ (0.002)}
 & \makecell{0.011 \\ (0.002)}
 & \makecell{0.011 \\ (0.002)}
 & 0.994 & 0.952 & 0.894 \\
\midrule
Observations
 & 18 & 18 & 18 & 18 &  &  &  \\[0.3em]
\bottomrule
\end{tabular}
}
\vspace{0.3em}

\begin{minipage}{\textwidth}
\scriptsize
\noindent \textit{Notes:} Each observation is an ad. The p-values are based on t-tests using heteroskedasticity-robust standard errors
\end{minipage}

\end{table}

\begin{table}[H]
\scriptsize
\centering
\caption{Audience Saturation under Alternative Audience Size and Reach Estimates}
\label{tab:audience_saturation}
\begin{tabular}{llcc}
\toprule
 &  & \multicolumn{2}{c}{Campaign Reach Estimate} \\
\cmidrule(lr){3-4}
 &  & Lower Bound & Upper Bound \\
\multicolumn{2}{l}{Audience Size Estimate} & 993,572 & 1,054,015 \\
\midrule
Lower Bound & \multicolumn{1}{l|}{904,700}   & 1.098 & 1.165 \\
Midpoint    & \multicolumn{1}{l|}{975,000}   & 1.019 & 1.081 \\
Upper Bound & \multicolumn{1}{l|}{1,045,300 $\space \space$} & 0.950 & 1.008 \\
\bottomrule
\end{tabular}

\vspace{0.3em}

\begin{minipage}{\textwidth}
\scriptsize
\noindent \textit{Notes:} Audience size estimates correspond to Meta's ``Estimated Audience Size,'' defined as the number of accounts meeting the specified targeting criteria. We record these estimates at the start of the campaign. Meta reports these values as potential reach ranges based on recent platform activity and targeting configuration. Estimates may fluctuate over time as platform usage and measurement systems evolve. Campaign reach estimates are measured either at the campaign level (lower bound) or at the ad level and then aggregated across ads (upper bound). Saturation is computed as the reach estimate divided by the corresponding audience estimate.
\end{minipage}
\end{table}

\begin{table}[H] \centering 
  \caption{The Impact of the Comment Section on On-platform User Engagement} 
  \label{main_combined} 
\scalebox{0.78}{\begin{tabular}{@{\extracolsep{5pt}}lcccccccc} 
\\[-1.8ex]\hline 
\hline \\[-1.8ex] 
 & \multicolumn{8}{c}{\textit{Dependent variable:}} \\ 
\cline{2-9} 
 & \multicolumn{2}{c}{All Engagement} & \multicolumn{2}{c}{Post Expansions} & \multicolumn{2}{c}{Interactions} & \multicolumn{2}{c}{Link Clicks} \\ 
\\[-1.8ex] & (1) & (2) & (3) & (4) & (5) & (6) & (7) & (8)\\ 
\hline \\[-1.8ex] 
 Any comments & 0.065$^{***}$ &  & 0.049$^{***}$ &  & 0.003 &  & 0.016$^{**}$ &  \\ 
  & (0.015) &  & (0.013) &  & (0.003) &  & (0.008) &  \\ 
  & [0.000] &  & [0.005] &  & [0.405] &  & [0.070] &  \\ 
 Opposing &  & 0.087$^{***}$ &  & 0.048$^{***}$ &  & 0.009$^{**}$ &  & 0.034$^{***}$ \\ 
  &  & (0.021) &  & (0.017) &  & (0.004) &  & (0.011) \\ 
  &  & [0.002] &  & [0.033] &  & [0.077] &  & [0.019] \\ 
 Mixed &  & 0.062$^{***}$ &  & 0.049$^{***}$ &  & 0.002 &  & 0.012 \\ 
  &  & (0.017) &  & (0.013) &  & (0.003) &  & (0.011) \\ 
  &  & [0.003] &  & [0.002] &  & [0.531] &  & [0.399] \\ 
 Supportive &  & 0.047$^{***}$ &  & 0.049$^{***}$ &  & $-$0.003 &  & 0.003 \\ 
  &  & (0.018) &  & (0.015) &  & (0.003) &  & (0.008) \\ 
  &  & [0.024] &  & [0.012] &  & [0.366] &  & [0.764] \\ 
 Constant & 0.538$^{***}$ & 0.539$^{***}$ & 0.158$^{**}$ & 0.158$^{**}$ & 0.004 & 0.005 & 0.379$^{***}$ & 0.380$^{***}$ \\ 
  & (0.120) & (0.115) & (0.063) & (0.063) & (0.010) & (0.010) & (0.070) & (0.066) \\ 
  & [0.030] & [0.047] & [0.063] & [0.065] & [0.692] & [0.669] & [0.036] & [0.049] \\ 
 \hline \\[-1.8ex] 
ZIP Code Set FEs & \checkmark & \checkmark & \checkmark & \checkmark & \checkmark & \checkmark & \checkmark & \checkmark \\ 
Controls & \checkmark & \checkmark & \checkmark & \checkmark & \checkmark & \checkmark & \checkmark & \checkmark \\ 
ZIP Code Set~$\times$~Controls & \checkmark & \checkmark & \checkmark & \checkmark & \checkmark & \checkmark & \checkmark & \checkmark \\ 
\hline \\[-1.8ex] 

Mean Y in Control & 0.485 & 0.485 & 0.254 & 0.254 & 0.020 & 0.020 & 0.228 & 0.228 \\ 
p(Support vs. Oppose) &  & 0.059 &  & 0.974 &  & 0.002 &  & 0.008 \\ 
Wild p(Support vs. Oppose) &  & 0.140 &  & 0.979 &  & 0.014 &  & 0.034 \\ 
p(Oppose vs. Mixed) &  & 0.216 &  & 0.935 &  & 0.091 &  & 0.106 \\ 
Wild p(Oppose vs. Mixed) &  & 0.333 &  & 0.944 &  & 0.167 &  & 0.193 \\ 
p(Support vs. Mixed) &  & 0.394 &  & 0.960 &  & 0.078 &  & 0.432 \\ 
Wild p(Support vs. Mixed) &  & 0.458 &  & 0.967 &  & 0.126 &  & 0.539 \\ 
Observations & 1,054,015 & 1,054,015 & 1,054,015 & 1,054,015 & 1,054,015 & 1,054,015 & 1,054,015 & 1,054,015 \\ 
R$^{2}$ & 0.0004 & 0.0004 & 0.0005 & 0.0005 & 0.0003 & 0.0003 & 0.0002 & 0.0002 \\ 
\hline 
\hline \\[-1.8ex] 
\end{tabular}
}
\vspace{0.3em}

\begin{minipage}{\textwidth}
\scriptsize
\noindent \textit{Notes:} Standard errors clustered at the advertisement level (\feTwoAdsN{} ads) are reported in parentheses; wild cluster bootstrap $p$-values are reported in square brackets. Asterisks refer to the cluster-robust standard errors. The unit of observation is the user, constructed from ad-level aggregate data. Controls include user gender and age, as well as their pairwise interactions. The outcome ``Interactions'' includes comments, reactions, and shares.
$^{*}p<0.10$; $^{**}p<0.05$; $^{***}p<0.01$.
\end{minipage}
\end{table}

\begin{table}[H]
\centering
\caption{The Impact of the Comment Section on Ad-level Cost Metrics}
\label{tab:costs}
\small
\resizebox{\textwidth}{!}{\begin{tabular}{lccccccc}
\toprule\toprule
& \multicolumn{4}{c}{\textit{Group Mean / (SD)}} & \multicolumn{3}{c}{\textit{t}--test \textit{p}--value} \\
\cmidrule(lr){2-5}\cmidrule(lr){6-8}
Variable &
(1) Control & (2) Supportive & (3) Mixed & (4) Opposing &
(1)–(2) & (1)–(3) & (1)–(4) \\[0.3em]
\midrule
Spend per engagement
 & \makecell{2.131 \\ (0.337)}
 & \makecell{1.924 \\ (0.208)}
 & \makecell{1.879 \\ (0.283)}
 & \makecell{1.813 \\ (0.392)}
 & 0.033 & 0.021 & 0.011 \\
\midrule
Spend per interaction
 & \makecell{68.587 \\ (42.266)}
 & \makecell{88.450 \\ (48.005)}
 & \makecell{54.910 \\ (29.021)}
 & \makecell{36.098 \\ (13.111)}
 & 0.217 & 0.288 & 0.005 \\
\midrule
Spend per link click
 & \makecell{4.685 \\ (1.264)}
 & \makecell{4.497 \\ (0.817)}
 & \makecell{4.449 \\ (0.991)}
 & \makecell{3.996 \\ (0.944)}
 & 0.598 & 0.549 & 0.066 \\
\midrule
Observations
 & 18 & 18 & 18 & 18 &  &  &  \\[0.3em]
\bottomrule
\end{tabular}
}
\vspace{0.3em}

\begin{minipage}{\textwidth}
\scriptsize
\noindent \textit{Notes:} Each observation is an ad. Statistics are weighted by ad reach. The p-values are based on t-tests using heteroskedasticity-robust standard errors. Opposing comments reduce spend per engagement by \$0.30 (15 percent relative to control, \(p<0.05\)); mixed and supportive comments also reduce spend per engagement (\(p<0.05\)), though by smaller magnitudes. Spend per interaction falls by \$32.5 under opposing comments (a 50 percent reduction, \(p<0.01\)), and spend per link click by \$0.69 (15 percent, \(p<0.10\)), with no significant effects in the remaining conditions.
\end{minipage}
\end{table}

\begin{figure}[H]
\caption{The Impact of the Comment Section on All and Supportive Interactions \\ \textit{Outcomes are expressed in \% of total reach}}
\label{fig:results_interactions_valence}
\centering
\includegraphics[width=0.6\textwidth]{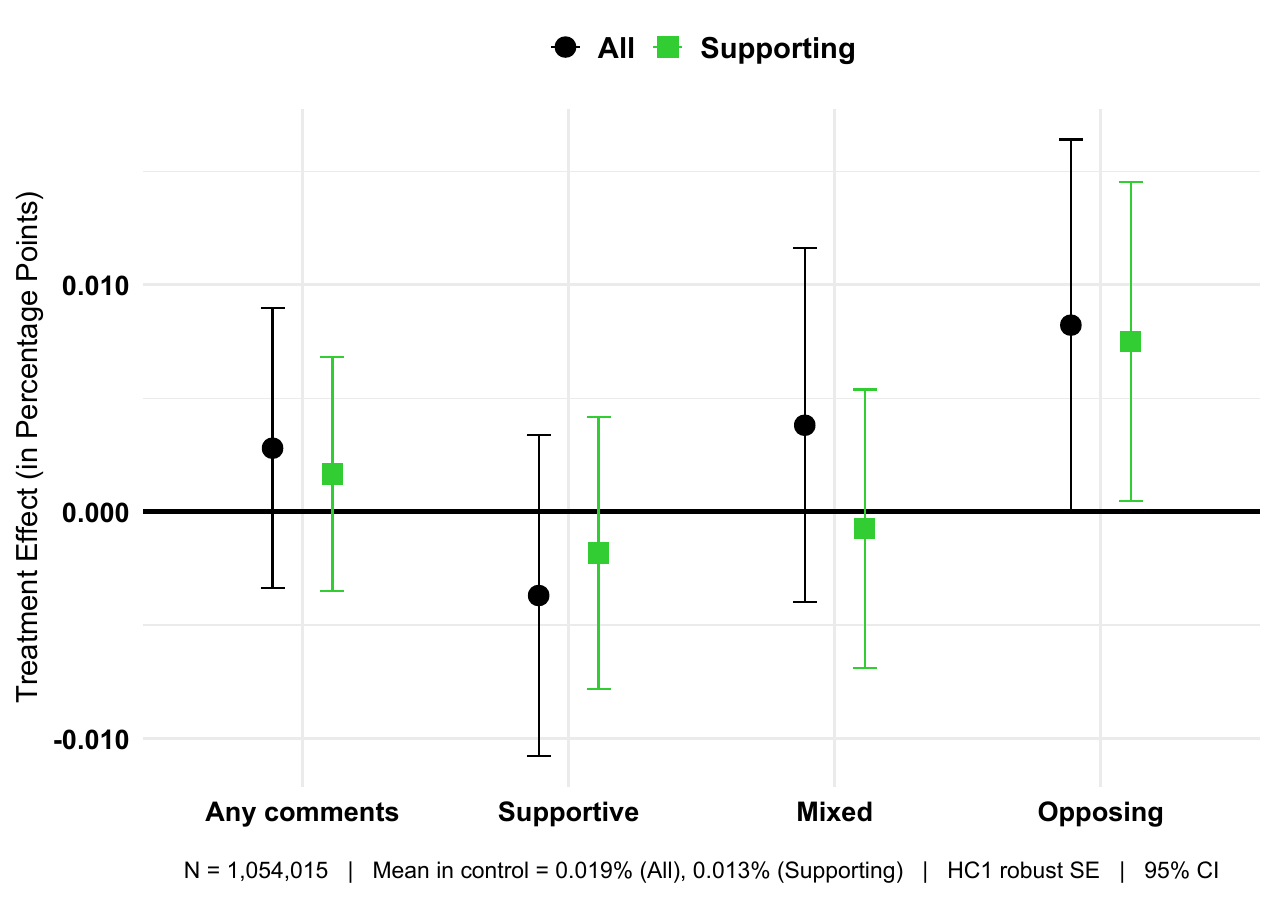}
\begin{minipage}{\textwidth}
\scriptsize
\noindent \textit{Notes:} This figure reports treatment effects of comment visibility and stance on all interactions and supportive interactions, expressed in percentage points of total reach. Interactions include comments, reactions, and shares. Supportive interactions include supportive comments, supportive reactions (likes, loves, and cares), and shares. Estimates are obtained from specifications using data collected directly from the posts, with heteroskedasticity-robust (HC1) standard errors. Vertical lines represent 95\% confidence intervals. The sample includes \feTwoReachN{} reached users.
\end{minipage}
\end{figure}

\begin{table}[H] \centering 
  \caption{The Impact of the Comment Section on the Valence of Subsequent Interactions (in \%)} 
  \label{tab:interactions_valence} 
\scalebox{0.9}{\begin{tabular}{@{\extracolsep{5pt}}lcccccc} 
\\[-1.8ex]\hline 
\hline \\[-1.8ex] 
 & \multicolumn{6}{c}{\textit{Dependent variable:}} \\ 
\cline{2-7} 
 & \multicolumn{2}{c}{All} & \multicolumn{2}{c}{Supportive} & \multicolumn{2}{c}{Non-supportive} \\ 
\\[-1.8ex] & (1) & (2) & (3) & (4) & (5) & (6)\\ 
\hline \\[-1.8ex] 
 Any comments & 0.0028 &  & 0.0017 &  & 0.0011 &  \\ 
  & (0.0032) &  & (0.0026) &  & (0.0017) &  \\ 
  & & & & & & \\ 
 Opposing &  & 0.0082$^{**}$ &  & 0.0075$^{**}$ &  & 0.0007 \\ 
  &  & (0.0042) &  & (0.0036) &  & (0.0021) \\ 
  & & & & & & \\ 
 Mixed &  & 0.0038 &  & $-$0.0008 &  & 0.0046$^{*}$ \\ 
  &  & (0.0040) &  & (0.0031) &  & (0.0025) \\ 
  & & & & & & \\ 
 Supportive &  & $-$0.0037 &  & $-$0.0018 &  & $-$0.0019 \\ 
  &  & (0.0036) &  & (0.0031) &  & (0.0019) \\ 
  & & & & & & \\ 
 Constant & 0.0210$^{***}$ & 0.0210$^{***}$ & 0.0143$^{***}$ & 0.0143$^{***}$ & 0.0067$^{***}$ & 0.0067$^{***}$ \\ 
  & (0.0034) & (0.0034) & (0.0027) & (0.0027) & (0.0019) & (0.0019) \\ 
  & & & & & & \\ 
\hline \\[-1.8ex] 
ZIP Code Set FEs & \checkmark & \checkmark &\checkmark &\checkmark & \checkmark & \checkmark\\ 
\hline \\[-1.8ex] 

Mean Y in Control (C) & 0.0190 & 0.0190 & 0.0133 & 0.0133 & 0.0057 & 0.0057 \\ 
p(Support vs. Oppose) &  & 0.003 &  & 0.008 &  & 0.186 \\ 
p(Oppose vs. Mixed) &  & 0.310 &  & 0.020 &  & 0.128 \\ 
p(Support vs. Mixed) &  & 0.048 &  & 0.722 &  & 0.005 \\ 
Observations & 1,054,015 & 1,054,015 & 1,054,015 & 1,054,015 & 1,054,015 & 1,054,015 \\ 
R$^{2}$ & 0.00001 & 0.00002 & 0.000005 & 0.00001 & 0.000002 & 0.00001 \\ 

\hline 
\hline \\[-1.8ex] 

\end{tabular}
}
\vspace{0.3em}

\begin{minipage}{\textwidth}
\scriptsize
\noindent \textit{Notes:} Heteroskedasticity-robust standard errors (HC1) are reported. The unit of observation is the user, constructed from ad-level aggregate data. Supportive interactions include supportive comments, supportive reactions, and shares. Because the stance of comments and the types of reactions are not available from the Meta Ads Manager but are instead retrieved directly from the posts, we cannot include user-level controls and cluster the standard errors at the ad level in this analysis. The outcome ``Interactions'' includes comments, reactions, and shares.$^{*}p<0.10$; $^{**}p<0.05$; $^{***}p<0.01$.
\end{minipage}
\end{table}

\begin{table}[H] \centering 
  \caption{The Impact of the Comment Section in Blue Areas} 
  \label{tab:blue_zip} 
\scalebox{0.8}{\begin{tabular}{@{\extracolsep{5pt}}lcccccccc} 
\\[-1.8ex]\hline 
\hline \\[-1.8ex] 
 & \multicolumn{8}{c}{\textit{Dependent variable:}} \\ 
\cline{2-9} 
 & \multicolumn{2}{c}{All Engagement} & \multicolumn{2}{c}{Post Expansions} & \multicolumn{2}{c}{Interactions} & \multicolumn{2}{c}{Link Clicks} \\ 
\\[-1.8ex] & (1) & (2) & (3) & (4) & (5) & (6) & (7) & (8)\\ 
\hline \\[-1.8ex] 
 Any comments & 0.031 &  & 0.014 &  & 0.003 &  & 0.009 &  \\ 
  & (0.025) &  & (0.024) &  & (0.005) &  & (0.015) &  \\ 
  & [0.309] &  & [0.651] &  & [0.651] &  & [0.589] &  \\ 
 Opposing &  & 0.053$^{*}$ &  & 0.023 &  & 0.009 &  & 0.017 \\ 
  &  & (0.030) &  & (0.028) &  & (0.007) &  & (0.021) \\ 
  &  & [0.199] &  & [0.530] &  & [0.293] &  & [0.542] \\ 
 Mixed &  & 0.053$^{*}$ &  & 0.029 &  & 0.002 &  & 0.022 \\ 
  &  & (0.032) &  & (0.026) &  & (0.006) &  & (0.025) \\ 
  &  & [0.234] &  & [0.381] &  & [0.800] &  & [0.548] \\ 
 Supportive &  & $-$0.014 &  & $-$0.010 &  & $-$0.003 &  & $-$0.012 \\ 
  &  & (0.031) &  & (0.028) &  & (0.005) &  & (0.015) \\ 
  &  & [0.714] &  & [0.792] &  & [0.640] &  & [0.486] \\ 
 Constant & 0.559$^{***}$ & 0.561$^{***}$ & 0.155$^{**}$ & 0.156$^{**}$ & $-$0.007 & $-$0.007 & 0.416$^{***}$ & 0.417$^{***}$ \\ 
  & (0.129) & (0.117) & (0.072) & (0.066) & (0.013) & (0.012) & (0.077) & (0.072) \\ 
  & [0.028] & [0.035] & [0.085] & [0.060] & [0.588] & [0.608] & [0.019] & [0.026] \\ 
 \hline \\[-1.8ex] 
ZIP Code Set FEs & \checkmark & \checkmark & \checkmark & \checkmark & \checkmark & \checkmark & \checkmark & \checkmark \\ 
Controls & \checkmark & \checkmark & \checkmark & \checkmark & \checkmark & \checkmark & \checkmark & \checkmark \\ 
ZIP Code Set~$\times$~Controls & \checkmark & \checkmark & \checkmark & \checkmark & \checkmark & \checkmark & \checkmark & \checkmark \\ 
\hline \\[-1.8ex] 

Mean Y in Control & 0.529 & 0.529 & 0.267 & 0.267 & 0.020 & 0.020 & 0.267 & 0.267 \\ 
p(Support vs. Oppose) &  & 0.052 &  & 0.191 &  & 0.029 &  & 0.195 \\ 
Wild p(Support vs. Oppose) &  & 0.148 &  & 0.319 &  & 0.093 &  & 0.334 \\ 
p(Oppose vs. Mixed) &  & 0.996 &  & 0.788 &  & 0.270 &  & 0.862 \\ 
Wild p(Oppose vs. Mixed) &  & 0.997 &  & 0.823 &  & 0.405 &  & 0.914 \\ 
p(Support vs. Mixed) &  & 0.065 &  & 0.074 &  & 0.363 &  & 0.184 \\ 
Wild p(Support vs. Mixed) &  & 0.202 &  & 0.182 &  & 0.465 &  & 0.376 \\ 
Observations & 329,844 & 329,844 & 329,844 & 329,844 & 329,844 & 329,844 & 329,844 & 329,844 \\ 
R$^{2}$ & 0.0004 & 0.0004 & 0.0005 & 0.0005 & 0.0003 & 0.0003 & 0.0002 & 0.0002 \\ 
\hline 
\hline \\[-1.8ex] 
\end{tabular}
}
\vspace{0.3em}

\begin{minipage}{\textwidth}
\scriptsize
\noindent \textit{Notes:} Standard errors clustered at the advertisement level (\feTwoBlueClustersN{} ads) are reported in parentheses; wild cluster bootstrap $p$-values are reported in square brackets. Asterisks refer to the cluster-robust standard errors. The unit of observation is the user, constructed from ad-level aggregate data. Controls include user gender and age, as well as their pairwise interactions. The outcome ``Interactions'' includes comments, reactions, and shares.
$^{*}p<0.10$; $^{**}p<0.05$; $^{***}p<0.01$.
\end{minipage}
\end{table}

\begin{table}[H] \centering 
  \caption{The Impact of the Comment Section in Swing Areas} 
  \label{tab:swing_zip}  
\scalebox{0.8}{\begin{tabular}{@{\extracolsep{5pt}}lcccccccc} 
\\[-1.8ex]\hline 
\hline \\[-1.8ex] 
 & \multicolumn{8}{c}{\textit{Dependent variable:}} \\ 
\cline{2-9} 
 & \multicolumn{2}{c}{All Engagement} & \multicolumn{2}{c}{Post Expansions} & \multicolumn{2}{c}{Interactions} & \multicolumn{2}{c}{Link Clicks} \\ 
\\[-1.8ex] & (1) & (2) & (3) & (4) & (5) & (6) & (7) & (8)\\ 
\hline \\[-1.8ex] 
 Any comments & 0.064$^{***}$ &  & 0.044$^{***}$ &  & 0.000 &  & 0.025$^{**}$ &  \\ 
  & (0.020) &  & (0.014) &  & (0.005) &  & (0.011) &  \\ 
  & [0.016] &  & [0.027] &  & [0.954] &  & [0.085] &  \\ 
 Opposing &  & 0.071$^{*}$ &  & 0.040$^{*}$ &  & 0.001 &  & 0.037$^{**}$ \\ 
  &  & (0.037) &  & (0.025) &  & (0.008) &  & (0.015) \\ 
  &  & [0.201] &  & [0.268] &  & [0.904] &  & [0.101] \\ 
 Mixed &  & 0.052$^{***}$ &  & 0.039$^{***}$ &  & 0.004 &  & 0.017 \\ 
  &  & (0.019) &  & (0.014) &  & (0.005) &  & (0.014) \\ 
  &  & [0.043] &  & [0.041] &  & [0.449] &  & [0.331] \\ 
 Supportive &  & 0.070$^{***}$ &  & 0.052$^{***}$ &  & $-$0.004 &  & 0.020 \\ 
  &  & (0.022) &  & (0.015) &  & (0.006) &  & (0.013) \\ 
  &  & [0.023] &  & [0.028] &  & [0.563] &  & [0.207] \\ 
 Constant & 0.317$^{***}$ & 0.317$^{***}$ & 0.181$^{***}$ & 0.181$^{***}$ & 0.021 & 0.021 & 0.110$^{**}$ & 0.109$^{**}$ \\ 
  & (0.067) & (0.068) & (0.047) & (0.047) & (0.017) & (0.017) & (0.051) & (0.053) \\ 
  & [0.028] & [0.032] & [0.007] & [0.004] & [0.344] & [0.389] & [0.089] & [0.092] \\ 
 \hline \\[-1.8ex] 
ZIP Code Set FEs & \checkmark & \checkmark & \checkmark & \checkmark & \checkmark & \checkmark & \checkmark & \checkmark \\ 
Controls & \checkmark & \checkmark & \checkmark & \checkmark & \checkmark & \checkmark & \checkmark & \checkmark \\ 
ZIP Code Set~$\times$~Controls & \checkmark & \checkmark & \checkmark & \checkmark & \checkmark & \checkmark & \checkmark & \checkmark \\ 
\hline \\[-1.8ex] 

Mean Y in Control & 0.478 & 0.478 & 0.254 & 0.254 & 0.025 & 0.025 & 0.209 & 0.209 \\ 
p(Support vs. Oppose) &  & 0.996 &  & 0.628 &  & 0.556 &  & 0.282 \\ 
Wild p(Support vs. Oppose) &  & 0.997 &  & 0.737 &  & 0.664 &  & 0.438 \\ 
p(Oppose vs. Mixed) &  & 0.604 &  & 0.956 &  & 0.789 &  & 0.237 \\ 
Wild p(Oppose vs. Mixed) &  & 0.730 &  & 0.969 &  & 0.839 &  & 0.406 \\ 
p(Support vs. Mixed) &  & 0.382 &  & 0.327 &  & 0.219 &  & 0.844 \\ 
Wild p(Support vs. Mixed) &  & 0.455 &  & 0.392 &  & 0.320 &  & 0.878 \\ 
Observations & 345,109 & 345,109 & 345,109 & 345,109 & 345,109 & 345,109 & 345,109 & 345,109 \\ 
R$^{2}$ & 0.0005 & 0.0005 & 0.0006 & 0.0006 & 0.0003 & 0.0003 & 0.0002 & 0.0002 \\ 
\hline 
\hline \\[-1.8ex] 
\end{tabular}
}
\vspace{0.3em}

\begin{minipage}{\textwidth}
\scriptsize
\noindent \textit{Notes:} Standard errors clustered at the advertisement level (\feTwoSwingClustersN{} ads) are reported in parentheses; wild cluster bootstrap $p$-values are reported in square brackets. Asterisks refer to the cluster-robust standard errors. The unit of observation is the user, constructed from ad-level aggregate data. Controls include user gender and age, as well as their pairwise interactions. The outcome ``Interactions'' includes comments, reactions, and shares.
$^{*}p<0.10$; $^{**}p<0.05$; $^{***}p<0.01$.
\end{minipage}
\end{table}

\begin{table}[H] \centering 
  \caption{The Impact of the Comment Section in Red Areas} 
  \label{tab:red_zip} 
\scalebox{0.8}{\begin{tabular}{@{\extracolsep{5pt}}lcccccccc} 
\\[-1.8ex]\hline 
\hline \\[-1.8ex] 
 & \multicolumn{8}{c}{\textit{Dependent variable:}} \\ 
\cline{2-9} 
 & \multicolumn{2}{c}{All Engagement} & \multicolumn{2}{c}{Post Expansions} & \multicolumn{2}{c}{Interactions} & \multicolumn{2}{c}{Link Clicks} \\ 
\\[-1.8ex] & (1) & (2) & (3) & (4) & (5) & (6) & (7) & (8)\\ 
\hline \\[-1.8ex] 
 Any comments & 0.096$^{***}$ &  & 0.083$^{***}$ &  & 0.004 &  & 0.014 &  \\ 
  & (0.030) &  & (0.024) &  & (0.004) &  & (0.014) &  \\ 
  & [0.012] &  & [0.026] &  & [0.418] &  & [0.402] &  \\ 
 Opposing &  & 0.132$^{***}$ &  & 0.077$^{**}$ &  & 0.015$^{***}$ &  & 0.045$^{**}$ \\ 
  &  & (0.041) &  & (0.034) &  & (0.005) &  & (0.021) \\ 
  &  & [0.015] &  & [0.103] &  & [0.031] &  & [0.137] \\ 
 Mixed &  & 0.077$^{**}$ &  & 0.076$^{***}$ &  & 0.001 &  & $-$0.003 \\ 
  &  & (0.031) &  & (0.024) &  & (0.005) &  & (0.016) \\ 
  &  & [0.068] &  & [0.026] &  & [0.869] &  & [0.876] \\ 
 Supportive &  & 0.080$^{**}$ &  & 0.097$^{***}$ &  & $-$0.002 &  & $-$0.000 \\ 
  &  & (0.031) &  & (0.027) &  & (0.004) &  & (0.014) \\ 
  &  & [0.058] &  & [0.020] &  & [0.681] &  & [0.995] \\ 
 Constant & 0.482$^{***}$ & 0.482$^{***}$ & 0.286$^{**}$ & 0.286$^{**}$ & $-$0.015$^{**}$ & $-$0.014$^{**}$ & 0.212$^{***}$ & 0.212$^{***}$ \\ 
  & (0.152) & (0.145) & (0.121) & (0.120) & (0.006) & (0.007) & (0.049) & (0.041) \\ 
  & [0.059] & [0.044] & [0.129] & [0.135] & [0.055] & [0.086] & [0.010] & [0.025] \\ 
 \hline \\[-1.8ex] 
ZIP Code Set FEs & \checkmark & \checkmark & \checkmark & \checkmark & \checkmark & \checkmark & \checkmark & \checkmark \\ 
Controls & \checkmark & \checkmark & \checkmark & \checkmark & \checkmark & \checkmark & \checkmark & \checkmark \\ 
ZIP Code Set~$\times$~Controls & \checkmark & \checkmark & \checkmark & \checkmark & \checkmark & \checkmark & \checkmark & \checkmark \\ 
\hline \\[-1.8ex] 

Mean Y in Control & 0.455 & 0.455 & 0.242 & 0.242 & 0.016 & 0.016 & 0.210 & 0.210 \\ 
p(Support vs. Oppose) &  & 0.145 &  & 0.539 &  & 0.000 &  & 0.027 \\ 
Wild p(Support vs. Oppose) &  & 0.290 &  & 0.655 &  & 0.003 &  & 0.113 \\ 
p(Oppose vs. Mixed) &  & 0.127 &  & 0.976 &  & 0.001 &  & 0.027 \\ 
Wild p(Oppose vs. Mixed) &  & 0.261 &  & 0.981 &  & 0.004 &  & 0.132 \\ 
p(Support vs. Mixed) &  & 0.919 &  & 0.349 &  & 0.370 &  & 0.849 \\ 
Wild p(Support vs. Mixed) &  & 0.924 &  & 0.432 &  & 0.449 &  & 0.876 \\ 
Observations & 379,062 & 379,062 & 379,062 & 379,062 & 379,062 & 379,062 & 379,062 & 379,062 \\ 
R$^{2}$ & 0.0005 & 0.0005 & 0.0006 & 0.0006 & 0.0002 & 0.0002 & 0.0001 & 0.0002 \\ 
\hline 
\hline \\[-1.8ex] 
\end{tabular}
}
\vspace{0.3em}

\begin{minipage}{\textwidth}
\scriptsize
\noindent \textit{Notes:} Standard errors clustered at the advertisement level (\feTwoRedClustersN{} ads) are reported in parentheses; wild cluster bootstrap $p$-values are reported in square brackets. Asterisks refer to the cluster-robust standard errors. The unit of observation is the user, constructed from ad-level aggregate data. Controls include user gender and age, as well as their pairwise interactions. The outcome ``Interactions'' includes comments, reactions, and shares.
$^{*}p<0.10$; $^{**}p<0.05$; $^{***}p<0.01$.
\end{minipage}
\end{table}

\begin{sidewaystable}[htp]
\centering\small
\caption{The Impact of Any Comment on On-platform User Engagement}
\label{tab:robustness_anycomment_controls}
\begin{adjustbox}{width=\textwidth,keepaspectratio}
\begin{threeparttable}

\begin{tabular}{l*{12}{c}}
\toprule
& \multicolumn{3}{c}{All Engagement}
& \multicolumn{3}{c}{Post Expansions}
& \multicolumn{3}{c}{Interactions}
& \multicolumn{3}{c}{Link Clicks} \\
\cmidrule(lr){2-4}\cmidrule(lr){5-7}\cmidrule(lr){8-10}\cmidrule(lr){11-13}
 & (1) & (2) & (3)
 & (4) & (5) & (6)
 & (7) & (8) & (9)
 & (10) & (11) & (12) \\
\midrule
Any comments
  & 0.065$^{***}$ & 0.065$^{***}$ & 0.066$^{***}$
  & 0.049$^{***}$ & 0.049$^{**}$  & 0.049$^{***}$
  & 0.003         & 0.003         & 0.003
  & 0.016$^{**}$  & 0.016         & 0.016 \\
  & (0.015) & (0.025) & (0.024)
  & (0.013) & (0.019) & (0.019)
  & (0.003) & (0.004) & (0.004)
  & (0.008) & (0.014) & (0.014) \\[0.5ex]

Constant
  & 0.566$^{***}$ & 0.418$^{***}$ & 0.485$^{***}$
  & 0.282$^{***}$ & 0.178$^{***}$ & 0.254$^{***}$
  & 0.034$^{***}$ & 0.006         & 0.020$^{***}$
  & 0.276$^{***}$ & 0.243$^{***}$ & 0.228$^{***}$ \\
  & (0.046) & (0.035) & (0.021)
  & (0.045) & (0.025) & (0.016)
  & (0.008) & (0.004) & (0.003)
  & (0.021) & (0.025) & (0.012) \\
\midrule
ZIP Code Set FEs
  & \multicolumn{1}{c}{\cmark} &  & 
  & \multicolumn{1}{c}{\cmark} &  &
  & \multicolumn{1}{c}{\cmark} &  &
  & \multicolumn{1}{c}{\cmark} &  & \\
Controls
  &  & \multicolumn{1}{c}{\cmark} & 
  &  & \multicolumn{1}{c}{\cmark} &
  &  & \multicolumn{1}{c}{\cmark} &
  &  & \multicolumn{1}{c}{\cmark} & \\
\midrule
Mean \(Y\) in Control (C)
  & 0.485 & 0.485 & 0.485
  & 0.254 & 0.254 & 0.254
  & 0.020 & 0.020 & 0.020
  & 0.228 & 0.228 & 0.228 \\

Observations
  & \multicolumn{12}{c}{1,054,015} \\
$R^{2}$
  & 0.0001 & 0.0002 & 0.00002
  & 0.0001 & 0.0002 & 0.00002
  & 0.00002 & 0.0001 & 0.00000
  & 0.0001 & 0.00004 & 0.00000 \\
\bottomrule
\end{tabular}

\vspace{0.3em}

\parbox{\linewidth}{\textit{Notes:} Standard errors are clustered at the advertisement level (\feTwoAdsN{} ads). The unit of observation is the user, constructed from ad-level aggregate data.  Columns (1), (4), (7), (10) include ZIP code set fixed effects only. Columns (2), (5), (8), (11) include gender and age controls only. Columns (3), (6), (9), (12) do not include ZIP code set fixed effects and controls. The outcome ``Interactions'' includes comments, reactions and shares. $^{*}p<0.10$; $^{**}p<0.05$; $^{***}p<0.01$.}
\end{threeparttable}
\end{adjustbox}
\end{sidewaystable}

\begin{sidewaystable}[htp]
\centering\small
\caption{The Impact of Comment Stance on On-platform User Engagement}
\label{tab:robustness_stance_controls}
\begin{adjustbox}{width=\textwidth,keepaspectratio}
\begin{threeparttable}
\begin{tabular}{l*{12}{c}}
\toprule
& \multicolumn{3}{c}{All Engagement}
& \multicolumn{3}{c}{Post Expansions}
& \multicolumn{3}{c}{Interactions}
& \multicolumn{3}{c}{Link Clicks} \\
\cmidrule(lr){2-4}\cmidrule(lr){5-7}\cmidrule(lr){8-10}\cmidrule(lr){11-13}
 & (1) & (2) & (3)
 & (4) & (5) & (6)
 & (7) & (8) & (9)
 & (10) & (11) & (12) \\
\midrule
Opposing
  & 0.086$^{***}$ & 0.087$^{***}$ & 0.087$^{***}$
  & 0.048$^{***}$ & 0.048$^{*}$   & 0.048$^{*}$
  & 0.009$^{**}$  & 0.009$^{*}$   & 0.009$^{*}$
  & 0.033$^{***}$ & 0.034$^{**}$  & 0.034$^{**}$ \\
  & (0.021) & (0.032) & (0.031)
  & (0.017) & (0.025) & (0.025)
  & (0.004) & (0.005) & (0.005)
  & (0.011) & (0.017) & (0.017) \\[0.5ex]

Mixed
  & 0.061$^{***}$ & 0.061$^{**}$ & 0.062$^{**}$
  & 0.049$^{***}$ & 0.049$^{**}$ & 0.050$^{**}$
  & 0.002         & 0.002        & 0.002
  & 0.011         & 0.011        & 0.011 \\
  & (0.017) & (0.029) & (0.028)
  & (0.013) & (0.021) & (0.020)
  & (0.003) & (0.004) & (0.004)
  & (0.011) & (0.018) & (0.018) \\[0.5ex]

Supportive
  & 0.047$^{***}$ & 0.047$^{*}$ & 0.048$^{*}$
  & 0.049$^{***}$ & 0.049$^{**}$ & 0.049$^{**}$
  & $-$0.003      & $-$0.003     & $-$0.003
  & 0.002         & 0.003        & 0.003 \\
  & (0.017) & (0.028) & (0.028)
  & (0.015) & (0.023) & (0.022)
  & (0.003) & (0.004) & (0.004)
  & (0.008) & (0.014) & (0.014) \\[0.5ex]

Constant
  & 0.566$^{***}$ & 0.418$^{***}$ & 0.485$^{***}$
  & 0.282$^{***}$ & 0.178$^{***}$ & 0.254$^{***}$
  & 0.034$^{***}$ & 0.006         & 0.020$^{***}$
  & 0.276$^{***}$ & 0.243$^{***}$ & 0.228$^{***}$ \\
  & (0.046) & (0.035) & (0.021)
  & (0.045) & (0.025) & (0.016)
  & (0.006) & (0.004) & (0.003)
  & (0.021) & (0.025) & (0.012) \\
\midrule
ZIP Code Set FEs
  & \multicolumn{1}{c}{\cmark} &  & 
  & \multicolumn{1}{c}{\cmark} &  &
  & \multicolumn{1}{c}{\cmark} &  &
  & \multicolumn{1}{c}{\cmark} &  & \\
Controls
  &  & \multicolumn{1}{c}{\cmark} & 
  &  & \multicolumn{1}{c}{\cmark} &
  &  & \multicolumn{1}{c}{\cmark} &
  &  & \multicolumn{1}{c}{\cmark} & \\
\midrule
Mean \(Y\) in Control (C)
  & 0.485 & 0.485 & 0.485
  & 0.254 & 0.254 & 0.254
  & 0.020 & 0.020 & 0.020
  & 0.228 & 0.228 & 0.228 \\

p(Support vs. Oppose)
  & 0.067 & 0.185 & 0.187
  & 0.935 & 0.978 & 0.964
  & 0.003 & 0.010 & 0.010
  & 0.007 & 0.027 & 0.027 \\

p(Oppose vs. Mixed)
  & 0.230 & 0.406 & 0.396
  & 0.914 & 0.948 & 0.944
  & 0.101 & 0.123 & 0.121
  & 0.105 & 0.218 & 0.217 \\

p(Support vs. Mixed)
  & 0.410 & 0.587 & 0.591
  & 0.987 & 0.968 & 0.982
  & 0.081 & 0.144 & 0.139
  & 0.427 & 0.585 & 0.590 \\

Observations
  & \multicolumn{12}{c}{1,054,015} \\

$R^{2}$
  & 0.0001 & 0.0002 & 0.00002
  & 0.0001 & 0.0002 & 0.00002
  & 0.00003 & 0.0001 & 0.00001
  & 0.0001 & 0.00005 & 0.00001 \\
\bottomrule
\end{tabular}

\vspace{0.3em}

\parbox{\linewidth}{\textit{Notes:} Standard errors are clustered at the advertisement level (\feTwoAdsN{} ads). The unit of observation is the user, constructed from ad-level aggregate data.  Columns (1), (4), (7), (10) include ZIP code set fixed effects only. Columns (2), (5), (8), (11) include gender and age controls only. Columns (3), (6), (9), (12) do not include ZIP code set fixed effects and controls. The outcome ``Interactions'' includes comments, reactions and shares. $^{*}p<0.10$; $^{**}p<0.05$; $^{***}p<0.01$.}
\end{threeparttable}
\end{adjustbox}
\end{sidewaystable}

\newpage
\begin{table}[H]
\centering
\caption{The Impact of the Comment Section on On-platform User Engagement (Logistic Regression)} 
\label{tab:logit}
\scalebox{0.65}{\begin{tabular}{@{\extracolsep{5pt}}p{4.2cm}*{8}{>{\centering\arraybackslash}p{2.3cm}}}
\\[-1.8ex]\hline
\hline \\[-1.8ex]
 & \multicolumn{8}{c}{\textit{Dependent variable (Odds Ratios)}} \\
 \cline{2-9}
 & \multicolumn{2}{c}{All Engagement}
 & \multicolumn{2}{c}{Post Expansions}
 & \multicolumn{2}{c}{Interactions}
 & \multicolumn{2}{c}{Link Clicks} \\
\\[-1.8ex]
 & (1) & (2) & (3) & (4) & (5) & (6) & (7) & (8) \\
\hline \\[-1.8ex]

Any comments
 & 1.136$^{***}$ & 
 & 1.193$^{***}$ &
 & 1.131 &
 & 1.070 &  \\
 & (0.036) & 
 & (0.052) &
 & (0.177) &
 & (0.050) & \\[0.8ex]

Opposing
 &  & 1.181$^{***}$
 &  & 1.191$^{***}$
 &  & 1.436$^{**}$
 &  & 1.148$^{**}$ \\
 &  & (0.045)
 &  & (0.063)
 &  & (0.257)
 &  & (0.064) \\

Mixed
 &  & 1.128$^{***}$
 &  & 1.196$^{***}$
 &  & 1.110
 &  & 1.051 \\
 &  & (0.043)
 &  & (0.063)
 &  & (0.210)
 &  & (0.060) \\

Supportive
 &  & 1.098$^{**}$
 &  & 1.193$^{***}$
 &  & 0.848
 &  & 1.012 \\
 &  & (0.043)
 &  & (0.063)
 &  & (0.172)
 &  & (0.058) \\[0.8ex]

Constant
 & 0.005$^{***}$ & 0.005$^{***}$
 & 0.002$^{***}$ & 0.002$^{***}$
 & 0.000 & 0.000
 & 0.004$^{***}$ & 0.004$^{***}$ \\
 & (0.001) & (0.001)
 & (0.001) & (0.001)
 & (0.00000) & (0.00000)
 & (0.001) & (0.001) \\

\hline \\[-1.8ex]

ZIP Code Set FEs
 & \checkmark & \checkmark
 & \checkmark & \checkmark
 & \checkmark & \checkmark
 & \checkmark & \checkmark \\

Controls
 & \checkmark & \checkmark
 & \checkmark & \checkmark
 & \checkmark & \checkmark
 & \checkmark & \checkmark \\

\mbox{ZIP Code Set~$\times$~Controls}
 & \checkmark & \checkmark
 & \checkmark & \checkmark
 & \checkmark & \checkmark
 & \checkmark & \checkmark \\

\hline \\[-1.8ex]

\mbox{Mean Y in Control (C)}
 & 0.485 & 0.485
 & 0.254 & 0.254
 & 0.020 & 0.020
 & 0.228 & 0.228 \\

Observations
 & 1,054,015 & 1,054,015
 & 1,054,015 & 1,054,015
 & 1,054,015 & 1,054,015
 & 1,054,015 & 1,054,015 \\

\hline
\hline \\[-1.8ex]

\end{tabular}
}
\vspace{0.3em}

\begin{minipage}{\textwidth}
\scriptsize
\noindent \textit{Notes:} Coefficients reported as odds ratios. The unit of observation is the user, constructed from ad-level aggregate data. Controls include user gender and age, as well as their pairwise interactions. The outcome ``Interactions'' includes comments, reactions and shares. $^{*}p<0.10$; $^{**}p<0.05$; $^{***}p<0.01$.
\end{minipage}

\end{table}

\begin{table}[H]
\centering
\caption{Ad-level Estimates}
\label{tab:outcomes_model_free}
\scriptsize
\resizebox{\textwidth}{!}{
\begin{tabular}{lcccccccc}
\toprule\toprule
& \multicolumn{4}{c}{\textit{Group Mean / (SD)}} & \multicolumn{4}{c}{t--test \textit{p}--value} \\
\cmidrule(lr){2-5}\cmidrule(lr){6-9}
Variable &
(1) Control & (2) Supportive & (3) Mixed & (4) Opposing &
(1)–(2) & (1)–(3) & (1)–(4) & (2)–(4) \\[0.3em]
\midrule
All Engagement
 & \makecell{0.485 \\ (0.091)}
 & \makecell{0.533 \\ (0.080)}
 & \makecell{0.547 \\ (0.078)}
 & \makecell{0.572 \\ (0.098)}
 & 0.096 & 0.032 & 0.008 & 0.201 \\
\midrule
Post expansions
 & \makecell{0.254 \\ (0.073)}
 & \makecell{0.303 \\ (0.065)}
 & \makecell{0.303 \\ (0.046)}
 & \makecell{0.302 \\ (0.080)}
 & 0.036 & 0.017 & 0.063 & 0.965 \\
\midrule
Interactions
 & \makecell{0.020 \\ (0.014)}
 & \makecell{0.017 \\ (0.014)}
 & \makecell{0.022 \\ (0.009)}
 & \makecell{0.029 \\ (0.015)}
 & 0.506 & 0.557 & 0.079 & 0.014 \\
\midrule
Link Clicks
 & \makecell{0.228 \\ (0.050)}
 & \makecell{0.230 \\ (0.035)}
 & \makecell{0.239 \\ (0.058)}
 & \makecell{0.261 \\ (0.048)}
 & 0.850 & 0.543 & 0.050 & 0.034 \\
\midrule
Observations
 & 18 & 18 & 18 & 18 &  &  &  &  \\[0.3em]
\bottomrule

\end{tabular}
}
\vspace{0.3em}

\begin{minipage}{\textwidth}
\scriptsize
\noindent \textit{Notes:} Each observation corresponds to an ad. Statistics are weighted by ad reach. The p-values are based on t-tests using heteroskedasticity-robust standard errors.
\end{minipage}
\end{table}

\begin{table}[H] \centering 
  \caption{The Impact of the Comment Section on On-platform User Engagement (clustered at the ZIP-code-set level)} 
  \label{tab:main_zip_cluster}
\scalebox{0.78}{\begin{tabular}{@{\extracolsep{5pt}}lcccccccc} 
\\[-1.8ex]\hline 
\hline \\[-1.8ex] 
 & \multicolumn{8}{c}{\textit{Dependent variable:}} \\ 
\cline{2-9} 
 & \multicolumn{2}{c}{All Engagement} & \multicolumn{2}{c}{Post Expansions} & \multicolumn{2}{c}{Interactions} & \multicolumn{2}{c}{Link Clicks} \\ 
\\[-1.8ex] & (1) & (2) & (3) & (4) & (5) & (6) & (7) & (8)\\ 
\hline \\[-1.8ex] 
 Any comments & 0.065$^{***}$ &  & 0.049$^{***}$ &  & 0.003 &  & 0.016$^{**}$ &  \\ 
  & (0.017) &  & (0.015) &  & (0.003) &  & (0.008) &  \\ 
  & [0.001] &  & [0.006] &  & [0.389] &  & [0.053] &  \\ 
 Opposing &  & 0.087$^{***}$ &  & 0.048$^{**}$ &  & 0.009$^{*}$ &  & 0.034$^{***}$ \\ 
  &  & (0.026) &  & (0.021) &  & (0.005) &  & (0.013) \\ 
  &  & [0.003] &  & [0.041] &  & [0.067] &  & [0.025] \\ 
 Mixed &  & 0.062$^{***}$ &  & 0.049$^{***}$ &  & 0.002 &  & 0.012 \\ 
  &  & (0.017) &  & (0.014) &  & (0.004) &  & (0.011) \\ 
  &  & [0.001] &  & [0.004] &  & [0.575] &  & [0.300] \\ 
 Supportive &  & 0.047$^{**}$ &  & 0.049$^{***}$ &  & $-$0.003 &  & 0.003 \\ 
  &  & (0.021) &  & (0.018) &  & (0.003) &  & (0.011) \\ 
  &  & [0.033] &  & [0.018] &  & [0.328] &  & [0.803] \\ 
 Constant & 0.538$^{***}$ & 0.539$^{***}$ & 0.158$^{***}$ & 0.158$^{***}$ & 0.004 & 0.005 & 0.379$^{***}$ & 0.380$^{***}$ \\ 
  & (0.017) & (0.017) & (0.016) & (0.016) & (0.005) & (0.005) & (0.013) & (0.013) \\ 
  & [0.239] & [0.237] & [0.459] & [0.465] & [0.564] & [0.551] & [0.210] & [0.204] \\ 
 \hline \\[-1.8ex] 
ZIP Code Set FEs & \checkmark & \checkmark & \checkmark & \checkmark & \checkmark & \checkmark & \checkmark & \checkmark \\ 
Controls & \checkmark & \checkmark & \checkmark & \checkmark & \checkmark & \checkmark & \checkmark & \checkmark \\ 
ZIP Code Set~$\times$~Controls & \checkmark & \checkmark & \checkmark & \checkmark & \checkmark & \checkmark & \checkmark & \checkmark \\ 
\hline \\[-1.8ex] 

Mean Y in Control & 0.485 & 0.485 & 0.254 & 0.254 & 0.020 & 0.020 & 0.228 & 0.228 \\ 
p(Support vs. Oppose) &  & 0.111 &  & 0.979 &  & 0.021 &  & 0.010 \\ 
Wild p(Support vs. Oppose) &  & 0.132 &  & 0.977 &  & 0.029 &  & 0.019 \\ 
p(Oppose vs. Mixed) &  & 0.331 &  & 0.945 &  & 0.167 &  & 0.241 \\ 
Wild p(Oppose vs. Mixed) &  & 0.347 &  & 0.942 &  & 0.178 &  & 0.256 \\ 
p(Support vs. Mixed) &  & 0.439 &  & 0.959 &  & 0.055 &  & 0.496 \\ 
Wild p(Support vs. Mixed) &  & 0.467 &  & 0.958 &  & 0.067 &  & 0.550 \\ 
Observations & 1,054,015 & 1,054,015 & 1,054,015 & 1,054,015 & 1,054,015 & 1,054,015 & 1,054,015 & 1,054,015 \\ 
R$^{2}$ & 0.0004 & 0.0004 & 0.0005 & 0.0005 & 0.0003 & 0.0003 & 0.0002 & 0.0002 \\ 
\hline 
\hline \\[-1.8ex] 
\end{tabular}
}
\vspace{0.3em}

\begin{minipage}{\textwidth}
\scriptsize
\noindent \textit{Notes:} Standard errors clustered at the ZIP-code-set level (\feTwoStrataN{} ZIP-code sets) are reported in parentheses; wild cluster bootstrap $p$-values are reported in square brackets. Asterisks refer to the cluster-robust standard errors. The unit of observation is the user, constructed from ad-level aggregate data. Controls include user gender and age, as well as their pairwise interactions. The outcome ``Interactions'' includes comments, reactions, and shares.
$^{*}p<0.10$; $^{**}p<0.05$; $^{***}p<0.01$.
\end{minipage}
\end{table}

\begin{table}[H] \centering 
  \caption{The Impact of the Comment Section on Landing-Page Views} 
  \label{tab:pageviews} 
\scalebox{0.8}{
\begin{tabular}{@{\extracolsep{5pt}}p{4cm}*{2}{>{\centering\arraybackslash}p{3cm}}} 
\\[-1.8ex]\hline 
\hline \\[-1.8ex] 
 & \multicolumn{2}{c}{\textit{Dependent variable (as \% of total reach):}} \\ 
 \cline{2-3} 
  & \multicolumn{2}{c}{Landing-Page Views} \\ 
\\[-1.8ex] & (1)  & (2)  \\ 
\hline \\[-1.8ex] 
 Any comments & 0.016$^{*}$ &  \\ 
  & (0.008) &  \\[0.5ex] 
 Opposing &  & 0.019$^{*}$ \\ 
  &  & (0.011) \\ 
 Mixed &  & 0.017 \\ 
  &  & (0.011) \\ 
 Supportive &  & 0.011 \\ 
  &  & (0.011) \\[0.5ex] 
 Constant & 0.347$^{**}$ & 0.347$^{***}$ \\ 
  & (0.135) & (0.134) \\ 
\hline \\[-1.8ex] 
ZIP Code Set FEs        & \checkmark & \checkmark \\ 
Controls                     & \checkmark & \checkmark \\ 
ZIP Code Set~$\times$~Controls & \checkmark & \checkmark \\ 
\hline \\[-1.8ex] 
Mean Y in Control (C) & 0.171 & 0.171 \\ 
p(Support vs. Oppose) &  & 0.523 \\ 
p(Oppose vs. Mixed)   &  & 0.880 \\ 
p(Support vs. Mixed)  &  & 0.611 \\ 
Observations          & 1,054,015 & 1,054,015 \\ 
R$^{2}$               & 0.0002 & 0.0002 \\ 
\hline 
\hline \\[-1.8ex] 
\end{tabular}
}
\vspace{0.3em}

\begin{minipage}{\textwidth}
\scriptsize
\noindent \textit{Notes:} Standard errors are clustered at the advertisement level (\feTwoAdsN{} ads). The unit of observation is the user, constructed from ad-level aggregate data. Controls include user gender and age, as well as their pairwise interactions.
$^{*}p<0.10$; $^{**}p<0.05$; $^{***}p<0.01$.
\end{minipage}
\end{table}

\begin{table}[H] \centering 
  \caption{Robustness Check: Excluding Toxic Comments (Threshold = 0.7)} 
  \label{tab:toxicity} 
\scalebox{0.8}{
\begin{tabular}{@{\extracolsep{5pt}}lcccc} 
\\[-1.8ex]\hline 
\hline \\[-1.8ex] 
 & \multicolumn{4}{c}{\textit{Dependent variable:}} \\ 
\cline{2-5} 
 & All Engagement & Post Expansions & Interactions & Link Clicks \\ 
\\[-1.8ex] & (1) & (2) & (3) & (4)\\ 
\hline \\[-1.8ex] 
 Opposing & 0.100$^{***}$ & 0.059$^{***}$ & 0.012$^{***}$ & 0.032$^{**}$ \\ 
  & (0.024) & (0.021) & (0.004) & (0.014) \\ 
  Mixed & 0.085$^{***}$ & 0.066$^{***}$ & 0.005 & 0.017 \\ 
  & (0.019) & (0.014) & (0.004) & (0.015) \\ 
  Supportive & 0.047$^{***}$ & 0.049$^{***}$ & $-$0.003 & 0.002 \\ 
  & (0.016) & (0.015) & (0.003) & (0.008) \\ 
  Constant & 0.521$^{***}$ & 0.144$^{**}$ & 0.002 & 0.378$^{***}$ \\ 
  & (0.111) & (0.061) & (0.010) & (0.066) \\ 
 \hline \\[-1.8ex] 
ZIP Code Set FEs & \cmark & \cmark & \cmark & \cmark \\
Controls              & \cmark & \cmark & \cmark & \cmark \\
ZIP Code Set~$\times$~Controls
  & \cmark & \cmark & \cmark & \cmark \\ 
\hline \\[-1.8ex] 
Mean Y in Control (C) & 0.485 & 0.254 & 0.020 & 0.228 \\ 
p(Support vs. Oppose) & 0.024 & 0.616 & 0.001 & 0.039 \\ 
p(Oppose vs. Mixed) & 0.519 & 0.723 & 0.164 & 0.392 \\ 
p(Support vs. Mixed) & 0.044 & 0.215 & 0.022 & 0.332 \\ 
Observations & 878,687 & 878,687 & 878,687 & 878,687 \\ 
R$^{2}$ & 0.0005 & 0.001 & 0.0003 & 0.0002 \\ 
\hline 
\hline \\[-1.8ex] 
\end{tabular}
}
\vspace{0.3em}

\begin{minipage}{\textwidth}
\scriptsize
\noindent \textit{Notes:} Standard errors are clustered at the advertisement level (\feTwoToxicityRestrictedAdsN{} ads). We excluded posts with any displayed comments that had toxicity scores above 0.7. The unit of observation is the user, constructed from ad-level aggregate data. Controls include user gender and age, as well as their pairwise interactions. The outcome ``Interactions'' includes comments, reactions and shares. $^{*}p<0.10$; $^{**}p<0.05$; $^{***}p<0.01$.
\end{minipage}
\end{table}

\begin{figure}[H]
\caption{Sensitivity to Excluding One ZIP-Code Set at a Time: All Engagement}
\label{all_Engagement_one_zipcode_set_out}
\centering
\includegraphics[width=0.8\textwidth]{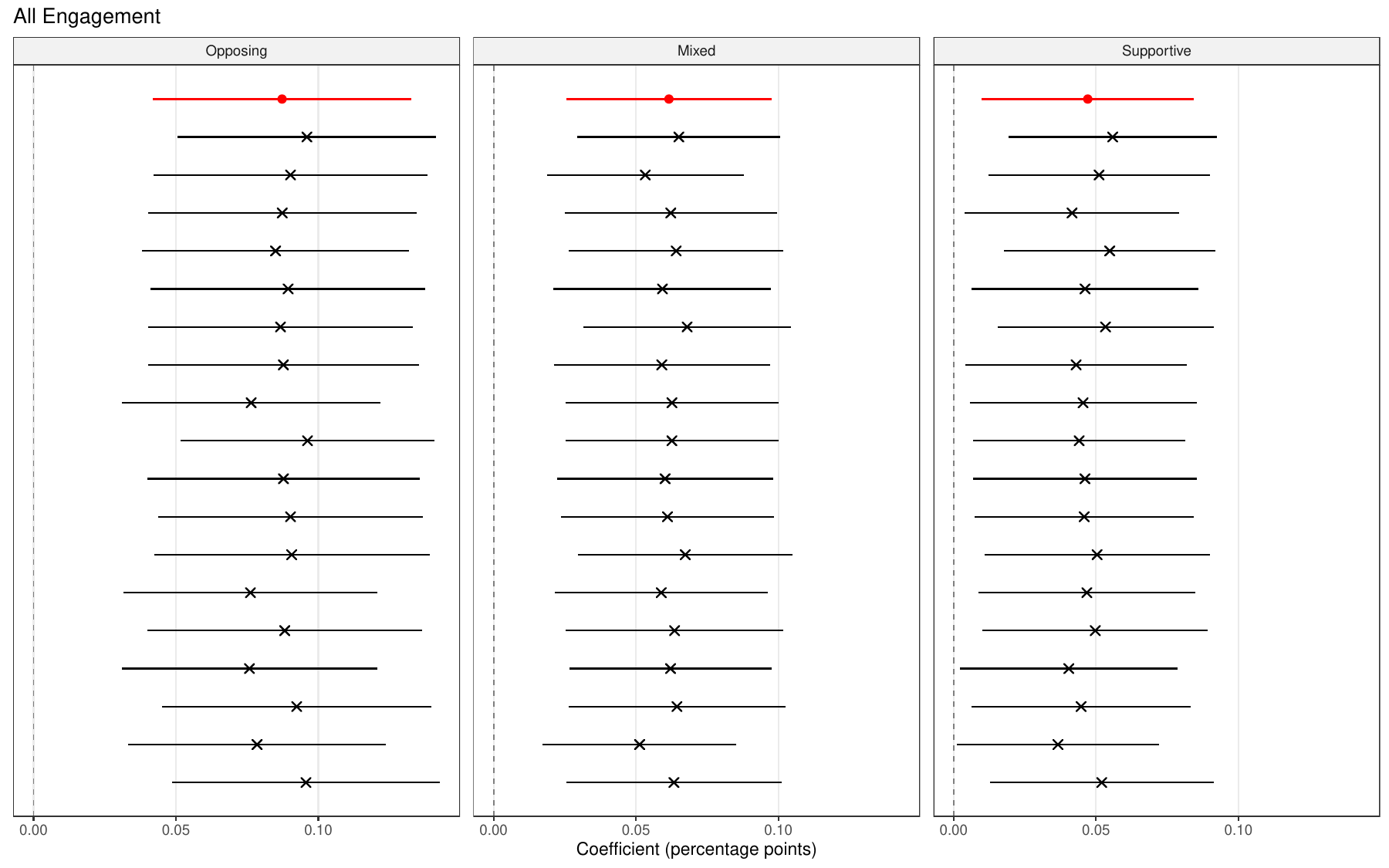} 
\end{figure}

\begin{figure}[H]
\caption{Sensitivity to Excluding One ZIP-Code Set at a Time: Post Expansions}
\label{post_expansions_one_zipcode_set_out}
\centering
\includegraphics[width=0.8\textwidth]{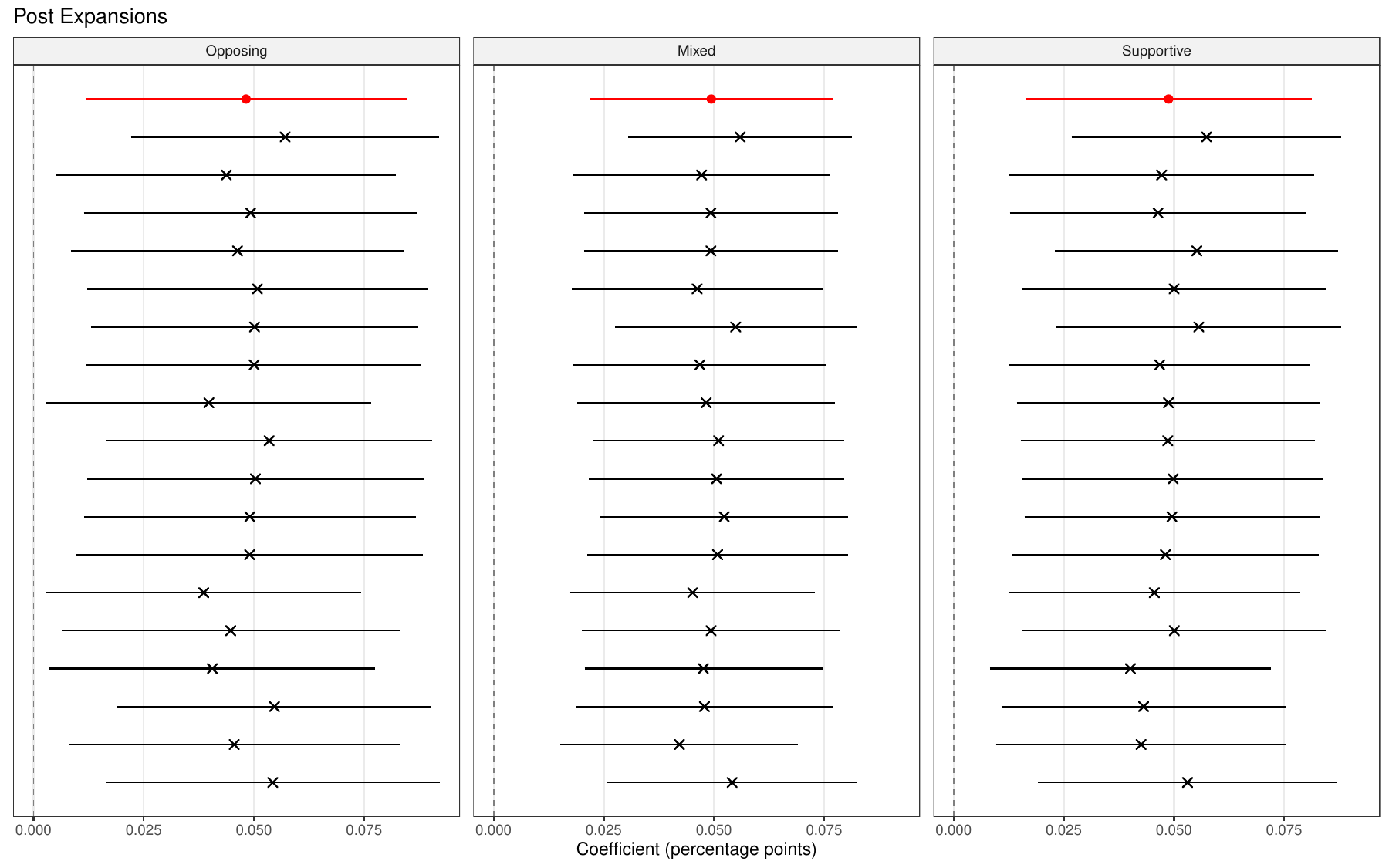} 
\end{figure}

\begin{figure}[H]
\caption{Sensitivity to Excluding One ZIP-Code Set at a Time: Interactions}
\label{interactions_one_zipcode_set_out}
\centering
\includegraphics[width=0.8\textwidth]{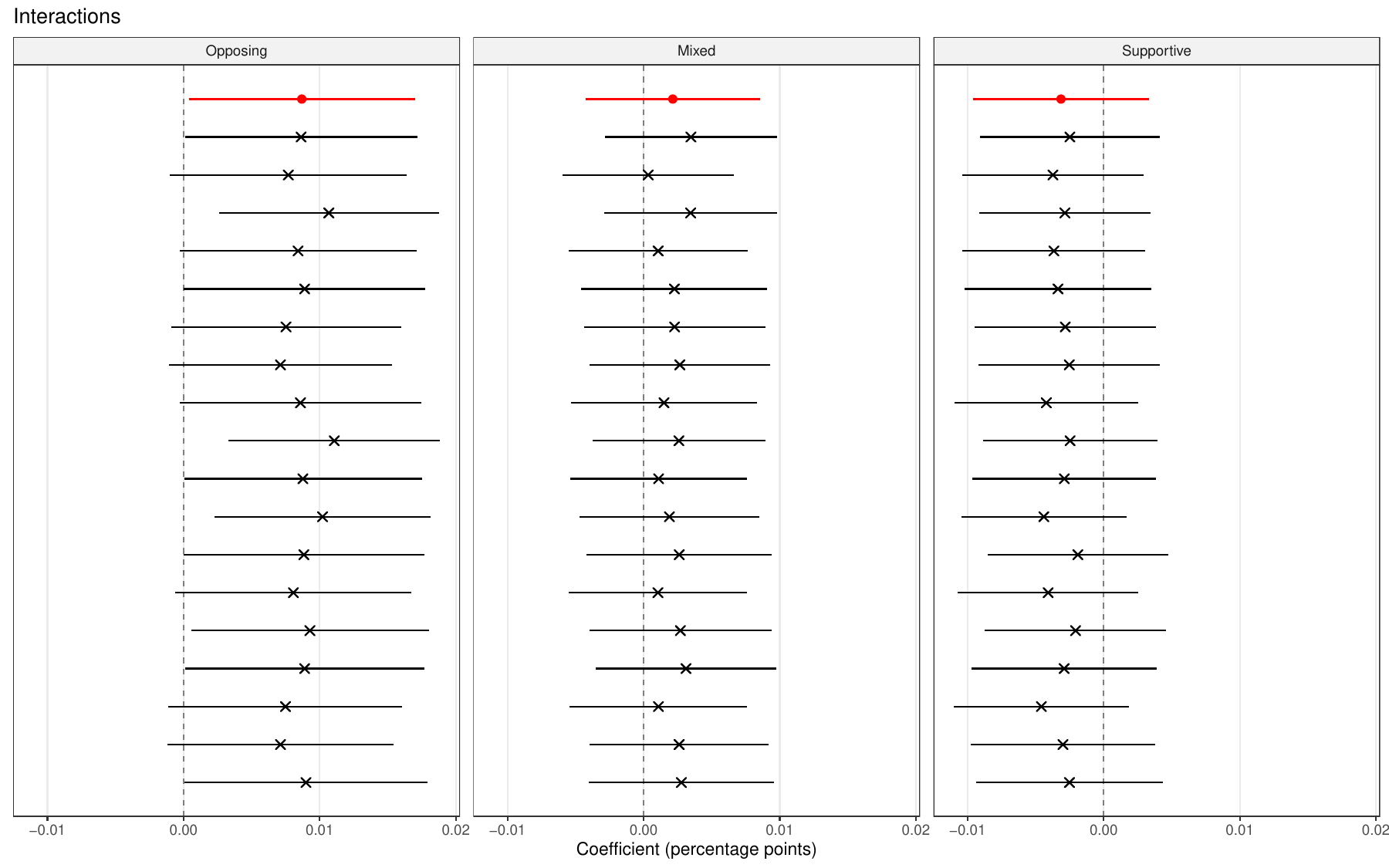} 
\end{figure}

\begin{figure}[H]
\caption{Sensitivity to Excluding One ZIP-Code Set at a Time: Link Clicks}
\label{link_clicks_one_zipcode_set_out}
\centering
\includegraphics[width=0.8\textwidth]{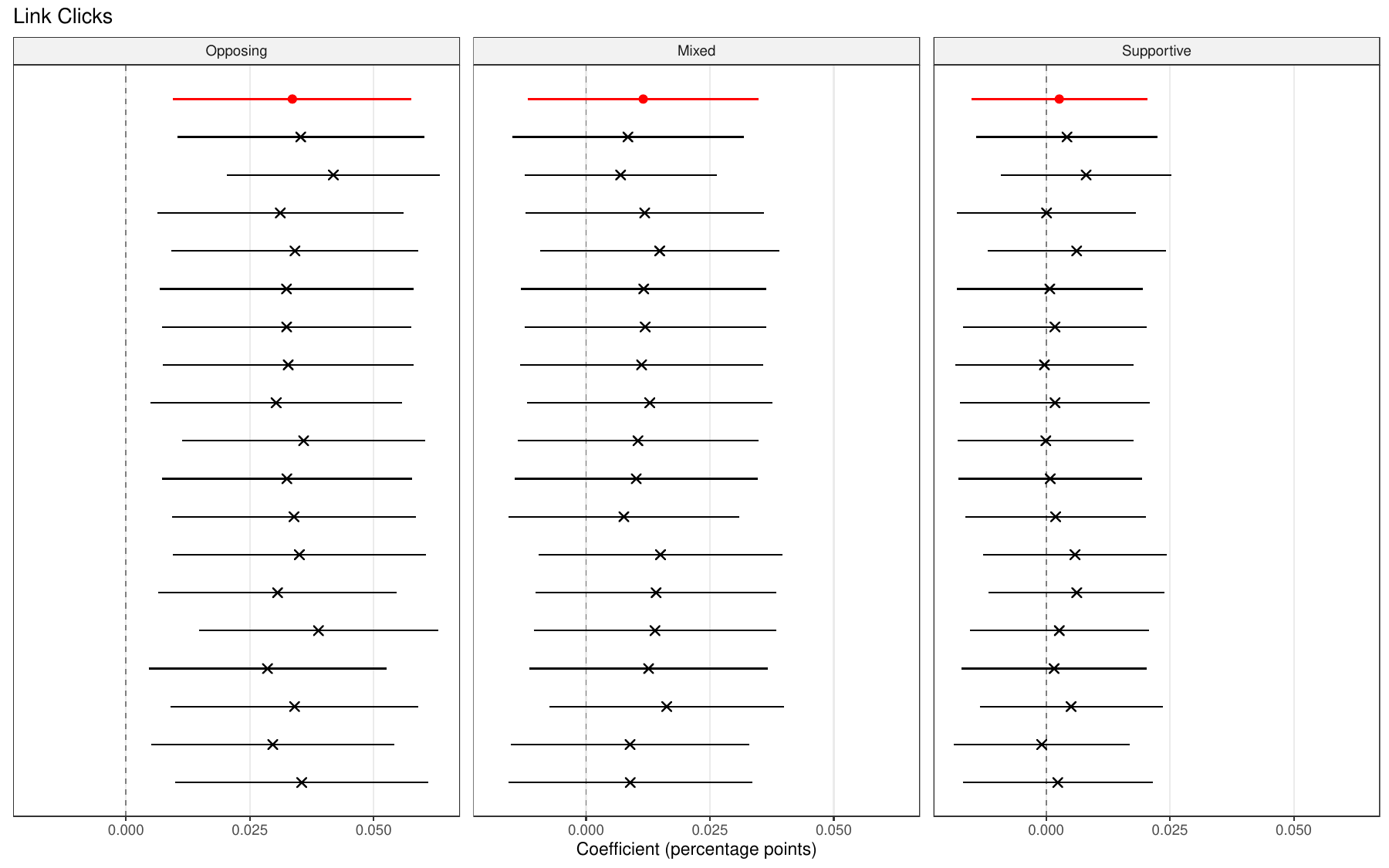} 
\begin{minipage}{\textwidth}
\scriptsize
\noindent \textit{Notes:} Figures~\ref{all_Engagement_one_zipcode_set_out}--\ref{link_clicks_one_zipcode_set_out} report a sensitivity analysis in which subsets of ZIP codes are iteratively excluded and the main treatment effects are re-estimated. Each point corresponds to an estimated coefficient from a re-estimated specification, and horizontal lines represent 95\% confidence intervals. The red point indicates the baseline estimate from the full sample.
\end{minipage}

\end{figure}

\begin{figure}[H]
\caption{Sensitivity to Excluding One Comment Section at a Time: All Engagement}
\label{all_Engagement_one_post_out}
\centering
\includegraphics[width=0.8\textwidth]{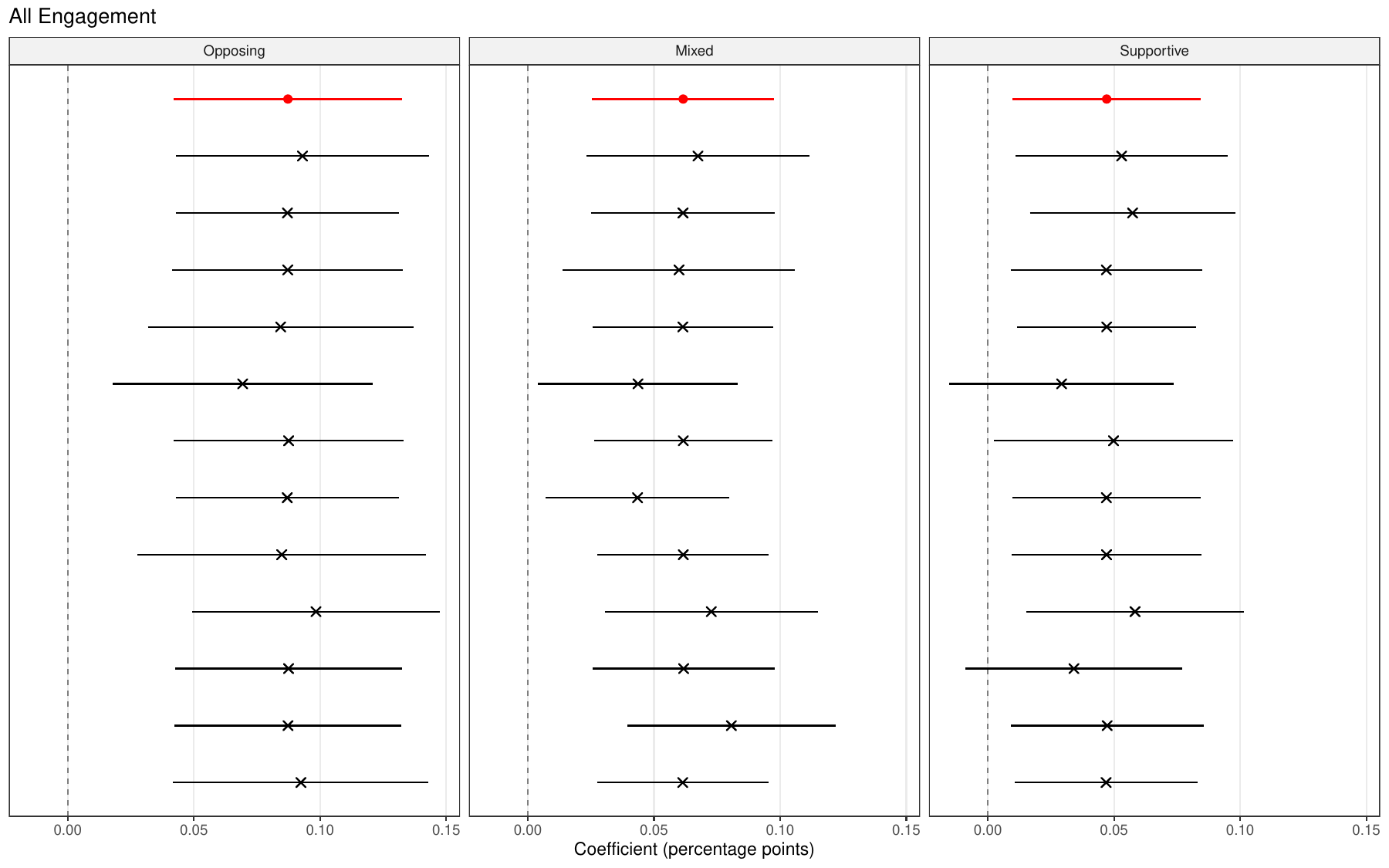} 
\end{figure}

\begin{figure}[H]
\caption{Sensitivity to Excluding One Comment Section at a Time: Post Expansions}
\label{post_expansions_one_post_out}
\centering
\includegraphics[width=0.8\textwidth]{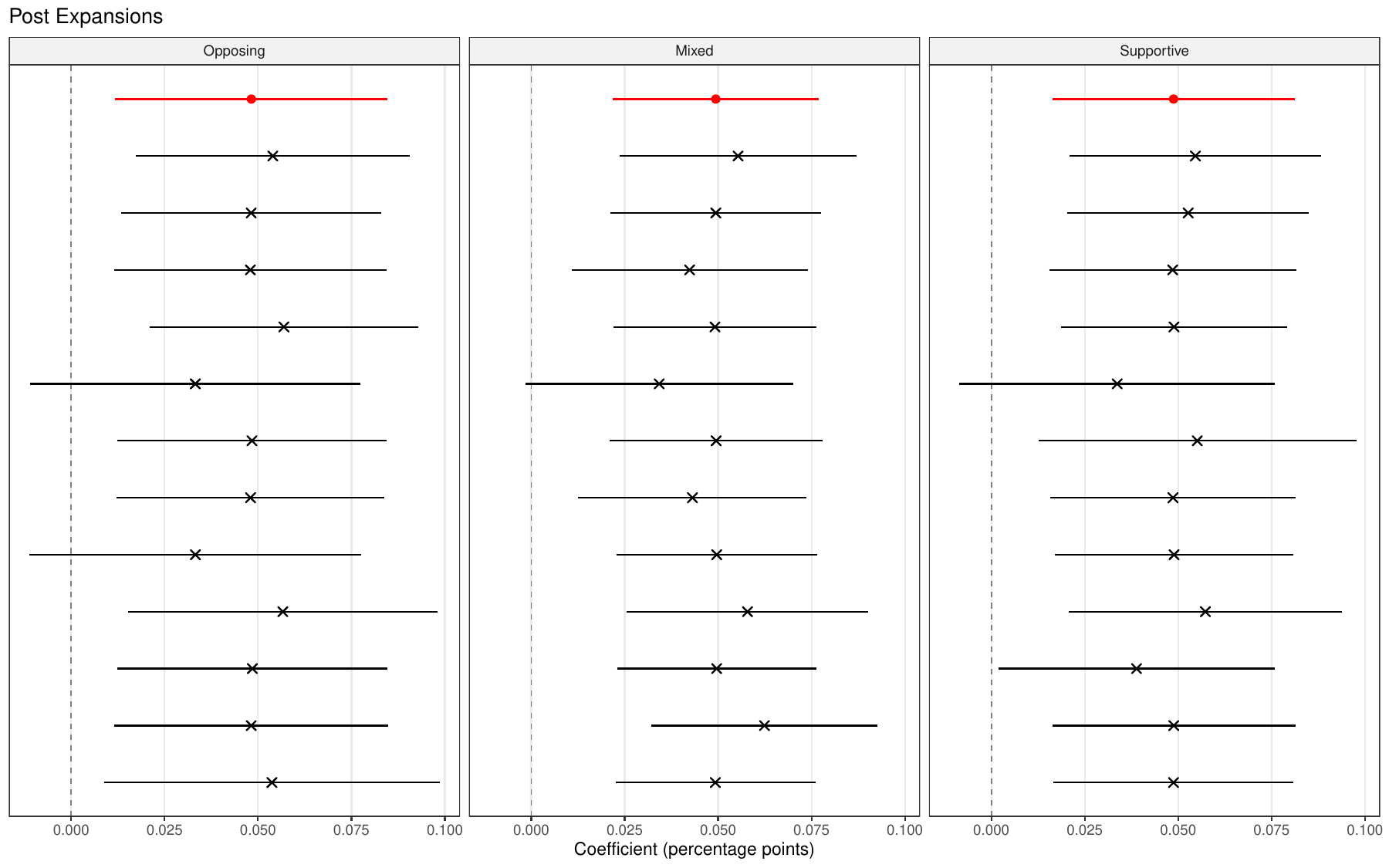} 
\end{figure}

\begin{figure}[H]
\caption{Sensitivity to Excluding One Comment Section at a Time: Interactions}
\label{interactions_one_post_out}
\centering
\includegraphics[width=0.8\textwidth]{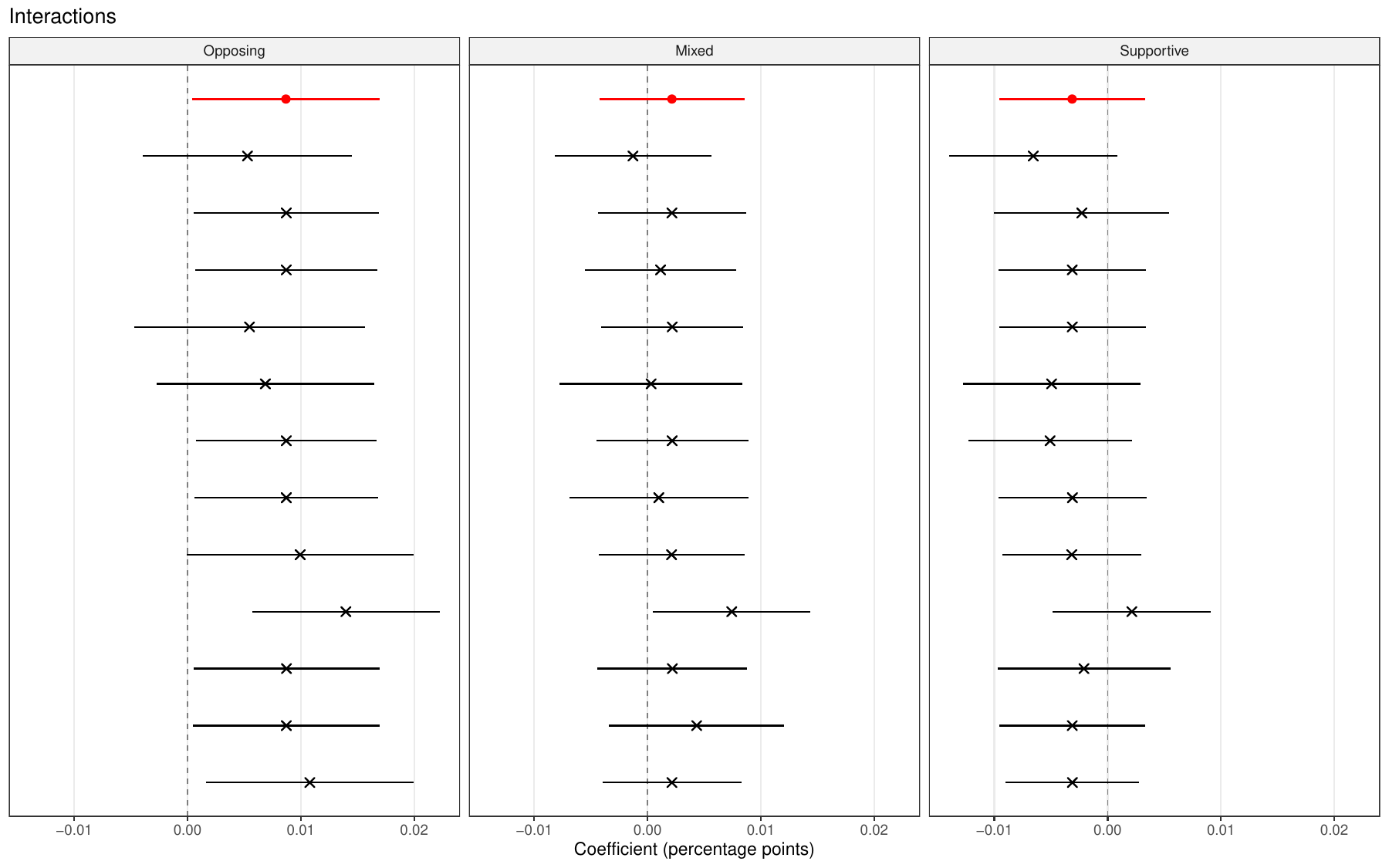} 
\end{figure}

\begin{figure}[H]
\caption{Sensitivity to Excluding One Comment Section at a Time: Link Clicks}
\label{link_clicks_one_post_out}
\centering
\includegraphics[width=0.8\textwidth]{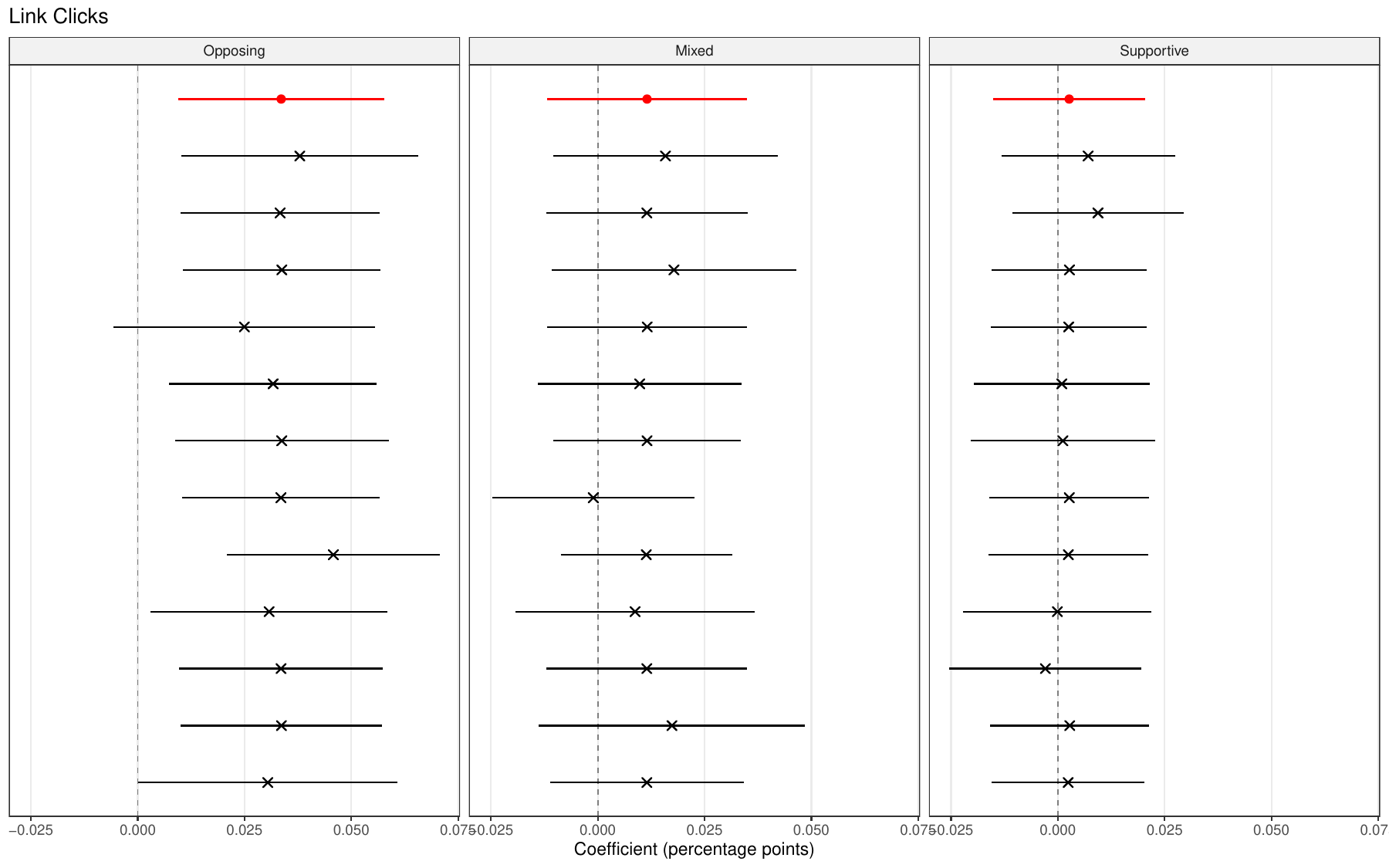} 
\begin{minipage}{\textwidth}
\scriptsize
\noindent \textit{Notes:} Figures~\ref{all_Engagement_one_post_out}--\ref{link_clicks_one_post_out} report a sensitivity analysis in which each comment section (post) is iteratively excluded and the main treatment effects are re-estimated. Each point corresponds to an estimated coefficient from a re-estimated specification, and horizontal lines represent 95\% confidence intervals. The red point indicates the baseline estimate from the full sample.
\end{minipage}

\end{figure}

\begin{figure}[H]
\caption{Randomization-Inference Tests: All Engagement}
\label{all_engagement_permutation}
\centering
\includegraphics[width=0.85\textwidth]{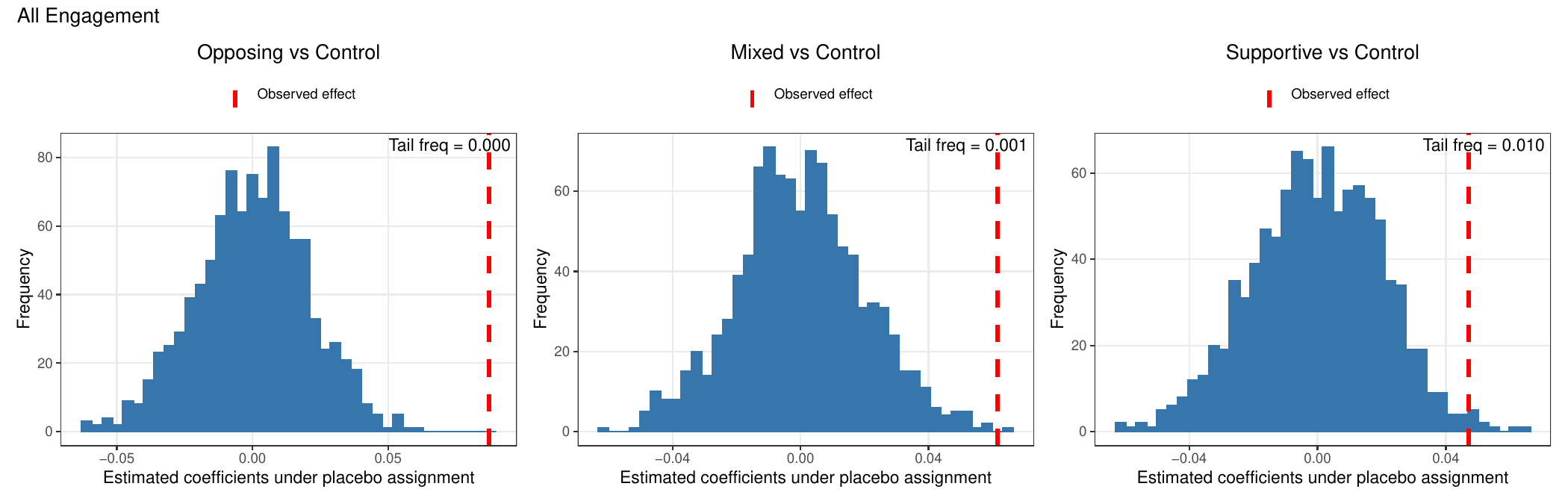} 
\end{figure}

\begin{figure}[H]
\caption{Randomization-Inference Tests: Post Expansions}
\label{post_expansions_permutation}
\centering
\includegraphics[width=0.85\textwidth]{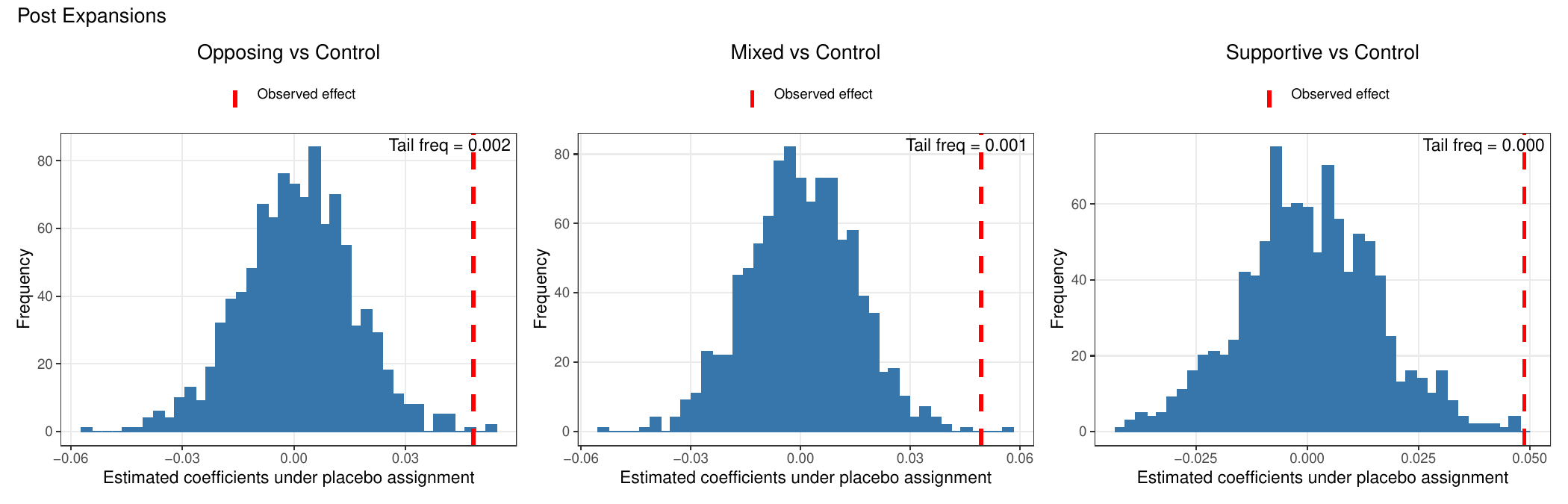} 
\end{figure}

\begin{figure}[H]
\caption{Randomization-Inference Tests: Interactions}
\label{interactions_permutation}
\centering
\includegraphics[width=0.85\textwidth]{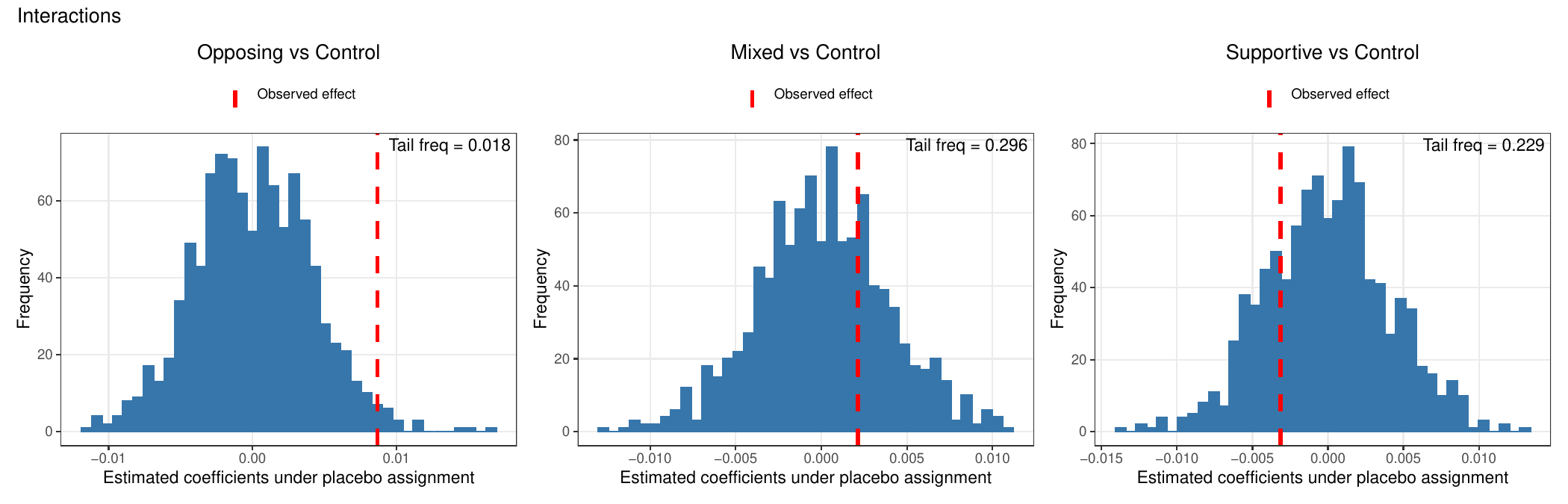} 
\end{figure}

\begin{figure}[H]
\caption{Randomization-Inference Tests: Link Clicks}
\label{link_clicks_permutation}
\centering
\includegraphics[width=0.85\textwidth]{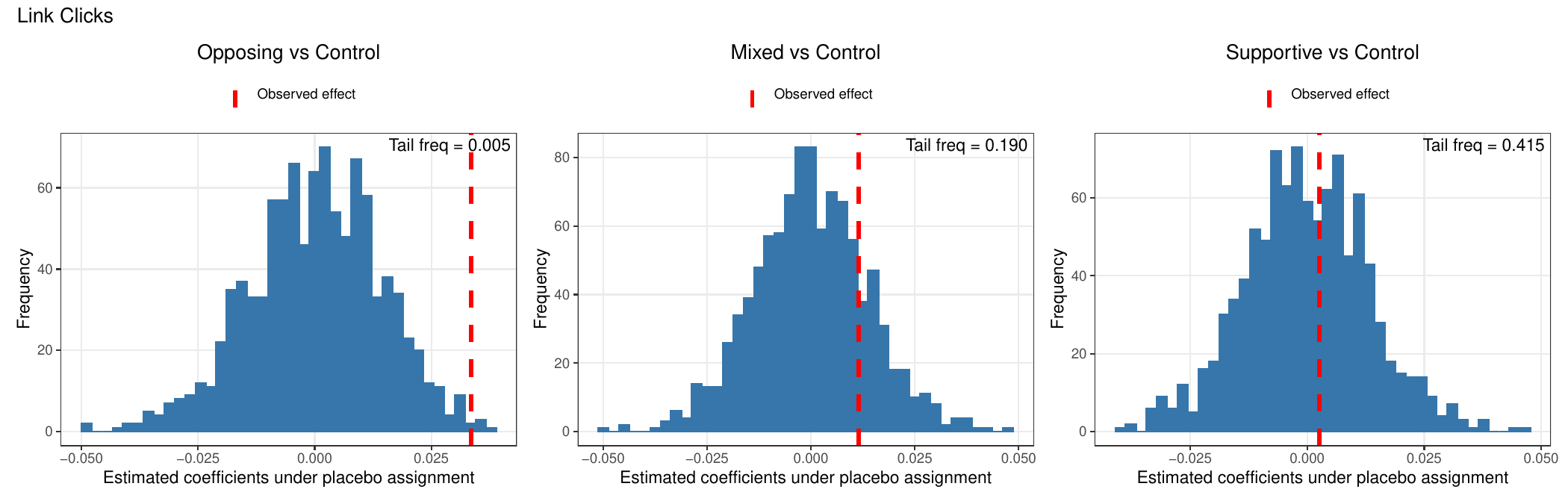} 
\begin{minipage}{\textwidth}
\scriptsize
\noindent \textit{Notes:} Figures~\ref{all_engagement_permutation}--\ref{link_clicks_permutation} report randomization-inference placebo tests. In each of 1,000 draws, individuals are randomly reassigned to the four arms within each ZIP-code set, with equal arm sizes, and the full specification of equation~\eqref{eq:main} is re-estimated. Histograms show the distribution of the resulting placebo estimates. The vertical dashed line marks the observed treatment effect in the data. Tail frequencies report the share of placebo estimates at least as extreme as the observed estimate, in its direction, and therefore correspond to one-sided permutation $p$-values.
\end{minipage}

\end{figure}

\section{Survey Experiment}
\setcounter{figure}{0}
\renewcommand{\thefigure}{F\arabic{figure}}
\renewcommand{\theHfigure}{F\arabic{figure}}

\setcounter{table}{0}
\renewcommand{\thetable}{F\arabic{table}}
\renewcommand{\theHtable}{F\arabic{table}}

\subsection{Design}
\begin{figure}[H]
    \centering
    \caption{Survey Experiment Stimuli}
    \label{fig:stimuli_survey}
    \begin{subfigure}[b]{0.25\textwidth}
        \centering
        \includegraphics[width=1\textwidth]{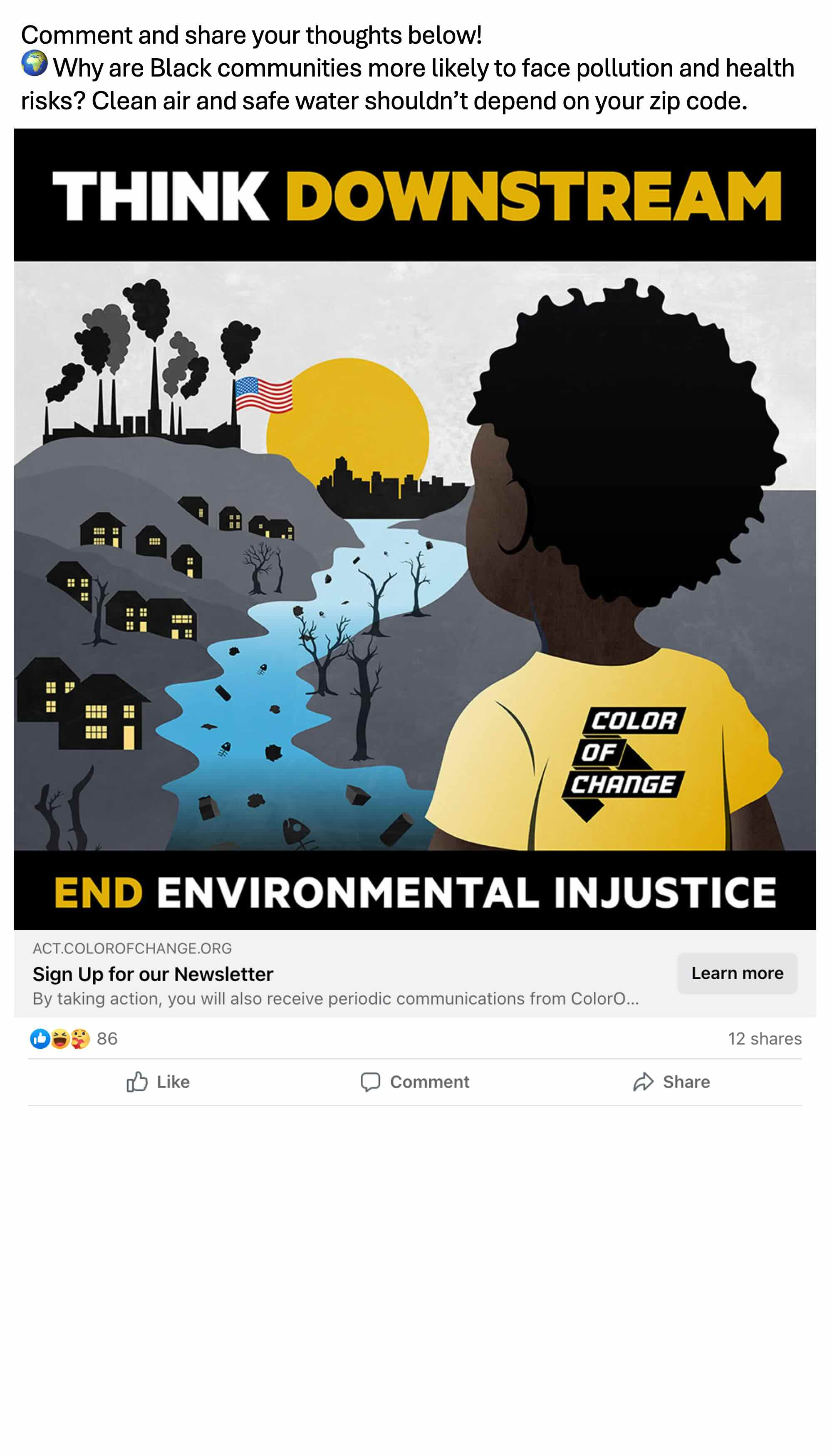} 
        \caption{Control}
    \end{subfigure}
    \, \,
    \begin{subfigure}[b]{0.25\textwidth}
        \centering
        \includegraphics[width=1\textwidth]{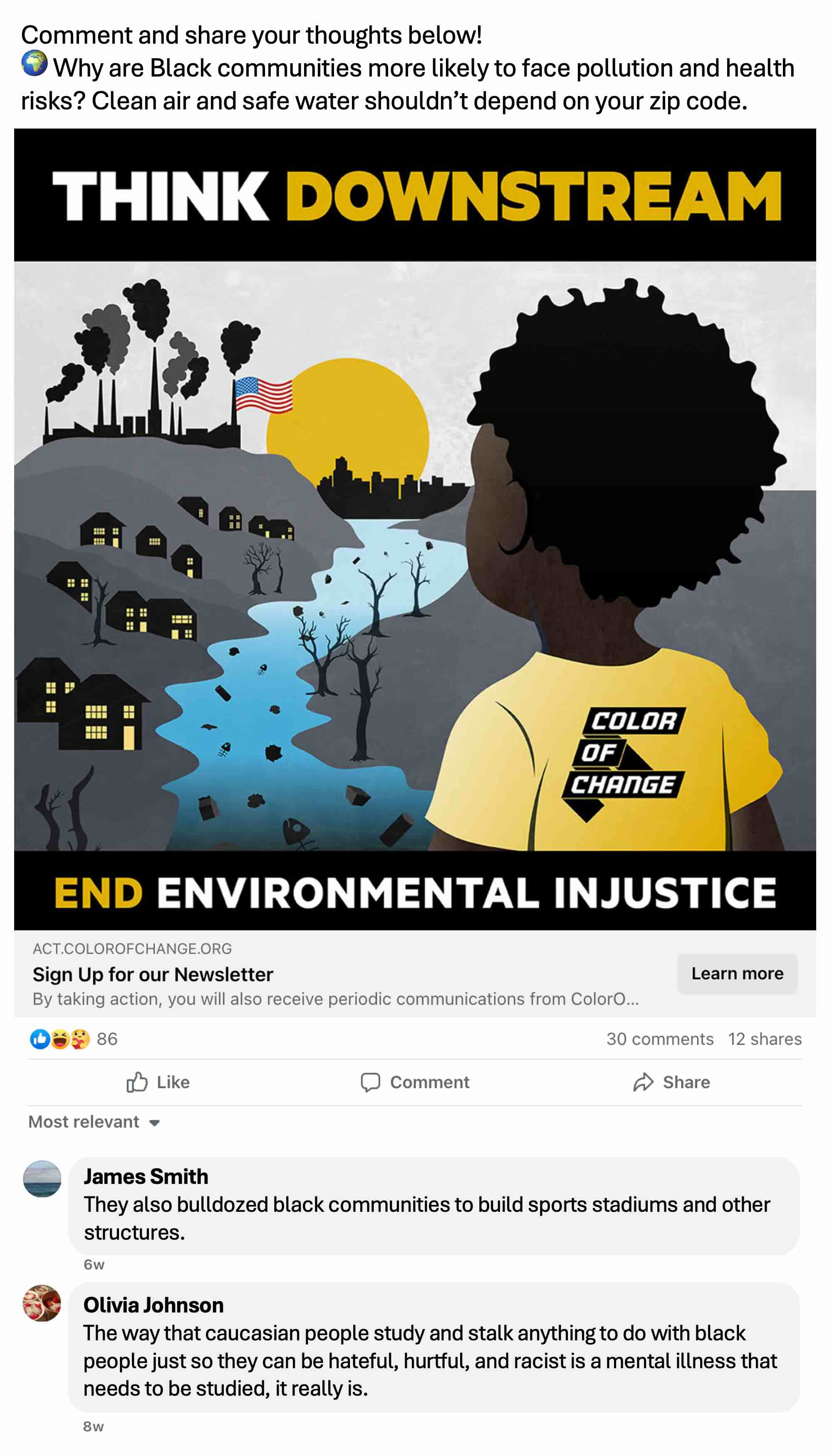}
        \caption{Supportive}
    \end{subfigure}
    \, \,
    \begin{subfigure}[b]{0.25\textwidth}
        \centering
        \includegraphics[width=1\textwidth]{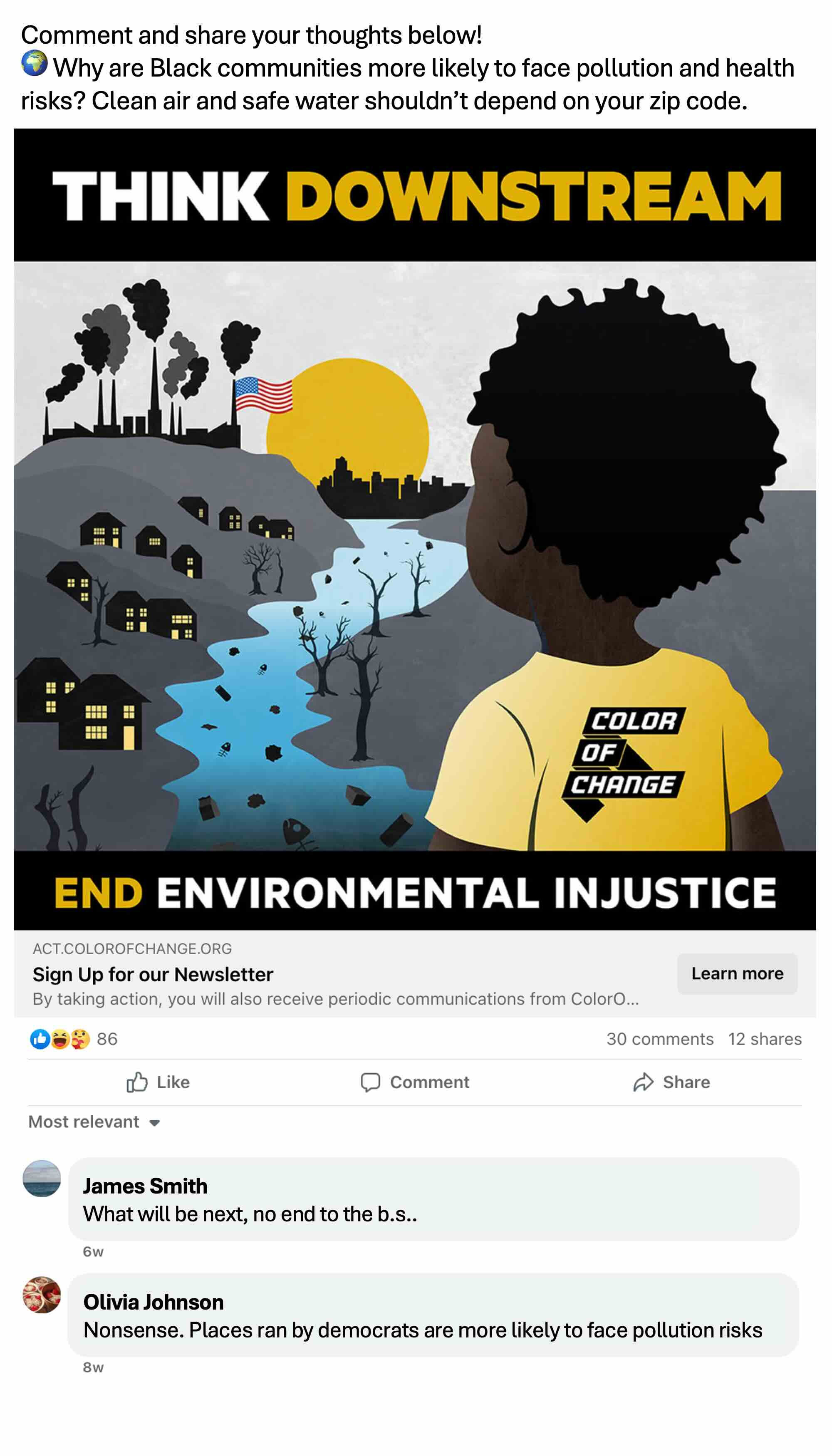} 
        \caption{Opposing}
    \end{subfigure}
\begin{minipage}{\textwidth}
\scriptsize
\noindent \textit{Notes:} This figure displays  screenshots of the survey experiment stimuli for the environment issue. Panel (a) shows the Control condition, in which the post is displayed without comments. Panel (b) shows the Supportive condition, in which two supportive comments appear below the post. Panel (c) shows the Opposing condition, in which two opposing comments appear below the post. The environment post is shown for illustration; participants were assigned to view three posts in random order under their assigned comment condition.
\end{minipage}
\end{figure}

\begin{table}[H]
    \centering
    \caption{Sample Demographics Relative to U.S. Adults}
    \label{tab:survey_sample_vs_us}
    \begin{adjustbox}{max width=\textwidth}
        \begin{tabular}{lcc}
\toprule
 & \shortstack{Analysis\\sample} & \shortstack{U.S.\\adults} \\
\midrule
Female & 0.493 & 0.505 \\
College (Bachelor's or higher) & 0.523 & 0.362 \\
\midrule
Age 18--34 & 0.470 & 0.373 \\
Age 35--54 & 0.423 & 0.420 \\
Age 55--64 & 0.107 & 0.206 \\
\midrule
White & 0.791 & 0.723 \\
Black or African American & 0.116 & 0.144 \\
Asian or Pacific Islander & 0.077 & 0.079 \\
American Indian or Alaska Native & 0.023 & 0.026 \\
Other or mixed race & 0.057 & 0.128 \\
Hispanic or Latino origin & 0.107 & 0.194 \\
\midrule
Democrat & 0.320 & 0.270 \\
Republican & 0.294 & 0.270 \\
Independent & 0.369 & 0.450 \\
\midrule
N & 3,868 &  \\
\bottomrule
\end{tabular}

    \end{adjustbox}
\begin{minipage}{\textwidth}
\scriptsize
\noindent \textit{Notes:} This table compares the composition of the analysis sample in the survey experiment with the U.S. adult population. The first column reports shares in the analysis sample, and the second column reports population benchmarks from 2023 American Community Survey (gender, age, race, and Hispanic origin for the total population; college attainment, bachelor's degree or higher, for adults aged 25 and over) and party identification from Gallup 2025. Age shares are computed among adults aged 18–64 to match the survey's inclusion criteria. Race shares use  the ACS ``race alone or in combination'' tabulation and need not sum to one.
\end{minipage}
\end{table}

\begin{table}[H]
    \centering
    \caption{Survey Completion Rates Across Experimental Arms}
    \label{tab:survey_response_rates}
    \begin{adjustbox}{max width=\textwidth}
        \begin{tabular}{lcccc}
\toprule
 & (1) Control & (2) Supportive & (3) Opposing & F-test \\
Variable & Share & Share & Share & $p$-value \\
\midrule
Completed Survey & 1.00 & 0.99 & 0.99 & 0.134 \\
\bottomrule
\end{tabular}

    \end{adjustbox}
\begin{minipage}{\textwidth}
\scriptsize
\noindent \textit{Notes:} This table reports the share of respondents who completed the survey in each experimental arm, among those who were randomized and passed the attention check. The final column reports the $p$-value from an F-test of joint equality of completion rates across the three arms.
\end{minipage}
\end{table}

\begin{table}[H]
\centering
\caption{Covariate Balance Across Experimental Arms}
\label{tab:survey_balance}
\begin{adjustbox}{max width=\textwidth}
\begin{tabular}{lccccccccc}
\toprule
 & (1) & (2) & (3) & \multicolumn{3}{c}{} & \multicolumn{3}{c}{} \\
 & Control & Supportive & Opposing & \multicolumn{3}{c}{SMD} & \multicolumn{3}{c}{t-test $p$-value} \\
\cmidrule(lr){5-7} \cmidrule(lr){8-10}
Variable & Mean/SD & Mean/SD & Mean/SD & (2)-(1) & (3)-(1) & (3)-(2) & (1)-(2) & (1)-(3) & (2)-(3) \\
\midrule
Age ($>$ Median) & 0.48 (0.50) & 0.49 (0.50) & 0.51 (0.50) & 0.025 & 0.064 & 0.039 & 0.54 & 0.11 & 0.31 \\
Baseline Racial Index ($>$ Median) & 0.49 (0.50) & 0.50 (0.50) & 0.48 (0.50) & 0.013 & -0.022 & -0.035 & 0.74 & 0.58 & 0.36 \\
Perceived Share of Non-Conservatives ($>$ Median) & 0.41 (0.49) & 0.38 (0.49) & 0.39 (0.49) & -0.072 & -0.047 & 0.025 & 0.07* & 0.24 & 0.51 \\
Perceived Share of Progressives ($>$ Median) & 0.49 (0.50) & 0.47 (0.50) & 0.47 (0.50) & -0.040 & -0.045 & -0.005 & 0.32 & 0.26 & 0.91 \\
Education ($>$ Median) & 0.16 (0.36) & 0.15 (0.36) & 0.18 (0.38) & -0.009 & 0.052 & 0.062 & 0.81 & 0.19 & 0.11 \\
Political view ($>$ Median) & 0.38 (0.49) & 0.40 (0.49) & 0.39 (0.49) & 0.045 & 0.026 & -0.020 & 0.26 & 0.51 & 0.61 \\
Social Media Use ($>$ Median) & 0.40 (0.49) & 0.41 (0.49) & 0.39 (0.49) & 0.016 & -0.032 & -0.049 & 0.68 & 0.41 & 0.21 \\
Reads Comments Frequency ($>$ Median) & 0.18 (0.38) & 0.19 (0.39) & 0.18 (0.39) & 0.023 & 0.003 & -0.020 & 0.57 & 0.94 & 0.60 \\
Female & 0.51 (0.50) & 0.49 (0.50) & 0.49 (0.50) & -0.040 & -0.037 & 0.003 & 0.32 & 0.35 & 0.94 \\
White & 0.78 (0.41) & 0.78 (0.41) & 0.81 (0.39) & 0.001 & 0.070 & 0.069 & 0.98 & 0.08* & 0.07* \\
Black / African American & 0.12 (0.32) & 0.12 (0.33) & 0.11 (0.31) & 0.015 & -0.036 & -0.052 & 0.71 & 0.36 & 0.18 \\
Asian American / Pacific Islander & 0.08 (0.28) & 0.07 (0.26) & 0.07 (0.26) & -0.041 & -0.043 & -0.002 & 0.31 & 0.27 & 0.96 \\
Hispanic & 0.09 (0.29) & 0.13 (0.33) & 0.10 (0.30) & 0.104 & 0.022 & -0.082 & 0.01*** & 0.59 & 0.03** \\
Democrat & 0.32 (0.47) & 0.31 (0.46) & 0.33 (0.47) & -0.007 & 0.029 & 0.036 & 0.87 & 0.46 & 0.36 \\
Republican & 0.30 (0.46) & 0.30 (0.46) & 0.28 (0.45) & 0.005 & -0.036 & -0.041 & 0.91 & 0.36 & 0.29 \\
Independent & 0.37 (0.48) & 0.37 (0.48) & 0.37 (0.48) & 0.003 & -0.004 & -0.007 & 0.94 & 0.91 & 0.85 \\
\midrule
N & 1,189 & 1,287 & 1,392 &  &  &  &  &  &  \\
F-test of joint significance ($p$-value) &  &  &  &  &  &  & 0.355 & 0.527 & 0.238 \\
F-test, number of observations &  &  &  &  &  &  & 2,476 & 2,581 & 2,679 \\
\bottomrule
\end{tabular}

\end{adjustbox}
\medskip{}
\begin{minipage}{\textwidth}
\scriptsize
\noindent \textit{Notes:} This table reports covariate balance across the three treatment arms of the survey experiment. Columns (1)--(3) report means with standard deviations in parentheses for the control, supportive comments, and opposing comments groups, respectively. The remaining columns report the standardized mean differences (SMD) and p-values from pairwise t-tests of equality of means. The final rows report F-tests of joint significance across all covariates for each pair of arms. * $p<0.10$, ** $p<0.05$, *** $p<0.01$.
\end{minipage}
\end{table}

\subsection{Additional Results}
\begin{figure}[H]
    \centering
    \caption{Reported Frequency of Reading or Checking Comments on Social Media}
    \label{fig:social_read_comments}

    \begin{subfigure}[t]{0.49\textwidth}
        \centering
        \caption*{(a) Full Sample}
        \includegraphics[width=\textwidth]{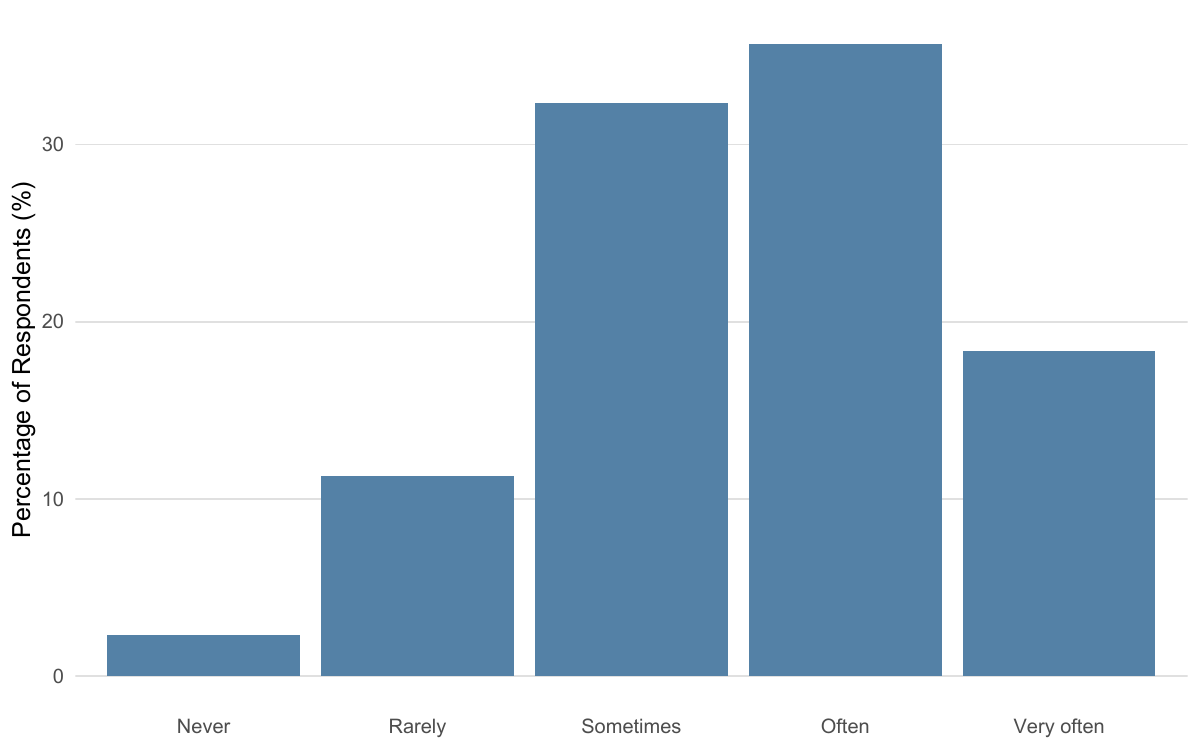}
    \end{subfigure}
    \hfill
    \begin{subfigure}[t]{0.5\textwidth}
        \centering
        \caption*{(b) By Social Media Use}
        \includegraphics[width=\textwidth]{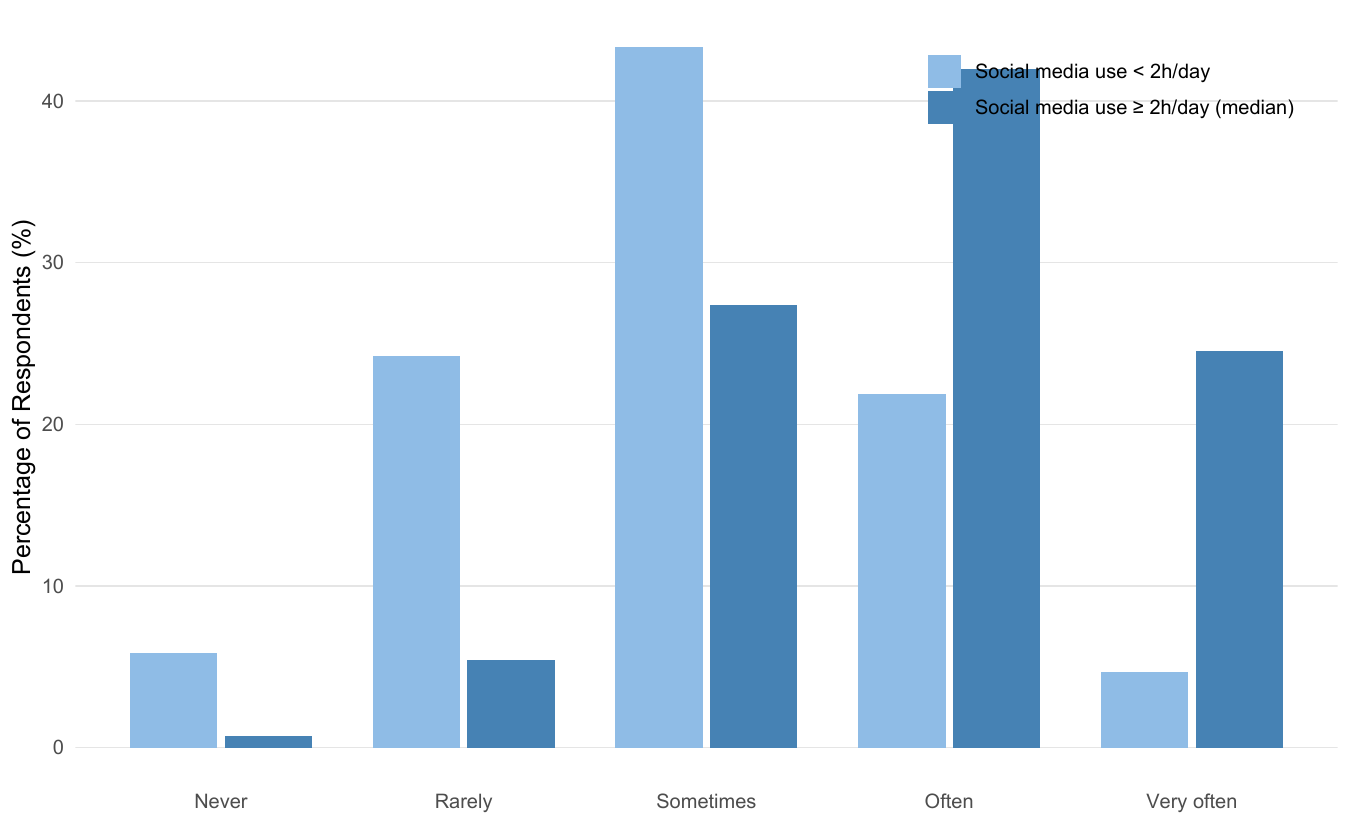}
    \end{subfigure}
    \begin{minipage}{\textwidth}
\scriptsize
\noindent \textit{Notes:} This figure reports the self-reported frequency with which survey participants read or check comments on social media. Panel (a) shows the distribution for the full sample. Panel (b) splits the sample by social media use, comparing participants above and below the median level of social media use.
\end{minipage}
\end{figure}

\begin{figure}[H]
    \centering
    \caption{Time Spent on Posts (in seconds)}
    \label{fig:survey_attention}
    \includegraphics[width=0.8\textwidth]{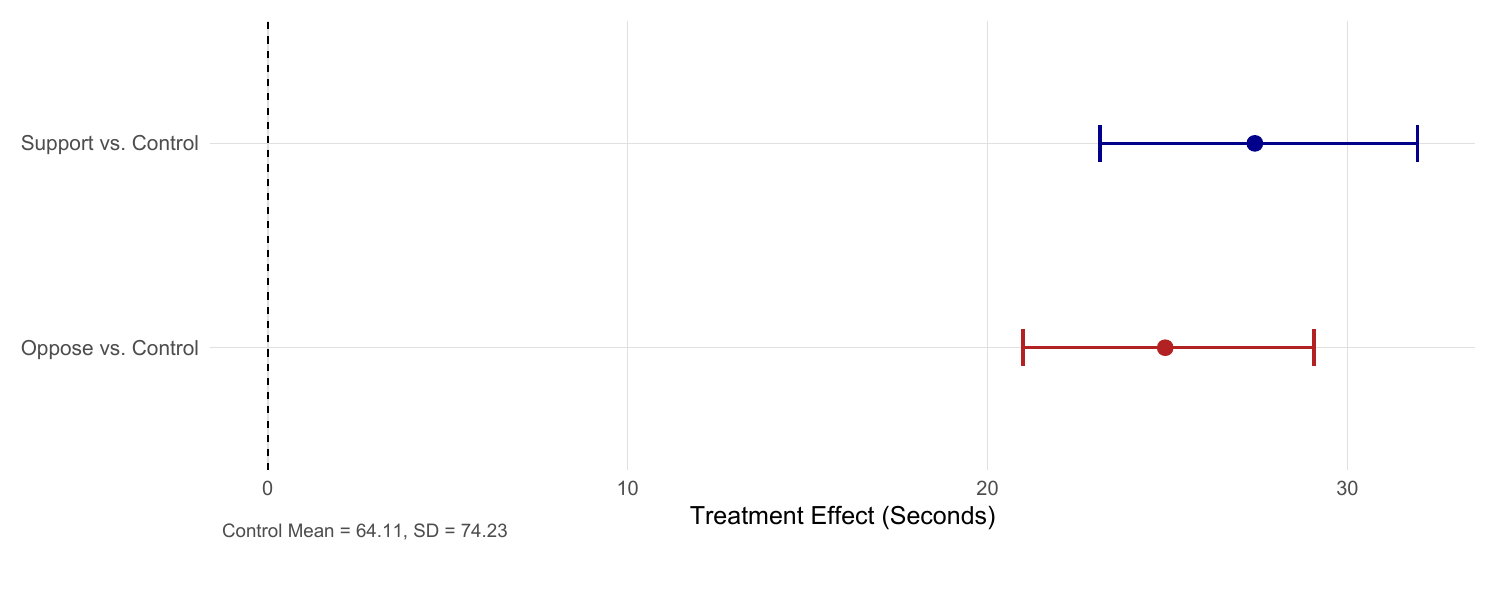}
\begin{minipage}{\textwidth}
\scriptsize
\noindent \textit{Notes:} This figure reports treatment effects of supportive and opposing comments, relative to the no-comments control, on time spent viewing the posts. Effects are expressed in seconds.
Horizontal lines represent 95\% confidence intervals.
\end{minipage}
\end{figure}

\begin{table}[H]
    \centering
    \caption{Effects of Comment Stance on Time Spent, Anger, Annoyance, and Curiosity}
    \label{tab:survey_arm_engagement}
    \begin{adjustbox}{max width=\textwidth}
        \begin{tabular}{l c c c}
\toprule
 & \textit{Time Spent} & \textit{Anger/Annoyance} & \textit{Curiosity/Reflection (Incl. Comments)} \\
\cmidrule(lr){2-2} \cmidrule(lr){3-3} \cmidrule(lr){4-4}
 & (1) & (2) & (3) \\
\midrule
Supportive vs.\ Control & 0.356$^{***}$ &  &  \\
 & (0.025) &  &  \\
Opposing vs.\ Control & 0.329$^{***}$ &  &  \\
 & (0.023) &  &  \\
Opposing vs.\ Supportive & -0.028 & 0.433$^{***}$ & -0.382$^{***}$ \\
 & (0.023) & (0.036) & (0.035) \\
\midrule
Controls & \checkmark & \checkmark & \checkmark \\
\midrule
Observations & 3,868 & 2,679 & 2,679 \\
$R^2$ & 0.1293 & 0.1562 & 0.1936 \\
Mean $Y$ (Control) & 3.931 &  &  \\
\bottomrule
\end{tabular}

    \end{adjustbox}
\begin{minipage}{\textwidth}
\scriptsize
\noindent \textit{Notes:} This table reports OLS estimates from equation~\eqref{eq:main_survey} for the time spent, anger/annoyance, and curiosity/reflection in the survey experiment. Column 1 shows time spent measured as $\log(1+t)$ where $t$ is total time on the three posts in seconds. Column 2 combines Likert-scale measures of anger and annoyance triggered by the comments, oriented so that higher values indicate a stronger negative emotional response. Because this measure is elicited only from respondents who saw a comment section, it is defined for the respondents in the two treatment arms. Column 3 combines interest in seeing additional comments with how thought-provoking respondents found the post and the comments. Because the comment item is asked only of respondents who saw a comment section, this index, like the anger/annoyance index, is estimated on the respondents in the two treatment arms. The indices are aggregated by inverse-covariance weighting following \citet{anderson2008multiple}. Neither index is observed in the control group, so both are standardized to mean zero and standard deviation one in the pooled sample of the two treatment arms. All columns include the pre-specified baseline controls listed in Section~\ref{sec:survey}, and standard errors robust to heteroskedasticity are in parentheses. $^{*}p<0.10$; $^{**}p<0.05$; $^{***}p<0.01$.
\end{minipage}
\end{table}

\begin{figure}[h]
    \centering
    \caption{Time Spent on Posts, log(seconds+1)}
    \label{fig:survey_attention_log}
    \includegraphics[width=0.8\textwidth]{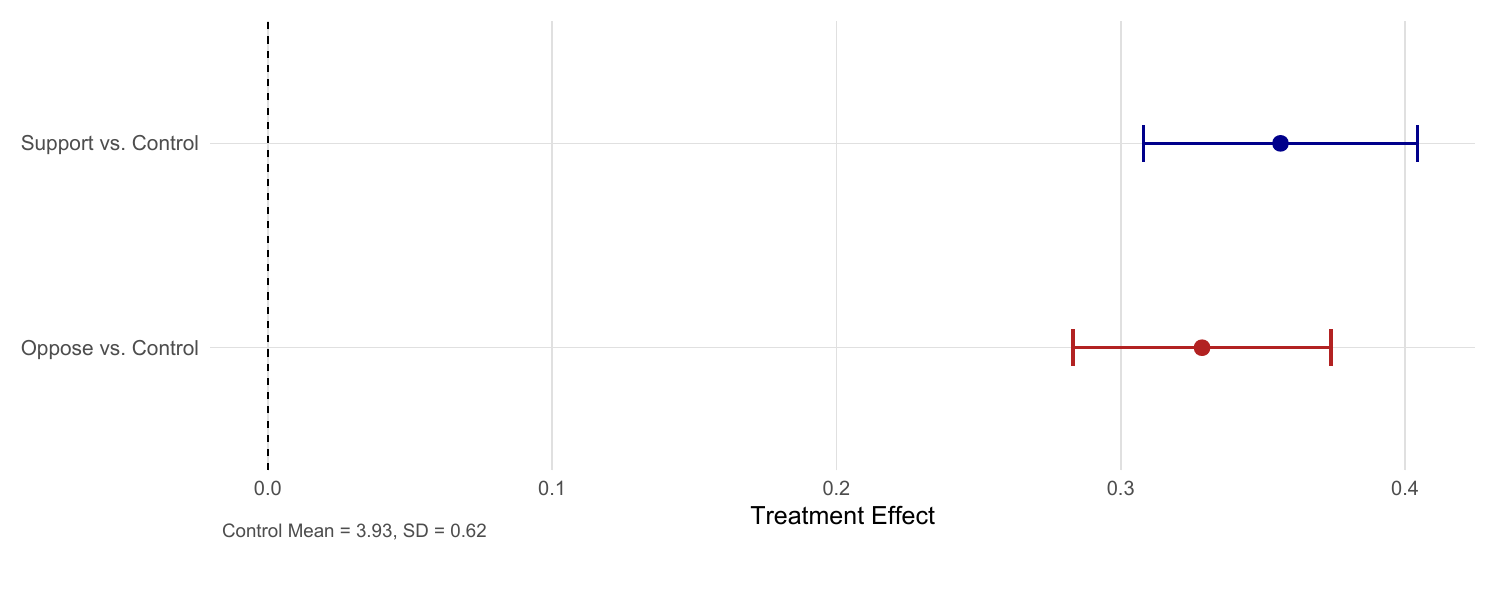}
\begin{minipage}{\textwidth}
\scriptsize
\noindent \textit{Notes:} This figure reports treatment effects of supportive and opposing comments, relative to the no-comments control, on time spent viewing the posts, expressed in log(seconds+1). Horizontal lines represent 95\% confidence intervals.
\end{minipage}
\end{figure}

\begin{figure}[H]
    \centering
    \caption{Donations (U.S.\ Dollars)}
    \label{fig:donations}
    \includegraphics[width=0.8\textwidth]{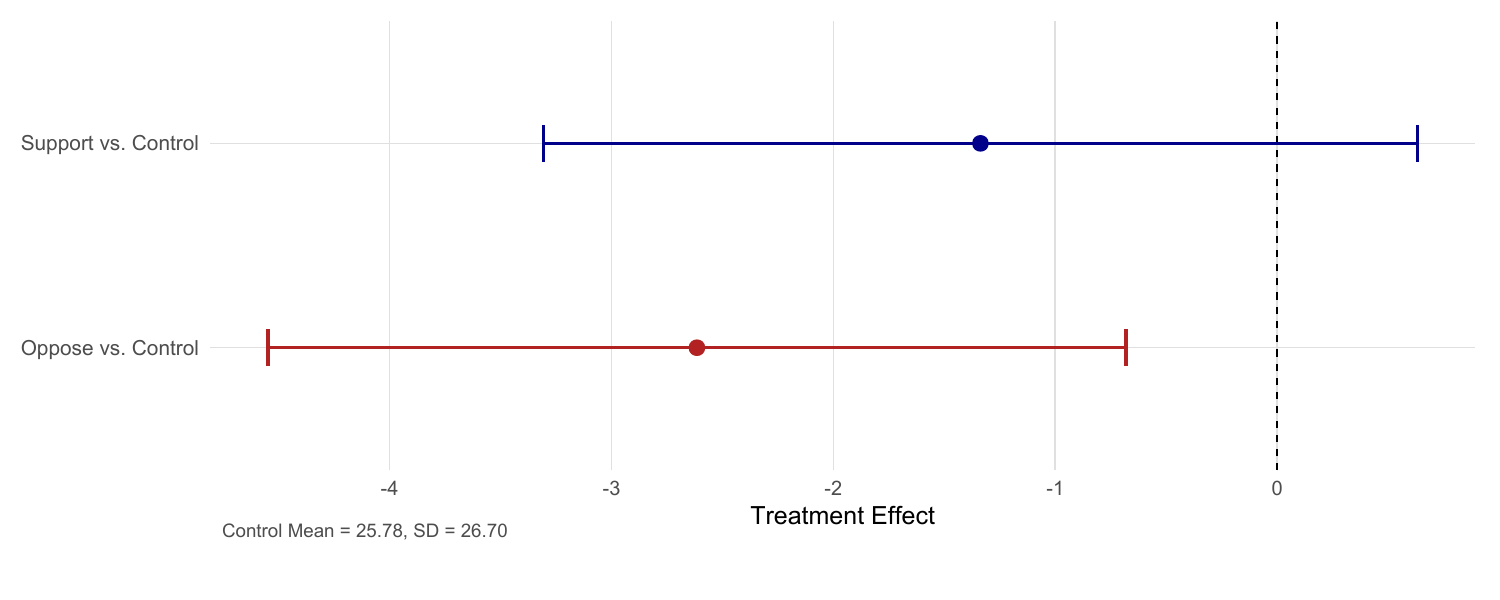}
\begin{minipage}{\textwidth}
\scriptsize
\noindent \textit{Notes:} This figure reports treatment effects of supportive and opposing comments, relative to the control, on the amount donated in the survey experiment, winsorized at the 95th percentile and expressed in U.S.\ dollars. Horizontal lines represent 95\% confidence intervals.
\end{minipage}

\end{figure}

\begin{table}[H]
    \centering
    \caption{Effects of Comment Stance on Attitudes, Donations, and Newsletter Sign-Up}
    \label{tab:survey_arm_main}
    \begin{adjustbox}{max width=\textwidth}
        \begin{tabular}{l c c c}
\toprule
 & \textit{Post-Exposure Attitudes} & \textit{Donated} & \textit{Newsletter Sign-Up} \\
\cmidrule(lr){2-2} \cmidrule(lr){3-3} \cmidrule(lr){4-4}
 & (1) & (2) & (3) \\
\midrule
Supportive vs.\ Control & -0.032 & -0.028 & 0.012 \\
 & (0.036) & (0.017) & (0.015) \\
Opposing vs.\ Control & -0.118$^{***}$ & -0.053$^{***}$ & 0.007 \\
 & (0.036) & (0.017) & (0.015) \\
Opposing vs.\ Supportive & -0.087$^{**}$ & -0.024 & -0.006 \\
 & (0.035) & (0.017) & (0.015) \\
\midrule
Controls & \checkmark & \checkmark & \checkmark \\
\midrule
Observations & 3,868 & 3,868 & 3,868 \\
$R^2$ & 0.1885 & 0.1001 & 0.1126 \\
Mean $Y$ (Control) & -0.000 & 0.733 & 0.193 \\
\bottomrule
\end{tabular}

    \end{adjustbox}
\begin{minipage}{\textwidth}
\scriptsize
\noindent \textit{Notes:} This table reports OLS estimates from equation~\eqref{eq:main_survey} for the three primary attitudinal and behavioral outcomes of the survey experiment. \textit{Post-Exposure Attitudes} is the attitudes index described in Section~\ref{sec:survey}, aggregated by inverse-covariance weighting following \citet{anderson2008multiple} and standardized to mean zero and standard deviation one in the control group, with higher values indicating greater alignment with the organization's position. \textit{Donated} is an indicator for committing a positive amount in the incentivized donation task, and \textit{Newsletter Sign-Up} is an indicator for providing an email address to the organization. All columns include the pre-specified baseline controls listed in Section~\ref{sec:survey}. Standard errors robust to heteroskedasticity are in parentheses. $^{*}p<0.10$; $^{**}p<0.05$; $^{***}p<0.01$.
\end{minipage}
\end{table}

\begin{figure}[H]
    \centering
        \caption{Heterogeneous Treatment Effects for Time Spent on the Post, log(seconds+1)}
    \label{fig:hte_time_spent}
    \includegraphics[width=0.85\textwidth]{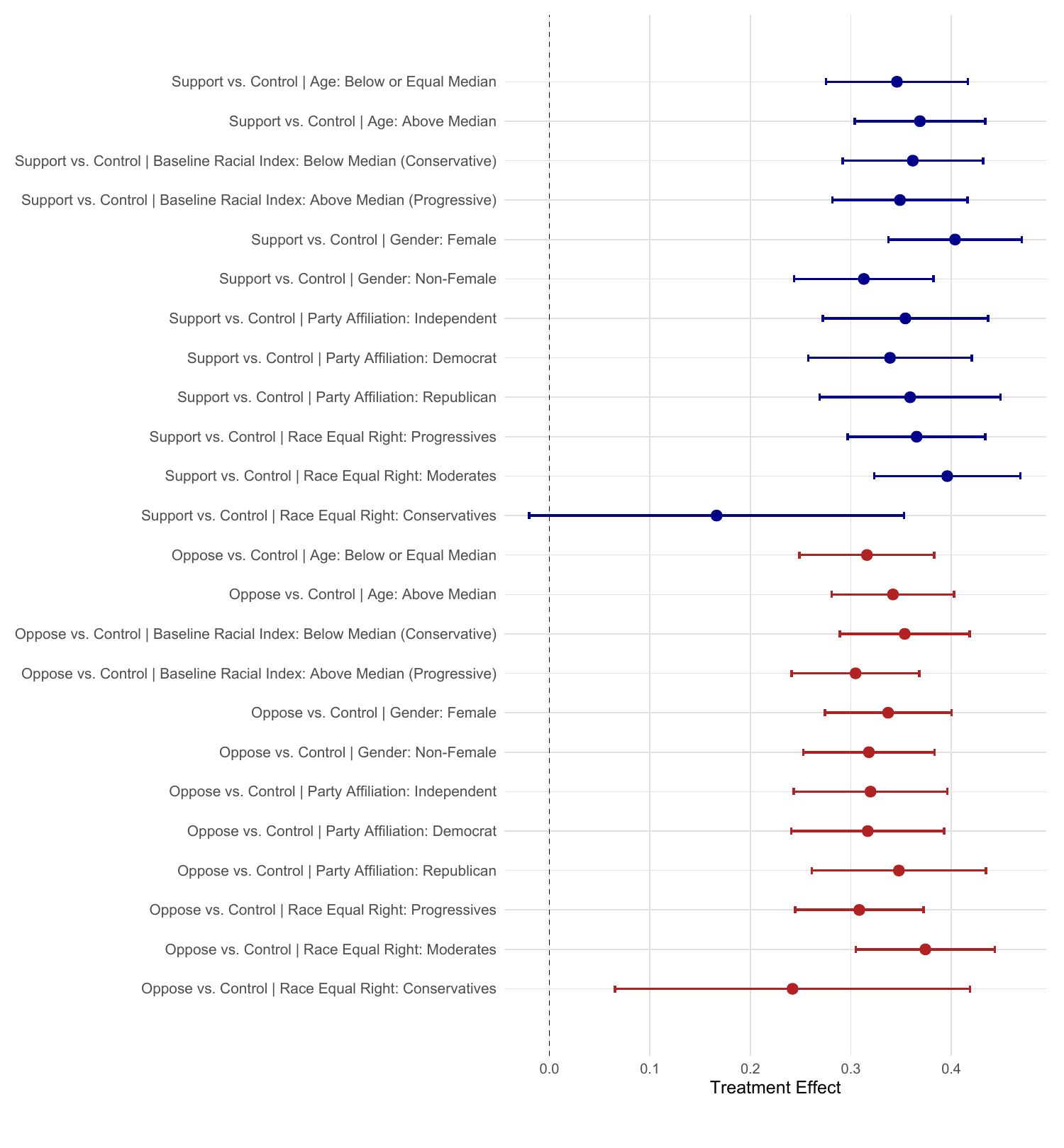}
\begin{minipage}{\textwidth}
\scriptsize
\noindent \textit{Notes:} This figure reports heterogeneous treatment effects in the survey experiment for time spent on the post, expressed in $\log(1+t)$ where $t$ is time in seconds. Effects are shown separately by baseline racial attitudes, party affiliation, gender, age, and ideology. Horizontal lines represent 95\% confidence intervals.
\end{minipage}

\end{figure}

\begin{figure}[H]
    \centering
        \caption{Heterogeneous Treatment Effects for Racial Attitudes Index}
    \label{fig:hte_attitudes}
    \includegraphics[width=0.85\textwidth]{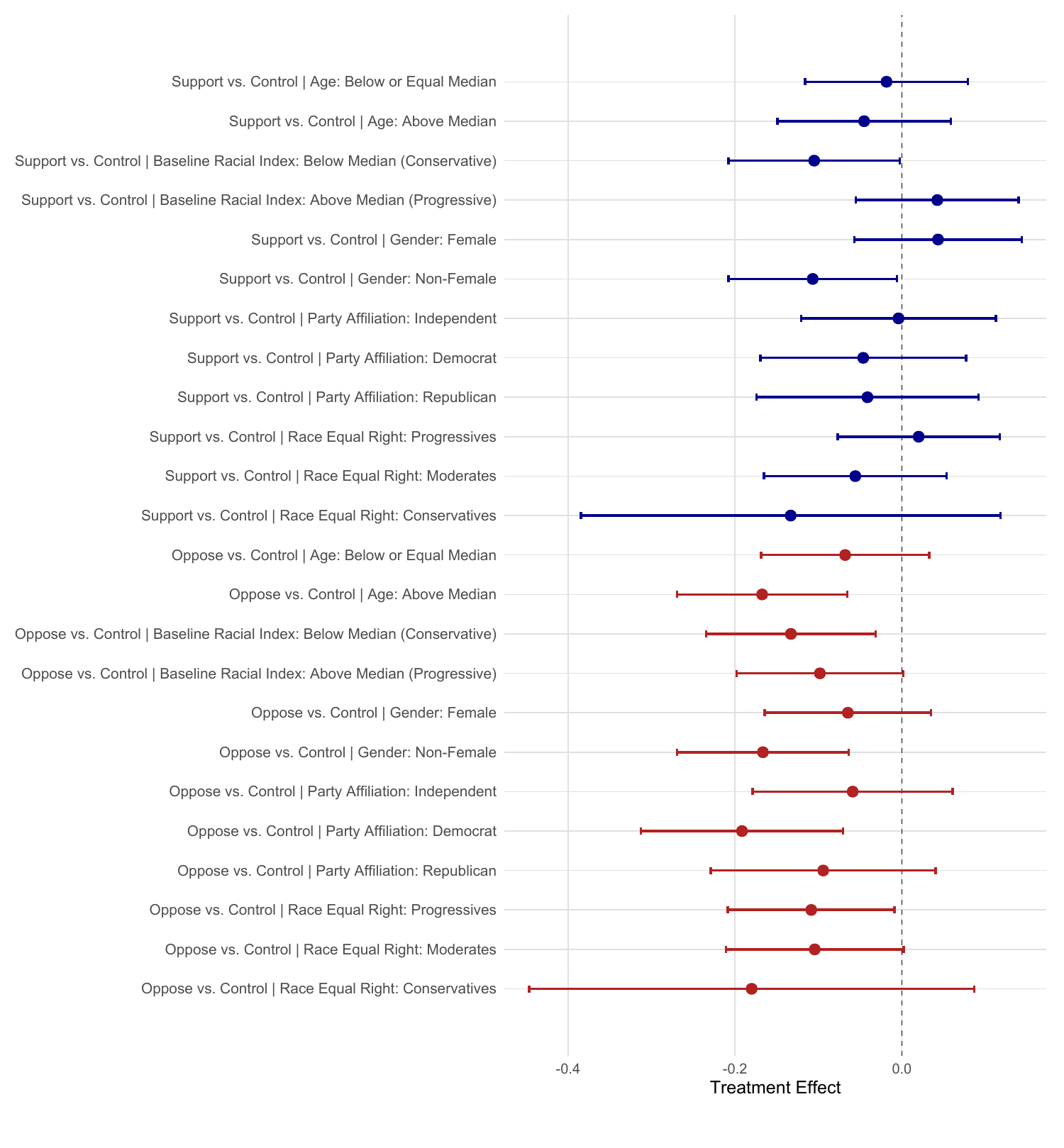}
\begin{minipage}{\textwidth}
\scriptsize
\noindent \textit{Notes:} This figure reports heterogeneous treatment effects in the survey experiment for the racial attitudes index. Effects are shown separately by baseline racial attitudes, party affiliation, gender, age, and ideology. Higher values indicate greater alignment with the organization’s position. Horizontal lines represent 95\% confidence intervals.
\end{minipage}
\end{figure}

\begin{figure}[H]
    \centering
        \caption{Heterogeneous Treatment Effects for Donation (Yes/No)}
    \label{fig:hte_donation}
    \includegraphics[width=0.85\textwidth]{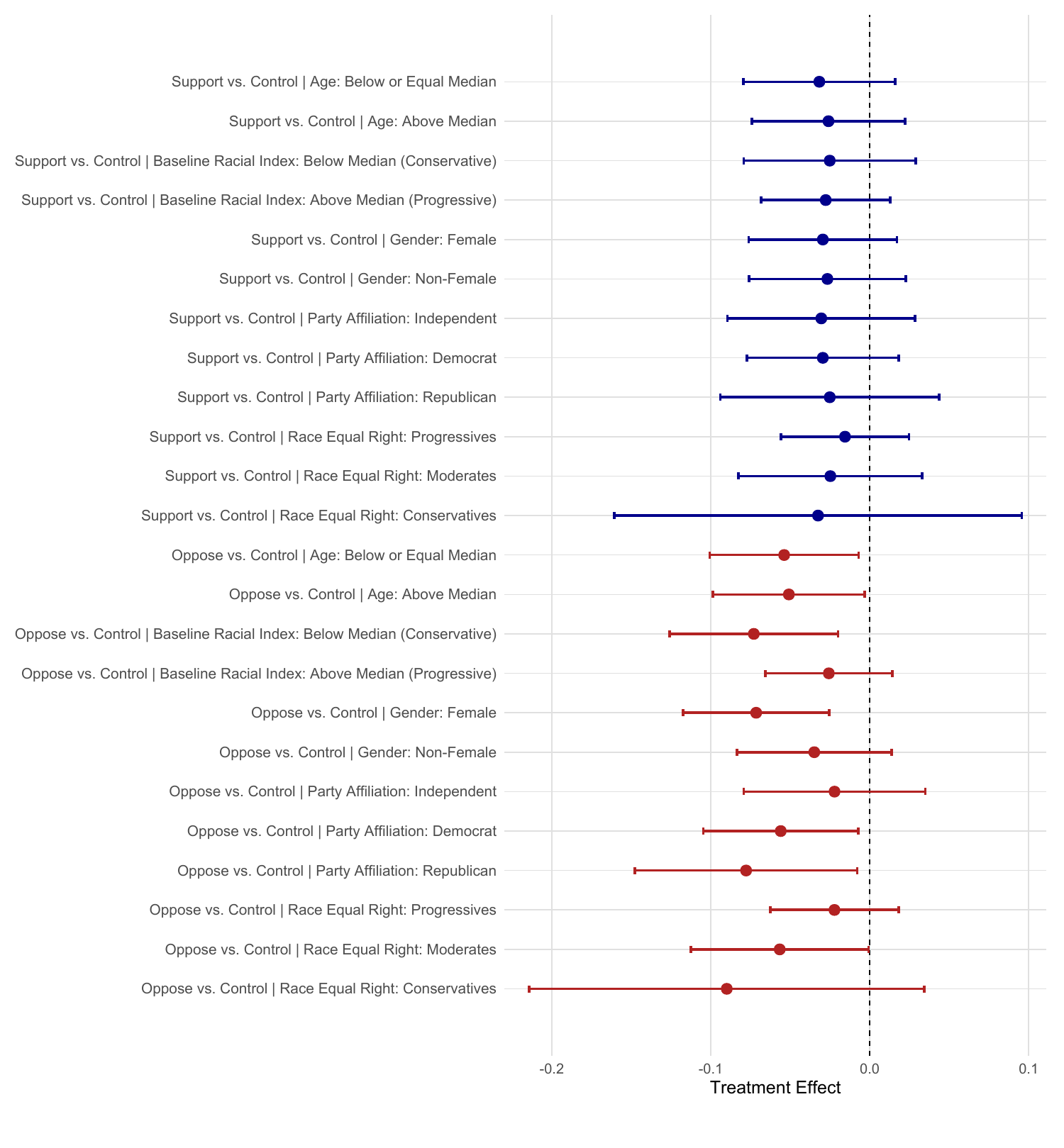}
\begin{minipage}{\textwidth}
\scriptsize
\noindent \textit{Notes:} This figure reports heterogeneous treatment effects in the survey experiment for a binary donation outcome. Effects are shown separately by baseline racial attitudes, party affiliation, gender, age, and ideology. Horizontal lines represent 95\% confidence intervals.
\end{minipage}
\end{figure}

\begin{figure}[H]
    \centering
        \caption{Heterogeneous Treatment Effects for Newsletter Sign-Up (Yes/No)}
    \label{fig:hte_newsletter}
    \includegraphics[width=0.85\textwidth]{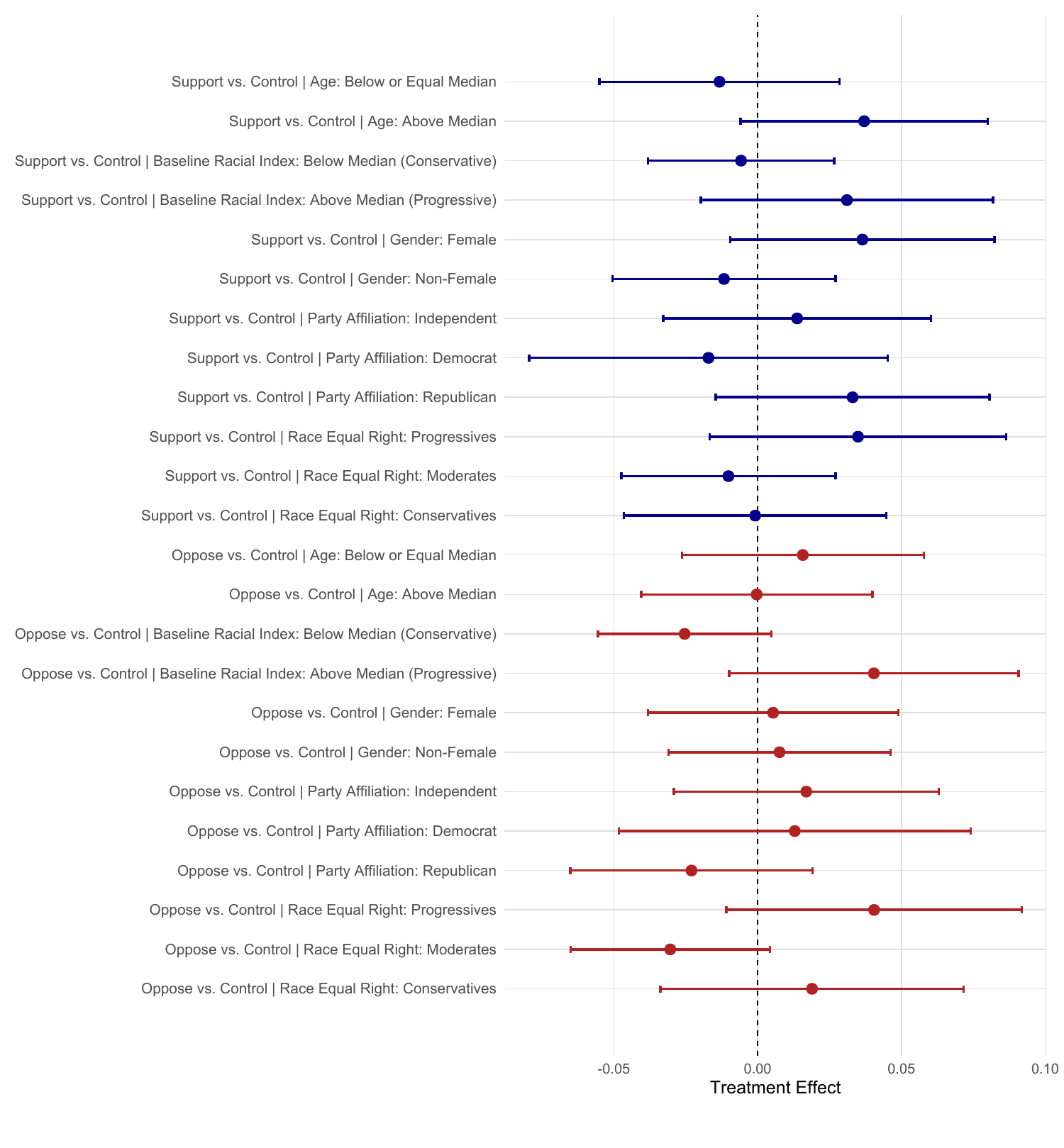}
\begin{minipage}{\textwidth}
\scriptsize
\noindent \textit{Notes:} This figure reports heterogeneous treatment effects in the survey experiment for a binary newsletter sign-up outcome. Effects are shown separately by baseline racial attitudes, party affiliation, gender, age, and ideology. Horizontal lines represent 95\% confidence intervals.
\end{minipage}

\end{figure}

\begin{figure}[H]
    \centering
    \caption{Heterogeneous Treatment Effects on Anger and Annoyance by Party Affiliation}
    \label{fig:survey_emotions}
    \includegraphics[width=0.8\textwidth]{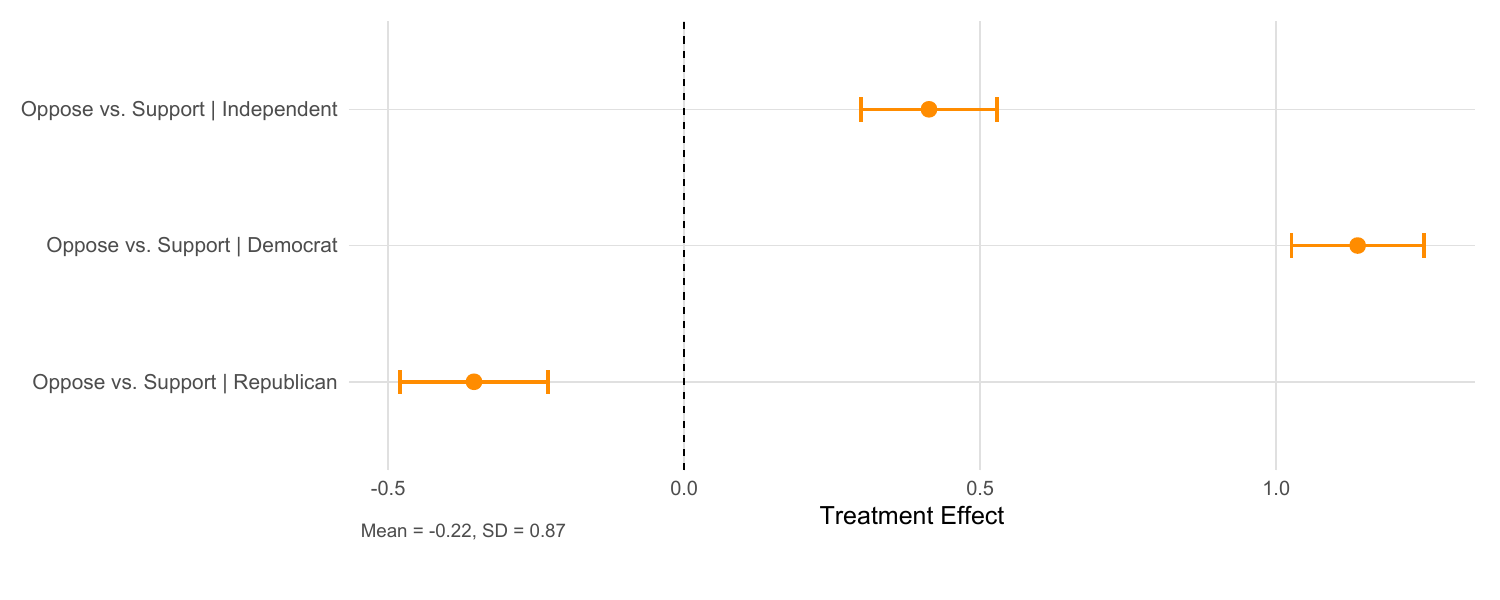}
\begin{minipage}{\textwidth}
\scriptsize
\noindent \textit{Notes:} This figure reports the differential effect of opposing versus supportive comments on an emotional response measure in the survey experiment, separately for Republicans, Democrats, and Independents. The outcome combines self-reported anger and annoyance triggered by the comments, with higher values indicating stronger negative emotional responses. The index is standardized to mean zero and standard deviation one in the pooled sample of the two treatment arms. The mean and standard deviation shown below the figure are for the Supportive arm. Horizontal lines represent 95\% confidence intervals.
\end{minipage}
\end{figure}

\begin{figure}[H]
    \centering
    \caption{Heterogeneous Treatment Effects on Curiosity and Reflection by Party Affiliation}
    \label{fig:cognitive_responses_primary_overall_by_party_affiliation}
    \includegraphics[width=0.85\textwidth]{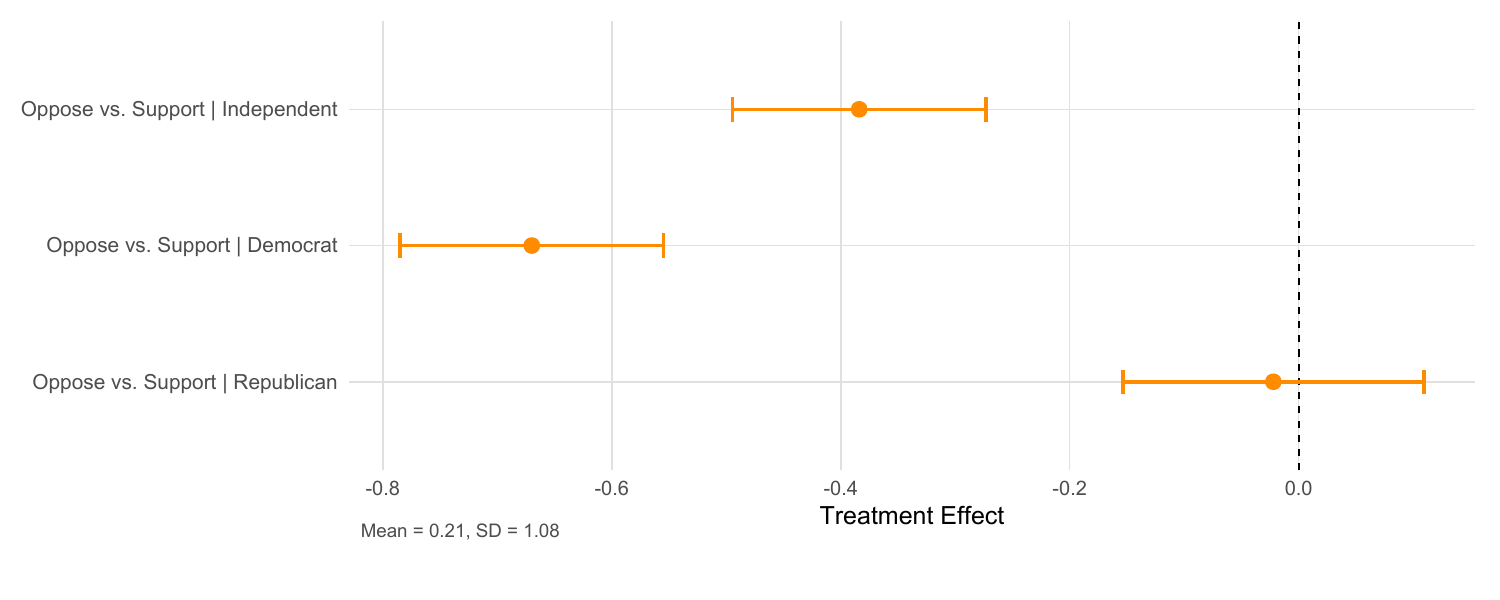}
\begin{minipage}{\textwidth}
\scriptsize
\noindent \textit{Notes:} This figure reports the effect of opposing versus supportive comments on an index of curiosity and reflection in the survey experiment, separately for Republicans, Democrats, and Independents. The outcome combines a Likert-scale measure of participants' interest in seeing additional comments on the post (interpreted as curiosity) with measures of how thought-provoking they find the post and the comments, with higher values indicating greater curiosity and reflection. The last item is asked only of respondents who saw comments, so this version of the index is used only to compare the two treatment conditions. The mean and standard deviation shown below the figure are for the Supportive arm. Horizontal lines represent 95\% confidence intervals.
\end{minipage}
\end{figure}

\begin{figure}[H]
    \centering
    \caption{Perceived Social Norms}
    \label{fig:survey_social_norms}
    \includegraphics[width=0.8\textwidth]{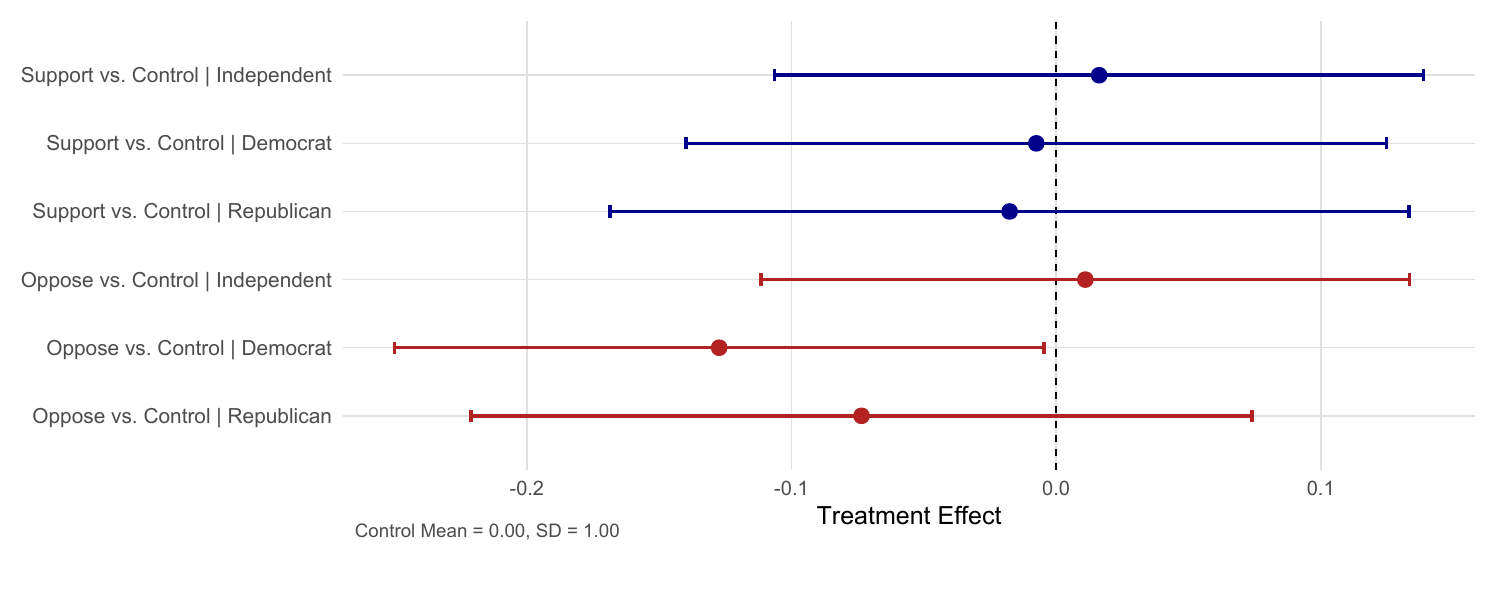}
\begin{minipage}{\textwidth}
\scriptsize
\noindent \textit{Notes:} This figure reports treatment effects of supportive and opposing comments, relative to the control, on perceived social norms in the survey experiment.  The outcome is respondents' incentivized estimate of the share of survey participants who agreed or strongly agreed with a conservative statement, reverse-coded so that higher values indicate more progressive perceived norms. Horizontal lines represent 95\% confidence intervals.
\end{minipage}
\end{figure}

\begin{table}[H]
    \centering
    \caption{Treatment-Arm Effects Across Alternative Control Sets}
    \label{tab:survey_arm_robustness}
    \begin{adjustbox}{max width=\textwidth}
        \begin{tabular}{l c c c c}
\toprule
 & (1) No Controls & (2) Main & (3) LASSO & (4) Full Controls \\
\midrule
\addlinespace[0.3em]
\multicolumn{5}{l}{\textbf{A. Post-Exposure Attitudes}} \\
\hspace{1em}Supportive vs.\ Control & -0.026 & -0.032 & -0.034 & -0.033 \\
 & (0.040) & (0.036) & (0.036) & (0.036) \\
\hspace{1em}Opposing vs.\ Control & -0.125$^{***}$ & -0.118$^{***}$ & -0.118$^{***}$ & -0.121$^{***}$ \\
 & (0.040) & (0.036) & (0.036) & (0.037) \\
\hspace{1em}Opposing vs.\ Supportive & -0.100$^{**}$ & -0.087$^{**}$ & -0.084$^{**}$ & -0.088$^{**}$ \\
 & (0.039) & (0.035) & (0.035) & (0.035) \\
\hspace{1em}Observations & 3,868 & 3,868 & 3,868 & 3,868 \\
\hspace{1em}$R^2$ & 0.0030 & 0.1885 & 0.1914 & 0.1958 \\
\addlinespace[0.3em]
\multicolumn{5}{l}{\textbf{B. Donated}} \\
\hspace{1em}Supportive vs.\ Control & -0.026 & -0.028 & -0.025 & -0.026 \\
 & (0.018) & (0.017) & (0.017) & (0.017) \\
\hspace{1em}Opposing vs.\ Control & -0.054$^{***}$ & -0.053$^{***}$ & -0.050$^{***}$ & -0.051$^{***}$ \\
 & (0.018) & (0.017) & (0.017) & (0.017) \\
\hspace{1em}Opposing vs.\ Supportive & -0.027 & -0.024 & -0.025 & -0.024 \\
 & (0.018) & (0.017) & (0.017) & (0.017) \\
\hspace{1em}Observations & 3,868 & 3,868 & 3,868 & 3,868 \\
\hspace{1em}$R^2$ & 0.0023 & 0.1001 & 0.1049 & 0.1069 \\
\addlinespace[0.3em]
\multicolumn{5}{l}{\textbf{C. Newsletter Sign-Up}} \\
\hspace{1em}Supportive vs.\ Control & 0.015 & 0.012 & 0.012 & 0.012 \\
 & (0.016) & (0.015) & (0.015) & (0.015) \\
\hspace{1em}Opposing vs.\ Control & 0.002 & 0.007 & 0.008 & 0.008 \\
 & (0.016) & (0.015) & (0.015) & (0.015) \\
\hspace{1em}Opposing vs.\ Supportive & -0.013 & -0.006 & -0.004 & -0.004 \\
 & (0.016) & (0.015) & (0.015) & (0.015) \\
\hspace{1em}Observations & 3,868 & 3,868 & 3,868 & 3,868 \\
\hspace{1em}$R^2$ & 0.0003 & 0.1126 & 0.1187 & 0.1199 \\
\bottomrule
\end{tabular}

    \end{adjustbox}
\begin{minipage}{\textwidth}
\scriptsize
\noindent \textit{Notes:} This table reports estimates from equation~\eqref{eq:main_survey} with separate indicators for the Supportive and Opposing conditions. The omitted category is the No Comments control. Column~(1) includes no covariates. Column~(2) is the main specification, and includes the pre-specified baseline controls: age, gender, education, race, ethnicity, party affiliation, political views, baseline racial attitudes, and beliefs about others' views. Column~(3) selects covariates from the full candidate set by LASSO. Column~(4) includes the full candidate set without selection. Panel~A reports the post-exposure attitudes index, aggregated by inverse-covariance weighting following \citet{anderson2008multiple} and standardized to mean zero and standard deviation one in the control group, with higher values indicating greater alignment with the organization's position. Panel~B reports \textit{Donated}, an indicator for committing a positive amount in the incentivized donation task, and Panel~C reports \textit{Newsletter Sign-Up}, an indicator for providing an email address to sign up for a newsletter. Standard errors robust to heteroskedasticity are in parentheses. $^{*}p<0.10$; $^{**}p<0.05$; $^{***}p<0.01$.

\end{minipage}
\end{table}

\newpage
\subsection{Results on Pre-specified Secondary Outcomes}\label{app:SecOutcomes}
This appendix reports results on pre-specified secondary outcomes. Figure~\ref{fig:post_exposure_attitudes_secondary_decomp} reports treatment effects on secondary attitudinal outcomes. The pattern mirrors the primary attitude index: opposing comments produce negative but imprecisely estimated effects on the index, driven primarily by reduced favorability towards Black Lives Matter. Effects on perceived importance of voting and technology issues — which are not directly covered by the stimuli — are small and close to zero. Supportive comments have little effect across all components.
Figure~\ref{fig:cognitive_responses_secondary_main} shows that the secondary curiosity index follows a similar pattern to the primary curiosity and reflection index, with opposing comments reducing curiosity relative to supportive comments. Turning to beliefs about commenters, Figure~\ref{fig:perceive_commenters_progressive_main} shows that supportive comments lead respondents to name progressives as the group more likely to comment on the post, while opposing comments have the opposite effect, consistent with respondents reading the stance of the comments they saw. Finally, Figure~\ref{fig:perceive_commenters_women_main} shows that exposure to a comment section, regardless of stance, leads respondents to name men as the group more likely to comment, relative to the control, even though each post displayed one comment under a male name and one under a female name.

\begin{figure}[H]
    \centering
    \caption{Treatment Effects on Secondary Attitude Index}
    \label{fig:post_exposure_attitudes_secondary_decomp}
    \includegraphics[width=0.85\textwidth]{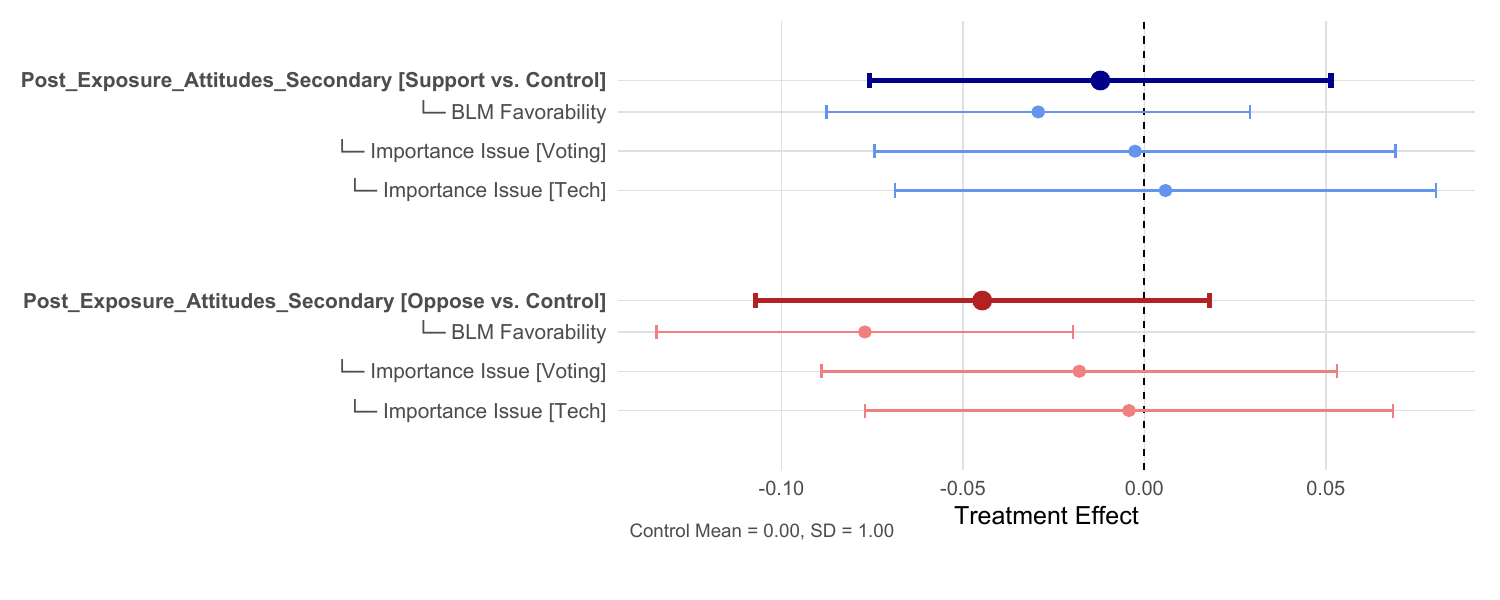}
\begin{minipage}{\textwidth}
\scriptsize
\noindent \textit{Notes:} This figure decomposes treatment effects on post-exposure attitudes (standardized, control mean = 0.00, SD = 1.00) into its component outcomes: favorability toward Black Lives Matter, and perceived importance of two racial justice sub-issues — voting and technology — that are not directly covered by the intervention posts. Effects are shown separately for the Oppose vs.\ Control and Support vs.\ Control comparisons, with the index estimate shown in bold and individual outcomes indented below. Horizontal lines represent 95\% confidence intervals.
\end{minipage}
\end{figure}

\begin{figure}[H]
    \centering
    \caption{Treatment Effects on Secondary Curiosity Index}
    \label{fig:cognitive_responses_secondary_main}
    \includegraphics[width=0.85\textwidth]{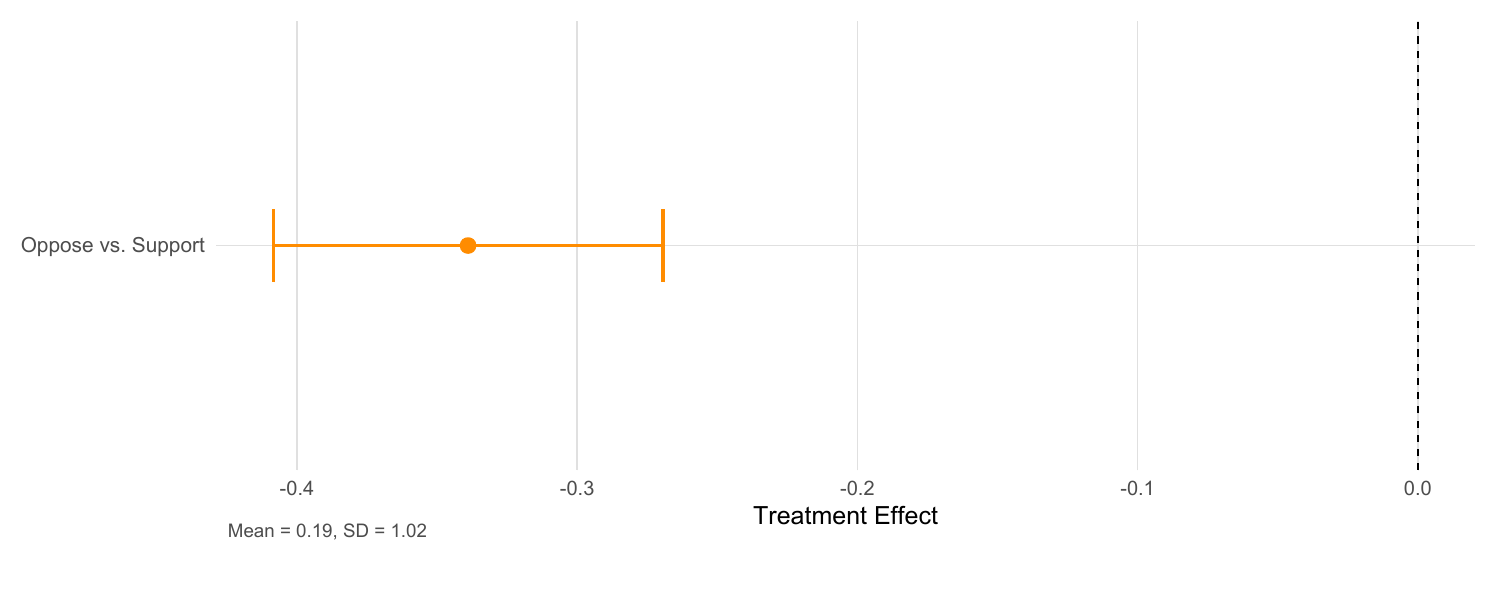}
\begin{minipage}{\textwidth}
\scriptsize
\noindent \textit{Notes:} This figure reports treatment effects on a secondary curiosity index  combining two measures: the extent to which comments increase curiosity about the underlying topic, and the extent to which they increase curiosity about the organization.  The effect is shown for the Oppose vs.\ Support comparison. The mean and standard deviation shown below the figure are for the Supportive arm. Horizontal lines represent 95\% confidence intervals.
\end{minipage}
\end{figure}

\begin{figure}[H]
    \centering
    \caption{Treatment Effects on Perceived Progressiveness of Commenters}
    \label{fig:perceive_commenters_progressive_main}
    \includegraphics[width=0.85\textwidth]{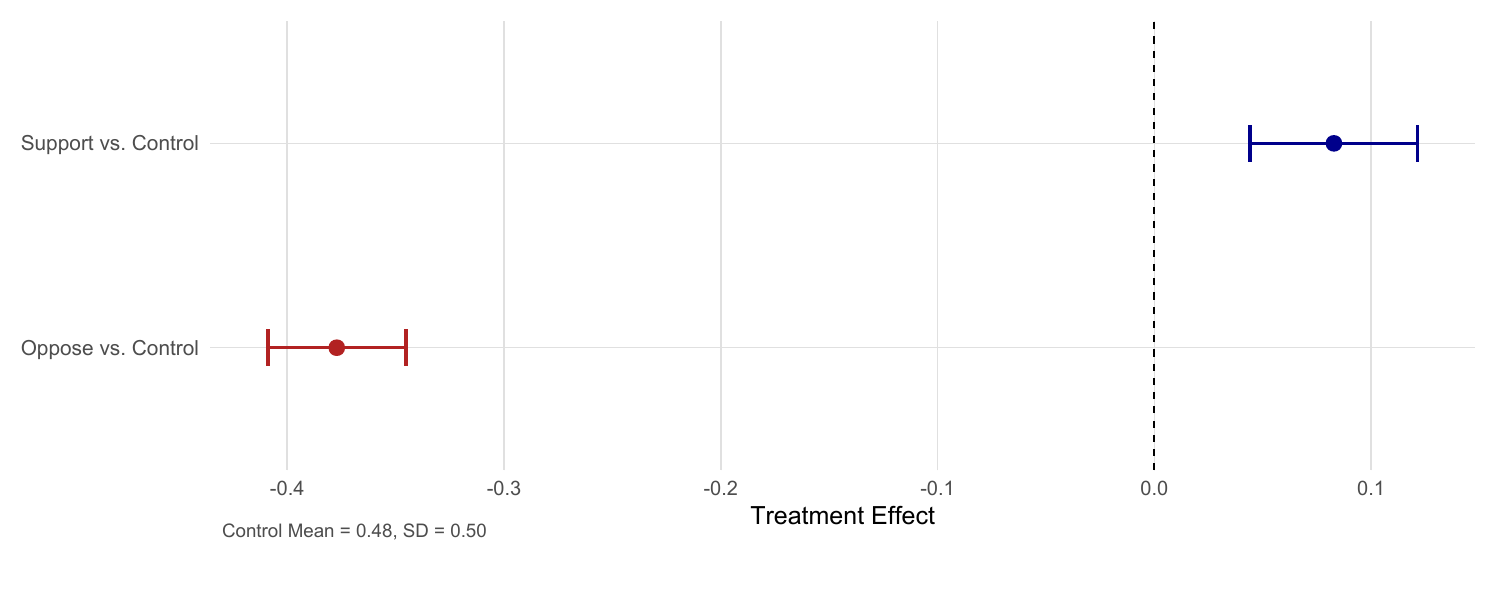}
\begin{minipage}{\textwidth}
\scriptsize
\noindent \textit{Notes:} This figure reports treatment effects on respondents' beliefs about the political ideology of commenters. Respondents were asked which group they believe is more likely to comment on the post; the outcome is coded as one if the respondent perceives progressives as more likely to comment. Effects are shown for the Oppose vs.\ Control and Support vs.\ Control comparisons.  Horizontal lines represent 95\% confidence intervals.
\end{minipage}
\end{figure}

\begin{figure}[H]
    \centering
    \caption{Treatment Effects on Perceived Gender of Commenters}
    \label{fig:perceive_commenters_women_main}
    \includegraphics[width=0.85\textwidth]{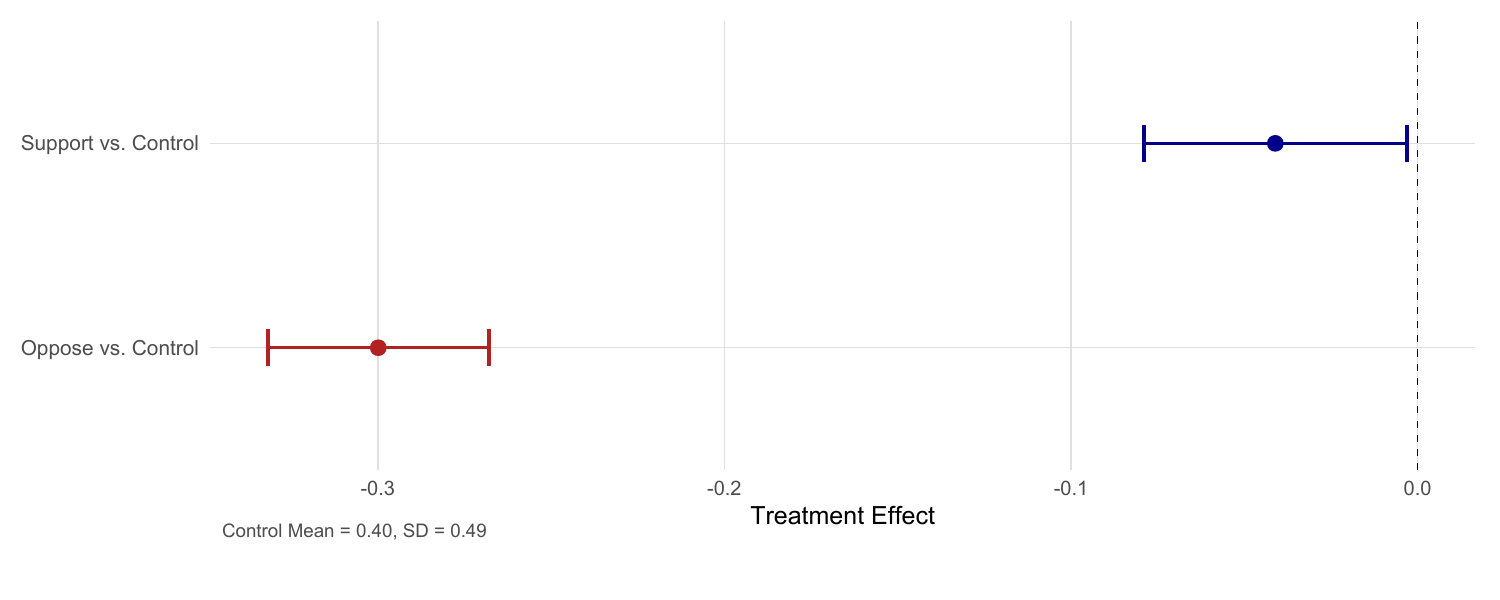}
\begin{minipage}{\textwidth}
\scriptsize
\noindent \textit{Notes:} This figure reports treatment effects on respondents' beliefs about the gender of commenters. Respondents were asked which group they believe is more likely to comment on the post; the outcome is coded as one if the respondent perceives women as more likely to comment. Effects are shown for the Oppose vs.\ Control and Support vs.\ Control comparisons. Horizontal lines represent 95\% confidence intervals.
\end{minipage}
\end{figure}

\newpage
\section{Cost-Benefit Analysis for Fundraising Campaigns \label{sec:cost-benefit}}
\setcounter{figure}{0}
\renewcommand{\thefigure}{G\arabic{figure}}
\renewcommand{\theHfigure}{G\arabic{figure}}

\setcounter{table}{0}
\renewcommand{\thetable}{G\arabic{table}}
\renewcommand{\theHtable}{G\arabic{table}}

We consider a nonprofit that chooses whether to tolerate opposing comments below its ads. We focus on opposing comments because they deliver the sharpest organizational trade-off in our setting: they increase clicks and website traffic, but reduce donations and shift attitudes in a less progressive direction. Let
$a \in \{C,O\}$ denote the comment policy, where $C$ is the control policy (no comments), and $O$ is the
opposing-comments policy. The organization is assumed to maximize expected donations generated by
a campaign with budget $B$:
\begin{equation}
\mathcal{D}_a
=
B \times r_a \times CTR_a \times CVR_a,
\label{eq:donations_general}
\end{equation}
where $r_a$ is the number of users reached per dollar spent, $CTR_a$ is the click-through rate out
of reached users, and $CVR_a$ is the donation conversion rate conditional on click.

Taking the ratio of donations under opposing comments relative to the control gives
\begin{equation}
\frac{\mathcal{D}_O}{\mathcal{D}_C}
=
\underbrace{\frac{r_O}{r_C}}_{\kappa}
\times
\underbrace{\frac{CTR_O}{CTR_C}}_{\tau}
\times
\underbrace{\frac{CVR_O}{CVR_C}}_{\text{conversion effect}}
\label{eq:ratio_start}
\end{equation}

We decompose the conversion effect into two components. First, we use the survey experiment to
proxy for the direct downstream effect of opposing comments on donation decisions once users have
already processed the ad. Let
\begin{equation}
\delta
\equiv
\frac{DonationRate_O}{DonationRate_C}
\label{eq:delta_def}
\end{equation}
Using the estimates in our survey experiment,
\begin{equation}
\delta = \frac{0.677}{0.73} = 0.927
\label{eq:delta_value}
\end{equation}

Second, we introduce a reduced-form ``traffic quality'' parameter, $q$, which captures any
additional change in conversion arising from the composition of users who click or are reached. In
particular, $q<1$ if opposing comments attract lower-intent users, or if the platform's delivery
algorithm shifts exposure toward users who are more likely to engage but less likely to donate.
Conversely, $q>1$ if opposing comments attract or reach users who are more likely to donate.

Combining these pieces,
\begin{equation}
\frac{CVR_O}{CVR_C} = \delta \times q,
\label{eq:cvr_decomp}
\end{equation}
so that
\begin{equation}
\frac{\mathcal{D}_O}{\mathcal{D}_C}
=
\kappa \times \tau \times \delta \times q.
\label{eq:main_ratio}
\end{equation}

The parameter $\kappa$ captures the extent to which opposing comments reduce advertising costs and
therefore increase reach per dollar. Under our benchmark specification, we set $\kappa=1$, which
corresponds to the case in which delivery efficiency is unaffected. This is a natural starting
point given our experimental design, which imposed fixed budgets and aimed to keep reach balanced
across arms. Once these constraints are relaxed, however, delivery efficiency may change. If the
platform treats engagement as a positive signal and reduces cost per reach under opposing comments,
then $\kappa>1$. If instead opposing comments make delivery less efficient, then $\kappa<1$.
In our scenario analysis, we vary $\kappa$ only in a narrow range around one in order to remain
conservative and to reflect the fact that large delivery-cost differences are not part of our benchmark.

The parameter $\tau$ is directly computed from the Facebook field experiment:
\begin{equation}
\tau
=
\frac{CTR_O}{CTR_C}
=
\frac{0.261\%}{0.228\%}
=
1.145.
\label{eq:tau_value}
\end{equation}

Substituting \eqref{eq:delta_value} and \eqref{eq:tau_value} into \eqref{eq:main_ratio} yields
\begin{equation}
\frac{\mathcal{D}_O}{\mathcal{D}_C}
=
1.145 \times 0.927 \times \kappa \times q
=
1.061 \times \kappa \times q.
\label{eq:final_ratio}
\end{equation}

Equation \eqref{eq:final_ratio} is our main back-of-the-envelope formula. It shows that the effect
of tolerating opposing comments depends on two parameters:
\begin{itemize}
    \item $\kappa$: a narrow cost-efficiency parameter, centered at one, capturing small deviations
    from the fixed-budget/fixed-reach benchmark;
    \item $q$: a traffic-quality parameter, capturing whether the users induced to click or reached
    by the platform are more or less likely to convert.
\end{itemize}

\paragraph{Base case}
In the benchmark case, we set
$\kappa = 1$, $q = 1$,
so that opposing comments affect donations only through the observed click effect in Facebook and
the direct downstream donation effect in the survey experiment. In this case,
\begin{equation}
\frac{\mathcal{D}_O}{\mathcal{D}_C} = 1.061,
\end{equation}
implying a $6.1\%$ increase in expected donations.

\paragraph{Scenario analysis}
To assess sensitivity to deviations in the other parameters, we vary $\kappa$ only slightly around one and allow for larger, asymmetric movements in
$q$. Specifically, we center the analysis at $q=1$, allow a modest upside case with $q>1$, and place greater weight on
$q<1$ because our evidence suggests that opposing comments attract relatively more clicks from users in less progressive areas, making a deterioration in traffic quality more plausible than an improvement.  In particular, we consider combinations of: 
$\kappa \in \{0.98, 1.00, 1.02\}$
and 
$q \in \{0.90, 1.00, 1.05\}$.
Table \ref{tab:scenario_tradeoff} shows the changes in campaign donation rates under different scenarios.

\begin{table}[htbp]
\centering
\caption{Changes in Campaign Donation Rates Under Alternative Scenarios}
\label{tab:scenario_tradeoff}
\begin{tabular}{lccc}
\toprule
& \multicolumn{3}{c}{Traffic-quality parameter $q$} \\
\cmidrule(lr){2-4}
Cost-efficiency parameter $\kappa$ & 0.90 & 1.00 & 1.05 \\
\midrule
0.98 & 0.936 \; ($-6.4\%$) & 1.040 \; ($+4.0\%$) & 1.092 \; ($+9.2\%$) \\
1.00 & 0.955 \; ($-4.5\%$) & 1.061 \; ($+6.1\%$) & 1.114 \; ($+11.4\%$) \\
1.02 & 0.974 \; ($-2.6\%$) & 1.082 \; ($+8.2\%$) & 1.136 \; ($+13.6\%$) \\
\bottomrule
\end{tabular}

\vspace{0.2cm}
\begin{minipage}{0.92\textwidth}
\scriptsize
\textit{Notes:} Each cell reports the implied ratio
$\mathcal D_O / \mathcal D_C = 1.061 \times \kappa \times q$.
The benchmark case is $(\kappa,q)=(1,1)$. We keep $\kappa$ in a tight range around one because the
experimental design split budgets evenly across arms and optimized for reach, so large delivery-cost
differences are not part of the maintained benchmark. By contrast, $q$ is allowed to vary more
widely because opposing comments may change the quality of induced traffic and, if delivery
responds to engagement, the composition of users reached. Values greater than one imply that
tolerating opposing comments increases expected donations; values below one imply that it decreases
expected donations.
\end{minipage}
\end{table}

The break-even condition is
\begin{equation}
\kappa \times q > \frac{1}{1.061} \approx 0.943.
\label{eq:breakeven}
\end{equation}
Thus, under the benchmark $\kappa=1$, a deterioration in traffic quality of only about $5.7\%$ is
enough to overturn the baseline gain from higher click-through rates.

\paragraph{Discussion}
The benchmark case $(\kappa,q)=(1,1)$ suggests 
that the overall donation rate increases by 6.1\%.
This benchmark is useful as
a reference point, but it likely corresponds to a relatively optimistic scenario in which the
additional traffic generated by opposing comments has the same propensity to donate as baseline traffic.
Our results suggest that this assumption may be unrealistic: opposing comments generate relatively
more traffic from users in less progressive areas, making it plausible that the induced traffic is
of lower quality for fundraising purposes. For this reason, cases with $q<1$ are likely to be more
informative than the benchmark case, even though we allow both parameters to vary in the scenario
analysis.

This exercise also abstracts from other potential objectives of the organization. In particular, we
do not incorporate effects on attitudes, newsletter sign-ups, or shifts in the platform's delivery algorithm,
even though these may also matter for advocacy organizations. We focus on donations
because they map naturally into a dollar-valued objective and therefore allow for a simple
back-of-the-envelope comparison. The table should therefore be interpreted as a stylized 
exercise focused on fundraising. Nonetheless, it highlights that modest gains in traffic can be offset---or overturned---if opposing comments attract lower-intent
users or shift delivery toward users who are less likely to convert.

\newpage
\section{Survey Instrument \label{app:SurveyInstrument}}

\vspace{12pt}
\emph{Screening}

Welcome! We have a few quick questions before we start.

This should take no more than 20 seconds. We will let you know if you are eligible for the study and provide details on participation payments.

\begin{enumerate}
    \item Do you live in the United States?\\
    \emph{[Yes / No]}

    \item What is your age? \emph{[number entry box]}

    \item What is your gender?\\
    \emph{[Male / Female / Non-binary / other / I prefer not to answer]}

    \item What is your race?\\
    \emph{[American Indian / Alaska Native / Asian / Pacific Islander / Black / African American / White / Other / Mixed Race]}

    \item Are you of Hispanic or Latino origin?\\
    \emph{[Yes / No]}
\end{enumerate}

\emph{[Continue if U.S.\ resident, aged 18--64.]}

\vspace{12pt}
\emph{Consent}
 
\emph{[Consent form]}

\begin{enumerate}
    \item[] I agree to participate, and I promise to read the questions carefully and answer honestly

    \item[] I do not agree to participate, or I cannot promise to read the questions carefully and answer honestly
\end{enumerate}

\vspace{12pt}
\emph{Baseline Opinions}

Thanks for agreeing to participate! We value your opinions and are interested in hearing what you think.

\begin{enumerate}
    \item In politics, as of today, do you consider yourself a Republican, a Democrat or an independent?\\
    \emph{[Republican / Democrat / Independent / Other]}

    \item We hear a lot of talk these days about liberals and conservatives. Which of the following best describe your political view?\\
    \emph{[Very liberal / Liberal / Moderate; middle of the road / Conservative / Very conservative / Haven't thought much about this/don't know]}

    \item When it comes to giving African Americans equal rights with white Americans, do you think our country has\ldots\\
    \emph{[Gone too far / Not gone far enough / Been about right]}

    \emph{[Randomize order of ``Gone too far'' and ``Not gone far enough.'']}

    \item Do you believe that the increased public attention to the history of slavery and racism is generally good or bad for our society?\\
    \emph{[Very good / Somewhat good / Neither good nor bad / Somewhat bad / Very bad]}

    \item \textbf{On the issues of race and racism,} my position is\ldots\\
    \emph{[Very progressive / Progressive / Moderate; middle of the road / Conservative / Very conservative / Haven't thought much about this / don't know]}

    \emph{[Randomly flip the order.]}

    \item Now we'd like to know your best guess about how people in a \textbf{representative sample of adults in the United States} answered this question in 2025.

    ``When it comes to giving African Americans equal rights with white Americans, do you think our country has\ldots''

    Please estimate what percentage of respondents chose each response. Your answers should add up to 100\%.

    \emph{[Constant sum question; entries for: \_\_\_\% Gone too far / \_\_\_\% Not gone far enough / \_\_\_\% Been about right; must sum to 100.]}
\end{enumerate}

\vspace{12pt}
\emph{Social Media Use}

\begin{enumerate}
    \item How much time do you spend on social media (e.g., Facebook, Instagram, TikTok, YouTube) excluding Messenger and WhatsApp, on an average day?\\
    \emph{[Less than 5 minutes a day / Between 5 and 30 minutes a day / Between 30 and 60 minutes a day / Between 1 and 2 hours / Between 2 and 4 hours / More than 4 hours]}

    \item How often do you read or check comments on social media?\\
    \emph{[Never / Rarely / Sometimes / Often / Very often]}
\end{enumerate}

\vspace{12pt}
\emph{Additional Demographics}

\begin{enumerate}
    \item In what zip code do you currently live? \emph{[text entry box; validation: U.S.\ ZIP code]}

    \item What is the highest degree or level of schooling that you have completed?\\
    \emph{[Less than a high school diploma / High school diploma or equivalent (for example: GED) / Some college but no degree / Associate's degree / Bachelor's degree / Graduate degree (for example: MA, MBA, JD, PhD)]}
\end{enumerate}

\vspace{12pt}
\emph{Attention Check}

\begin{enumerate}
    \item In order to facilitate our research, we are interested in knowing certain factors about you. Specifically, we are interested in whether you actually take the time to read the instructions; if not, then the data we collect based on your responses will be invalid. So, in order to demonstrate that you have read the instructions, please ignore the next question, and simply write ``I read the instructions'' in the ``Any comments?'' box below. Thank you very much.

    What is your marital status?\\
    \emph{[Single / Married / Other]}

    Any comments? \emph{[text box]}
\end{enumerate}

\vspace{12pt}
\emph{Intervention}

\emph{[Participants randomized into 3 groups: No Comments, Supportive, Opposing.]}

\emph{[3 posts of the same treatment type are shown, varying order of topics: education, environment, and police.]}

We are interested in your reactions to social media posts about racial justice. You will be shown three posts from Color of Change, a leading U.S.-based racial justice advocacy organization.

\emph{[Next page.]}

Here is a post from Color of Change:

\emph{[Screenshot shown.]}

\emph{[Next page.]}

Now we will ask you some questions about this post.

\emph{[Post shown again.]}

\begin{enumerate}
    \item Which of the following emojis would you react with if you saw this post on Facebook?\\
    \emph{[Like / Love / Care / Haha / Wow / Sad / Angry / I would not react]}

    \item Would you comment on this post if you saw it on Facebook?\\
    \emph{[Yes / No]}

    \item \emph{[If yes to previous question]} What comment would you make? \emph{[open text]}

    \item Would you click on this post to visit the website if you saw it on Facebook?\\
    \emph{[Yes / No]}
\end{enumerate}

\emph{[Screenshots of remaining posts shown; questions above repeated for each post.]}

\vspace{12pt}
\emph{Newsletter}

\begin{enumerate}
    \item Would you like to sign up for the Color of Change newsletter?\\
    \emph{[Yes / No]}
\end{enumerate}

\vspace{12pt}
\emph{Donation}

\begin{enumerate}
    \item You have been automatically enrolled in a lottery to win up to \$100. If you win, you have the option to donate some or all of your winnings to Color of Change.

    The payment will be made to you as a bonus, so no further action is required on your part. If you are one of the lottery winners, you will be paid, in addition to your participation payment, \$100 minus the amount you donated. We will directly pay your desired donation amount to Color of Change.

    How much, if any, would you be willing to donate to Color of Change in case you won \$100? \emph{[number entry]}
\end{enumerate}

\vspace{12pt}
\emph{Post-Exposure Opinion}

\begin{enumerate}
    \item How would you describe your overall opinion of Color of Change?\\
    \emph{[Very unfavorable / Somewhat unfavorable / Neutral / Somewhat favorable / Very favorable]}

    \item How would you describe your overall opinion of Black Lives Matter?\\
    \emph{[Very unfavorable / Somewhat unfavorable / Neutral / Somewhat favorable / Very favorable]}

    \item How willing would you be to discuss political issues with someone who has \textbf{progressive views} on racial issues?\\
    \emph{[Very unwilling / Somewhat unwilling / Neither willing nor unwilling / Somewhat willing / Very willing]}

    \item How willing would you be to discuss political issues with someone who has \textbf{conservative views} on racial issues?\\
    \emph{[Very unwilling / Somewhat unwilling / Neither willing nor unwilling / Somewhat willing / Very willing]}

    \item People differ in how important they consider different racial justice issues.

    How important is each of the following issues to you personally?\\
    \emph{[Matrix; rows: Voter suppression and voting rights / Criminal-justice reform (e.g., policing, sentencing, incarceration) / Education equity (e.g., school funding, achievement gaps) / Environmental justice (e.g., pollution exposure, clean air/water access) / Technology fairness (e.g., algorithmic bias, digital discrimination); 5-point scale: Not at all important / Slightly important / Moderately important / Very important / Extremely important]}

    \item Please tell us the extent to which you agree or disagree with the statement below.

    It's really a matter of some people not trying hard enough. Black people could be just as well off as white people if they would only try harder.\\
    \emph{[Strongly disagree / Disagree / Slightly disagree / Neither agree nor disagree / Slightly agree / Agree / Strongly agree]}

    \item What do you think is the most important issue facing Black people today? \emph{[open text]}
\end{enumerate}

\vspace{12pt}
\emph{Opinion about Others}

Now we are going to ask you to make a guess.

You can earn a Guess Bonus of up to \$1 based on the accuracy of your estimate. We will compare your estimate to the actual percentages observed in this study. The formula we use rewards you more when your estimate is closer to the true value. \textbf{The closer your guess is to the correct percentage, the larger your bonus.}

\textbf{Your best strategy is to give your honest, best estimate.}

(You do \emph{not} need to know the formula to earn the bonus.)

\emph{[Next page.]}

\begin{enumerate}
    \item What percentage of participants in this U.S.\ adult sample do you think \textbf{agreed or strongly agreed} with the following statement?

    ``It's really a matter of some people not trying hard enough. Black people could be just as well off as white people if they would only try harder.''

    Please enter a number between \textbf{0 and 100}.\\
    You can earn a \textbf{bonus} based on how close your estimate is to the true value.

    \textbf{\_\_\_\%} \emph{[number entry]}

    \emph{[Pop-up window with quadratic scoring rule: Guess Bonus = \$1 $-$ ((Your Answer $-$ True Value)/100)$^2$. Small errors reduce your bonus slightly; larger errors reduce it much more. If your estimate is exactly correct, you will earn \$1. Your best strategy is to give your honest, best estimate.]}
\end{enumerate}

\emph{[Environment post shown again.]}

\begin{enumerate}
    \setcounter{enumi}{1}
    \item How thought-provoking do you find this post?\\
    \emph{[Not at all / A little / Somewhat / Very / Extremely]}
\end{enumerate}

\vspace{12pt}
\emph{Opinion about Comments in Post} \emph{[treatment arms only]}

Now consider the \textbf{comments} below the following social media post.

\emph{[Environment post with comments shown again.]}

\begin{enumerate}
    \item To what extent do the comments make you feel:\\
    \emph{[Matrix; rows: Angry / Annoyed; 5-point scale: Not at all / A little / Somewhat / Very / Extremely]}

    \item How thought-provoking do you find these comments?\\
    \emph{[Not at all / A little / Somewhat / Very / Extremely]}

    \item To what extent do the comments make you feel:\\
    \emph{[Matrix; rows: Curious about the topic / Curious about the organization; 5-point scale: Not at all / A little / Somewhat / Very / Extremely]}

    \item How representative do you think these comments are of what people generally think about this issue?\\
    \emph{[Not at all representative / Slightly representative / Moderately representative / Very representative / Extremely representative]}
\end{enumerate}

\vspace{12pt}
\emph{Opinion about Comments in Post} \emph{[all arms]}

\emph{[Environment post shown again.]}

\begin{enumerate}
    \item How interested would you be in opening the comment section to see more comments on this post?\\
    \emph{[Not at all interested / Slightly interested / Moderately interested / Very interested / Extremely interested]}

    \item In your opinion, which group is more likely to comment on this post?\\
    \emph{[Men / Women / Men and women are equally likely / Not sure]}

    \item In your opinion, which group is more likely to comment on this post?\\
    \emph{[Progressives / Conservatives / Both groups are equally likely / Not sure]}
\end{enumerate}

\vspace{12pt}
\emph{Opinion about Comments in General}

\begin{enumerate}
    \item What are the main reasons you read or look at comments on social media? \emph{[open text]}
\end{enumerate}

\vspace{12pt}
\emph{Newsletter Sign-Up}

\emph{[Shown only to participants who answered Yes to the newsletter question.]}

You said that you like to sign up for the Color of Change newsletter.

\begin{enumerate}
    \item Please enter your email address below. Your email will only be shared with Color of Change, and only for the purpose of subscribing you to their newsletter.

    \emph{[text entry box]}
\end{enumerate}

\vspace{12pt}
\emph{AI Use}

\begin{enumerate}
    \item Did you use AI at all to help you fill out this survey?\\
    \emph{[Yes / No]}

    \item \emph{[If yes]} What question did you use AI to help you answer? \emph{[open text]}
\end{enumerate}

\vspace{12pt}
\emph{Final Feedback}

Thank you! We really appreciate you for participating in this research!

Please let us know if you have any other feedback.

\textbf{Make sure to click the next arrow to submit your survey responses.}

\emph{[text box]}

\newpage
\section{Codebook for NLP-Based Measures \label{app:Codebook}}

This appendix documents the NLP-based measures used in the paper. For all GPT-based classifications, we use GPT-4 with temperature set to 0. Toxicity is measured separately using the Google Perspective API.

\subsection{GPT-Based Conversation Measures}

\paragraph{Reciprocity, Justification, and Respect.}
These three measures are coded once for each direct comment.
The input consists of the original Facebook post, one direct comment, and all replies to that comment, so that the comment is assessed in the context of the exchange it started. GPT-4 is instructed to return a JSON object with three binary indicators:
\begin{itemize}
    \item \textbf{Reciprocity} equals 1 when participants engage the existing conversation and remain on topic.
    \item \textbf{Justification} equals 1 when participants who advocate a position provide reasons or arguments.
    \item \textbf{Respect} equals 1 when the exchange remains civil and free of threats, insults, humiliating language, or silencing expressions.
\end{itemize}

\subsection{GPT-Based Comment-Level Measures}

\paragraph{Political Stance (1--5).}
GPT-4 classifies each comment on a five-point ideological scale relative to the issue discussed in the ad:
\begin{enumerate}
    \item Strongly progressive or left-leaning
    \item Slightly or moderately progressive
    \item Centrist, unclear, or no explicit stance
    \item Slightly or moderately conservative
    \item Strongly conservative or right-leaning
\end{enumerate}

\paragraph{Political Stance (binary).}
A second GPT-4 measure collapses ideological stance into a binary variable:
\begin{itemize}
    \item 1 = progressive or left-leaning
    \item 2 = conservative or right-leaning
\end{itemize}
When the comment is neutral or ambiguous, GPT-4 is instructed to choose the closest side based on the language used.

\paragraph{Sentiment.}
GPT-4 classifies each comment from 1 to 5 according to its sentiment toward the main statement of the ad:
\begin{enumerate}
    \item Highly negative
    \item Somewhat negative
    \item Neutral or mixed
    \item Somewhat positive
    \item Very positive
\end{enumerate}

\paragraph{Informativeness.}
GPT-4 classifies each comment from 1 to 3 according to how much useful or fact-based information it provides:
\begin{enumerate}
    \item Not informative
    \item Somewhat informative
    \item Highly informative
\end{enumerate}

\paragraph{Offensiveness.}
GPT-4 classifies each comment from 1 to 3 according to the presence of offensive or insulting language:
\begin{enumerate}
    \item Not offensive
    \item Mildly offensive
    \item Highly offensive
\end{enumerate}

\paragraph{Agreement.}
GPT-4 classifies each comment from 1 to 3 according to whether it agrees with the message of the ad:
\begin{enumerate}
    \item Clear agreement
    \item Neutral or unclear
    \item Clear disagreement
\end{enumerate}

\subsection{API-Based Measure}

\paragraph{Toxicity Score.}
Toxicity is measured using the Google Perspective API, requesting the \texttt{TOXICITY} attribute for each comment. The API returns a continuous score between 0 and 1, where higher values indicate a higher likelihood that the comment is perceived as toxic. 

\end{document}